\documentclass[11pt]{article}

\usepackage[colorlinks, linkcolor=blue, citecolor=blue]{hyperref}            
\usepackage{color}
\usepackage{graphicx,subfigure,amsmath,amssymb,amsfonts,bm,epsfig,epsf,url,dsfont}
\usepackage{amsthm}
\usepackage{tikz}
\usepackage{times}
\usepackage{bbm}      
\usepackage{booktabs}
\usepackage{cases}
\usepackage{fullpage}
\usepackage[small,bf]{caption}
\usepackage{natbib}
\usepackage[top=1in,bottom=1in,left=1in,right=1in]{geometry}
\usepackage{fancybox}
\usepackage[bottom]{footmisc}

\newcommand{\bR}{\mathbb{R}}
\newcommand{\mD}{\mathcal{D}}
\newcommand{\mL}{\mathcal{L}}
\newcommand{\mO}{\mathcal{O}}
\newcommand{\mG}{\mathcal{G}}
\newcommand{\mN}{\mathcal{N}}

\newcommand{\mP}{\mathcal{P}}
\newcommand{\mQ}{\mathcal{Q}}
\newcommand{\mR}{\mathcal{R}}
\newcommand{\mW}{\mathcal{W}}
\newcommand{\mU}{\mathcal{U}}

\newcommand{\TV}{d_{\textnormal{TV}}}
\newcommand{\KL}{d_{\textnormal{KL}}}

\newcommand{\bP}{\mathbb{P}}
\newcommand{\bE}{\mathbb{E}}

\newcommand{\pr}[1]{\textsc{#1}}

\makeatletter
\newtheorem*{rep@theorem}{\rep@title}
\newcommand{\newreptheorem}[2]{%
\newenvironment{rep#1}[1]{%
 \def\rep@title{#2 \ref{##1}}%
 \begin{rep@theorem}}%
 {\end{rep@theorem}}}
\makeatother

\newreptheorem{theorem}{Theorem}
\newreptheorem{corollary}{Corollary}
\newreptheorem{definition}{Definition}

\newtheorem{question}{Question}[section]
\newtheorem{condition}{Condition}[section]
\newtheorem{theorem}{Theorem}[section]
\newtheorem{definition}[theorem]{Definition}

\newtheorem{fact}[theorem]{Fact}
\newtheorem{proposition}[theorem]{Proposition}
\newtheorem{corollary}[theorem]{Corollary}
\newtheorem{conjecture}[theorem]{Conjecture}
\newtheorem{lemma}[theorem]{Lemma}

\newcommand{\la}{\langle}
\newcommand{\ra}{\rangle}

\newcommand{\cS}{\mathcal{S}}
\newcommand{\cF}{\mathcal{F}}
\newcommand{\cD}{\mathcal{D}}
\newcommand{\cZ}{\mathcal{Z}}
\newcommand{\cX}{\mathcal{X}}

\newcommand{\Dh}{\widehat{D}}

\renewcommand{\P}{\mathbb{P}}
\newcommand{\E}{\mathbb{E}}

\newcommand{\Ot}{\tilde{O}}

\newcommand{\inv}{^{-1}}

\newcommand{\lr}{\mathsf{LR}}
\newcommand{\lrd}{\lr^{\leq D}}

\newcommand{\kpc}{{$k$-\textsc{pc}}}

\newenvironment{fminipage}%
  {\begin{Sbox}\begin{minipage}}%
  {\end{minipage}\end{Sbox}\fbox{\TheSbox}}

\newenvironment{algbox}[0]{\vskip 0.2in
\noindent 
\begin{fminipage}{6.3in}
}{
\end{fminipage}
\vskip 0.2in
}

\begin{document}

\title{Reducibility and Statistical-Computational Gaps \\ from Secret Leakage}

\author{Matthew Brennan
\thanks{Massachusetts Institute of Technology. Department of EECS. Email: \texttt{brennanm@mit.edu}.}
\and 
Guy Bresler
\thanks{Massachusetts Institute of Technology. Department of EECS. Email: \texttt{guy@mit.edu}.}
}
\date{\today}

\maketitle

\begin{abstract}
Inference problems with conjectured statistical-computational gaps are ubiquitous throughout modern statistics, computer science, statistical physics and discrete probability. While there has been success evidencing these gaps from the failure of restricted classes of algorithms, progress towards a more traditional reduction-based approach to computational complexity in statistical inference has been limited. These average-case problems are each tied to a different natural distribution, high-dimensional structure and conjecturally hard parameter regime, leaving reductions among them technically challenging. Despite a flurry of recent success in developing such techniques, existing reductions have largely been limited to inference problems with similar structure -- primarily mapping among problems representable as a sparse submatrix signal plus a noise matrix, which is similar to the common starting hardness assumption of planted clique (\pr{pc}).

The insight in this work is that a slight generalization of the planted clique conjecture -- secret leakage planted clique ($\pr{pc}_\rho$), wherein a small amount of information about the hidden clique is revealed -- gives rise to a variety of new average-case reduction techniques, yielding a web of reductions relating statistical problems with very different structure. Based on generalizations of the planted clique conjecture to specific forms of $\pr{pc}_\rho$, we deduce tight statistical-computational tradeoffs for a diverse range of problems including robust sparse mean estimation, mixtures of sparse linear regressions, robust sparse linear regression, tensor PCA, variants of dense $k$-block stochastic block models, negatively correlated sparse PCA, semirandom planted dense subgraph, detection in hidden partition models and a universality principle for learning sparse mixtures. This gives the first reduction-based evidence supporting a number of statistical-computational gaps observed in the literature \cite{li2017robust, balakrishnan2017computationally, diakonikolas2017statistical, chen2016statistical, hajek2015computational, brennan2018reducibility, fan2018curse, liu2018high, richard2014statistical, hopkins2015tensor, wein2019kikuchi, azizyan2013minimax, verzelen2017detection}.

We introduce a number of new average-case reduction techniques that also reveal novel connections to combinatorial designs based on the incidence geometry of $\mathbb{F}_r^t$ and to random matrix theory. In particular, we show a convergence result between Wishart and inverse Wishart matrices that may be of independent interest. The specific hardness conjectures for $\pr{pc}_\rho$ implying our statistical-computational gaps all are in correspondence with natural graph problems such as $k$-partite, bipartite and hypergraph variants of $\pr{pc}$. Hardness in a $k$-partite hypergraph variant of $\pr{pc}$ is the strongest of these conjectures and sufficient to establish all of our computational lower bounds. We also give evidence for our $\pr{pc}_\rho$ hardness conjectures from the failure of low-degree polynomials and statistical query algorithms. Our work raises a number of open problems and suggests that previous technical obstacles to average-case reductions may have arisen because planted clique is not the right starting point. An expanded set of hardness assumptions, such as $\pr{pc}_\rho$, may be a key first step towards a more complete theory of reductions among statistical problems.
\end{abstract}

\pagebreak

\tableofcontents

\pagebreak

\part{Summary of Results}
\label{part:intro}

\section{Introduction}
\label{sec:1-intro}

Computational complexity has become a central consideration in statistical inference as focus has shifted to high-dimensional structured problems. A primary aim of the field of mathematical statistics is to determine how much data is needed for various estimation tasks, and to analyze the performance of practical algorithms. For a century, the focus has been on \emph{information-theoretic} limits. However, the study of high-dimensional structured estimation problems over the last two decades has revealed that the much more relevant quantity -- the amount of data needed by \emph{computationally efficient} algorithms -- may be significantly higher than what is achievable without computational constraints. These \emph{statistical-computational gaps} were first observed to exist more than two decades ago \cite{valiant1984theory,servedio1999computational,decatur2000computational} but only recently have emerged as a trend ubiquitous in problems throughout modern statistics, computer science, statistical physics and discrete probability \cite{bottou2008tradeoffs,chandrasekaran2013computational,jordan2015machine}. Prominent examples arise in estimating sparse vectors from linear observations, estimating low-rank tensors, community detection, subgraph and matrix recovery problems, random constraint satisfiability, sparse principal component analysis and robust estimation.

Because statistical inference problems are formulated with probabilistic models on the observed data, there are natural barriers to basing their computational complexity as average-case problems on worst-case complexity assumptions such as $\text{P}\neq \text{NP}$ \cite{feigenbaum1993random, bogdanov2006worst,applebaum2008basing}. To cope with this complication, a number of different approaches have emerged to provide evidence for conjectured statistical-computational gaps. These can be roughly classified into two categories:
\begin{enumerate}
\item \textbf{Failure of Classes of Algorithms:} Showing that powerful classes of efficient algorithms, such as statistical query algorithms, the sum of squares (SOS) hierarchy and low-degree polynomials, fail up to the conjectured computational limit of the problem. 
\item \textbf{Average-Case Reductions:} The traditional complexity-theoretic approach showing the existence of polynomial-time reductions relating statistical-computational gaps in problems to one another.
\end{enumerate}
The line of research providing evidence for statistical-computational gaps through the failure of powerful classes of algorithms has seen a lot of progress in the past few years. A breakthrough work of \cite{barak2016nearly} developed the general technique of pseudocalibration for showing SOS lower bounds, and used this method to prove tight lower bounds for planted clique (\pr{pc}). In \cite{hopkinsThesis}, pseudocalibration motivated a general conjecture on the optimality of low-degree polynomials for hypothesis testing that has been used to provide evidence for a number of additional gaps \cite{hopkins2017efficient,kunisky2019notes,bandeira2019computational}. There have also been many other recent SOS lower bounds \cite{grigoriev2001linear,deshpande2015improved,ma2015sum,meka2015sum,kothari2017sum,hopkins2018integrality,raghavendra2018high,hopkins2017power,mohanty2019lifting}. Other classes of algorithms for which there has been progress in a similar vein include statistical query algorithms \cite{feldman2013statistical,feldman2015complexity,diakonikolas2017statistical,diakonikolas2019efficient}, classes of circuits \cite{razborov1997natural,rossman2008constant,rossman2014monotone}, local algorithms \cite{gamarnik2017limits,linial1992locality} and message-passing algorithms \cite{zdeborova2016statistical,lesieur2015mmse,lesieur2016phase,krzakala2007gibbs,ricci2018typology,bandeira2018notes}. Another line of work has aimed to provide evidence for computational limits by establishing properties of the energy landscape of solutions that are barriers to natural optimization-based approaches \cite{achlioptas2008algorithmic, gamarnik2017high,arous2017landscape, arous2018algorithmic,ros2019complex,chen2019suboptimality, gamarnik2019landscape}. 

While there has been success evidencing statistical-computational gaps from the failure of these classes of algorithms, progress towards a traditional reduction-based approach to computational complexity in statistical inference has been more limited. This is because reductions between average-case problems are more constrained and overall very different from reductions between worst-case problems. Average-case combinatorial problems have been studied in computer science since the 1970's \cite{karp1977probabilistic,kuvcera1977expected}. In the 1980's, Levin introduced his theory of average-case complexity \cite{levin1986average}, formalizing the notion of an average-case reduction and obtaining abstract completeness results. Since then, average-case complexity has been studied extensively in cryptography and complexity theory. A survey of this literature can be found in \cite{bogdanov2006average} and \cite{goldreich2011notes}. As discussed in \cite{Barak2017} and \cite{goldreich2011notes}, average-case reductions are notoriously delicate and there is a lack of available techniques. Although technically difficult to obtain, average-case reductions have a number of advantages over other approaches. Aside from the advantage of being future-proof against new classes of algorithms, showing that a problem of interest is hard by reducing from $\pr{pc}$ effectively \emph{subsumes} hardness for classes of algorithms known to fail on $\pr{pc}$ and thus gives stronger evidence for hardness. Reductions preserving gaps also directly relate phenomena across problems and reveal insights into how parameters, hidden structures and noise models correspond to one another.

Worst-case reductions are only concerned with transforming the \emph{hidden structure} in one problem to another. For example, a worst-case reduction from $\pr{3-sat}$ to $k\pr{-independent-set}$ needs to ensure that the hidden structure of a satisfiable $\pr{3-sat}$ formula is mapped to a graph with an independent set of size $k$, and that an unsatisfiable formula is not. Average-case reductions need to not only transform the structure in one problem to that of another, but also precisely map between the \emph{natural distributions} associated with problems. In the case of the example above, all classical worst-case reductions use gadgets that map random $\pr{3-sat}$ formulas to a very unnatural distribution on graphs. Average-case problems in statistical inference are also fundamentally \emph{parameterized}, with parameter regimes in which the problem is information-theoretically impossible, possible but conjecturally computationally hard and computationally easy. To establish the strongest possible lower bounds, reductions need to exactly fill out one of these three parameter regimes -- the one in which the problem is conjectured to be computationally hard. These subtleties that arise in devising average-case reductions will be discussed further in Section \ref{subsec:1-desiderata}.

Despite these challenges, there has been a flurry of recent success in developing techniques for average-case reductions among statistical problems. Since the seminal paper of \cite{berthet2013complexity} showing that a statistical-computational gap for a distributionally-robust formulation of sparse PCA follows from the \pr{pc} conjecture, there have been a number of average-case reductions among statistical problems. Reductions from $\pr{pc}$ have been used to show lower bounds for RIP certification \cite{wang2016average, koiran2014hidden}, biclustering detection and recovery \cite{ma2015computational, cai2015computational, caiwu2018, brennan2019universality}, planted dense subgraph \cite{hajek2015computational,brennan2019universality}, testing $k$-wise independence \cite{alon2007testing}, matrix completion \cite{chen2015incoherence} and sparse PCA \cite{berthet2013optimal, berthet2013complexity, wang2016statistical, gao2017sparse, brennan2019optimal}. Several reduction techniques were introduced in \cite{brennan2018reducibility}, providing the first web of average-case reductions among a number of problems involving sparsity. More detailed surveys of these prior average-case reductions from $\pr{pc}$ can be found in the introduction section of \cite{brennan2018reducibility} and in \cite{wu2018statistical}. There also have been a number of average-case reductions in the literature starting with different assumptions than the \pr{pc} conjecture. Hardness conjectures for random CSPs have been used to show hardness in improper learning complexity \cite{daniely2014average}, learning DNFs \cite{daniely2016complexityDNF} and hardness of approximation \cite{feige2002relations}. Recent reductions also map from a 3-uniform hypergraph variant of the \pr{pc} conjecture to SVD for random 3-tensors \cite{zhang2017tensor} and between learning two-layer neural networks and tensor decomposition \cite{mondelli2018connection}.

A common criticism to the reduction-based approach to computational complexity in statistical inference is that, while existing reductions have introduced nontrivial techniques for mapping precisely between different natural distributions, they are not yet capable of transforming between problems with dissimilar \emph{high-dimensional structures}. In particular, the vast majority of the reductions referenced above map among problems representable as a \emph{sparse submatrix signal plus a noise matrix}, which is similar to the common starting hardness assumption \pr{pc}. Such a barrier would be fatal to a satisfying reduction-based theory of statistical-computational gaps, as the zoo of statistical problems with gaps contains a broad range of very different high-dimensional structures. This leads directly to the following central question that we aim to address in this work.

\begin{question}
Can statistical-computational gaps in problems with different high-dimensional structures be related to one another through average-case reductions?
\end{question}

\subsection{Overview}
\label{subsec:1-overview}

\begin{figure}[t!]
	\begin{center}
	\includegraphics[width=\textwidth]{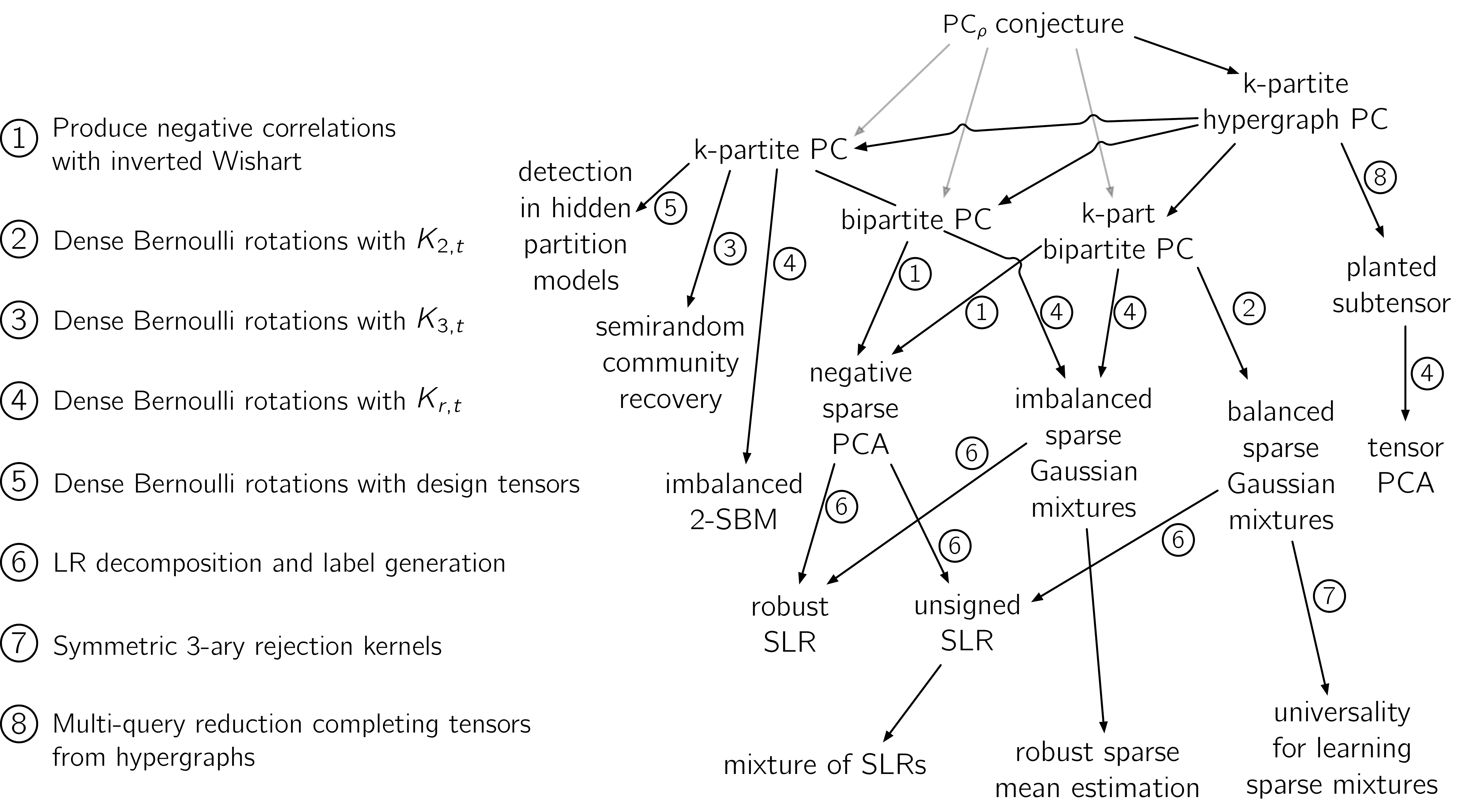}
	\end{center}
	\vspace{-4mm}
	\caption{The web of reductions carried out in this paper. An edge indicates existence of a reduction transferring computational hardness from the tail to the head. Edges are labeled with associated reduction techniques and unlabelled edges correspond to simple reductions or specializing a problem to a particular case.}
	\label{fig:web}
\end{figure}

The main objective of this paper is to provide the first evidence that relating differently structured statistical problems through reductions is possible. We show that mild generalizations of the \pr{pc} conjecture to $k$-partite and bipartite variants of \pr{pc} are naturally suited to a number of new average-case reduction techniques. These techniques map to problems breaking out of the sparse submatrix plus noise structure that seemed to constrain prior reductions. They thus show that revealing a tiny amount of information about the hidden clique vertices substantially increases the reach of the reductions approach, providing the first web of reductions among statistical problems with significantly different structure. Our techniques also yield reductions beginning from hypergraph variants of \pr{pc} which, along with the $k$-partite and bipartite variants mentioned above, can be unified under a single assumption that we introduce -- the secret leakage planted clique ($\pr{pc}_\rho$) conjecture. This conjecture makes a precise prediction of what information about the hidden clique can be revealed while $\pr{pc}$ remains hard.

A summary of our web of average-case reductions is shown in Figure \ref{fig:web}. Our reductions yield tight statistical-computational gaps for a range of differently structured problems, including robust sparse mean estimation, variants of dense stochastic block models, detection in hidden partition models, semirandom planted dense subgraph, negatively correlated sparse PCA, mixtures of sparse linear regressions, robust sparse linear regression, tensor PCA and a universality principle for learning sparse mixtures. This gives the first reduction-based evidence supporting a number of gaps observed in the literature \cite{li2017robust, balakrishnan2017computationally, diakonikolas2017statistical, chen2016statistical, hajek2015computational, brennan2018reducibility, fan2018curse, liu2018high, richard2014statistical, hopkins2015tensor, wein2019kikuchi, azizyan2013minimax, verzelen2017detection}. In particular, there are no known reductions deducing these gaps from the ordinary $\pr{pc}$ conjecture. Similar to \cite{brennan2018reducibility}, several average-case problems emerge as natural intermediates in our reductions, such as negative sparse PCA and imbalanced sparse Gaussian mixtures. The specific instantiations of the $\pr{pc}_\rho$ conjecture needed to obtain these lower bounds correspond to natural $k$-partite, bipartite and hypergraph variants of $\pr{pc}$. Among these hardness assumptions, we show that hardness in a $k$-partite hypergraph variant of $\pr{pc}$ ($k\pr{-hpc}^s$) is the strongest and sufficient to establish all of our computational lower bounds. We also give evidence for our hardness conjectures from the failure of low-degree polynomials and statistical query algorithms.

Our results suggest that \pr{pc} may not be the right starting point for average-case reductions among statistical problems. However, surprisingly mild generalizations of $\pr{pc}$ are all that are needed to break beyond the structural constraints of previous reductions. Generalizing to either $\pr{pc}_\rho$ or $k\pr{-hpc}^s$ unifies all of our reductions under a single hardness assumption, now capturing reductions to a range of dissimilarly structured problems including supervised learning tasks and problems over tensors. This suggests $\pr{pc}_\rho$ and $k\pr{-hpc}^s$ are both much more powerful candidate starting points than $\pr{pc}$ and, more generally, that these may be a key first step towards a more complete theory of reductions among statistical problems. Although we often will focus on providing evidence for statistical-computational gaps, we emphasize that our main contribution is more general -- our reductions give a new set of techniques for relating differently structured statistical problems that seem likely to have applications beyond the problems we consider here.

The rest of the paper is structured as follows. The next section gives general background on average-case reductions and several criteria that they must meet in order to show strong computational lower bounds for statistical problems. In Section \ref{sec:1-PC}, we introduce the $\pr{pc}_\rho$ conjecture and the specific instantiations of this conjecture that imply our computational lower bounds, such as $k\pr{-hpc}^s$. In Section \ref{sec:1-problems} we formally introduce the problems in Figure \ref{fig:web} and state our main theorems. In Section \ref{sec:1-techniques}, we describe the key ideas underlying our techniques and we conclude Part \ref{part:intro} by discussing a number of questions arising from these techniques in Section \ref{sec:1-open-problems}. Parts \ref{part:reductions} and \ref{part:lower-bounds} are devoted to formally introducing our reduction techniques and applying them, respectively. Part \ref{part:reductions} begins with Section \ref{sec:2-preliminaries}, which introduces reductions in total variation and the corresponding hypothesis testing formulation for each problem we consider that it will suffice to reduce to. In the rest of Part \ref{part:reductions}, we introduce our main reduction techniques and give several initial applications of these techniques to reduce to a subset of the problems that we consider. Part \ref{part:lower-bounds} begins with a further discussion of the $\pr{pc}_\rho$ conjecture, where we show that $k\pr{-hpc}^s$ is our strongest assumption and provide evidence for the $\pr{pc}_\rho$ conjecture from the failure of low-degree tests and the statistical query model. The remainder of Part \ref{part:lower-bounds} is devoted to our other reductions and deducing the computational lower bounds in our main theorems from Section \ref{sec:1-problems}. At the end of Part \ref{part:lower-bounds}, we discuss the implications of our reductions to estimation and recovery formulations of the problems that we consider. Reading Part \ref{part:intro}, Section \ref{sec:2-secret-leakage} and the pseudocode for our reductions gives an accurate summary of the theorems and ideas in this work. We note that a preliminary draft of this work containing a small subset of our results appeared in \cite{brennan2019average}.

\subsection{Desiderata for Average-Case Reductions}
\label{subsec:1-desiderata}

As discussed in the previous section, average-case reductions are delicate and more constrained than their worst-case counterparts. In designing average-case reductions between problems in statistical inference, the essential challenge is to reduce to instances that are \textit{hard up to the conjectured computational barrier}, without destroying the \textit{naturalness} of the distribution over instances. Dissecting this objective further yields four general criteria for a reduction between the problems $\mP$ and $\mP'$ to be deemed to show strong computational lower bounds for $\mP'$. These objectives are to varying degrees at odds with one another, which is what makes devising reductions a challenging task. To illustrate these concepts, our running example will be our reduction from $\pr{pc}_\rho$ to robust sparse linear regression (SLR). Some parts of this discussion are slightly simplified for clarity. The following are our four criteria.
\begin{enumerate}
\item \textbf{Aesthetics:} If $\mP$ and $\mP'$ each have a specific canonical distribution then a reduction must faithfully map these distributions to one another. In our example, this corresponds to mapping the independent $0$-$1$ edge indicators in a random graph to noisy Gaussian samples of the form $y = \langle \beta, X \rangle + \mN(0, 1)$ with $X \sim \mN(0, I_d)$ and where an $\epsilon$-fraction are corrupted.
\item \textbf{Mapping Between Different Structures:} A reduction must simultaneously map all possible latent signals of $\mP$ to that of $\mP'$. In our example, this corresponds to mapping each possible clique position in $\pr{pc}_\rho$ to a specific mixture over the hidden vector $\beta$. A reduction in this case would also need to map between possibly very differently structured data, e.g., in robust SLR the dependence of $(X, y)$ on $\beta$ is intricate and the $\epsilon$-fraction of corrupted samples also produces latent structure across samples. These are both very different than the planted signal plus noise form of the clique in $\pr{pc}_\rho$.
\item \textbf{Tightness to Algorithms:} A reduction showing computational lower bounds that are tight against what efficient algorithms can achieve needs to map the conjectured computational limits of $\mP$ to those of $\mP'$. In our example, $\pr{pc}_\rho$ in general has a conjectured limit depending on $\rho$, which may for instance be at $K = o(\sqrt{N})$ when the clique is of size $K$ in a graph with $N$ vertices. In contrast, robust SLR has the conjectured limit at $n = \tilde{o}(k^2 \epsilon^2/\tau^4)$, where $\tau$ is the $\ell_2$ error to which we wish to estimate $\beta$, $k$ is the sparsity of $\beta$ and $n$ is the number of samples. 
\item \textbf{Strong Lower Bounds for Parameterized Problems:} In order to show that a certain constraint $\mathcal{C}$ \emph{defines} the computational limit of $\mP'$ through this reduction, we need the reduction to fill out the possible parameter sequences within $\mathcal{C}$. For example, to show that $n = \tilde{o}(k^2 \epsilon^2/\tau^4)$ truly captures the correct dependence in our computational lower bound for robust SLR, it does not suffice to produce a single sequence of points $(n, k, d, \tau, \epsilon)$ for which this is true, or even a one parameter curve. There are four parameters in the conjectured limit and a reduction showing that this is the correct dependence needs to fill out any possible combination of growth rates in these parameters permitted by $n = \tilde{o}(k^2 \epsilon^2/\tau^4)$. The fact that the initial problem $\mP$ has a conjectured limit depending on only two parameters can make achieving this criterion challenging.
\end{enumerate}
We remark that the third criterion requires that reductions are \emph{information preserving} in the sense that they do not degrade the underlying level signal used by optimal efficient algorithms. This necessitates that the amount of additional randomness introduced in reductions to achieve aesthetic requirements is negligible. The fourth criterion arises from the fact that statistical problems are generally described by a tuple of parameters and are therefore actually an entire family of problems. A full characterization of the computational feasibility of a problem therefore requires addressing all possible scalings of the parameters. 

All of the reductions carried out in this paper satisfy all four desiderata. Several of the initial reductions from $\pr{pc}$ in the literature met most but not all of these criteria. For example, the reductions in \cite{berthet2013complexity, wang2016statistical} to sparse PCA map to a distribution in a distributionally robust formulation of the problem as opposed to the canonical Gaussian formulation in the spiked covariance model. Similarly \cite{cai2015computational} reduces to a distributionally robust formulation of submatrix localization. The reduction in \cite{gao2017sparse} only shows tight computational lower bounds for sparse PCA at a particular point in the parameter space when $\theta = \tilde{\Theta}(1)$ and $n = \tilde{\Theta}(k^2)$. However, a number of reductions in the literature have successfully met all of these four criteria \cite{ma2015computational, hajek2015computational, zhang2017tensor, brennan2018reducibility, brennan2019optimal, brennan2019universality}. 

We remark that it can be much easier to only satisfy some of these desiderata -- in particular, many natural reduction ideas meet a subset of these four criteria but fail to show nontrivial computational lower bounds. For instance, it is often straightforward to construct a reduction that degrades the level of signal. The simple reduction that begins with $\pr{pc}$ and randomly subsamples edges with probability $n^{-\alpha}$ yields an instance of planted dense subgraph with the correct distributional aesthetics. However, this reduction fails to be tight to algorithms and furthermore fails to show any meaningful tradeoff between the size of the planted dense subgraph and the sparsity of the graph.

Another natural reduction to robust sparse mean estimation first maps from $\pr{pc}$ to Gaussian biclustering using one of the reductions in \cite{ma2015computational, brennan2018reducibility, brennan2019universality}, computes the sum $v$ of all of the rows of this matrix, then uses Gaussian cloning as in \cite{brennan2018reducibility} to produce $n$ weak copies of $v$ and finally outputs these copies with an an $\epsilon$-fraction corrupted. This reduction can be verified to produce a valid instance of robust sparse mean estimation in its canonical Gaussian formulation, but fails to show any nontrivial hardness above its information-theoretic limit. Conceptually, this is because the reduction is generating the $\epsilon$-fraction of the corruptions itself. On applying a robust sparse mean estimation blackbox to solve $\pr{pc}$, the reduction could just as easily have revealed which samples it corrupted. This would allow the blackbox to only have to solve sparse mean estimation, which has no statistical-computational gap. In general, a reduction showing tight computational lower bounds cannot generate a non-negligible amount randomness that produces the hardness of the target problem. Instead, this $\epsilon$-fraction must come from the hidden clique in the input $\pr{pc}$ instance. In Section \ref{subsec:1-tech-encoding}, we discuss how our reductions obliviously encode cliques into the hidden structures in the problems we consider.

We also remark that many problems that appear to be similar from the perspective of designing efficient algorithms can be quite different to reduce to. This arises from differences in their underlying stochastic models that efficient algorithms do not have to make use of. For example, although ordinary sparse PCA and sparse PCA with a negative spike can be solved by the same efficient algorithms, the former has a signal plus noise decomposition while the latter does not and has negatively correlated as opposed to positively correlated planted entries. We will see that these subtle differences are significant in designing reductions.

\section{Planted Clique and Secret Leakage}
\label{sec:1-PC}

In this section, we introduce planted clique and our generalization of the planted clique conjecture. In the \emph{planted clique problem} ($\pr{pc}$), the task is to find the vertex set of a $k$-clique planted uniformly at random in an $n$-vertex Erd\H{o}s-R\'{e}nyi graph $G$. Planted clique can equivalently be formulated as a testing problem $\pr{pc}(n, k, 1/2)$ \cite{alon2007testing} between the two hypotheses
$$H_0: G \sim \mG(n, 1/2) \quad \text{and} \quad H_1: G \sim \mG(n, k, 1/2)$$
where $\mG(n, 1/2)$ denotes the $n$-vertex Erd\H{o}s-R\'{e}nyi graph with edge density $1/2$ and $\mG(n, k, 1/2)$ the distribution resulting from planting a $k$-clique uniformly at random in $\mG(n, 1/2)$. This problem can be solved in quasipolynomial time by searching through all vertex subsets of size $(2 + \epsilon) \log_2 n$ if $k > (2 + \epsilon) \log_2 n$. The \emph{Planted Clique Conjecture} is that there is no polynomial time algorithm solving $\pr{pc}(n, k, 1/2)$ if $k = o(\sqrt{n})$.

There is a plethora of evidence in the literature for the $\pr{pc}$ conjecture. Spectral algorithms, approximate message passing, semidefinite programming, nuclear norm minimization and several other polynomial-time combinatorial approaches all appear to fail to solve $\pr{pc}$ exactly when $k = o(\sqrt{n})$ \cite{alon1998finding, feige2000finding, mcsherry2001spectral, feige2010finding, ames2011nuclear, dekel2014finding, deshpande2015finding, chen2016statistical}. Lower bounds against low-degree sum of squares relaxations \cite{barak2016nearly} and statistical query algorithms \cite{feldman2013statistical} have also been shown up to $k = o(\sqrt{n})$.

\paragraph{Secret Leakage $\pr{pc}$.} We consider a slight generalization of the planted clique problem, where the input graph $G$ comes with some information about the vertex set of the planted clique. This corresponds to the vertices in the $k$-clique being chosen from some distribution $\rho$ other than the uniform distribution of $k$-subsets of $[n]$, as formalized in the following definition.

\begin{definition}[Secret Leakage $\pr{pc}_\rho$]
Given a distribution $\rho$ on $k$-subsets of $[n]$, let $\mG_\rho(n, k, 1/2)$ be the distribution on $n$-vertex graphs sampled by first sampling $G \sim \mG(n, 1/2)$ and $S \sim \rho$ independently and then planting a $k$-clique on the vertex set $S$ in $G$. Let $\pr{pc}_\rho(n, k, 1/2)$ denote the resulting hypothesis testing problem between $H_0: G \sim \mG(n, 1/2)$ and $H_1: G \sim \mG_\rho(n, k, 1/2)$.
\end{definition}

All of the $\rho$ that we will consider will be uniform over the $k$-subsets that satisfy some constraint. In the cryptography literature, modifying a problem such as $\pr{pc}$ with a promise of this form is referred to as information leakage about the secret. There is a large body of work on leakage-resilient cryptography recently surveyed in \cite{kalai2019survey}. The hardness of the Learning with Errors (LWE) problem has been shown to be unconditionally robust to leakage \cite{dodis2010public, goldwasser10}, and it is left as an interesting open problem to show that a similar statement holds true for \textsc{pc}.

Both $\pr{pc}$ and $\pr{pc}_\rho$ fall under the class of general parameter recovery problems where the task is to find $P_S$ generating the observed graph from a family of distributions $\{ P_S \}$. In the case of $\pr{pc}$, $P_S$ denotes the distribution $\mG(n, k, 1/2)$ conditioned on the $k$-clique being planted on $S$. Observe that the conditional distributions $\{ P_S \}$ are the same in $\pr{pc}$ and $\pr{pc}_\rho$. Secret leakage can be viewed as placing a prior on the parameter $S$ of interest, rather than changing the main average-case part of the problem -- the family $\{ P_S \}$. When $\rho$ is uniform over a family of $k$-subsets, secret leakage corresponds to imposing a worst-case constraint on $S$. In particular, consider the maximum likelihood estimator (MLE) for a general parameter recovery problem given by
$$\hat{S} = \arg \max_{S \in \text{supp}(\rho)} P_S(G)$$
As $\rho$ varies, only the search space of the MLE changes while the objective remains the same. We make the following precise conjecture of the hardness of $\pr{pc}_\rho(n, k, 1/2)$ for the distributions $\rho$ we consider. Given a distribution $\rho$, let $p_{\rho}(s) = \bP_{S, S' \sim \rho^{\otimes 2}}[|S \cap S'| = s]$ be the probability mass function of the size of the intersection of two independent random sets $S$ and $S'$ drawn from $\rho$.

\begin{conjecture}[Secret Leakage Planted Clique ($\pr{pc}_\rho$) Conjecture] \label{conj:sl-conj}
Let $\rho$ be one of the distributions on $k$-subsets of $[n]$ given below in Conjecture \ref{conj:hard-conj}. Suppose that there is some $p_0 = o_n(1)$ and constant $\delta > 0$ such that $p_{\rho}(s)$ satisfies the tail bounds
$$p_{\rho}(s) \le p_0 \cdot \left\{ \begin{array}{ll} 2^{-s^2} &\textnormal{if } 1 \le s^2 < d \\ s^{-2d-4} &\textnormal{if } s^2 \ge d \end{array} \right.$$
for any parameter $d = O_n((\log n)^{1 + \delta})$. Then there is no polynomial time algorithm solving $\pr{pc}_\rho(n, k, 1/2)$.
\end{conjecture}

While this conjecture is only stated for the specific $\rho$ corresponding to the hardness assumptions used in our reductions, we believe it should hold for a wide class of $\rho$ with sufficient symmetry. The motivation for the decay condition on $p_\rho$ in the $\pr{pc}_\rho$ conjecture is from low-degree polynomials, which we show in Section \ref{subsec:2-low-degree} fail to solve $\pr{pc}_\rho$ subject to this condition. The \textit{low-degree conjecture} -- that low-degree polynomials predict the computational barriers for a broad class of inference problems -- has been shown to match conjectured statistical-computational gaps in a number of problems \cite{hopkins2017efficient, hopkinsThesis, kunisky2019notes, bandeira2019computational}. We discuss this conjecture, the technical conditions arising in its formalizations and how these relate to $\pr{pc}_\rho$ in Section \ref{subsec:2-low-degree}. Specifically, we discuss the importance of symmetry and the requirement on $d$ in generalizing Conjecture \ref{conj:sl-conj} to further $\rho$. In contrast to low-degree polynomials, because the SQ model only concerns problems with a notion of samples, it seems ill-suited to accurately predict the computational barriers in $\pr{pc}_\rho$ for every $\rho$. However, in Section \ref{subsec:2-sq}, we show SQ lower bounds supporting the $\pr{pc}_\rho$ conjecture for specific $\rho$ related to our hardness assumptions. We also remark that the distribution $p_{\rho}$ is an overlap distribution, which has been linked to conjectured statistical-computational gaps using techniques from statistical physics \cite{zdeborova2016statistical}.

\paragraph{Hardness Conjectures for Specific $\rho$.} In our reductions, we will only need the $\pr{pc}_\rho$ conjecture for specific $\rho$, all of which are simple and correspond to their own hardness conjectures in natural mild variants of $\pr{pc}$. Secret leakage can be viewed as a way to conceptually unify these different assumptions. These $\rho$ all seem to avoid revealing enough information about $S$ to give rise to new polynomial time algorithms to solve $\pr{pc}_{\rho}$. In particular, spectral algorithms consistently seem to match our conjectured computational limits for $\pr{pc}_\rho$ for the different $\rho$ we consider.

We now introduce these specific hardness assumptions and briefly outline how each can be produced from an instance of $\pr{pc}_\rho$. This is more formally discussed in Section \ref{subsec:2-sl-verifying}. Let $\mG_{B}(m, n, 1/2)$ denote the distribution on bipartite graphs $G$ with parts of size $m$ and $n$ wherein each edge between the two parts is included independently with probability $1/2$.
\begin{itemize}
\item \textbf{$k$-partite \pr{pc}:} Suppose that $k$ divides $n$ and let $E$ be a partition of $[n]$ into $k$ parts of size $n/k$. Let $k\pr{-pc}_E(n, k, 1/2)$ be $\pr{pc}_\rho(n, k, 1/2)$ where $\rho$ is uniformly distributed over all $k$-sets intersecting each part of $E$ in exactly one element.
\item \textbf{bipartite \pr{pc}:} Let $\pr{bpc}(m, n, k_m, k_n, 1/2)$ be the problem of testing between $H_0 : G \sim \mG_B(m, n, 1/2)$ and $H_1$ under which $G$ is formed by planting a complete bipartite graph with $k_m$ and $k_n$ vertices in the two parts, respectively, in a graph sampled from $\mG_B(m, n, 1/2)$. This problem can be realized as a bipartite subgraph of an instance of $\pr{pc}_\rho$.
\item \textbf{$k$-part bipartite \pr{pc}:} Suppose that $k_n$ divides $n$ and let $E$ be a partition of $[n]$ into $k_n$ parts of size $n/k_n$. Let $k\pr{-bpc}_E(m, n, k_m, k_n, 1/2)$ be $\pr{bpc}$ where the $k_n$ vertices in the part of size $n$ are uniform over all $k_n$-sets intersecting each part of $E$ in exactly one element, as in the definition of $k\pr{-pc}_E$. As with $\pr{bpc}$, this problem can be realized as a bipartite subgraph of an instance of $\pr{pc}_\rho$, now with additional constraints on $\rho$ to enforce the $k$-part restriction.
\item \textbf{$k$-partite hypergraph \pr{pc}:} Let $k, n$ and $E$ be as in the definition of $k\pr{-pc}$. Let $k\pr{-hpc}^s_E(n, k, 1/2)$ where $s \ge 3$ be the problem of testing between $H_0$, under which $G$ is an $s$-uniform Erd\H{o}s-R\'{e}nyi hypergraph where each hyperedge is included independently with probability $1/2$, and $H_1$, under which $G$ is first sampled from $H_0$ and then a $k$-clique with one vertex chosen uniformly at random from each part of $E$ is planted in $G$. This problem has a simple correspondence with $\pr{pc}_\rho$: there is a specific $\rho$ that corresponds to unfolding the adjacency tensor of this hypergraph problem into a matrix. We will show more formally how to produce $k\pr{-hpc}^s_E(n, k, 1/2)$ from $\pr{pc}_\rho$ in Section \ref{subsec:2-sl-verifying}.
\end{itemize}

Since $E$ is always revealed in these problems, it can without loss of generality be taken to be any partition of $[n]$ into $k$ equally-sized parts. Consequently, we will often simplify notation by dropping the subscript $E$ from the above notation. We conjecture the following computational barriers for these graph problems, each of which matches the decay rate condition on of $p_{\rho}(s)$ in $\pr{pc}_\rho$ conjecture, as we will show in Section \ref{subsec:2-sl-verifying}.

\begin{conjecture}[Specific Hardness Assumptions] \label{conj:hard-conj}
Suppose that $m$ and $n$ are polynomial in one another. Then there is no $\textnormal{poly}(n)$ time algorithm solving the following problems:
\begin{enumerate}
\item $k\pr{-pc}(n, k, 1/2)$ when $k = o(\sqrt{n})$;
\item $\pr{bpc}(m, n, k_m, k_n, 1/2)$ when $k_n = o(\sqrt{n})$ and $k_m = o(\sqrt{m})$;
\item $k\pr{-bpc}(m, n, k_m, k_n, 1/2)$ when $k_n = o(\sqrt{n})$ and $k_m = o(\sqrt{m})$; and
\item $k\pr{-hpc}^s(n, k, 1/2)$ for $s \ge 3$ when $k = o(\sqrt{n})$.
\end{enumerate}
\end{conjecture}

From an entropy viewpoint, the $k$-partite assumption common to these variants of $\pr{pc}_\rho$ only reveals a very small amount of information about the location of the clique. In particular, both the uniform distribution over $k$-subsets and over $k$-subsets respecting a given partition $E$ have $(1 + o(1))k \log_2 n$ bits of entropy. We also remark that the $\pr{pc}_\rho$ conjecture, as stated, implies the thresholds in the conjecture above up to arbitrarily small polynomial factors i.e. where the thresholds are $k = O(n^{1/2 - \epsilon})$, $k_n = O(n^{1/2 - \epsilon})$ and $k_m = O(m^{1/2 - \epsilon})$ for arbitrarily small $\epsilon > 0$. As we will discuss in \ref{subsec:2-low-degree}, the low-degree conjecture also supports the stronger thresholds in Conjecture \ref{conj:hard-conj}. We also note that our reductions continue to show tight hardness up to arbitrarily small polynomial factors even under these weaker assumptions. As mentioned in Section \ref{subsec:1-overview}, our hardness assumption for $k\pr{-hpc}^s$ is the strongest of those in Conjecture \ref{conj:hard-conj}. Specifically, in Section \ref{subsec:2-sl-verifying} we give simple reductions showing that (4) in Conjecture \ref{conj:hard-conj} implies (1), (2) and (3).

We remark that the discussion in this section also applies to planted dense subgraph ($\pr{pds}$) problems. In the $\pr{pds}$ variant of a $\pr{pc}$ problem, instead of planting a $k$-clique in $\mG(n, 1/2)$, a dense subgraph $\mG(k, p)$ is planted in $\mG(n, q)$ where $p > q$. We conjecture that all of the hardness assumptions remain true for $\pr{pds}$ with constant edge densities $0 < q < p \le 1$. Note that $\pr{pc}$ is an instance of $\pr{pds}$ with $p = 1$ and $q = 1/2$. All of the reductions beginning with $\pr{pc}_\rho$ in this work will also yield reductions beginning from secret leakage planted dense subgraph problems $\pr{pds}_\rho$. In particular, they will continue to apply with a small loss in the amount of signal when $q = 1/2$ and $p = 1/2 + n^{-\epsilon}$ for a small constant $\epsilon > 0$. As discussed in \cite{brennan2019optimal}, $\pr{pds}$ conjecturally has no quasipolynomial time algorithms in this regime and thus our reductions would transfer lower bounds above polynomial time. In this parameter regime, the barriers of $\pr{pds}$ also appear to be similar to those of detection in the sparsely spiked Wigner model, which also conjecturally has no quasipolynomial time algorithms \cite{hopkins2017power}. Throughout this work, we will denote the $\pr{pds}$ variants of the problems introduced above by $k\pr{-pds}(n, k, p, q)$, $\pr{bpds}(m, n, k_m, k_n, p, q)$, $k\pr{-bpds}(m, n, k_m, k_n, p, q)$ and $k\pr{-hpds}^s(n, k, p, q)$.

\section{Problems and Statistical-Computational Gaps}
\label{sec:1-problems}

In this section, we introduce the problems we consider and give informal statements of our main theorems, each of which is a tight computational lower bound implied by a conjecture in the previous section. These statistical-computational gaps follow from a variety of different average-case reduction techniques that are outlined in the next section and will be the focus in the rest of this work. Before stating our main results, we clarify precisely what we mean by \emph{solving} and showing a \emph{computational lower bound} for a problem. All of the computational lower bounds in this section are implied by one of the assumptions in Conjecture \ref{conj:hard-conj}. As mentioned previously, they also follow from $\pr{pds}$ variants of these assumptions or only from the hardness of $k\pr{-hpc}^s$, which is the strongest assumption.

\paragraph{Statistical Problems and Algorithms.} Every problem $\mP(n, a_1, a_2, \dots, a_t)$ we consider is parameterized by a natural parameter $n$ and has several other parameters $a_1(n), a_2(n), \dots, a_t(n)$, which will typically be implicit functions of $n$. If $\mP$ is a hypothesis testing problem with observation $X$ and hypotheses $H_0$ and $H_1$, an algorithm $\mathcal{A}$ is deemed to solve $\mP$ subject to the constraints $\mathcal{C}$ if it has asymptotic Type I$+$II error bounded away from $1$ when $(n, a_1, a_2, \dots, a_t) \in \mathcal{C}$ i.e. if $\bP_{H_0}\left[ \mathcal{A}(X) = H_1 \right] + \bP_{H_1}\left[ \mathcal{A}(X) = H_0 \right] = 1 - \Omega_n(1)$. Furthermore, we say that there is no algorithm solving $\mP$ in polynomial time under the constraints $\mathcal{C}$ if for any sequence of parameters $\{(n, a_1, a_2, \dots, a_t)\}_{n = 1}^\infty \subseteq \mathcal{C}$, there is no polynomial time algorithm solving $\mP(n, a_1, a_2, \dots, a_t)$ with Type I$+$II error bounded away from $1$ as $n \to \infty$. If $\mP$ is an estimation problem with a parameter $\theta$ of interest and loss $\ell$, then $\mathcal{A}$ solves $\mP$ subject to the constraints $\mathcal{C}$ if $\ell(\mathcal{A}(X), \theta) \le \epsilon$ is true with probability $1 - o_n(1)$ when $(n, a_1, a_2, \dots, a_t, \epsilon) \in \mathcal{C}$, where $\epsilon = \epsilon(n)$ is a function of $n$.

\paragraph{Computational Lower Bounds.} We say there is a computational lower bound for $\mathcal{P}$ subject to the constraint $\mathcal{C}$ if for any sequence of parameters $\{(n, a_1(n), a_2(n), \dots, a_t(n))\}_{n = 1}^\infty \subseteq \mathcal{C}$ there is another sequence given by $\{(n_i, a'_1(n_i), a'_2(n_i), \dots, a'_t(n_i))\}_{i = 1}^\infty \subseteq \mathcal{C}$ such that $\mP(n_i, a'_1(n_i), a'_2(n_i), \dots, a'_t(n_i))$ cannot be solved in $\text{poly}(n_i)$ time and $\lim_{i \to \infty} \log a_k'(n_i)/\log a_k(n_i) = 1$. In other words, there is a lower bound at $\mathcal{C}$ if, for any sequence $s$ in $\mathcal{C}$, there is another sequence of parameters that cannot be solved in polynomial time and whose growth matches the growth of a subsequence of $s$. Thus all of our computational lower bounds are \emph{strong lower bounds} in the sense that rather than show that a single sequence of parameters is hard, we show that parameter sequences filling out \emph{all possible growth rates} in $\mathcal{C}$ are hard.

The constraints $\mathcal{C}$ will typically take the form of a system of asymptotic inequalities. Furthermore, each of our computational lower bounds for estimation problems will be established through a reduction to a hypothesis testing problem which then implies the desired lower bound. The exact formulations for these intermediate hypothesis testing problems can be found in Section \ref{subsec:2-formulations} and how they also imply lower bounds for estimation and recovery variants of our problems is discussed in Section \ref{subsec:2-estimation}. Throughout this work, we will use the terms detection and hypothesis testing interchangeably. We say that two parameters $a$ and $b$ are polynomial in one another if there is a constant $C > 0$ such that $a^{1/C} \le b \le a^C$ as $a \to \infty$. Throughout the paper, we adopt the standard asymptotic notation $O(\cdot), \Omega(\cdot), o(\cdot), \omega(\cdot)$ and $\Theta(\cdot)$. We let $\tilde{O}(\cdot)$ and analogous variants denote these relations up to $\text{polylog}(n)$ factors. Here, $n$ is the natural parameter of the problem under consideration and will typically be clear from context. We remark that the argument of $\tilde{O}(\cdot)$ will often be polynomially large or small in $n$, in which case our notation recovers the typical definition of $\tilde{O}(\cdot)$. Furthermore, all of these definitions also apply to the discussion in the previous section.

\paragraph{Canonical Simplest Average-Case Formulations.} All of our reductions are to the canonical simplest average-case formulations of the problems we consider. For example, all $k$-sparse unit vectors in our lower bounds are binary and in $\{0, 1/\sqrt{k} \}^d$, and the rank-1 component in our lower bound for tensor PCA is sampled from a Rademacher prior. Our reductions are all also to the canonical simple vs. simple hypothesis testing formulation for each of our problems and, as discussed in \cite{brennan2018reducibility}, this yields strong computational lower bounds, is often technically more difficulty and crucially allows reductions to naturally be composed with one another.

\subsection{Robust Sparse Mean Estimation}
\label{subsec:1-problems-rsme}

The study of robust estimation began with Huber's contamination model \cite{huber1992robust, huber1965robust} and observations of Tukey \cite{tukey1975mathematics}. Classical robust estimators have typically either been computationally intractable or heuristic \cite{huber2011robust, tukey1975mathematics, yatracos1985rates}. Recent breakthrough works \cite{diakonikolas2016robust, lai2016agnostic} gave the first efficient algorithms for high-dimensional robust estimation, which sparked an active line of research into robust algorithms for other high-dimensional problems \cite{awasthi2014power, li2017robust, balakrishnan2017computationally, charikar2017learning, diakonikolas2018robustly, klivans2018efficient, diakonikolas2019efficient, hopkins2019hard, dong2019quantum}. The most canonical high-dimensional robust estimation problem is robust sparse mean estimation, which has an intriguing statistical-computational gap induced by robustness. 

In sparse mean estimation, the observations $X_1, X_2, \dots, X_n$ are $n$ independent samples from $\mN(\mu, I_d)$ where $\mu$ is an unknown $k$-sparse vector in $\mathbb{R}^d$ of bounded $\ell_2$ norm and the task is to estimate $\mu$ within an $\ell_2$ error of $\tau$. This is a gapless problem, as taking the largest $k$ coordinates of the empirical mean runs in $\text{poly}(d)$ time and achieves the information-theoretically optimal sample complexity of $n = \Theta(k \log d/\tau^2)$.

If an $\epsilon$-fraction of these samples are corrupted arbitrarily by an adversary, this yields the robust sparse mean estimation problem $\pr{rsme}(n, k, d, \tau, \epsilon)$. As discussed in \cite{li2017robust, balakrishnan2017computationally}, for $\| \mu - \mu' \|_2$ sufficiently small, it holds that $\TV\left( \mN(\mu, I_d), \mN(\mu', I_d) \right) = \Theta(\| \mu - \mu' \|_2)$. Furthermore, an $\epsilon$-corrupted set of samples can simulate distributions within $O(\epsilon)$ total variation from $\mN(\mu, I_d)$. Therefore $\epsilon$-corruption can simulate $\mN(\mu', I_d)$ if $\|\mu' - \mu\|_2 = O(\epsilon)$ and it is impossible to estimate $\mu$ with $\ell_2$ distance less than this $O(\epsilon)$. This implies that the minimax rate of estimation for $\mu$ is $O(\epsilon)$, even for very large $n$. As shown in \cite{li2017robust, balakrishnan2017computationally}, the information-theoretic threshold for estimating at this rate in the $\epsilon$-corrupted model remains at $n = \Theta(k \log d/\epsilon^2)$ samples. However, the best known polynomial-time algorithms from \cite{li2017robust, balakrishnan2017computationally} require $n = \tilde{\Theta}(k^2 \log d/\epsilon^2)$ samples to estimate $\mu$ within $\tau = \Theta(\epsilon \sqrt{\log \epsilon^{-1}})$ in $\ell_2$. In Sections \ref{subsec:3-rsme-reduction} and \ref{subsec:3-rsme}, we give a reduction showing that these polynomial time algorithms are optimal, yielding the first average-case evidence for the $k$-to-$k^2$ statistical-computational gap conjectured in \cite{li2017robust, balakrishnan2017computationally}. Our reduction applies to more general rates $\tau$ and obtains the following tradeoff.

\begin{theorem}[Lower Bounds for $\pr{rsme}$] \label{thm:rsme-lb}
If $k, d$ and $n$ are polynomial in each other, $k = o(\sqrt{d})$ and $\epsilon < 1/2$ is such that $(n, \epsilon^{-1})$ satisfies condition \pr{(t)}, then the $k\pr{-bpc}$ conjecture implies that there is a computational lower bound for $\pr{rsme}(n, k, d, \tau, \epsilon)$ at all sample complexities $n = \tilde{o}(k^2 \epsilon^2/\tau^4)$.
\end{theorem}

For example, taking $\epsilon = 1/3$ and $\tau = \tilde{O}(1)$ shows that there is a $k$-to-$k^2$ gap between the information-theoretically optimal sample complexity of $n = \tilde{\Theta}(k)$ and the computational lower bound of $n = \tilde{o}(k^2)$. Note that taking $\tau = O(\epsilon)$ in Theorem \ref{thm:rsme-lb} recovers exactly the tradeoff in \cite{li2017robust, balakrishnan2017computationally}, with the dependence on $\epsilon$. Our reduction to $\pr{rsme}$ is based on dense Bernoulli rotations and constructions of combinatorial design matrices based on incidence geometry in $\mathbb{F}_r^t$, as is further discussed in Sections \ref{sec:1-techniques} and \ref{sec:2-bernoulli-rotations}.

In Theorem \ref{thm:rsme-lb}, \pr{(t)} denotes a technical condition arising from number-theoretic constraints in our reduction that require that $\epsilon^{-1} = n^{o(1)}$ or $\epsilon^{-1} = \tilde{\Theta}(n^{-1/2t})$ for some positive integer $t$. As $\epsilon^{-1} = n^{o(1)}$ is the primary regime of interest in the $\pr{rsme}$ literature, this condition is typically trivial. We discuss the condition \pr{(t)} in more detail in Section \ref{sec:3-robust-and-supervised} and give an alternate reduction removing it from Theorem \ref{thm:rsme-lb} in the case where $\epsilon = \tilde{\Theta}(n^{-c})$ for some constant $c \in [0, 1/2]$.

Our result also holds in the stronger Huber's contamination model where an $\epsilon$-fraction of the $n$ samples are chosen at random and replaced with i.i.d. samples from another distribution $\mathcal{D}$. The prior work of \cite{diakonikolas2017statistical} showed that SQ algorithms require $n = \tilde{\Omega}(k^2)$ samples to solve $\pr{rsme}$, establishing the conjectured $k$-to-$k^2$ gap in the SQ model. However, our work is the first to make a precise prediction of the computational barrier in $\pr{rsme}$ as a function of both $\epsilon$ and $\tau$. As will be discussed in Section \ref{subsec:3-rsme-reduction}, our reduction from $k\pr{-pc}$ maps to the instance of $\pr{rsme}$ under the adversary introduced in \cite{diakonikolas2017statistical}.

\subsection{Dense Stochastic Block Models}
\label{subsec:1-problems-sbm}

The stochastic block model (SBM) is the canonical model for community detection, having independently emerged in the machine learning and statistics \cite{holland1983stochastic}, computer science \cite{bui1987graph, dyer1989solution, boppana1987eigenvalues}, statistical physics \cite{decelle2011asymptotic} and mathematics communities \cite{bollobas2007phase}. It has been the subject of a long line of research, which has recently been surveyed in \cite{abbe2017community, moore2017computer}. In the $k$-block SBM, a vertex set of size $n$ is uniformly at random partitioned into $k$ latent communities $C_1, C_2, \dots, C_k$ each of size $n/k$ and edges are then included in the graph $G$ independently such that intra-community edges appear with probability $p$ while inter-community edges appear with probability $q < p$. The exact recovery problem entails finding $C_1, C_2, \dots, C_k$ and the weak recovery problem, also known as community detection, entails outputting nontrivial estimates $\hat{C}_1, \hat{C}_2, \dots, \hat{C}_k$ with $|C_i \cap \hat{C}_i| \ge (1 + \Omega(1))n/k$.

Community detection in the SBM is often considered in the sparse regime, where $p = a/n$ and $q = b/n$. In \cite{decelle2011asymptotic}, non-rigorous arguments from statistical physics were used to form the precise conjecture that weak recovery begins to be possible in $\text{poly}(n)$ time exactly at the \textit{Kesten-Stigum} threshold $\pr{snr} = (a - b)^2/k(a + (k - 1)b) > 1$. When $k = 2$, the algorithmic side of this conjecture was confirmed with methods based on belief propagation \cite{mossel2018proof}, spectral methods and non-backtracking walks \cite{massoulie2014community, bordenave2015non}, and it was shown to be information-theoretically impossible to solve weak recovery below the Kesten-Stigum threshold in \cite{mossel2015reconstruction, deshpande2015asymptotic}. The algorithmic side of this conjecture for general $k$ was subsequently resolved with approximate acyclic belief propagation in \cite{abbe2015detection, abbe2016achieving, abbe2018proof} and has also been shown using low-degree polynomials, tensor decomposition and color coding \cite{hopkins2017efficient}. A statistical-computational gap is conjectured to already arise at $k = 4$ \cite{abbe2018proof} and the information-theoretic limit for community detection has been shown to occur for large $k$ at $\pr{snr} = \Theta(\log k/k)$, which is much lower than the Kesten-Stigum threshold \cite{banks2016information}. Rigorous evidence for this statistical-computational gap has been much more elusive and has only been shown for low-degree polynomials \cite{hopkins2017efficient} and variants of belief propagation. Another related line of work has exactly characterized the thresholds for exact recovery in the regime $p, q = \Theta(\log n/n)$ when $k = 2$ \cite{abbe2015exact, hajek2016achieving, hajek2016achievingb}.

The $k$-block SBM for general edge densities $p$ and $q$ has also been studied extensively under the names graph clustering and graph partitioning in the statistics and computer science communities. A long line of work has developed algorithms recovering the latent communities in this regime, including a wide range of spectral and convex programming techniques \cite{boppana1987eigenvalues, dyer1989solution, condon2001algorithms, mcsherry2001spectral, bollobas2004max, coja2010graph, rohe2011spectral, chaudhuri2012spectral, nadakuditi2012graph, chen2012clustering, ames2014guaranteed, anandkumar2014tensor, chen2014improved, chen2016statistical}. A comparison and survey of these results can be found in \cite{chen2014improved}. As discussed in \cite{chen2016statistical}, for growing $k$ satisfying $k = O(\sqrt{n})$ and $p$ and $q$ with $p = \Theta(q)$ and $1 - p = \Theta(1 - q)$, the best known $\text{poly}(n)$ time algorithms all only work above
$$\frac{(p - q)^2}{q(1 - q)} \gtrsim \frac{k^2}{n}$$
which is an asymptotic extension of the Kesten-Stigum threshold to general $p$ and $q$. In contrast, the statistically optimal rate of recovery is again roughly a factor of $k$ lower at $\tilde{\Omega}(k/n)$. Furthermore, up to $\log n$ factors, the Kesten-Stigum threshold is both when efficient exact recovery algorithms begin to work and where the best efficient weak recovery algorithms are conjectured to fail \cite{chen2016statistical}. 

In this work, we show computational lower bounds matching the Kesten-Stigum threshold up to a constant factor in a mean-field analogue of recovering a first community $C_1$ in the $k$-SBM, where $p$ and $q$ are bounded away from zero and one. Consider a sample $G$ from the $k$-SBM restricted to the union of the other communities $C_2, \dots, C_k$. This subgraph has average edge density approximately given by $\hat{q} = (p - q) \cdot (k - 1) \cdot (n/k)^2 \cdot (n - n/k)^{-2} + q = (k - 1)^{-1} \cdot p + (1 - (k - 1)^{-1}) \cdot q$. Now consider the task of recovering the community $C_1$ in the graph $G'$ in which the subgraph on $C_2, \dots, C_k$ is replaced by the corresponding mean-field Erd\H{o}s-R\'{e}nyi graph $\mG(n - n/k, \hat{q})$. Formally, let $G'$ be the graph formed by first choosing $C_1$ at random and sampling edges as follows:
\begin{itemize}
\item include edges within $C_1$ with probability $P_{11} = p$;
\item include edges between $C_1$ and $[n]\backslash C_1$ with probability $P_{12} = q$; and
\item includes edges within $[n]\backslash C_1$ with probability $P_{22}$ where $P_{22} = (k - 1)^{-1} \cdot p + (1 - (k - 1)^{-1}) \cdot q$.
\end{itemize}
We refer to this model as the imbalanced SBM and let $\pr{isbm}(n, k, P_{11}, P_{12}, P_{22})$ denote the problem of testing between this model and Erd\H{o}s-R\'{e}nyi graphs of the form $\mG(n, P_0)$. As we will discuss in Section \ref{subsec:2-formulations}, lower bounds for this formulation also imply lower bounds for weakly and exactly recovering $C_1$. We remark that under our notation for $\pr{isbm}$, the hidden community $C_1$ has size $n/k$ and $k$ is the number of communities in the analogous $k$-block SBM described above.

As we will discuss in Section \ref{sec:3-community}, $\pr{isbm}$ can also be viewed as a model of single community detection with uniformly calibrated expected degrees. Note that the expected degree of a vertex in $C_1$ is $nP_{22} - p$ and the expected degree of a vertex in $C_1 \backslash [n]$ is $(n - 1)P_{22}$, which differ by at most $1$. Similar models with two imbalanced communities and calibrated expected degrees have appeared previously in \cite{neeman2014non, verzelen2015community, perry2017semidefinite, caltagirone2018recovering}. As will be discussed in Section \ref{subsec:1-problems-semicr}, the simpler planted dense subgraph model of single community recovery has a detection threshold that differs from the Kesten-Stigum threshold, even though the Kesten-Stigum threshold is conjectured to be the barrier for recovering the planted dense subgraph. This is because non-uniformity in expected degrees gives rise to simple edge-counting tests that do not lead to algorithms for recovering the planted subgraph. Our main result for $\pr{isbm}$ is the following lower bound up to the asymptotic Kesten-Stigum threshold.

\begin{theorem}[Lower Bounds for $\pr{isbm}$] \label{thm:isbm-lb}
Suppose that $(n, k)$ satisfy condition \pr{(t)}, that $k$ is prime or $k = \omega_n(1)$ and $k = o(n^{1/3})$, and suppose that $q \in (0, 1)$ satisfies $\min\{q, 1 - q \} = \Omega_n(1)$. If $P_{22} = (k - 1)^{-1} \cdot p + (1 - (k - 1)^{-1}) \cdot q$, then the $k\pr{-pc}$ conjecture implies that there is a computational lower bound for $\pr{isbm}(n, k, p, q, P_{22})$ at all levels of signal below the Kesten-Stigum threshold of $\frac{(p - q)^2}{q(1 - q)} = \tilde{o}(k^2/n)$.
\end{theorem}

This directly provides evidence for the conjecture that $(p - q)^2/q(1 - q) = \tilde{\Theta}(k^2/n)$ defines the computational barrier for community recovery in general $k$-SBMs made in \cite{chen2016statistical}. While the statistical-computational gaps in $\pr{pc}$ and $k$-SBM are the two most prominent conjectured gaps in average-case problems over graphs, they are very different from an algorithmic perspective and evidence for computational lower bounds up to the Kesten-Stigum threshold has remained elusive. Our reduction yields a first step towards understanding the relationship between these gaps.


\subsection{Testing Hidden Partition Models}
\label{subsec:1-problems-hidden-partition}

We also introduce two testing problems we refer to as the Gaussian and bipartite hidden partition models. We give a reduction and algorithms that show these problems have a statistical-computational gap, and we tightly characterize their computational barriers based on the $k\pr{-pc}$ conjecture. The main motivation for introducing these problems is to demonstrate the versatility of our reduction technique dense Bernoulli rotations in transforming hidden structure. A description of dense Bernoulli rotations and the construction of a key design tensor used in our reduction can be found in Section \ref{sec:2-bernoulli-rotations}.

The task in the bipartite hidden partition model problem is to test for the presence of a planted $rK$-vertex subgraph, sampled from an $r$-block stochastic block model, in an $n$-vertex random bipartite graph. The Gaussian hidden partition model problem is a corresponding Gaussian analogue. These are both multi-community variants of the subgraph stochastic block model considered in \cite{brennan2018reducibility}, which corresponds to the setting in which $r = 2$. The multi-community nature of the planted subgraph yields a more intricate hidden structure, and the additional free parameter $r$ yields a more complicated computational barrier. The work of \cite{chen2016statistical} considered the related task of recovering the communities in the Gaussian and bipartite hidden partition models. We remark that conjectured computational limits for this recovery task differ from the detection limits we consider.

Formally, our hidden partition problems are defined as follows. Let $C = (C_1, C_2, \dots, C_r)$ and $D = (D_1, D_2, \dots, D_r)$ are chosen independently and uniformly at random from the set of all sequences of length $r$ consisting of disjoint $K$-subsets of $[n]$. Consider the random matrix $M$ sampled by first sampling $C$ and $D$ and then sampling
$$M_{ij} \sim \left\{ \begin{array}{ll} \mN(\gamma, 1) &\textnormal{if } i \in C_h \textnormal{ and } j \in D_h \textnormal{ for some } h \in [r] \\ \mN\left(-\frac{\gamma}{r - 1}, 1 \right) &\textnormal{if } i \in C_{h_1} \textnormal{ and } j \in D_{h_2} \textnormal{ where } h_1 \neq h_2 \\ \mN(0, 1) &\textnormal{otherwise} \end{array} \right.$$
independently for each $1 \le i, j \le n$. The problem $\pr{ghpm}(n, r, K, \gamma)$ is to test between $H_0 : M \sim \mN(0, 1)^{\otimes n \times n}$ and an alternative hypothesis $H_1$ under which $M$ is sampled as outlined above. The problem $\pr{bhpm}(n, r, K, P_0, \gamma)$ is a bipartite graph analogue of this problem with ambient edge density $P_0$, edge density $P_0 + \gamma$ within the communities in the subgraph and $P_0 - \frac{\gamma}{r - 1}$ on the rest of the subgraph.

As we will show in Section \ref{sec:3-hidden-partition}, an empirical variance test succeeds above the threshold $\gamma_{\text{comp}}^2 = \tilde{\Theta}(n/rK^2)$ and an exhaustive search succeeds above $\gamma_{\text{IT}}^2 = \tilde{\Theta}(1/K)$ in $\pr{ghpm}$ and $\pr{bhpm}$ where $P_0$ is bounded away from $0$ and $1$. Thus our main lower bounds for these two problems confirm that this empirical variance test is approximately optimal among efficient algorithms and that both problems have a statistical-computational gap assuming the $k\pr{-pc}$ conjecture.

\begin{theorem}[Lower Bounds for $\pr{ghpm}$ and $\pr{bhpm}$] \label{thm:ghpm-lb}
Suppose that $r^2 K^2 = \tilde{\omega}(n)$ and $(\lceil r^2 K^2/n \rceil, r)$ satisfies condition \pr{(t)}, suppose $r$ is prime or $r = \omega_n(1)$ and suppose that $P_0 \in (0, 1)$ satisfies $\min\{P_0, 1 - P_0 \} = \Omega_n(1)$. Then the $k\pr{-pc}$ conjecture implies that there is a computational lower bound for each of $\pr{ghpm}(n, r, K, \gamma)$ for all levels of signal $\gamma^2 = \tilde{o}(n/rK^2)$. This same lower bound also holds for $\pr{bhpm}(n, r, K, P_0, \gamma)$ given the additional condition $n = o(rK^{4/3})$.
\end{theorem}

We also remark that the empirical variance and exhaustive search tests along with our lower bound do not support the existence of a statistical-computational gap in the case when the subgraph is the entire graph with $n = rK$, which is our main motivation for considering this subgraph variant. We remark that a number of the technical conditions in the theorem such as condition \pr{(t)} and $n = o(rK^{4/3})$ are trivial in the parameter regime where the number of communities is not very large with $r = n^{o(1)}$ and when the total size of the hidden communities is large with $rK = \tilde{\Theta}(n^{c})$ where $c > 3/4$. In this regime, these problems have a nontrivial statistical-computational gap that our result tightly characterizes.

\subsection{Semirandom Planted Dense Subgraph and the Recovery Conjecture}
\label{subsec:1-problems-semicr}

\newcommand{\colorImp}{gray}
\newcommand{\colorEasy}{green}
\newcommand{\colorHard}{blue}

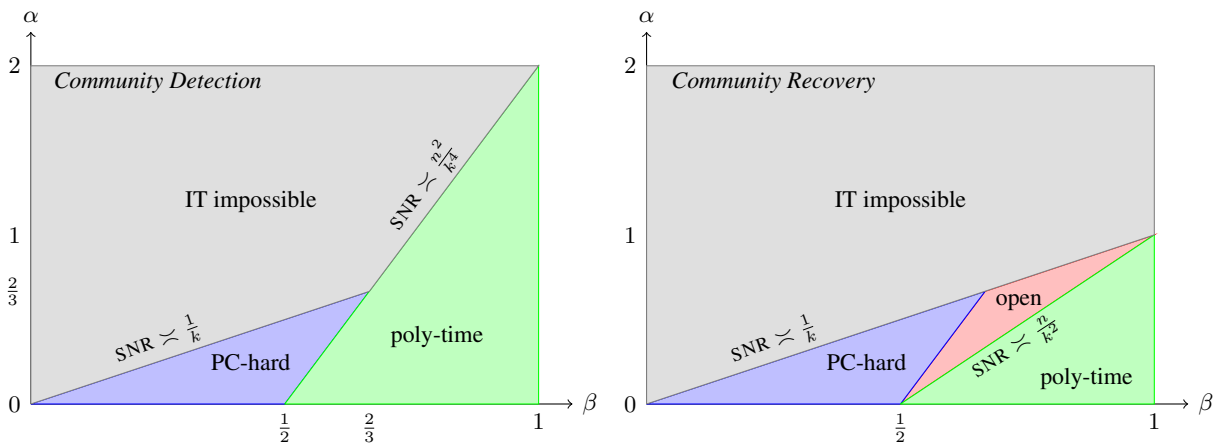
\begin{figure*}[t!]
\centering
\begin{tikzpicture}[scale=0.45]
\tikzstyle{every node}=[font=\footnotesize]
\def\xmin{0}
\def\xmax{16}
\def\ymin{0}
\def\ymax{11}

\draw[->] (\xmin,\ymin) -- (\xmax,\ymin) node[right] {$\beta$};
\draw[->] (\xmin,\ymin) -- (\xmin,\ymax) node[above] {$\alpha$};

\node at (15, 0) [below] {$1$};
\node at (7.5, 0) [below] {$\frac{1}{2}$};
\node at (0, 0) [left] {$0$};
\node at (0, 10) [left] {$2$};
\node at (0, 5) [left] {$1$};
\node at (0, 3.33) [left] {$\frac{2}{3}$};
\node at (10, 0) [below] {$\frac{2}{3}$};

\filldraw[fill=\colorHard!25, draw=\colorHard] (0, 0) -- (7.5, 0) -- (10, 3.33) -- (0, 0);
\filldraw[fill=\colorEasy!25, draw=\colorEasy] (7.5, 0) -- (15, 10) -- (15, 0) -- (7.5, 0);
\filldraw[fill=\colorImp!25, draw=\colorImp] (0, 0) -- (10, 3.33) -- (15, 10) -- (0, 10) -- (0, 0);

\node at (3.75, 9.5) {\textit{Community Detection}};
\node at (12.2, 6)[rotate=54, anchor=south] {$\pr{snr} \asymp \frac{n^2}{k^4}$};
\node at (4, 1.2)[rotate=20, anchor=south] {$\pr{snr} \asymp \frac{1}{k}$};
\node at (6.5, 6) {IT impossible};
\node at (12, 2) {poly-time};
\node at (6.5, 1.25) {PC-hard};
\end{tikzpicture}
\begin{tikzpicture}[scale=0.45]
\tikzstyle{every node}=[font=\footnotesize]
\def\xmin{0}
\def\xmax{16}
\def\ymin{0}
\def\ymax{11}

\draw[->] (\xmin,\ymin) -- (\xmax,\ymin) node[right] {$\beta$};
\draw[->] (\xmin,\ymin) -- (\xmin,\ymax) node[above] {$\alpha$};

\node at (15, 0) [below] {$1$};
\node at (7.5, 0) [below] {$\frac{1}{2}$};
\node at (0, 0) [left] {$0$};
\node at (0, 10) [left] {$2$};
\node at (0, 5) [left] {$1$};

\filldraw[fill=red!25, draw=red] (7.5, 0) -- (15, 5) -- (10, 3.33) -- (7.5, 0);
\filldraw[fill=\colorHard!25, draw=\colorHard] (0, 0) -- (7.5, 0) -- (10, 3.33) -- (0, 0);
\filldraw[fill=\colorEasy!25, draw=\colorEasy] (7.5, 0) -- (15, 5) -- (15, 0) -- (7.5, 0);
\filldraw[fill=\colorImp!25, draw=\colorImp] (0, 0) -- (15, 5) -- (15, 10) -- (0, 10) -- (0, 0);

\node at (3.75, 9.5) {\textit{Community Recovery}};
\node at (11.3, 1.1)[rotate=33, anchor=south] {$\pr{snr} \asymp \frac{n}{k^2}$};
\node at (4, 1.2)[rotate=20, anchor=south] {$\pr{snr} \asymp \frac{1}{k}$};
\node at (7.5, 6) {IT impossible};
\node at (13, 0.75) {poly-time};
\node at (6.5, 1.25) {PC-hard};
\node at (11, 3) {open};
\end{tikzpicture}

\caption{Prior computational and statistical barriers in the detection and recovery of a single hidden community from the \pr{pc} conjecture \cite{hajek2015computational, brennan2018reducibility, brennan2019universality}. The axes are parameterized by $\alpha$ and $\beta$ where $\pr{snr} = \frac{(P_1 - P_0)^2}{P_0(1 - P_0)} = \tilde{\Theta}(n^{-\alpha})$ and $k = \tilde{\Theta}(n^\beta)$. The red region is conjectured to be hard but no $\textsc{pc}$ reductions showing this are known.}
\label{fig:pdsdetrecgap}
\end{figure*}

In the planted dense subgraph model of single community recovery, the observation is a sample from $\mG(n, k, P_1, P_0)$ which is formed by planting a random subgraph on $k$ vertices from $\mG(k, P_1)$ inside a copy of $\mG(n, P_0)$, where $P_1 > P_0$ are allowed to vary with $n$ and satisfy that $P_1 = O(P_0)$. Detection and recovery of the hidden community in this model have been studied extensively \cite{arias2014community, butucea2013detection, verzelen2015community, hajek2015computational,chen2016statistical, hajek2016information, montanari2015finding, candogan2018finding} and this model has emerged as a canonical example of a problem with a detection-recovery computational gap. While it is possible to efficiently detect the presence of a hidden subgraph of size $k=\tilde \Omega(\sqrt{n})$ if $(P_1 - P_0)^2/P_0(1 - P_0) = \tilde{\Omega}(n^2/k^4)$, the best known polynomial time algorithms to \emph{recover} the subgraph require a higher signal at the Kesten-Stigum threshold of $(P_1 - P_0)^2/P_0(1 - P_0) = \tilde{\Omega}(n/k^2)$.

In each of \cite{hajek2015computational, brennan2018reducibility} and \cite{brennan2019universality}, it has been conjectured that the recovery problem is hard below this threshold of $\tilde{\Theta}(n/k^2)$. This \pr{pds} Recovery Conjecture was even used in \cite{brennan2018reducibility} as a hardness assumption to show detection-recovery gaps in other problems including biased sparse PCA and Gaussian biclustering. A line of work has tightly established the conjectured detection threshold through reductions from the \pr{pc} conjecture \cite{hajek2015computational, brennan2018reducibility, brennan2019universality}, while the recovery threshold has remained elusive. Planted clique maps naturally to the detection threshold in this model, so it seems unlikely that the \pr{pc} conjecture could also yield lower bounds at the tighter recovery threshold, given that recovery and detection are known to be equivalent for \pr{pc} \cite{alon2007testing}. These prior lower bounds and the conjectured detection-recovery gap in $\pr{pds}$ are depicted in Figure \ref{fig:pdsdetrecgap}.

We show that the $k\pr{-pc}$ conjecture implies the \pr{pds} Recovery Conjecture for \textit{semirandom} community recovery in the regime where $P_0 = \Theta(1)$. Semirandom adversaries provide an alternate notion of robustness against constrained modifications that heuristically appear to increase the signal strength \cite{blum1995coloring}. Algorithms and lower bounds in semirandom problems have been studied for a number of problems, including the stochastic block model \cite{feige2001heuristics, moitra2016robust}, planted clique \cite{feige2000finding}, unique games \cite{kolla2011play}, correlation clustering \cite{mathieu2010correlation, makarychev2015correlation}, graph partitioning \cite{makarychev2012approximation}, 3-coloring \cite{david2016effect} and clustering mixtures of Gaussians \cite{vijayaraghavan2018clustering}. Formally we consider the problem $\pr{semi-cr}(n, k, P_1, P_0)$ where a semirandom adversary is allowed to remove edges outside of the planted subgraph from a graph sampled from $\mG(n, k, P_1, P_0)$. The task is to test between this model and an Erd\H{o}s-R\'{e}nyi graph $\mG(n, P_0)$ similarly perturbed by a semirandom adversary. As we will discuss in Section \ref{subsec:2-formulations}, lower bounds for this formulation extend to approximately recovering the hidden community under a semirandom adversary. In Section \ref{sec:semirandom}, we prove the following theorem -- that the computational barrier in the detection problem shifts to the recovery threshold in $\pr{semi-cr}$.

\begin{theorem}[Lower Bounds for $\pr{semi-cr}$] \label{thm:semi-cr-lb}
If $k$ and $n$ are polynomial in each other with $k = \Omega(\sqrt{n})$ and $0 < P_0 < P_1 \le 1$ where $\min\{P_0, 1 - P_0 \} = \Omega(1)$, then the $k\pr{-pc}$ conjecture implies that there is a computational lower bound for $\pr{semi-cr}(n, k, P_1, P_0)$ at $\frac{(P_1 - P_0)^2}{P_0(1 - P_0)} = \tilde{o}(n/k^2)$.
\end{theorem}

A related reference is the reduction in \cite{cai2015computational}, which proves a detection-recovery gap in the context of sub-Gaussian submatrix localization based on the hardness of finding a planted $k$-clique in a random $n/2$-regular graph. The relationship between our lower bound and that of \cite{cai2015computational} is discussed in more detail in Section \ref{sec:semirandom}. From an algorithmic perspective, the convexified maximum likelihood algorithm from \cite{chen2016statistical} complements our lower bound -- a simple monotonicity argument shows that it continues to solve the community recovery problem above the Kesten-Stigum threshold under a semirandom adversary.

\subsection{Negatively Correlated Sparse Principal Component Analysis}
\label{subsec:1-problems-negspca}

In sparse principal component analysis (PCA), the observations $X_1, X_2, \dots, X_n$ are $n$ independent samples from $\mN(0, \Sigma)$ where the eigenvector $v$ corresponding to the largest eigenvalue of $\Sigma$ is $k$-sparse, and the task is to estimate $v$ in $\ell_2$ norm or find its support. Sparse PCA has many applications ranging from online visual tracking \cite{wang2013online} and image compression \cite{majumdar2009image} to gene expression analysis \cite{zou2006sparse, chun2009expression, parkhomenko2009sparse, chan2010using}. Showing lower bounds for sparse PCA can be reduced to analyzing detection in the spiked covariance model \cite{johnstoneSparse04}, which has hypotheses
$$H_0:X \sim \mN(0, I_d)^{\otimes n} \quad\text{ and }\quad  H_1:X \sim \mN(0, I_d + \theta vv^\top)^{\otimes n}$$
Here, $H_1$ is the composite hypothesis where $v \in \mathbb{R}^d$ is unknown and allowed to vary over all $k$-sparse unit vectors. The information-theoretically optimal rate of detection is at the level of signal $\theta = \Theta(\sqrt{k \log d/n})$ \cite{berthet2013optimal, cai2015optimal, wang2016statistical}. However, when $k = o(\sqrt{d})$, the best known polynomial time algorithms for sparse PCA require that $\theta = \Omega(\sqrt{k^2/n})$. Since the seminal paper of \cite{berthet2013complexity} initiated the study of statistical-computational gaps through the $\pr{pc}$ conjecture, this $k$-to-$k^2$ gap for sparse PCA has been shown to follow from the $\pr{pc}$ conjecture in a sequence of papers \cite{berthet2013optimal, berthet2013complexity, wang2016statistical, gao2017sparse, brennan2018reducibility, brennan2019optimal}.

In negatively correlated sparse PCA, the eigenvector $v$ of interest instead corresponds to the \emph{smallest eigenvalue} of $\Sigma$. Negative sparse PCA can similarly be formulated as a hypothesis testing problem $\pr{neg-spca}(n, k, d, \theta)$, where the alternative hypothesis is instead given by $H_1: X \sim \mN(0, I_d - \theta vv^\top)^{\otimes n}$. Similar algorithms as in ordinary sparse PCA continue to work in the negative setting -- the information-theoretic limit of the problem remains at $\theta = \Theta(\sqrt{k \log d/n})$ and the best known efficient algorithms still require $\theta = \Omega(\sqrt{k^2/n})$. However, negative sparse PCA is stochastically a \emph{very differently structured} problem than ordinary sparse PCA. A sample from the ordinary spiked covariance model can be expressed as
$$X_i = \sqrt{\theta} \cdot gv + \mN(0, I_d)$$
where $g \sim \mN(0, 1)$ is independent of the $\mN(0, I_d)$ term. This signal plus noise representation is a common feature in many high-dimensional statistical models and is crucially used in the reductions showing hardness for sparse PCA in \cite{berthet2013optimal, berthet2013complexity, wang2016statistical, gao2017sparse, brennan2018reducibility, brennan2019optimal}. Negative sparse PCA does not admit a representation of this form, making it an atypical planted problem and different from ordinary sparse PCA, despite the deceiving similarity between their optimal algorithms. The lack of this representation makes reducing to Negative sparse PCA technically challenging. Negatively spiked PCA was also recently related to the hardness of finding approximate ground states in the Sherrington-Kirkpatrick model \cite{bandeira2019computational}. However, ordinary PCA does not seem to share this connection. In Section \ref{sec:2-neg-spca}, we give a reduction obtaining the following computational lower bound for $\pr{neg-spca}$ from the $\pr{bpc}$ conjecture.

\begin{theorem}[Lower Bounds for $\pr{neg-spca}$] \label{thm:neg-spca-lb}
If $k, d$ and $n$ are polynomial in each other, $k = o(\sqrt{d})$ and $k = o(n^{1/6})$, then the $\pr{bpc}$ conjecture implies a computational lower bound for $\pr{neg-spca}(n, k, d, \theta)$ at all levels of signal $\theta = \tilde{o}(\sqrt{k^2/n})$.
\end{theorem}

We deduce this theorem and discuss its conditions in detail in Section \ref{subsec:3-neg-spca}. A key step in our reduction to $\pr{neg-spca}$ involves randomly rotating the positive semidefinite square root of the inverse of an empirical covariance matrix. In analyzing this step, we  prove a novel convergence result in random matrix theory, which may be of independent interest. Specifically, we characterize when a Wishart matrix and its inverse converge in KL divergence. This is where the parameter constraint $k = o(n^{1/6})$ in the theorem above arises. We believe that this is an artefact of our techniques and extending the theorem to hold without this condition is an interesting open problem. A similar condition arose in the strong lower bounds of \cite{brennan2019optimal}. We remark that conditions of this form \emph{do not affect the tightness} of our lower bounds, but rather only impose a constraint on the level of sparsity $k$. More precisely, for each fixed level of sparsity $k = \tilde{\Theta}(n^{\alpha})$, there is conjectured statistical-computational gap in $\theta$ between the information-theoretic barrier of $\theta = \Theta(\sqrt{k \log d/n})$ and computational barrier of $\theta = \tilde{o}(\sqrt{k^2/n})$. Our reduction tightly establishes this gap for all $\alpha \in (0, 1/6]$. Our main motivation for considering $\pr{neg-spca}$ is that it seems to have a fundamental connection to the structure of \emph{supervised problems} where ordinary sparse PCA does not. In particular, our reduction to $\pr{neg-spca}$ is a crucial subroutine in reducing to mixtures of sparse linear regressions and robust sparse linear regression. This is discussed further in Sections \ref{sec:1-techniques}, \ref{sec:2-neg-spca} and \ref{sec:2-supervised}.

\subsection{Unsigned and Mixtures of Sparse Linear Regressions}
\label{subsec:1-problems-mslr}

In learning mixtures of sparse linear regressions (SLR), the task is to learn $L$ sparse linear functions capturing the relationship between features and response variables in heterogeneous samples from $L$ different sparse regression problems. Formally, the observations $(X_1, y_1), (X_2, y_2), \dots, (X_n, y_n)$ are $n$ independent sample-label pairs given by $y_i = \langle \beta, X_i \rangle + \eta_i$ where $X_i \sim \mN(0, I_d)$, $\eta_i \sim \mN(0, 1)$ and $\beta$ is chosen from a mixture distribution $\nu$ over a finite set $k$-sparse vectors $\{\beta_1, \beta_2, \dots, \beta_L\}$ of bounded $\ell_2$ norm. The task is to estimate the components $\beta_j$ that are sufficiently likely under $\nu$ in $\ell_2$ norm i.e. to within an $\ell_2$ distance of $\tau$.

Mixtures of linear regressions, also known as the hierarchical mixtures of experts model in the machine learning community \cite{jordan1994hierarchical}, was first introduced in \cite{quandt1978estimating} and has been studied extensively in the past few decades \cite{de1989mixtures, wedel1995mixture, mclachlan2004finite, zhu2004hypothesis, faria2010fitting}. Recent work on mixtures of linear regressions has focussed on efficient algorithms with finite-sample guarantees \cite{chaganty2013spectral, chen2014convex, yi2014alternating, balakrishnan2017statistical, chen2017convex, li2018learning}. The high-dimensional setting of mixtures of SLRs was first considered in \cite{stadler2010l}, which proved an oracle inequality for an $\ell_1$-regularization approach, and variants of the EM algorithm for mixtures of SLRs were analyzed in \cite{wang2014high, yi2015regularized}. Recent work has also studied a different setting for mixtures of SLRs where the covariates $X_i$ can be designed by the learner \cite{yin2018learning, krishnamurthy2019sample}.

We show that a statistical-computational gap emerges for mixtures of SLRs even in the simplest case where there are $L = 2$ components, the mixture distribution $\nu$ is known to sample each component with probability $1/2$ and the task is to estimate even just one of the components $\{ \beta_1, \beta_2\}$ to within $\ell_2$ norm $\tau$. We refer to this simplest setup for learning mixtures of SLRs as $\pr{mslr}(n, k, d, \tau)$. The following computational lower bound is deduced in Section \ref{subsec:3-slr} and is a consequence of the reduction in Section \ref{sec:2-supervised}.

\begin{theorem}[Lower Bounds for $\pr{mslr}$] \label{thm:mslr-lb}
If $k, d$ and $n$ are polynomial in each other, $k = o(\sqrt{d})$ and $k = o(n^{1/6})$, then the $k\pr{-bpc}$ conjecture implies that there is a computational lower bound for $\pr{mslr}(n, k, d, \tau)$ at all sample complexities $n = \tilde{o}(k^2/\tau^4)$.
\end{theorem}

As we will discuss in Section \ref{subsec:2-formulations}, we will prove this theorem by reducing to the problem of testing between the mixtures of SLRs model when $\beta_1 = - \beta_2$ and a null hypothesis under which $y$ and $X$ are independent. A closely related work \cite{fan2018curse} studies a nearly identical testing problem in the statistical query model. They tightly characterize the information-theoretic limit of this problem, showing that it occurs at the sample complexity $n = \tilde{\Theta}(k \log d /\tau^4)$. Therefore our reduction establishes a $k$-to-$k^2$ statistical-computational gap in this model of learning mixtures of SLRs. In \cite{fan2018curse}, it is also shown that efficient algorithms in the statistical query model suffer from this same $k$-to-$k^2$ gap.

Our reduction to the hypothesis testing formulation of $\pr{mslr}$ above is easily seen to imply that the same computational lower bound holds for an unsigned variant $\pr{uslr}(n, k, d, \tau)$ of SLR, where the $n$ observations $(X_1, y_1), (X_2, y_2), \dots, (X_n, y_n)$ now of the form $y_i = |\langle \beta, X_i \rangle + \eta_i|$ for a fixed unknown $\beta$. Note that by the symmetry of $\mN(0, 1)$, $y_i$ is equidistributed to $||\langle \beta, X_i \rangle | + \eta_i|$ and thus is a noisy observation of $|\langle \beta, X_i \rangle |$. In general, noisy observations of the phaseless modulus $|\langle \beta, X_i \rangle |$ from some conditional link distribution $\bP( \cdot \, | \, |\langle \beta, X_i \rangle | )$ yields a general instance of phase retrieval \cite{mondelli2018fundamental, celentano2020estimation}. As observed in \cite{fan2018curse}, the problem $\pr{uslr}$ is close to the canonical formulation of sparse phase retrieval (SPR) where $\bP( \cdot \, | \, |\langle \beta, X_i \rangle | )$ is $\mN(|\langle \beta, X_i \rangle |^2, \sigma^2)$, which has been studied extensively and has a conjectured $k$-to-$k^2$ statistical-computational gap \cite{li2013sparse, schniter2014compressive, candes2015phase, cai2016optimal, wang2017sparse, hand2018phase, barbier2019optimal, celentano2020estimation}. Our lower bounds provide partial evidence for this conjecture and it is an interesting open problem to give a reduction to the canonical formulation of SPR and other sparse GLMs through average-case reductions.

The reduction to $\pr{mslr}$ showing Theorem \ref{thm:mslr-lb} in Section \ref{sec:2-supervised} is our capstone reduction. It showcases a wide range of our techniques including dense Bernoulli rotations, constructions of combinatorial design matrices from $\mathbb{F}_r^t$, our reduction to $\pr{neg-spca}$ and its connection to random matrix theory, and an additional technique of combining instances of different unsupervised problems into a supervised problem. We give an overview of these techniques in Section \ref{sec:1-techniques}. Furthermore, $\pr{mslr}$ is a very differently structured problem from any of our variants of $\pr{pc}$ and it is surprising that the tight statistical-computational gap for $\pr{mslr}$ can be derived from their hardness. We remark that our lower bounds for $\pr{mslr}$ inherit the technical condition that $k = o(n^{1/6})$ from our reduction to $\pr{neg-spca}$. As before, this does not affect the fact that we show tight hardness and it is an interesting open problem to remove this condition.

\subsection{Robust Sparse Linear Regression}
\label{subsec:1-problems-robust-slr}

In ordinary SLR, the observations $(X_1, y_1), (X_2, y_2), \dots, (X_n, y_n)$ are independent sample-label pairs given by $y_i = \langle \beta, X_i \rangle + \eta_i$ where $X_i \sim \mN(0, \Sigma)$, $\eta_i \sim \mN(0, 1)$ and $\beta$ is an unknown $k$-sparse vector with bounded $\ell_2$ norm. The task is to estimate $\beta$ to within $\ell_2$ norm $\tau$. When $\Sigma$ is well-conditioned, SLR is a gapless problem with the computationally efficient LASSO attaining the information-theoretically optimal sample complexity of $n = \Theta(k \log d/\tau^2)$ \cite{tibshirani1996regression,bickel2009simultaneous,raskutti2010restricted}. When $\Sigma$ is not well-conditioned, SLR has a statistical-computational gap based on its restricted eigenvalue constant \cite{zhang2014lower}. As with robust sparse mean estimation, the robust SLR problem $\pr{rslr}(n, k, d, \tau, \epsilon)$ is obtained when a computationally-unbounded adversary corrupts an arbitrary $\epsilon$-fraction of the observed sample-label pairs. In this work, we consider the simplest case of $\Sigma = I_d$ where SLR is gapless but, as we discuss next, robustness seems to induce a statistical-computational gap.

Robust regression is a well-studied classical problem in statistics \cite{rousseeuw2005robust}. Efficient algorithms remained elusive for decades, but recent breakthroughs in sum of squares algorithms \cite{klivans2018efficient, karmalkar2019list, raghavendra2020list}, filtering approaches \cite{diakonikolas2019efficient} and robust gradient descent \cite{chen2017distributed, prasad2018robust, diakonikolas2019sever} have led to the first efficient algorithms with provable guarantees. A recent line of work has also studied efficient algorithms and barriers in the high-dimensional setting of robust SLR \cite{chen2013robust, balakrishnan2017computationally, liu2018high, liu2019high}. Even in the simplest case of $\Sigma = I_d$ where the covariates $X_i$ have independent entries, the best known polynomial time algorithms suggest robust SLR has a $k$-to-$k^2$ statistical-computational gap. As shown in \cite{gao2020robust}, similar to $\pr{rsme}$, robust SLR is only information-theoretically possible if $\tau = \Omega(\epsilon)$. In \cite{balakrishnan2017computationally, liu2018high}, it is shown that polynomial-time ellipsoid-based algorithms solve robust SLR with $n = \tilde{\Theta}(k^2 \log d/\epsilon^2)$ samples when $\tau = \tilde{\Theta}(\epsilon)$. Furthermore, \cite{liu2018high} shows that an $\pr{rsme}$ oracle can be used to solve robust SLR with only a $\tilde{\Theta}(1)$ factor loss in $\tau$ and the required number of samples $n$. As noted in \cite{li2017robust}, $n = \Omega(k \log d/\epsilon^2)$ samples suffice to solve $\pr{rsme}$ inefficiently when $\tau = \Theta(\epsilon)$. Combining these observations yields an inefficient algorithm for robust SLR with sample complexity $n = \tilde{\Theta}(k \log d/\epsilon^2)$ samples when $\tau = \tilde{\Theta}(\epsilon)$, confirming that the best known efficient algorithms suggest a $k$-to-$k^2$ statistical-computational gap. In \cite{chen2013robust, liu2019high}, efficient algorithms are shown to succeed in an alternative regime where $n = \tilde{\Theta}(k \log d)$, $\epsilon = \tilde{O}(1/\sqrt{k})$ and $\tau = \tilde{O}(\epsilon \sqrt{k})$.

All of these algorithms suggest that the correct computational sample complexity for robust SLR is $n = \tilde{\Omega}(k^2 \epsilon^2/\tau^4)$. In Section \ref{subsec:3-slr}, we deduce the following tight computational lower bound for $\pr{rslr}$ providing evidence for this conjecture.

\begin{theorem}[Lower Bounds for $\pr{rslr}$] \label{thm:rslr-lb}
If $k, d$ and $n$ are polynomial in each other, $k = o(n^{1/6})$, $k = o(\sqrt{d})$ and $\epsilon < 1/2$ is such that $\epsilon = \tilde{\Omega}(n^{-1/2})$, then the $k\pr{-bpc}$ conjecture implies that there is a computational lower bound for $\pr{rslr}(n, k, d, \tau, \epsilon)$ at all sample complexities $n = \tilde{o}(k^2 \epsilon^2/\tau^4)$.
\end{theorem}

We present the reductions to $\pr{mslr}$ and $\pr{rslr}$ together as a single unified reduction $k\pr{-pds-to-mslr}$ in Section \ref{sec:2-supervised}. As is discussed in Section \ref{subsec:3-slr}, $\pr{mslr}$ and $\pr{rslr}$ are obtained by setting $r = \epsilon^{-1} = 2$ and $\epsilon < 1/2$, respectively. The theorem above follows from a slightly modified version of this reduction, $k\pr{-pds-to-mslr}_R$, that removes the technical condition \pr{(t)} that otherwise arises in applying $k\pr{-pds-to-mslr}$ with $r = n^{\Omega(1)}$. This turns out to be more important here than in the context of $\pr{rsme}$ because, as in the reduction to $\pr{mslr}$, this reduction to $\pr{rslr}$ inherits the technical condition that $k = o(n^{1/6})$ from our reduction to $\pr{neg-spca}$. This condition implicitly imposes a restriction on $\epsilon$ to satisfy that $\epsilon = \tilde{O}(n^{-1/3})$, since $\tau = \Omega(\epsilon)$ must be true for the problem to not be information-theoretically impossible. Thus our regime of interest for $\pr{rslr}$ is a regime where the technical condition \pr{(t)} is nontrivial. 

As in the case of $\pr{mslr}$ and $\pr{neg-spca}$, we emphasize that the condition $k = o(n^{1/6})$ does not affect the tightness of our lower bounds, merely restricting their regime of application. In particular, the theorem above yields a tight nontrivial statistical-computational gap in the entire parameter regime when $k = o(n^{1/6})$, $\tau = \Omega(\epsilon)$ and $\epsilon = \tilde{\Theta}(n^{-c})$ where $c$ is any constant in the interval $[1/3, 1/2]$. We remark that the condition $k = o(n^{1/6})$ seems to be an artefact of our techniques rather than necessary.

In the context of $\pr{rslr}$, we view our main contribution as a set of reduction techniques relating $\pr{pc}_\rho$ to the very differently structured problem $\pr{rslr}$, rather than the resulting computation lower bound itself. A byproduct of our reduction is the explicit construction of an adversary modifying an $\epsilon$-fraction of the samples in robust SLR that produces the $k$-to-$k^2$ statistical-computational gap in the theorem above. This adversary turns out to be surprisingly nontrivial on its own, but is a direct consequence of the structure of the reduction. This is discussed in more detail in Sections \ref{subsec:2-mixtures-slr} and \ref{subsec:3-slr}.

\subsection{Tensor Principal Component Analysis}
\label{subsec:1-problems-tpca}

In Tensor PCA, the observation is a single order $s$ tensor $T$ with dimensions $n^{\otimes s} = n \times n \times \cdots \times n$ given by $T \sim \theta v^{\otimes s} + \mN(0, 1)^{\otimes n^{\otimes s}}$, where $v$ has a Rademacher prior and is distributed uniformly over $\{-1, 1\}^n$ \cite{richard2014statistical}. The task is to recover $v$ within nontrivial $\ell_2$ error $o(\sqrt{n})$ and is only information-theoretically possible if $\theta = \tilde{\omega}\left(n^{(1 - s)/2}\right)$ \cite{richard2014statistical, lesieur2017statistical, chen2018phase, jagannath2018statistical, chen2019phase, perry2020statistical}, in which case $v$ can be recovered through exhaustive search. The best known polynomial-time algorithms all require the higher signal strength $\theta = \tilde{\Omega}(n^{-s/4})$, at which point $v$ can be recovered through spectral algorithms \cite{richard2014statistical}, the sum of squares hierarchy \cite{hopkins2015tensor, hopkins2016fast} and spectral algorithms based on the Kikuchi hierarchy \cite{wein2019kikuchi}. Lower bounds up to this conjectured computational barrier have been shown in the sum of squares hierarchy \cite{hopkins2015tensor, hopkins2017power} and for low-degree polynomials \cite{kunisky2019notes}. A number of natural ``local'' algorithms have also been shown to fail given much stronger levels of signal up to $\theta = \tilde{o}(n^{-1/2})$, including approximate message passing, the tensor power method, Langevin dynamics and gradient descent \cite{richard2014statistical, anandkumar2014tensor, arous2018algorithmic}.

We give a reduction showing that the $\pr{pc}_\rho$ conjecture implies an optimal computational lower bound at $\theta = \tilde{\Omega}(n^{-s/4})$ for tensor PCA. We show this lower bound against efficient algorithms with a low false positive probability of error in the hypothesis testing formulation of tensor PCA where $T \sim \mN(0, 1)^{\otimes n^{\otimes s}}$ under $H_0$ and $T$ is sampled from the tensor PCA distribution described above under $H_1$. More precisely, we prove the following theorem in Sections \ref{sec:2-hypergraph-planting} and \ref{sec:3-tensor}.

\begin{theorem}[Lower Bounds for $\pr{tpca}$] \label{thm:tpca-lb}
Let $n$ be a parameter and $s \ge 3$ be a constant, then the $k\pr{-hpc}^s$ conjecture implies a computational lower bound for $\pr{tpca}^s(n, \theta)$ when $\theta = \tilde{o}(n^{-s/4})$ against $\textnormal{poly}(n)$ time algorithms $\mathcal{A}$ solving $\pr{tpca}^s(n, \theta)$ with a low false positive probability of $\bP_{H_0}[\mathcal{A}(T) = H_1] = O(n^{-s})$.
\end{theorem}

Lemma \ref{lem:one-side-estimation} in Section \ref{sec:3-tensor} shows that any $\text{poly}(n)$ time algorithm solving the recovery formulation of tensor PCA yields such an algorithm $\mathcal{A}$, and thus this theorem implies our desired computational lower bound. This low false positive probability of error condition on $\mathcal{A}$ arises from the fact that our reduction to $\pr{tpca}$ is a \textit{multi-query} average-case reduction, requiring multiple calls to a tensor PCA blackbox to solve $k\pr{-hpc}^s$. This feature is a departure from the rest of our reductions and the other average-case reductions to statistical problems in the literature, all of which are reductions in total variation, as will be described in Section \ref{subsec:2-tvreductions}, and thus only require a single query. This feature is a requirement of our technique for completing hypergraphs that will be described further in Sections \ref{subsec:1-tech-completing} and \ref{sec:2-hypergraph-planting}. 

We note that most formulations of tensor PCA in the literature also assume that the noise tensor of standard Gaussians is symmetric \cite{richard2014statistical, wein2019kikuchi}. However, given that the planted rank-1 component $v^{\otimes s}$ is symmetric as it is in our formulation, the symmetric and asymmetric noise models have a simple equivalence up to a constant factor loss in $\theta$. Averaging the entries of the asymmetric model over all permutations of its $s$ coordinates shows one direction of this equivalence, and the other is achieved by reversing this averaging procedure through Gaussian cloning as in Section 10 of \cite{brennan2018reducibility}. A closely related work is that of \cite{zhang2017tensor}, which gives a reduction from $\pr{hpc}^3$ to the problem of detecting a planted rank-1 component in a 3-tensor of Gaussian noise. Aside being obtained through different techniques, their result differs from ours in two ways: (1) the rank-1 components they considered were sparse, rather than sampled from a Rademacher prior; and (2) their reduction necessarily produces asymmetric rank-1 components. Although the limits of tensor PCA when $s \ge 3$ with sparse and Rademacher priors are similar, they can be very different in other problems. For example, in the matrix case when $s = 2$, a sparse prior yields a problem with a statistical-computational gap while a Rademacher prior does not. We also remark that ensuring the symmetry of the planted rank-1 component is a technically difficult step and part of the motivation for our completing hypergraphs technique in Section \ref{sec:2-hypergraph-planting}.

\subsection{Universality for Learning Sparse Mixtures}
\label{subsec:1-problems-universality}

When $\epsilon = 1/2$, our reduction to robust sparse mean estimation also implicitly shows tight computational lower bounds at $n = \tilde{o}(k^2/\tau^4)$ for learning sparse Gaussian mixtures. In this problem the task is to estimate two vectors $\mu_1, \mu_2$ up to $\ell_2$ error $\tau$, where the $\mu_i$ have bounded $\ell_2$ norms and a $k$-sparse difference $\mu_1 - \mu_2$, given samples from an even mixture of $\mN(\mu_1, I_d)$ and $\mN(\mu_2, I_d)$. In general, learning in Gaussian mixture models with sparsity has been studied extensively over the past two decades \cite{raftery2006variable, pan2007penalized, maugis2009variable, maugis2011non, azizyan2013minimax, azizyan2015efficient, malsiner2016model, verzelen2017detection, fan2018curse}. Recent work has established finite-sample guarantees for efficient and inefficient algorithms and proven information-theoretic lower bounds for the two-component case \cite{azizyan2013minimax, verzelen2017detection, fan2018curse}. These works conjectured that this problem has the $k$-to-$k^2$ statistical-computational gap shown by our reduction. In \cite{fan2018curse}, a tight computational lower bound matching ours was established in the statistical query model.

So far, despite having a variety of different hidden structures, the problems we have considered have all had either Gaussian or Bernoulli noise distributions. As we will describe in Section \ref{sec:1-techniques}, our techniques also crucially use a number of properties of the Gaussian distribution. This naturally raises the question: do our techniques have implications beyond simple noise distributions? Our final reduction answers this affirmatively, showing that our lower bound for learning sparse Gaussian mixtures implies computational lower bounds for a wide universality class of noise distributions. This lower bound includes the optimal gap in learning sparse Gaussian mixtures and the optimal gaps in \cite{berthet2013optimal, berthet2013complexity, wang2016statistical, gao2017sparse, brennan2018reducibility} for sparse PCA as special cases. This reduction requires introducing a new type of rejection kernel, that we refer to as symmetric 3-ary rejection kernels, and is described in Sections \ref{subsec:1-tech-universality} and \ref{subsec:srk}.

In Section \ref{sec:universality}, we show computational lower bounds for the \textit{generalized learning sparse mixtures} problem $\pr{glsm}$. In $\pr{glsm}(n, k, d, \mU)$ where $\mathcal{U} = (\mD, \mQ, \{ \mP_{\nu} \}_{\nu \in \mathbb{R}})$, the elements of the family $\{\mP_{\nu}\}_{\nu \in \mathbb{R}}$ and $\mQ$ are distributions on a measurable space, such that the pairs $(\mP_{\nu}, \mQ)$ all satisfy mild conditions permitting efficient computation outlined in Section \ref{subsec:srk}, and $\mD$ is a mixture distribution on $\mathbb{R}$. The observations in $\pr{glsm}$ are $n$ independent samples $X_1, X_2, \dots, X_n$ formed as follows:
\begin{itemize}
\item for each sample $X_i$, draw some latent variable $\nu_i \sim \mD$ and
\item sample $(X_i)_j \sim \mP_{\nu_i}$ if $j \in S$ and $(X_i)_j \sim \mQ$ otherwise, independently
\end{itemize}
where $S$ is some unknown subset containing $k$ of the $d$ coordinates. The task is to recover $S$ or distinguish from an $H_0$ in which all of the data is drawn i.i.d. from $\mQ$. Given a collection of distributions $\mU$, we define $\mU$ to be in our universality class $\pr{uc}(N)$ with level of signal $\tau_{\mU}$ if it satisfies the following conditions.

\begin{definition}[Universality Class and Level of Signal] \label{defn:univ-signal}
Given a parameter $N$, define the collection of distributions $\mathcal{U} = (\mD, \mQ, \{ \mP_{\nu} \}_{\nu \in \mathbb{R}})$ implicitly parameterized by $N$ to be in the universality class $\pr{uc}(N)$ if
\begin{itemize}
\item the pairs $(\mP_{\nu}, \mQ)$ are all computable pairs, as in Definition \ref{def:computable}, for all $\nu \in \mathbb{R}$;
\item $\mD$ is a symmetric distribution about zero and $\bP_{\nu \sim \mD}[\nu \in [-1, 1]] = 1 - o(N^{-1})$; and
\item there is a level of signal $\tau_{\mathcal{U}} \in \mathbb{R}$ such that for all $\nu \in [-1, 1]$ such that for any fixed constant $K > 0$, it holds that
$$\left| \frac{d\mP_{\nu}}{d\mQ} (x) - \frac{d\mP_{-\nu}}{d\mQ} (x) \right| = O_N\left(\tau_{\mathcal{U}} \right) \quad \textnormal{and} \quad \left|\frac{d\mP_{\nu}}{d\mQ} (x) + \frac{d\mP_{-\nu}}{d\mQ} (x) - 2 \right| = O_N\left( \tau_{\mathcal{U}}^2 \right)$$
with probability at least $1 - O\left(N^{-K}\right)$ over each of $\mP_{\nu}, \mP_{-\nu}$ and $\mQ$.
\end{itemize}
\end{definition}

Our main result establishes a computational lower bound for $\pr{glsm}$ instances with $\mU \in \pr{uc}(n)$ in terms of the level of signal $\tau_{\mU}$. As mentioned above, this theorem implies optimal lower bounds for learning sparse mixtures of Gaussians, sparse PCA and many more natural problem formulations described in Section \ref{subsec:universalitydiscussion}.

\begin{theorem}[Computational Lower Bounds for \pr{glsm}] \label{thm:glsm-lb}
Let $n, k$ and $d$ be polynomial in each other and such that $k = o(\sqrt{d})$. Suppose that the collections of distributions $\mU = (\mD, \mQ, \{ \mP_{\nu} \}_{\nu \in \mathbb{R}})$ is in $\pr{uc}(n)$. Then the $k\pr{-bpc}$ conjecture implies a computational lower bound for $\pr{glsm}\left(n, k, d, \mU \right)$ at all sample complexities $n = \tilde{o}\left(\tau_{\mU}^{-4}\right)$.
\end{theorem}


\section{Technical Overview}
\label{sec:1-techniques}

We now outline our main technical contributions and the central ideas behind our reductions. These techniques will be formally introduced in Part \ref{part:reductions} and applied in our problem-specific reductions to deduce our main theorems stated in the previous section in Part \ref{part:lower-bounds}.

\subsection{Rejection Kernels}
\label{subsec:1-tech-rk}

Rejection kernels are a reduction primitive introduced in \cite{brennan2018reducibility, brennan2019universality} for \textit{algorithmic changes of measure}. Related reduction primitives for changes of measure to Gaussians and binomial random variables appeared earlier in \cite{ma2015computational, hajek2015computational}. Given two input Bernoulli probabilities $0 < q < p \le 1$, a rejection kernel simultaneously maps $\text{Bern}(p)$ and $\text{Bern}(q)$ approximately in total variation to samples from two arbitrary distributions $\mP$ and $\mQ$. Note that in this setup, the rejection kernel primitive is oblivious to whether the true distribution of its input is $\text{Bern}(p)$ or $\text{Bern}(q)$. The main idea behind rejection kernels is that, under suitable conditions on $\mP$ and $\mQ$, this can be achieved through a rejection sampling scheme that samples $x \sim \mQ$ and rejects with a probability that depends on $x$ and on whether the input was $0$ or $1$. Rejection kernels are discussed in more depth in Section \ref{sec:2-rejection-kernels}. In this work, we will need the following two instantiations of the framework developed in \cite{brennan2018reducibility, brennan2019universality}:
\begin{itemize}
\item \textit{Gaussian Rejection Kernels:} Rejection kernels mapping $\text{Bern}(p)$ and $\text{Bern}(q)$ to within $O(R_{\pr{rk}})$ total variation of $\mN(\mu, 1)$ and $\mN(0, 1)$ where $\mu = \Theta\left(1/\sqrt{\log R_{\text{rk}}^{-1}}\right)$ and $p, q$ are fixed constants.
\item \textit{Bernoulli Cloning:} A rejection kernel mapping $\text{Bern}(p)$ and $\text{Bern}(q)$ exactly to $\text{Bern}(P)^{\otimes t}$ and $\text{Bern}(Q)^{\otimes t}$ where
$$\frac{1 - p}{1 - q} \le \left( \frac{1 - P}{1 - Q} \right)^t \quad \text{and} \quad \left( \frac{P}{Q} \right)^t \le \frac{p}{q}$$
\end{itemize}
By performing computational changes of measure, these primitives are crucial in mapping to desired distributional aesthetics. However, they also play an important role in transforming hidden structure. Gaussian rejection kernels grant access to an arsenal of measure-preserving transformations of high-dimensional Gaussian vectors for mapping between different hidden structures while preserving independence in the noise distribution. Bernoulli cloning is crucial in removing the symmetry in adjacency matrices of $\pr{pc}$ instances and adjacency tensors of $\pr{hpc}$ instances, as in the $\pr{To-Submatrix}$ procedure in \cite{brennan2019universality}. We introduce a $k$-partite variant of this procedure that maps the adjacency matrix of $k\pr{-pds}$ to a matrix of independent Bernoulli random variables while respecting the constraint that there is one planted entry per block of the $k$-partition. This procedure is discussed in more detail in Section \ref{subsec:1-tech-completing} and will serve as a crucial preprocessing step for dense Bernoulli rotations, which involves taking linear combinations of functions of entries of this matrix that crucially must be independent.

\subsection{Dense Bernoulli Rotations}
\label{subsec:1-tech-dbr}

This technique is introduced in Section \ref{sec:2-bernoulli-rotations} and is one of our main primitives for \textit{transforming hidden structure} that will be applied repeatedly throughout our reductions. Let $\pr{pb}(n, i, p, q)$ denote the planted bit distribution over $V \in \{0, 1\}^n$ with independent entries satisfying that $V_j \sim \text{Bern}(q)$ unless $j = i$, in which case $V_i \sim \text{Bern}(p)$. Given an input vector $V \in \{0, 1\}^n$, the goal of dense Bernoulli rotations is to output a vector $V' \in \mathbb{R}^m$ such that, for each $i \in [n]$, $V'$ is close in total variation to $\mN(c \cdot A_i, I_m)$ if $V \sim \pr{pb}(n, i, p, q)$. Here, $A_1, A_2, \dots, A_n \in \mathbb{R}^m$ are a given sequence of target mean vectors, $p$ and $q$ are fixed constants and $c$ is a scaling factor with $c = \tilde{\Theta}(1)$. The reduction must satisfy these approximate Markov transition conditions oblivious to the planted bit $i$ and also preserve independent noise, by mapping $\text{Bern}(q)^{\otimes n}$ to $\mN(0, I_m)$ approximately in total variation.

Let $A \in \mathbb{R}^{m \times n}$ denote the matrix with columns $A_1, A_2, \dots, A_n$. If the rows of $A$ are orthogonal unit vectors, then the goal outlined above can be achieved using the isotropy of the distribution $\mN(0, I_n)$. More precisely, consider the reduction that form $V_1 \in \mathbb{R}^n$ by applying Gaussian rejection kernels entrywise to $V$ and then outputs $AV_1$. If $V \sim \pr{pb}(n, i, p, q)$, then the rejection kernels ensure that $V_1$ is close in total variation to $\mN(\mu \cdot \mathbf{1}_i, I_n)$ and thus $V' = AV_1$ is close to $\mN(\mu \cdot A_i, I_m)$. However, if the rows of $A$ are not orthogonal, then the entries of the output are potentially very dependent and have covariance matrix $AA^\top$ instead of $I_m$. This can be remedied by adding a \textit{noise-correction term} to the output: generate $U \sim \mN(0, I_m)$ and instead output
$$V' = \lambda^{-1} \cdot AV_1 + \left( I_m - \lambda^{-2} \cdot AA^\top \right)^{1/2} \cdot U$$
where $\lambda$ is an upper bound on the largest singular value of $A$ and $\left( I_m - \lambda^{-2} AA^\top \right)^{1/2}$ is the positive semidefinite square root of $I_m - \lambda^{-2} \cdot AA^\top$. If $V \sim \pr{pb}(n, i, p, q)$, it now follows that $V'$ is close in total variation to $\mN(\mu \lambda^{-1} \cdot A_i, I_m)$ where $\mu$ can be taken to be $\mu = \Theta(1/\sqrt{\log n})$. This reduction also preserves independent noise, mapping $\text{Bern}(q)^{\otimes n}$ approximately to $\mN(0, I_m)$.

Dense Bernoulli rotations thus begin with a random vector of independent entries and one unknown elevated bit and produce a vector with independent entries and an unknown elevated \textit{pattern} from among an arbitrary prescribed set $A_1, A_2, \dots, A_n$. Furthermore, the dependence of the signal strength $\mu \lambda^{-1}$ in the output instance $V'$ on these $A_1, A_2, \dots, A_n$ is entirely through the singular values of $A$. This yields a general structure-transforming primitive that will be used throughout our reductions. Each such use will consist of many local applications of dense Bernoulli rotations that will be stitched together to produce a target distribution. These local applications will take three forms:
\begin{itemize}
\item \textit{To Rows Restricted to Column Parts:} The adjacency matrix of $k\pr{-bpc}$ consists of $k_n k_m$ blocks each consisting of the edge indicators in $E_i \times F_j$ for each pair of the parts $E_i, F_j$ from the given partitions of $[n]$ and $[m]$. In our reductions to robust sparse mean estimation, mixtures of SLRs, robust SLR and universality for learning sparse mixtures, we apply dense Bernoulli rotations separately to each row in each of these blocks.
\item \textit{To Vectorized Adjacency Matrix Blocks:} In our reductions to dense stochastic block models, testing hidden partition models and semirandom single community detection, we first pre-process the adjacency matrix of $k\pr{-pc}$ with $\pr{To-}k\pr{-Partite-Submatrix}$. We then apply dense Bernoulli rotations to $\mathbb{R}^{h^2}$ vectorizations of each $h \times h$ block in this matrix corresponding a pair of parts in the given partition i.e. of the form $E_i \times E_j$.
\item \textit{To Vectorized Adjacency Tensor Blocks:} In our reduction to tensor PCA with order $s$, after completing the adjacency tensor of the input $k\pr{-hpc}$ instance, we apply dense Bernoulli rotations to $\mathbb{R}^{h^s}$ vectorizations of each $h \times h \times \cdots \times h$ block corresponding to an $s$-tuple of parts. 
\end{itemize}
We remark that while dense Bernoulli rotations heavily rely on distributional properties of isotropic Gaussian vectors, their implications extend far beyond statistical problems with Gaussian noise. Entrywise thresholding produces planted graph problems and we will show that multiple thresholds followed by applying 3-ary symmetric rejection kernels maps to a large universality class of noise distributions. These applications of dense Bernoulli rotations generally reduce the problem of transforming hidden structure to a constrained combinatorial construction problem -- the task of designing a set of mean output vectors $A_1, A_2, \dots, A_n$ that have nearly orthogonal rows and match the combinatorial structure in the target statistical problem.

\subsection{Design Matrices and Tensors}
\label{subsec:1-tech-design-matrices}

\paragraph{Design Matrices.} To construct these vectors $A_1, A_2, \dots, A_n$ for our applications of dense Bernoulli rotations, we introduce several families of matrices based on the incidence geometry of finite fields. In our reduction to robust sparse mean estimation, we will show that the adversary that corrupts an $\epsilon$-fraction of the samples by resampling them from $\mN(-c \cdot \mu, I_d)$ produces the desired $k$-to-$k^2$ statistical-computational gap. This same adversarial construction was used in \cite{diakonikolas2017statistical}. Here, $\mu \in \mathbb{R}^d$ denotes the $k$-sparse mean of interest. As will be further discussed at the beginning of Section \ref{sec:2-bernoulli-rotations}, on applying dense Bernoulli rotations to rows restricted to parts of the partition of column partition, our desiderata for the mean vectors $A_1, A_2, \dots, A_n$ reduce to the following:
\begin{itemize}
\item $A$ contains two distinct values $\{x, y\}$, and an $\epsilon'$-fraction of each column is $y$ where $\epsilon \ge \epsilon' = \Theta(\epsilon)$;
\item the rows of $A$ are unit vectors and nearly orthogonal with $\lambda = O(1)$; and
\item $A$ is nearly an isometry as a linear transformation from $\mathbb{R}^n \to \mathbb{R}^m$.
\end{itemize}
The first criterion above is enough to ensure the correct distributional aesthetics and hidden structure in the output of our reduction. The second and third criteria turn out to be necessary and sufficient for the reduction to show tight computational lower bounds up to the conjectured barrier of $n = \tilde{o}(k^2 \epsilon^2/\tau^4)$. We remark that the third criterion also is equivalent to $m = \tilde{\Theta}(n)$ given the second. Thus our task is to design nearly square, nearly orthogonal matrices containing two distinct entries with an $\epsilon'$-fraction of one present in each column. Note that if $\epsilon = 1/2$, this is exactly achieved by Hadamard matrices. For $\epsilon < 1/2$, our desiderata are nearly met by the following natural generalization of Hadamard matrices that we introduce. Note that the rows of a Hadamard matrix can be generated as a reweighted incidence matrix between the hyperplanes and points of $\mathbb{F}_2^t$. Let $r$ be a prime number with $\epsilon^{-1} \le r = O(\epsilon^{-1})$ and consider the $\ell \times r^t$ matrix $A$ where $\ell = \frac{r^t - 1}{r - 1}$ with entries given by
$$A_{ij} = \frac{1}{\sqrt{r^t(r - 1)}} \cdot \left\{ \begin{matrix} 1 & \textnormal{if } P_j \not \in V_i \\ 1 - r & \textnormal{if } P_j \in V_i \end{matrix} \right.$$
where $V_1, V_2, \dots, V_{\ell}$ is an enumeration of the $(t - 1)$-dimensional subspaces of $\mathbb{F}_r^t$ and $P_1, P_2, \dots, P_{r^t}$ is an enumeration of the points in $\mathbb{F}_r^t$. This construction nearly meets our three criteria, with one minor issue that the column corresponding to $0 \in \mathbb{F}_r^t$ only contains one entry. A more serious issue is that $\ell = \Theta(r^{t - 1})$ and $A$ is far from an isometry if $r \gg 1$, which leads to a suboptimal computational lower bound for $\pr{rsme}$.

These issues are both remedied by adding in additional rows for all affine shifts of the hyperplanes $V_1, V_2, \dots, V_{\ell}$. The resulting matrix has dimensions $r\ell \times r^t$ and, although its rows are no longer orthogonal, its largest singular value is $\sqrt{1 + (r - 1)^{-1}}$. The resulting matrix $K_{r, t}$ is used in our applications of dense Bernoulli rotations to reduce to robust sparse mean estimation, mixtures of SLRs, robust SLR and to show universality for learning sparse mixtures. Note that for any two rows $r_i$ and $r_j$ of $K_{r, t}$, the outer product $r_i r_j^\top$ is a zero-centered mean adjacency matrix of an imbalanced 2-block stochastic block model. This observation suggests that the Kronecker product $K_{r, t} \otimes K_{r, t}$ can be used in dense Bernoulli rotations to map to these SBMs. Surprisingly, this overall reduction yields tight computational lower bounds up to the Kesten-Stigum threshold for dense SBMs, and using the matrix $(K_{3, t} \otimes I_s) \otimes (K_{3, t} \otimes I_s)$ yields tight computational lower bounds for semirandom single community detection. We remark that, in this case, it is again crucial that $K_{r, t}$ is approximately square -- if the matrix $A$ defined above were used in place of $K_{r, t}$, our reduction would show a lower bound suboptimal to the Kesten-Stigum threshold by a factor of $r$. Our reduction to order $s$ tensor PCA applies dense Bernoulli rotations to vectorizations of each tensor block with the $s$th order Kronecker product $K_{2, t} \otimes K_{2, t} \otimes \cdots \otimes K_{2, t}$. We remark that these instances of $K_{2, t}$ in this Kronecker product could be replaced by Hadamard matrices in dimension $2^t$.

In Section \ref{subsec:2-Rne}, we introduce a natural alternative to $K_{r, t}$ -- a random matrix $R_{n, \epsilon}$ that \textit{approximately} satisfies the three desiderata above. In our reductions to $\pr{rsme}$ and $\pr{rslr}$, this random matrix has the advantage of eliminating the number-theoretic condition \pr{(t)} arising from applying dense Bernoulli rotations with $K_{r, t}$, which has nontrivial restrictions in the very small $\epsilon$ regime when $\epsilon = n^{-\Omega(1)}$. However, the approximate properties of $R_{n, \epsilon}$ are insufficient to map exactly to our formulations of $\pr{isbm}, \pr{semi-cr}, \pr{ghpm}$ and $\pr{bhpm}$, where the sizes of the hidden communities are known. A more detailed comparison of $K_{r, t}$ and $R_{n, \epsilon}$ can be found in Section \ref{subsec:2-Rne}. The random matrix $R_{n, \epsilon}$ is closely related to the adjacency matrices of sparse random graphs, and establishing $\lambda = O(1)$ requires results on their spectral concentration from the literature. For a consistent and self-contained exposition, we present our reductions with $K_{r, t}$, which has a comparatively simple analysis, and only outline extensions of our reductions using $R_{n, \epsilon}$.

\paragraph{Design Tensors.}  Our final reduction using dense Bernoulli rotations is to testing hidden partition models. This reduction requires a more involved construction for $A$ that we only sketch here and defer a detailed discussion to Section \ref{subsec:2-design-tensors}. Again applying dense Bernoulli rotations to vectorizations of each block of the input $k\pr{-pc}$ instance, our goal is to construct a tensor $T_{r, t}$ such that each slice has the same block structure as an $r$-block SBM and the slices are approximately orthogonal under the matrix inner product. A natural construction is as follows: index each slice by a pair of hyperplanes $(V_i, V_j)$, label the rows and columns of each slice by $\mathbb{F}_r^t$ and plant $r$ communities on the entries with indices in $(V_i + au_i) \times (V_j + au_j)$ for each $a \in \mathbb{F}_r$. Here $u_i$ and $u_j$ are arbitrary vectors not in $V_i$ and $V_j$, respectively, and thus $V_i + au_i$ ranges over all affine shifts of $V_i$ for $a \in \mathbb{F}_r$. An appropriate choice of weights $x$ and $y$ on and off of these communities yields slices that are exactly orthogonal.

However, this construction suffers from the same issue as the construction of $A$ above -- there are $O(r^{2t - 2})$ slices each of which has $r^{2t}$ entries, making the matrix formed by vectorizing the slices of this tensor far from square. This can be remedied by creating additional slices further indexed by a nonconstant linear function $L : \mathbb{F}_r \to \mathbb{F}_r$ such that communities are now planted on $(V_i + au_i) \times (V_j + L(a) \cdot u_j)$ for each $a \in \mathbb{F}_r$. There are $r(r - 1)$ such linear functions $L$, making the vectorization of this tensor nearly square. Furthermore, it is shown in Section \ref{subsec:2-design-tensors} that this matrix has largest singular value $\sqrt{1 + (r - 1)^{-1}}$. We remark that this property is quite brittle, as substituting other families of bijections for $L$ can cause this largest singular value to increase dramatically. Taking the Kronecker product of each slice of this tensor $T_{r, t}$ with $I_s$ now yields the family of matrices used in our reduction to testing hidden partition models.

We remark that in all of these reductions with both design matrices and design tensors, dense Bernoulli rotations are applied locally within the blocks induced by the partition accompanying the $\pr{pc}_\rho$ instance. In all cases, our constructions ensure that the fact that the planted bits within these blocks take the form of a submatrix is sufficient to stitch together the outputs of these local applications of dense Bernoulli rotations into a single instance with the desired hidden structure. While we did not discuss this constraint in choosing the design matrices $A$ for each of our reductions, it will be a key consideration in the proofs throughout this work. Surprisingly, the linear functions $L$ in the construction of $T_{r, t}$ directly lead to a community alignment property proven in Section \ref{subsec:2-design-tensors} that allow slices of this tensor to be consistently stitched together. Furthermore, we note that unlike $K_{r, t}$, the tensor $T_{r, t}$ does not seem to have a random matrix analogue that is tractable to bound in spectral norm.

\paragraph{Parameter Correspondence with Dense Bernoulli Rotations.} In several of our reductions using dense Bernoulli rotations, a simple heuristic predicts our computational lower bound in the target problem. Let $X$ be a data tensor, normalized and centered so that each entry has mean zero and variance $1$, and then consider the $\ell_2$ norm of the expected tensor $\bE[X]$. Our applications of rejection kernels typically preserve this $\ell_2$ norm up to $\text{polylog}(n)$ factors. Since our design matrices are approximate isometries, most of our applications of dense Bernoulli rotations also approximately preserve this $\ell_2$ norm. Thus comparing the $\ell_2$ norms of the input $\pr{pc}_\rho$ instance and output instance in our reductions yields a heuristic for predicting the resulting computational lower bound. For example, our adversary in $\pr{rsme}$ produces a matrix $\bE[X] \in \mathbb{R}^{d \times n}$ consisting of columns of the form $\tau \cdot k^{-1/2} \cdot \mathbf{1}_S$ and $\epsilon^{-1} (1 - \epsilon)\tau \cdot k^{-1/2} \cdot \mathbf{1}_S$, up to constant factors where $S$ is the hidden support of $\mu$. The $\ell_2$ norm of this matrix is $\Theta(\tau \sqrt{n/\epsilon})$. The $\ell_2$ norm of the matrix $\bE[X]$ corresponding to the starting $k\pr{-bpc}$ instance can be verified to be just below $o(k^{1/2} n^{1/4})$, when the $k\pr{-bpc}$ instance is nearly at its computational barrier. Equating these two $\ell_2$ norms yields the relation $n = \Theta(k^2 \epsilon^2/\tau^4)$, which is exactly our computational barrier for $\pr{rsme}$. Similar heuristic derivations of our computational barriers are produced for $\pr{isbm}$, $\pr{ghpm}$, $\pr{bhpm}$, $\pr{semi-cr}$ and $\pr{tpca}$ at the beginnings of Sections \ref{sec:3-all-community} and \ref{sec:3-tensor}. We remark that for some of our problems with central steps other than dense Bernoulli rotations, such as $\pr{mslr}$, $\pr{rslr}$ and $\pr{glsm}$, this heuristic does not apply. 

\subsection{Decomposing Linear Regression and Label Generation}
\label{subsec:1-tech-decomposing}

Our reductions to mixtures of SLRs and robust SLR in Section \ref{sec:2-supervised} are motivated by the following simple initial observation. Suppose $(X, y)$ is a single sample from unsigned SLR with $y = \gamma R \cdot \langle v, X \rangle + \mN(0, 1)$ where $R \in \{-1, 1\}$ is a Rademacher random variable, $v \in \mathbb{R}^d$ is a $k$-sparse unit vector, $X \sim \mN(0, I_d)$ and $\gamma \in (0, 1)$. A standard conditioning property of Gaussian vectors yields that the conditional distribution of $X$ given $R$ and $y$ is another jointly Gaussian vector, as shown below. Our observation is that this conditional distribution can be decomposed into a sum of our adversarial construction for robust sparse mean estimation and an independent instance of negative sparse PCA. More formally, we have that
\begin{align*}
X | R, y &\sim \mN\left( \frac{R\gamma \cdot y}{1 + \gamma^2} \cdot v, \, I_d - \frac{\gamma^2}{1 + \gamma^2} \cdot vv^\top \right) \\
&\sim \underbrace{\frac{1}{\sqrt{2}} \cdot \mN\left( R\tau \cdot v, \, I_d \right)}_{\text{Our } \pr{rsme} \text{ adversary with } \epsilon \, = \, 1/2} + \, \, \, \, \underbrace{\frac{1}{\sqrt{2}} \cdot \mN\left( 0, \, I_d - \theta vv^\top \right)}_{\text{Negative Sparse PCA}}
\end{align*}
where $\tau = \tau(y) = \frac{\gamma \sqrt{2}}{1 + \gamma^2} \cdot y$ and $\theta = \frac{2\gamma^2}{1 + \gamma^2}$. Note that the marginal distribution of $y$ is $\mN(0, 1 + \gamma^2)$ and thus it typically holds that $|y| = \Theta(1)$. When this unsigned SLR instance is at its computational barrier of $n = \tilde{\Theta}(k^2/\gamma^4)$ and $|y| = \Theta(1)$, then $n = \tilde{\Theta}(k^2/\tau^4)$ and $\theta = \tilde{\Theta}(\sqrt{k^2/n})$. Therefore surprisingly, both of the $\pr{rsme}$ and $\pr{neg-spca}$ in the decomposition above are also at their computational barriers. 

Now consider task of instead reducing from $k\pr{-bpc}$ to the problem of estimating $v$ from $n$ independent samples from the conditional distribution $\mL(X | \, |y| = 1)$. In light of the observations above, it suffices to first use Bernoulli cloning to produce two independent copies of $k\pr{-bpc}$, reduce these two copies as outlined below and then take the sum of the two outputs of these reductions.
\begin{itemize}
\item \textit{Producing Our} \pr{rsme} \textit{Adversary}: One of the two copies of $k\pr{-bpc}$ should be mapped to a tight instance of our adversarial construction for $\pr{rsme}$ with $\epsilon = 1/2$ through local applications of dense Bernoulli rotations with design matrix $K_{r, t}$ or $R_{n, \epsilon}$, as described previously.
\item \textit{Producing} \pr{neg-spca}: The other copy should be mapped to a tight instance of negative sparse PCA. This requires producing negatively correlated data from positively correlated data, and will need new techniques that we discuss next.
\end{itemize}
We remark that while these two output instances must be independent, it is important that they share the same latent vector $v$. Bernoulli cloning ensures that the two independent copies of $k\pr{-pc}$ have the same clique vertices and thus the output instances have this desired property.

This reduction can be extended to reduce to the true joint distribution of $(X, y)$ as follows. Consider replacing each sample $X_1$ of the output $\pr{rsme}$ instance by
$$X_2 = cy \cdot X_1 + \sqrt{1 - c^2y^2} \cdot \mN(0, I_d)$$
where $c$ is some scaling factor and $y$ is independently sampled from $\mN(0, 1 + \gamma^2)$, truncated to lie in the interval $[-T, T]$ where $cT \le 1$. Observe that if $X_1 \sim \mN(R\tau \cdot v, I_d)$, then $X_2 \sim \mN(cR\tau y \cdot v, I_d)$ conditioned on $y$. In Section \ref{subsec:2-mixtures-slr}, we show that a suitable choice of $c, T$ and tweaking $\tau$ in the reduction above tightly maps to the desired distribution of mixtures of SLRs. Analogous observations and performing the $\pr{rsme}$ sub-reduction with $\epsilon < 1/2$ can be used to show tight computational lower bounds for robust SLR. We remark that this produces a more complicated adversarial construction for robust SLR that may be of independent interest. The details of this adversary can be found in Section \ref{subsec:2-mixtures-slr}.

\subsection{Producing Negative Correlations and Inverse Wishart Matrices}
\label{subsec:1-tech-inverse-wishart}

To complete our reductions to mixtures of SLRs and robust SLR, it suffices to give a tight reduction from $k\pr{-bpc}$ to $\pr{neg-spca}$. Although $\pr{neg-spca}$ and ordinary $\pr{spca}$ share the same conjectured computational barrier at $\theta = \Theta(\sqrt{k^2/n})$ and can be solved by similar efficient algorithms above this barrier, as stochastic models, the two are very different. As discussed in Section \ref{subsec:1-problems-negspca}, ordinary $\pr{spca}$ admits a signal plus noise representation while $\pr{neg-spca}$ does not. This representation was crucially used in prior reductions showing optimal computational lower bounds for $\pr{spca}$ in \cite{berthet2013optimal, berthet2013complexity, wang2016statistical, gao2017sparse, brennan2018reducibility, brennan2019optimal}. Furthermore, the planted entries in a $\pr{neg-spca}$ sample are \textit{negatively correlated}. In contrast, the edge indicators of $\pr{pc}_\rho$ are positively correlated and all prior reductions from $\pr{pc}$ have only produced hidden structure that is also positively correlated.

We first simplify the task of reducing to $\pr{neg-spca}$ with an observation used in the reduction to $\pr{spca}$ in \cite{brennan2019optimal}. Suppose that $n \ge m + 1$ and let $m$ be such that $m/k^2$ tends slowly to infinity. If $X$ is an $m \times n$ matrix with columns $X_1, X_2, \dots, X_n \sim_{\text{i.i.d.}} \mN(0, \Sigma)$ where $\Sigma \in \mathbb{R}^{m \times m}$ is positive semidefinite, then the conditional distribution of $X$ given its rescaled empirical covariance matrix $\hat{\Sigma} = \sum_{i = 1}^n X_i X_i^\top$ is $\hat{\Sigma}^{1/2} R$ where $R$ is an independent $m \times n$ matrix sampled from Haar measure on the Stiefel manifold. This implies that it suffices to reduce to $\hat{\Sigma}$ in the case where $\Sigma = I_d - \theta vv^\top$ in order to map to $\pr{neg-spca}$, as $X$ can be generated from $\hat{\Sigma}$ by randomly sampling this Haar measure. This measure can then be sampled efficiently by applying Gram-Schmidt to the rows of an $m \times n$ matrix of independent standard Gaussians.

Let $\mW_m(n, \Sigma)$ be the law of $\hat{\Sigma}$, or in other words the Wishart distribution with covariance matrix $\Sigma$, and let $\mW_m^{-1}(n, \Sigma)$ denote the distribution of its inverse. The matrices $\mW_m(n, \Sigma)$ and $\mW_m^{-1}(n, \beta \cdot \Sigma^{-1})$ where $\beta^{-1} = n(n - m - 1)$ have a number of common properties including close low-order moments. Furthermore, if $\Sigma = I_d - \theta vv^\top$ then $\Sigma^{-1} = I_d + \theta' vv^\top$ where $\theta' = \frac{\theta}{1 - \theta}$, which implies that $\mW_m^{-1}(n, \beta \cdot \Sigma^{-1})$ is a rescaling of the inverse of the empirical covariance matrix of a set of samples from ordinary $\pr{spca}$. This motivates our main reduction to $\pr{neg-spca}$ in Section \ref{subsec:2-neg-spca-reduction}, which roughly proceeds in the following two steps.
\begin{enumerate}
\item Begin with a small instance of $\pr{bpc}$ with $m = \omega(k^2)$ vertices on the left and $n$ on the right. Apply either the reduction of \cite{brennan2018reducibility} or \cite{brennan2019optimal} to reduce to an ordinary $\pr{spca}$ instance $(X_1, X_2, \dots, X_n)$ in dimension $m$ with $n$ samples and signal strength $\theta'$.
\item Form the rescaled empirical covariance matrix $\hat{\Sigma} = \sum_{i = 1}^n X_i X_i^\top$ and
$$Y = \sqrt{n(n - m - 1)} \cdot \hat{\Sigma}^{-1/2} R$$
Output the columns of $Y$ after padding them to be $d$-dimensional with i.i.d. $\mN(0, 1)$ random variables.
\end{enumerate}
The key detail in this reduction is that $\hat{\Sigma}^{1/2}$ in process of regenerating $X$ from $\hat{\Sigma}$ described above has been replaced by the positive semidefinite square root $\hat{\Sigma}^{-1/2}$ of a rescaling of the empirical precision matrix. As we will show in Section \ref{subsec:2-neg-spca-reduction}, establishing total variation guarantees for this reduction amounts to answering the following nonasymptotic question from random matrix theory that may be of independent interest: when do $\mW_m(n, \Sigma)$ and $\mW_m^{-1}(n, \beta \cdot \Sigma^{-1})$ converge in total variation for all positive semidefinite matrices $\Sigma$? A simple reduction shows that the general case is equivalent to the isotropic case when $\Sigma = I_m$. In Section \ref{subsec:2-inverse-wishart}, we answer this question, showing that these two matrices converge in KL divergence if and only if $n \gg m^3$. Our result is of the same flavor as a number of recent results in random matrix theory showing convergence in total variation between Wishart and $\pr{goe}$ matrices \cite{jiang2015approximation, bubeck2016testing, bubeck2016entropic, racz2019smooth}. This condition amounts to constraining our reduction to the low-sparsity regime $k \ll n^{1/6}$. As discussed in Section \ref{subsec:1-problems-negspca}, this condition does not affect the tightness of our lower bounds and seems to be an artefact of our techniques that possibly can be removed.

\subsection{Completing Tensors from Hypergraphs and Tensor PCA}
\label{subsec:1-tech-completing}

As alluded to in the above discussion of rejection kernels, it is important that the entries in the vectors to which we apply dense Bernoulli rotations are independent and that none of these entries is missing. In the context of reductions beginning with $k\pr{-pc}$, $k\pr{-hpc}$, $\pr{pc}$ and $\pr{hpc}$, establishing this entails pre-processing steps to remove the symmetry of the input adjacency matrix and add in missing entries. As discussed in Section 1.1 of \cite{brennan2019universality}, these missing entries in the matrix case have led to technical complications in the prior reductions in \cite{hajek2015computational, brennan2018reducibility, brennan2019universality, brennan2019optimal}. In reductions to tensor PCA, completing these pre-processing steps in the tensor case seems unavoidable in order to produce the canonical formulation of tensor PCA with a symmetric rank-1 spike $v^{\otimes s}$ as discussed in Section \ref{subsec:1-problems-tpca}.

In order to motivate our discussion of the tensor case, we first consider the matrix case. Asymmetrizing the adjacency matrix of an input $\pr{pc}$ instance can be achieved through a simple application of Bernoulli cloning, but adding in the missing diagonal entries is more subtle. Note that the desired diagonal entries contain nontrivial information about the vertices in the planted clique -- they are constrained to be $1$ along the vertices of the clique and independent $\text{Bern}(1/2)$ random variables elsewhere. This is roughly the information gained on revealing a single vertex from the planted clique. In the matrix case, the following trick effectively produces an instance of $\pr{pc}$ with the diagonal entries present. Add in $1$'s along the entire diagonal and randomly embed the resulting matrix as a principal minor in a larger matrix with off-diagonal entries sampled from $\text{Bern}(1/2)$ and on-diagonal entries sampled so that the total number of $1$'s on the diagonal has the correct binomial distribution. This trick appeared in the $\pr{To-Submatrix}$ procedure in \cite{brennan2019universality} for general $\pr{pds}$ instances, and is adapted in this work for $k\pr{-pds}$ as the reduction $\pr{To-}k\pr{-Partite-Submatrix}$ in Section \ref{sec:2-rejection-kernels}. This reduction is an important pre-processing step in mapping to dense stochastic block models, testing hidden partition models and semirandom planted dense subgraph.

The tensor case is not as simple as the matrix case. While asymmetrizing can be handled similarly with Bernoulli cloning, the missing entries of the adjacency tensor of $\pr{hpc}$ are now more numerous and correspond to any entry with two equal indices. Unlike in the matrix case, the information content in these entries alone is enough to solve $\pr{hpc}$. For example, in 3-uniform $\pr{hpc}$, the missing set of entries $(i, i, j)$ should have the same distribution as the completed adjacency matrix of an entire instance of planted clique with the same hidden clique vertices. Thus a reduction that randomly generates these missing entries as in the matrix case is no longer possible without knowing the solution to the input $\pr{hpc}$ instance. However, if an oracle were to have revealed a single vertex of the hidden clique, we would be able to use the hyperedges containing this vertex to complete the missing entries of the adjacency tensor. In general, given an $\pr{hpc}$ instance of arbitrary order $s$, a more involved cloning and embedding procedure detailed in Section \ref{sec:2-hypergraph-planting} completes the missing entries of the adjacency tensor given oracle access to $s - 1$ vertices of the hidden clique. Our reduction to tensor PCA in Sections \ref{sec:2-hypergraph-planting} and \ref{sec:3-tensor} iterates over all $(s - 1)$-tuples of vertices in the input $\pr{hpc}$ instance, uses this procedure to complete the missing entries of the adjacency tensor, applies dense Bernoulli rotations as described previously and then feeds the output instance to a blackbox solving tensor PCA. The reduction only succeeds in mapping to the correct distribution on tensor PCA in iterations that successfully guess $s - 1$ vertices of the planted clique. However, we show that this is sufficient to deduce tight computational lower bounds for tensor PCA. We remark that this reduction is the first reduction in total variation from $\pr{pc}_{\rho}$ that seems to require multiple calls to a blackbox solving the target problem.

\subsection{Symmetric 3-ary Rejection Kernels and Universality}
\label{subsec:1-tech-universality}

So far, all of our reductions have been to problems with Gaussian or Bernoulli data and our techniques have often relied heavily on the properties of jointly Gaussian vectors. Our last reduction technique shows that the consequences of these reductions extend far beyond Gaussian and Bernoulli problems. We introduce a new rejection kernel in Section \ref{subsec:srk} and show in Section \ref{sec:universality} that, when applied entrywise to the output of our reduction to $\pr{rsme}$ when $\epsilon = 1/2$, this rejection kernel yields a universal computational lower bound for a general variant of learning sparse mixtures with nearly arbitrary marginals.

Because sparse mixture models necessarily involve at least three distinct marginal distributions, a deficit in degrees of freedom implies that the existing framework for rejection kernels with binary entries cannot yield nontrivial hardness. We resolve this issue by considering rejection kernels with a slightly larger input space, and introduce a general framework for 3-ary rejection kernels with entries in $\{-1, 0, 1\}$ in Section \ref{subsec:srk}. We show in Section \ref{sec:universality} that first mapping each entry of our $\pr{rsme}$ instance with $\epsilon = 1/2$ into $\{-1, 0, 1\}$ by thresholding at intervals of the form $(-\infty, -T], (-T, T)$ and $[T, \infty)$ with $T = \Theta(1)$ and then applying 3-ary rejection kernels entrywise is a nearly lossless reduction. In particular, it yields new computational lower bounds for a wide universality class that tightly recover optimal computational lower bounds for sparse PCA, learning mixtures of exponentially distributed data, the original $\pr{rsme}$ instance with $\epsilon = 1/2$ and many other sparse mixture formulations. The implications of this reduction are discussed in detail in Section \ref{subsec:universalitydiscussion}.

\subsection{Encoding Cliques as Structural Priors}
\label{subsec:1-tech-encoding}

As discussed in Section \ref{subsec:1-desiderata}, reductions from $\pr{pc}_\rho$ showing tight computational lower bounds cannot generate a non-negligible part of the hidden structure in the target problem themselves, but instead must encode the hidden clique of the input instance into this structure. In this section, we outline how our reductions implicitly encode hidden cliques. Note that the hidden subset of vertices corresponding to a clique in $\pr{pc}_{\rho}$ has $\Theta(k \log n)$ bits of entropy while the distribution over the hidden structure in the target problems that we consider can have much higher entropy. For example, the Rademacher prior on the planted vector $v$ in Tensor PCA has $n$ bits of entropy and the distribution over hidden partitions in testing partition models has entropy $\Theta(r^2 K^2 \log n \log r)$.

Although our reductions inject randomness to produce the desired noise distributions of target problems, the induced maps encoding the clique in $\pr{pc}_\rho$ as a new hidden structure typically do not inject randomness. Consequently, our reductions generally show hardness for priors over the hidden structure in our target problems with entropy $\Theta(k \log n)$. This then implies a lower bound for our target problems, because the canonical uniform priors with which they are defined are the \textit{hardest priors}. For example, every instance of $\pr{pc}_\rho$ reduces to uniform prior over cliques as in $\pr{pc}$ by randomly relabelling nodes. Similarly, a tensor PCA instance with a fixed planted vector $v$ reduces to the formulation in which $v$ is uniformly distributed on $\{-1, 1\}^n$ by taking the entrywise product of the tensor PCA instance with $u^{\otimes s}$ where $u$ is chosen u.a.r. from $\{-1, 1\}^n$. Thus our reductions actually show slightly stronger computational lower bounds than those stated in our main theorems -- they show lower bounds for our target problems with \textit{nonuniform} priors on their hidden structures. These nonuniform priors arise from the encodings of planted cliques into target hidden structure implicitly in our reductions, several of which we summarize below. Our reductions often involve aesthetic pre-processing and post-processing steps to reduce to canonical uniform priors and often subsample the output instance. To simplify our discussion, we omit these steps in describing the clique encodings induced by our reductions.
\begin{itemize}
\item \textbf{Robust Sparse Mean Estimation and SLR:} Let $S_L$ and $S_R$ be the sets of left and right clique vertices of the input $k\pr{-bpc}$ instance and let $[N] = E_1 \cup E_2 \cup \cdots \cup E_{k_N}$ be the given partition of the right vertices. The support of the $k$-sparse vector in our output $\pr{rsme}$ and $\pr{rslr}$ instances is simply $S_L$. Let $r$ be a prime and let $E_1' \cup E_2' \cup \cdots \cup E_{k_N}'$ be a partition of the output $n$ samples into parts of size $r\ell$ where $\ell = \frac{r^t - 1}{r - 1}$. Label each of element of $E_i'$ with a affine shift of a hyperplane in $\mathbb{F}_r^t$ and each element of $E_i$ with a point of $\mathbb{F}_r^t$. For each $i$, our adversary corrupts each sample in $E_i'$ corresponding to an affine shift of a hyperplane containing the point corresponding to the unique element in $S_R \cap E_i$.
\item \textbf{Dense Stochastic Block Models:} Let $S$ be the set of clique vertices of the input $k\pr{-pc}$ instance and let $E$ be the given partition of the its vertices $[N]$. Let $E'$ be a partition of the output $n$ vertices again into parts of size $r\ell$. Label elements in each part as above. Our output $\pr{isbm}$ instance has its smaller community supported on the union of the vertices across all $E_i'$ corresponding to affine shifts containing the points in $\mathbb{F}_r^t$ corresponding to the vertices $S$.
\item \textbf{Mixtures of SLRs and Generalized Learning Sparse Mixtures:} Let $S_L, S_R, k, k_N, N, n$ and $E$ be as above. The support of the $k$-sparse vector in our output $\pr{mslr}$ and $\pr{glsm}$ instances is again simply $S_L$. Let $H_1, H_2, \dots, H_{2^t - 1} \in \{-1, 1\}^{2^t}$ be the zero-sum rows of a Hadamard matrix and let $E'$ be a partition of the output $n$ samples into $k_N$ blocks of size $2^t$. The output instance sets the $j$th sample in $E_i'$ to be from the first part of the mixture if and only if the $j$th entry of $H_{s}$ is $1$ where $s$ is the unique element in $S_R \cap E_i$. In other words, the mixture pattern along $E_i'$ is given by the $(S_R \cap E_i)$th row of a Hadamard matrix.
\item \textbf{Tensor PCA:} Let $S$ be the set of clique vertices of the input $k\pr{-hpc}$ instance and let $E$ and $N$ be as above. Similarly to $\pr{mslr}$ and $\pr{glsm}$, the planted vector $v$ of our output $\pr{tpca}$ instance is the concatenation of the $(S \cap E_i)$th rows of a Hadamard matrix.
\end{itemize}
Our reduction to testing hidden partition models induces a more intricate encoding of cliques similar to that of dense stochastic block models described above. We remark that each of these encodings arises directly from design matrices and tensors based on $K_{r, t}$ used in the dense Bernoulli rotation step of our reductions.

\section{Further Directions and Open Problems}
\label{sec:1-open-problems}

In this section, we describe several further directions and problems left open in this work. These directions mainly concern the $\pr{pc}_\rho$ conjecture and our reduction techniques.

\paragraph{Further Evidence for $\pr{pc}_\rho$ Conjectures.} In this work, we give evidence for the $\pr{pc}_\rho$ conjecture from the failure of low-degree polynomials and for specific instantiations of the $\pr{pc}_\rho$ conjecture from the failure of SQ algorithms. An interesting direction for future work is to show sum of squares lower bounds for $\pr{pc}_\rho$ and $k\pr{-hpc}^s$ supporting this conjecture. A priori, this seems to be a technically difficult task as the sum of squares lower bounds in \cite{barak2016nearly} only apply to the prior in planted clique where every vertex is included in the clique independently with probability $k/n$. Thus it even remains open to extend these lower bounds to the uniform prior over $k$-subsets of $[n]$.

\paragraph{How do Priors on Hidden Structure Affect Hardness?} In this work, we showed that slightly altering the prior over the hidden structure of $\pr{pc}$ gave rise to a problem much more amenable to average-case reductions. This raises a broad question: for general problems $\mP$ with hidden structure, how does changing the prior over this hidden structure affect its hardness? In other words, for natural problems other than $\pr{pc}$, how does the conjectured computational barrier change with $\rho$? Another related direction for future work is whether other choices of $\rho$ in the $\pr{pc}_\rho$ conjecture give meaningful assumptions that can be mapped to more natural problems than the ones we consider here. Furthermore, it would be interesting to study how reductions carry ensembles of problems with a general prior $\rho$ to one another. For instance, is there a reduction between $\pr{pc}$ and another problem, such as $\pr{spca}$, such that every hard prior in $\pr{pc}_\rho$ is mapped to a corresponding hard prior in $\pr{spca}$?

\paragraph{Generalizations of Dense Bernoulli Rotations.} In this work, dense Bernoulli rotations were an extremely important subroutine, serving as our simplest primitive for transforming hidden structure. An interesting technical direction for future work is to find similar transformations mapping to other distributions. More concretely, dense Bernoulli rotations approximately mapped from $\pr{pb}(n, i, 1, 1/2)$ to the $n$ distributions $\mD_i = \mN(c \cdot A_i, I_m)$, respectively, and mapped from $\text{Bern}(1/2)^{\otimes m}$ to $\mD = \mN(0, I_m)$. Are there other similar reductions mapping from these planted bit distributions to different ensembles of $\mD, \mD_1, \mD_2, \dots, \mD_n$? Furthermore, can these maps be used to show tight computational lower bounds for natural problems? For example, two possibly interesting ensembles of $\mD, \mD_1, \mD_2, \dots, \mD_n$ are:
\begin{enumerate}
\item $\mD_i = \otimes_{j = 1}^m \text{Bern}(P_{ij} n^{-\alpha})$ and some $\mD$ where $P \in [0, 1]^{n \times m}$ is a fixed matrix of constants and $\alpha > 0$.
\item $\mD_i = \mN(c \cdot A_i, I_m - c^2 A_i A_i^\top)$ and $\mD = \mN(0, I_m)$.
\end{enumerate}
The first example above corresponds to whether or not there is a \textit{sparse} analogue of Bernoulli rotations that can be used to show tight computational lower bounds. A natural approach to (1) is to apply dense Bernoulli rotations and map each entry into $\{0, 1\}$ by thresholding at some large real number $T = \Theta(\sqrt{\log n})$. While this maps to an ensemble of the form in (1), this reduction seems \textit{lossy}, in the sense that it discards signal in the input instance, and it does not appear to show tight computational lower bounds for any natural problem. The second example above presents a set of $\mD_i$ with the same expected covariance matrices as $\mD$. Note that in ordinary dense Bernoulli rotations the expected covariance matrices for each $i$ are $I_m + c^2 \cdot A_i A_i^\top$ and often a degree-2 polynomial suffices to distinguish them from $\mD$. More generally, a natural question is: are there analogues of dense Bernoulli rotations that are tight to algorithms given by polynomials of degree higher than 2?

\paragraph{General Reductions to Supervised Problems.} Our last open problem is more concrete than the previous two. In our reductions to $\pr{mslr}$ and $\pr{rslr}$, we crucially use a subroutine mapping to $\pr{neg-spca}$. This subroutine requires that $k = \tilde{o}(n^{1/6})$ in order to show convergence in KL divergence between the Wishart and inverse Wishart distributions. Is there a reduction that relaxes this requirement to $k = \tilde{o}(n^{\alpha})$ where $1/6 < \alpha < 1/2$? Providing a reduction for $\alpha$ arbitrarily close to $1/2$ would essentially fill out all parameter regimes of interest in our computational lower bounds for $\pr{mslr}$ and $\pr{rslr}$. Any reduction relaxing this constraint to some $\alpha$ with $\alpha > 1/6$ seems as though it would require new techniques and be technically interesting. Another question related to our reductions to $\pr{mslr}$ and $\pr{rslr}$ is: can our label generation technique be generalized to handle more general link functions $\sigma$ i.e. generalized linear models where each sample-label pair $(X, y)$ satisfies $y = \sigma(\langle \beta, X \rangle) + \mN(0, 1)$? In particular, is there a reduction mapping to the canonical formulation of sparse phase retrieval with $\sigma(t) = t^2$? Although the statistical-computational gap for this formulation of sparse phase retrieval seems closely related to our computational lower bound for $\pr{mslr}$, any such reduction seems as though it would be interesting from a technical viewpoint.

\pagebreak

\part{Average-Case Reduction Techniques}
\label{part:reductions}

\section{Preliminaries and Problem Formulations}
\label{sec:2-preliminaries}

In this section, we establish notation and some preliminary observations for proving our main theorems from Section \ref{sec:1-problems}. We already defined our notion of computational lower bounds and solving detection and recovery problems in Section \ref{sec:1-problems}. In this section, we begin by stating our conventions for detection problems and adversaries. In Section \ref{subsec:2-tvreductions}, we introduce the framework for reductions in total variation to show computational lower bounds for detection problems. In Section \ref{subsec:2-formulations}, we then state detection formulations for each of our problems of interest that it will suffice to exhibit reductions to. Finally, in Section \ref{subsec:2-notation}, we introduce the key notation that will be used throughout the paper. Later in Section \ref{subsec:2-estimation}, we discuss how our reductions and lower bounds for the detection formulations in Section \ref{subsec:2-formulations} imply lower bounds for natural estimation and recovery variants of our problems.

\subsection{Conventions for Detection Problems and Adversaries}
\label{subsec:2-definitions}

We begin by describing our general setup for detection problems and the notions of robustness and types adversaries that we consider.

\paragraph{Detection Problems.} In a detection task $\mP$, the algorithm is given a set of observations and tasked with distinguishing between two hypotheses:
\begin{itemize}
\item a \emph{uniform} hypothesis $H_0$ corresponding to the natural noise distribution for the problem; and
\item a \emph{planted} hypothesis $H_1$, under which observations are generated from this distribution but with a latent planted structure.
\end{itemize}
Both $H_0$ and $H_1$ can either be simple hypotheses consisting of a single distribution or a composite hypothesis consisting of multiple distributions. Our problems typically are such that either: (1) both $H_0$ and $H_1$ are simple hypotheses; or (2) both $H_0$ and $H_1$ are composite hypotheses consisting of the set of distributions that can be induced by some constrained adversary. 

As discussed in \cite{brennan2018reducibility} and \cite{hajek2015computational}, when detection problems need not be composite by definition, average-case reductions to natural simple vs. simple hypothesis testing formulations are stronger and technically more difficult. In these cases, composite hypotheses typically arise because a reduction gadget precludes mapping to the natural simple vs. simple hypothesis testing formulation. We remark that simple vs. simple formulations are the hypothesis testing problems that correspond to average-case decision problems $(L, \mathcal{D})$ as in Levin's theory of average-case complexity. A survey of average-case complexity can be found in \cite{bogdanov2006average}.

\paragraph{Adversaries.} The robust estimation literature contains a number of adversaries capturing different notions of model misspecification. We consider the following three central classes of adversaries:
\begin{enumerate}
\item \textbf{$\epsilon$-corruption}: A set of samples $(X_1, X_2, \dots, X_n)$ is an $\epsilon$-corrupted sample from a distribution $\mD$ if they can be generated by giving a set of $n$ samples drawn i.i.d. from $\mD$ to an adversary who then changes at most $\epsilon n$ of them arbitrarily.
\item \textbf{Huber's contamination model}: A set of samples $(X_1, X_2, \dots, X_n)$ is an $\epsilon$-contamination of $\mD$ in Huber's model if
$$X_1, X_2, \dots, X_n \sim_{\text{i.i.d.}} \pr{mix}_{\epsilon}(\mD, \mD_O)$$
where $\mD_O$ is an unknown outlier distribution chosen by an adversary. Here, $\pr{mix}_{\epsilon}(\mD, \mD_O)$ denotes the $\epsilon$-mixture distribution formed by sampling $\mD$ with probability $(1 - \epsilon)$ and $\mD_O$ with probability $\epsilon$.
\item \textbf{Semirandom adversaries}: Suppose that $\mD$ is a distribution over collections of observations $\{ X_i \}_{i \in I}$ such that an unknown subset $P \subseteq I$ of indices correspond to a planted structure. A sample $\{ X_i \}_{i \in I}$ is semirandom if it can be generated by giving a sample from $\mD$ to an adversary who is allowed decrease $X_i$ for any $i \in I \backslash P$. Some formulations of semirandom adversaries in the literature also permit increases in $X_i$ for any $i \in P$. Our lower bounds apply to both adversarial setups.
\end{enumerate}
All adversaries in these models of robustness are computationally unbounded and have access to randomness -- meaning that they also have access to any hidden structure in a problem that can be recovered information theoretically. Given a single distribution $\mD$ over a set $X$, any one of these three adversaries produces a set of distributions $\pr{adv}(\mD)$ that can be obtained after corruption. When formulated as detection problems, the hypotheses $H_0$ and $H_1$ are of the form $\pr{adv}(\mD)$ for some $\mD$. We remark that $\epsilon$-corruption can simulate contamination in Huber's model at a slightly smaller $\epsilon'$ within $o(1)$ total variation. This is because a sample from Huber's model has $\text{Bin}(n, \epsilon')$ samples from $\mD_O$. An adversary resampling $\min\{\text{Bin}(n, \epsilon'), \epsilon n\}$ samples from $\mD_O$ therefore simulates Huber's model within a total variation distance bounded by standard concentration for the Binomial distribution.


\subsection{Reductions in Total Variation and Computational Lower Bounds}
\label{subsec:2-tvreductions}

In this section, we introduce our framework for reductions in total variation, state a general condition for deducing computational lower bounds from reductions in total variation and state a number of properties of total variation that we will use in analyzing our reductions.

\paragraph{Average-Case Reductions in Total Variation.} We give approximate reductions in total variation to show that lower bounds for one hypothesis testing problem imply lower bounds for another. These reductions yield an exact correspondence between the asymptotic Type I$+$II errors of the two problems. This is formalized in the following lemma, which is Lemma 3.1 from \cite{brennan2018reducibility} stated in terms of composite hypotheses $H_0$ and $H_1$. The main quantity in the statement of the lemma can be interpreted as the smallest total variation distance between the reduced object $\mathcal{A}(X)$ and the closest mixture of distributions from either $H_0'$ or $H_1'$. The proof of this lemma is short and follows from the definition of total variation. Given a hypothesis $H_i$, we let $\Delta(H_i)$ denote the set of all priors over the set of distributions valid under $H_i$.

\begin{lemma}[Lemma 3.1 in \cite{brennan2018reducibility}] \label{lem:3a}
Let $\mP$ and $\mP'$ be detection problems with hypotheses $H_0, H_1$ and  $H_0', H_1'$, respectively. Let $X$ be an instance of $\mathcal{P}$ and let $Y$ be an instance of $\mP'$. Suppose there is a polynomial time computable map $\mathcal{A}$ satisfying
$$\sup_{P \in H_0} \inf_{\pi \in \Delta(H_0')} \TV\left( \mL_{P}(\mathcal{A}(X)), \bE_{P' \sim \pi} \, \mL_{P'}(Y) \right) + \sup_{P \in H_1} \inf_{\pi \in \Delta(H_1')} \TV\left( \mL_{P}(\mathcal{A}(X)), \bE_{P' \sim \pi} \, \mL_{P'}(Y) \right) \le \delta$$
If there is a randomized polynomial time algorithm solving $\mP'$ with Type I$+$II error at most $\epsilon$, then there is a randomized polynomial time algorithm solving $\mP$ with Type I$+$II error at most $\epsilon + \delta$.
\end{lemma}

If $\delta = o(1)$, then given a blackbox solver $\mathcal{B}$ for $\mathcal{P}'_D$, the algorithm that applies $\mathcal{A}$ and then $\mathcal{B}$ solves $\mathcal{P}_D$ and requires only a single query to the blackbox. We now outline the computational model and conventions we adopt throughout this paper. An algorithm that runs in randomized polynomial time refers to one that has access to $\text{poly}(n)$ independent random bits and must run in $\text{poly}(n)$ time where $n$ is the size of the instance of the problem. For clarity of exposition, in our reductions we assume that explicit real-valued expressions can be exactly computed and that we can sample a biased random bit $\text{Bern}(p)$ in polynomial time. We also assume that the sampling and density oracles described in Definition \ref{def:computable} can be computed in $\text{poly}(n)$ time. For simplicity of exposition, we assume that we can sample $\mN(0, 1)$ in $\text{poly}(n)$ time. 

\paragraph{Deducing Strong Computational Lower Bounds for Detection from Reductions.} Throughout Part \ref{part:lower-bounds}, we will use the guarantees for our reductions to show computational lower bounds. For clarity and to avoid redundancy, we will outline a general recipe for showing these hardness results. All lower bounds that will be shown in Part \ref{part:lower-bounds} are \emph{computational lower bounds} in the sense introduced in the beginning of Section \ref{subsec:2-definitions}. Consider a problem $\mP$ with parameters $(n, a_1, a_2, \dots, a_t)$ and hypotheses $H_0$ and $H_1$ with a conjectured computationally hard regime captured by the constraint set $\mathcal{C}$. In order to show a computational lower bound at $\mathcal{C}$ based on one of our hardness assumptions, it suffices to show that the following is true:

\begin{condition}[Computational Lower Bounds from Reductions] \label{cond:lb}
\textup{
For all sequences of parameters satisfying the lower bound constraints $\{ (n, a_1(n), a_2(n), \dots, a_t(n)) \}_{n = 1}^\infty \subseteq \mathcal{C}$, there are:
\begin{enumerate}
\item another sequence of parameters $\{(n_i, a'_1(n_i), a'_2(n_i), \dots, a'_t(n_i))\}_{i = 1}^\infty \subseteq \mathcal{C}$ such that
$$\lim_{i \to \infty} \frac{\log a_k'(n_i)}{\log a_k(n_i)} = 1$$
\item a sequence of instances $\{G_i\}_{i = 1}^\infty$ of a problem $\pr{pc}_\rho$ with hypotheses $H_0'$ and $H_1'$ that cannot be solved in polynomial time according to Conjecture \ref{conj:hard-conj}; and
\item a polynomial time reduction $\mathcal{R}$ such that if $\mP(n_i, a'_1(n_i), a'_2(n_i), \dots, a'_t(n_i))$ has an instance denoted by $X_i$, then
$$\TV\left( \mathcal{R}(G_i | H_0'), \mL(X_i | H_0) \right) = o_{n_i}(1) \quad \text{and} \quad \TV\left( \mathcal{R}(G_i | H_1'), \mL(X_i | H_1) \right) = o_{n_i}(1)$$
\end{enumerate}
}
\end{condition}

This can be seen to suffice as follows. Suppose that $\mathcal{A}$ solves $\mP$ for some possible growth rate in $\mathcal{C}$ i.e. there is a sequence $\{(n_i, a'_1(n_i), a'_2(n_i), \dots, a'_t(n_i))\}_{i = 1}^\infty \subseteq \mathcal{C}$ with this growth rate such that $\mathcal{A}$ has Type I$+$II error $1 - \Omega_{n_i}(1)$ on $\mP(n_i, a'_1(n_i), a'_2(n_i), \dots, a'_t(n_i))$. By Lemma \ref{lem:3a}, it follows that $\mathcal{A} \circ \mathcal{R}$ also has Type I$+$II error $1 - \Omega_{n_i}(1)$ on the sequence of inputs $\{G_i\}_{i = 1}^\infty$, which contradicts the conjecture that they are hard instances. The three conditions above will be verified in a number of theorems in Part \ref{part:lower-bounds}.

\paragraph{Remarks on Deducing Computational Lower Bounds.} We make several important remarks on the recipe outlined above. In all of our applications of Condition \ref{cond:lb}, the second sequence of parameters $(n_i, a'_1(n_i), a'_2(n_i), \dots, a'_t(n_i))$ will either be exactly a subsequence of the original parameter sequence $(n, a_1(n), a_2(n), \dots, a_t(n))$ or will have one parameter $a_i' \neq a_i$ different from the original. However, the ability to pass to a subsequence will be crucial in a number of cases where number-theoretic constraints on parameters impact the tightness of our computational lower bounds. These constraints will arise in our reductions to robust sparse mean estimation, robust SLR and dense stochastic block models. They are discussed more in Section \ref{sec:3-robust-and-supervised}.

\paragraph{Properties of Total Variation.} The analysis of our reductions will make use of the following well-known facts and inequalities concerning total variation distance.

\begin{fact} \label{tvfacts}
The distance $\TV$ satisfies the following properties:
\begin{enumerate}
\item (Tensorization) Let $P_1, P_2, \dots, P_n$ and $Q_1, Q_2, \dots, Q_n$ be distributions on a measurable space $(\mathcal{X}, \mathcal{B})$. Then
$$\TV\left( \prod_{i = 1}^n P_i, \prod_{i = 1}^n Q_i \right) \le \sum_{i = 1}^n \TV\left( P_i, Q_i \right)$$
\item (Conditioning on an Event) For any distribution $P$ on a measurable space $(\mathcal{X}, \mathcal{B})$ and event $A \in \mathcal{B}$, it holds that
$$\TV\left( P(\cdot | A), P \right) = 1 - P(A)$$
\item (Conditioning on a Random Variable) For any two pairs of random variables $(X, Y)$ and $(X', Y')$ each taking values in a measurable space $(\mathcal{X}, \mathcal{B})$, it holds that
$$\TV\left( \mL(X), \mL(X') \right) \le \TV\left( \mL(Y), \mL(Y') \right) + \bE_{y \sim Y} \left[ \TV\left( \mL(X | Y = y), \mL(X' | Y' = y) \right)\right]$$
where we define $\TV\left( \mL(X | Y = y), \mL(X' | Y' = y) \right) = 1$ for all $y \not \in \textnormal{supp}(Y')$.
\end{enumerate}
\end{fact}

Given an algorithm $\mathcal{A}$ and distribution $\mP$ on inputs, let $\mathcal{A}(\mP)$ denote the distribution of $\mathcal{A}(X)$ induced by $X \sim \mP$. If $\mathcal{A}$ has $k$ steps, let $\mathcal{A}_i$ denote the $i$th step of $\mathcal{A}$ and $\mathcal{A}_{i\text{-}j}$ denote the procedure formed by steps $i$ through $j$. Each time this notation is used, we clarify the intended initial and final variables when $\mathcal{A}_{i}$ and $\mathcal{A}_{i\text{-}j}$ are viewed as Markov kernels. The next lemma from \cite{brennan2019universality} encapsulates the structure of all of our analyses of average-case reductions. Its proof is simple and included in Appendix \ref{subsec:appendix-2-tv} for completeness.

\begin{lemma}[Lemma 4.2 in \cite{brennan2019universality}] \label{lem:tvacc}
Let $\mathcal{A}$ be an algorithm that can be written as $\mathcal{A} = \mathcal{A}_m \circ \mathcal{A}_{m-1} \circ \cdots \circ \mathcal{A}_1$ for a sequence of steps $\mathcal{A}_1, \mathcal{A}_2, \dots, \mathcal{A}_m$. Suppose that the probability distributions $\mP_0, \mP_1, \dots, \mP_m$ are such that $\TV(\mathcal{A}_i(\mP_{i-1}), \mP_i) \le \epsilon_i$ for each $1 \le i \le m$. Then it follows that
$$\TV\left( \mathcal{A}(\mP_0), \mP_m \right) \le \sum_{i = 1}^m \epsilon_i$$
\end{lemma}

The next lemma bounds the total variation between unplanted and planted samples from binomial distributions. This will serve as a key computation in the proof of correctness for the reduction primitive $\pr{To-}k\textsc{-Partite-Submatrix}$. We remark that the total variation upper bound in this lemma is tight in the following sense. When all of the $P_i$ are the same, the expected value of the sum of the coordinates of the first distribution is $k(P_i - Q)$ higher than that of the second. The standard deviation of the second sum is $\sqrt{kmQ(1 - Q)}$ and thus when $k(P_i - Q)^2 \gg mQ(1 - Q)$, the total variation below tends to one. The proof of this lemma can be found in Appendix \ref{subsec:appendix-2-tv}.

\begin{lemma} \label{lem:bernproduct}
If $k, m \in \mathbb{N}$, $P_1, P_2, \dots, P_k \in [0, 1]$ and $Q \in (0, 1)$, then
$$\TV\left( \otimes_{i = 1}^k \left( \textnormal{Bern}(P_i) + \textnormal{Bin}(m - 1, Q) \right), \textnormal{Bin}(m, Q)^{\otimes k} \right) \le \sqrt{\sum_{i = 1}^k \frac{(P_i - Q)^2}{2mQ(1 - Q)}}$$
\end{lemma}

Here, $\mL_1 + \mL_2$ denotes the convolution of two given probability measures $\mL_1$ and $\mL_2$. The next lemma bounds the total variation between two binomial distributions. Its proof can be found in Appendix \ref{subsec:appendix-2-tv}.

\begin{lemma} \label{lem:bintv}
Given $P \in [0, 1]$, $Q \in (0, 1)$ and $n \in \mathbb{N}$, it follows that
$$\TV\left( \textnormal{Bin}(n, P), \textnormal{Bin}(n, Q) \right) \le |P - Q| \cdot \sqrt{\frac{n}{2Q(1 - Q)}}$$
\end{lemma}

\subsection{Problem Formulations as Detection Tasks}
\label{subsec:2-formulations}

In this section, we formulate each problem for which we will show computational lower bounds as a detection problem. More precisely, for each problem $\mP$ introduced in Section \ref{sec:1-problems}, we introduce a detection variant $\mP'$ such that a blackbox for $\mP$ also solves $\mP'$. Some of these formulations were already implicitly introduced or will be reintroduced in future sections. We gather all of these formulations here for convenience. Throughout this work, to simplify notation, we will refer to problems $\mP$ and their detection formulations $\mP'$ introduced in this section using the same notation. Furthermore, we will often denote the distribution over instances under the alternative hypothesis $H_1$ of the detection formulation for $\mP$ with the notation $\mP_D$, when $H_1$ is a simple hypothesis. We will also often parameterize $\mP_D$ by $\theta$ to denote $\mP_D$ conditioned on the latent hidden structure $\theta$. When $H_1$ is composite, $\mP_D$ denotes the set of distributions permitted under $H_1$. These general conventions are introduced on a per problem basis in this section. In Section \ref{subsec:2-estimation}, we show that our reductions and lower bounds for these detection formulations also imply lower bounds for analogous estimation and recovery variants.

\paragraph{Robust Sparse Mean Estimation.} Our hypothesis testing formulation for the problem $\pr{rsme}(n, k, d, \tau, \epsilon)$ has hypotheses given by
\begin{align*}
H_0 : (X_1, X_2, \dots, X_n) &\sim_{\textnormal{i.i.d.}} \mN(0, I_d) \\
H_1 : (X_1, X_2, \dots, X_n) &\sim_{\textnormal{i.i.d.}} \pr{mix}_{\epsilon}\left( \mN(\tau \cdot \mu_R, I_d), \mD_O \right)
\end{align*}
where $\mD_O$ is any adversarially chosen outlier distribution on $\mathbb{R}^d$, where $\mu_R \in \mathbb{R}^d$ is a random $k$-sparse unit vector chosen uniformly at random from all such vectors with entries in $\{0, 1/\sqrt{k}\}$. Note that $H_1$ is a composite hypothesis here since $\mD_O$ is arbitrary. Note also that this is a formulation of $\pr{rsme}$ in Huber's contamination model, and therefore lower bounds for this detection problem imply corresponding lower bounds under stronger $\epsilon$-corruption adversaries.

As discussed in Section \ref{subsec:1-problems-rsme}, $\pr{rsme}$ is only information-theoretically feasible when $\tau = \Omega(\epsilon)$. Consider any algorithm that produces some estimate $\hat{\mu}$ satisfying that $\| \hat{\mu} - \mu \|_2 < \tau/2$ with probability $1/2 + \Omega(1)$ in the estimation formulation for $\pr{rsme}$ with hidden $k$-sparse vector $\mu$, as described in Section \ref{subsec:1-problems-rsme}. This algorithm would necessarily output some $\hat{\mu}$ with $\| \hat{\mu} \|_2 < \tau/2$ under $H_0$ and some $\hat{\mu}$ with $\| \hat{\mu} \|_2 > \tau/2$ under $H_1$ with probability $1/2 + \Omega(1)$ in the hypothesis testing formulation above, thus solving it in the sense of Section \ref{sec:1-problems}. Thus any computational lower bounds for this hypothesis testing formulation also implies a lower bound for the typical estimation formulation of $\pr{rsme}$.

\paragraph{Dense Stochastic Block Models.} Given a subset $C_1 \subseteq [n]$ of size $n/k$, let $\pr{isbm}_D(n, C_1, P_{11}, P_{12}, P_{22})$ denote the distribution on $n$-vertex graphs $G'$ introduced in Section \ref{subsec:1-problems-sbm} conditioned on $C_1$. Furthermore, let $\pr{isbm}_D(n, k, P_{11}, P_{12}, P_{22})$ denote the mixture of these distributions induced by choosing $C_1$ uniformly at random from the $(n/k)$-subsets of $[n]$. The problem $\pr{isbm}(n, k, P_{11}, P_{12}, P_{22})$ introduced in Section \ref{subsec:1-problems-sbm} is already a hypothesis testing problem, with hypotheses
$$H_0 : G \sim \mG\left(n, P_0 \right) \quad \text{and} \quad H_1 : G \sim \pr{isbm}_D(n, k, P_{11}, P_{12}, P_{22})$$
where $H_0$ is a composite hypothesis and $P_0$ can vary over all edge densities in $(0, 1)$. As we will discuss at the end of this section, computational lower bounds for this hypothesis testing problem imply lower bounds for the problem of recovering the hidden community $C_1$.

\paragraph{Testing Hidden Partition Models.} Let $C = (C_1, C_2, \dots, C_r)$ and $D = (D_1, D_2, \dots, D_r)$ be two fixed sequences, each consisting of disjoint $K$-subsets of $[n]$. Let $\pr{ghpm}_D(n, r, C, D, \gamma)$ denote the distribution over random matrices $M \in \mathbb{R}^{n \times n}$ introduced in Section \ref{subsec:1-problems-hidden-partition} conditioned on the fixed sequences $C$ and $D$. We denote the mixture over these distributions induced by choosing $C$ and $D$ independently and uniformly at random from all admissible such sequences as $\pr{ghpm}_D(n, r, K, \gamma)$. Similarly, we let $\pr{bhpm}_D(n, r, C, P_0, \gamma)$ denote the distribution over bipartite graphs $G$ with two parts of size $n$, each indexed by $[n]$ with edges included independently with probability
$$\bP\left[ (i, j) \in E(G) \right] = \left\{ \begin{array}{ll} P_0 + \gamma &\textnormal{if } i \in C_h \textnormal{ and } j \in D_h \textnormal{ for some } h \in [r] \\ P_0 - \frac{\gamma}{r - 1} &\textnormal{if } i \in C_{h_1} \textnormal{ and } j \in D_{h_2} \textnormal{ where } h_1 \neq h_2 \\ P_0 &\textnormal{otherwise} \end{array} \right.$$
where $P_0, \gamma \in (0, 1)$ be such that $\gamma/r \le P_0 \le 1 - \gamma$. Then let $\pr{bhpm}_D(n, r, K, P_0, \gamma)$ denote the mixture formed by choosing $C$ and $D$ randomly as in $\pr{ghpm}_D$. The problems $\pr{ghpm}(n, r, C, D, \gamma)$ and $\pr{bhpm}(n, r, K, P_0, \gamma)$ are simple hypothesis testing problems given by
$$\begin{array}{lll}
H_0: M \sim \mN(0, 1)^{\otimes n \times n} &\text{and} &H_1: M \sim \pr{ghpm}_D(n, r, K, \gamma) \\
H_0: G \sim \mG_B(n, n, P_0) &\text{and} &H_1: G \sim \pr{bhpm}_D(n, r, K, P_0, \gamma)
\end{array}$$
where $\mG_B(n, n, P_0)$ denotes the Erd\H{o}s-R\'{e}nyi distribution over bipartite graphs with two parts each indexed by $[n]$ and where each edge is included independently with probability $P_0$.

\paragraph{Semirandom Planted Dense Subgraph.} Our hypothesis testing formulation for $\pr{semi-cr}(n, k, P_1, P_0)$ has observation $G \in \mG_n$ and two composite hypotheses given by
\begin{align*}
&H_0 : G \sim \mathbb{P}_0 \quad \textnormal{for some } \mathbb{P}_0 \in \pr{adv}\left(\mG(n, P_0)\right) \\
&H_1 : G \sim \mathbb{P}_1 \quad \textnormal{for some } \mathbb{P}_1 \in \pr{adv}\left(\mG(n, k, P_1, P_0)\right)
\end{align*}
Here, $\pr{adv}\left(\mG(n, k, P_1, P_0)\right)$ denotes the set of distributions induced by a semirandom adversary that can only remove edges outside of the planted dense subgraph $S$. Similarly, the set $\pr{adv}\left(\mG(n, P_0)\right)$ corresponds to an adversary that can remove any edges from the Erd\H{o}s-R\'{e}nyi graph $\mG(n, P_0)$. We will discuss at the end of this section, how computational lower bounds for this hypothesis testing formulation imply lower bounds for the problem of approximately recovering the vertex subset corresponding to the planted dense subgraph.

\paragraph{Negative Sparse PCA.} Our hypothesis testing formulation for $\pr{neg-spca}(n, k, d, \theta)$ is the spiked covariance model introduced in \cite{johnstoneSparse04} and used to formulate ordinary $\pr{spca}$ in \cite{gao2017sparse, brennan2018reducibility, brennan2019optimal}. This problem has hypotheses given by
\begin{align*}
H_0 : (X_1, X_2, \dots, X_n) &\sim_{\textnormal{i.i.d.}} \mN(0, I_d) \\
H_1 : (X_1, X_2, \dots, X_n) &\sim_{\textnormal{i.i.d.}} \mN\left( 0, I_d - \theta vv^\top \right)
\end{align*}
where $v \in \mathbb{R}^d$ is a $k$-sparse unit vector with entries in $\{0, 1/\sqrt{k}\}$ chosen uniformly at random.

\paragraph{Unsigned and Mixtures of SLRs.} Given a vector $v \in \mathbb{R}^d$, let $\pr{lr}_d(v)$ be the distribution of a single sample-label pair $(X, y) \in \mathbb{R}^d \times \mathbb{R}$ given by
$$y = \langle v, X \rangle + \eta \quad \text{where } X \sim \mN(0, I_d) \text{ and } \eta \sim \mN(0, 1) \text{ are independent}$$
Given a subset $S \subseteq [n]$, let $\pr{mslr}_D(n, S, d, \tau, 1/2)$ denote the distribution over $n$ independent sample-label pairs $(X_1, y_1), (X_2, y_2), \dots, (X_n, y_n)$ each distributed as
$$(X_i, y_i) \sim \pr{lr}_d(\tau s_i v_S) \quad \text{where } s_i \sim_{\text{i.i.d.}} \text{Rad}$$
where $v_S = |S|^{-1/2} \cdot \mathbf{1}_S$ and $\text{Rad}$ denotes the Rademacher distribution which is uniform over $\{-1, 1\}$. Note that this is a even mixture of sparse linear regressions with hidden unit vectors $v_S$ and $-v_S$ and signal strength $\tau$. Let $\pr{mslr}_D(n, k, d, \tau, 1/2)$ denote the mixture of these distributions induced by choosing $S$ uniformly at random from all $k$-subsets of $[n]$. Our hypothesis testing formulation for $\pr{mslr}(n, k, d, \tau)$ has two simple hypotheses given by
\begin{align*}
H_0 : \left\{ (X_i, y_i) \right\}_{i \in [n]} &\sim \left( \mN(0, I_d) \otimes \mN\left(0, 1 + \tau^2\right) \right)^{\otimes n} \\
H_1 : \left\{ (X_i, y_i) \right\}_{i \in [n]} &\sim \pr{mslr}_D(n, k, d, \tau, 1/2)
\end{align*}
Our hypothesis testing formulation of $\pr{uslr}(n, k, d, \tau)$ is a simple derivative of this formulation obtained by replacing each observation $(X_i, y_i)$ with $(X_i, |y_i|)$. We remark that, unlike $\pr{rsme}$ where an estimation algorithm trivially solved the hypothesis testing formulation, the hypothesis $H_0$ here is not an instance of $\pr{mslr}$ corresponding to a hidden vector of zero. This is because the labels $y_i$ under $H_0$ have variance $1 + \tau^2$, whereas they would have variance $1$ if they were this instance of $\pr{mslr}$. However, this detection problem still yields hardness for the estimation variants of $\pr{mslr}$ and $\pr{uslr}$ described in Section \ref{subsec:1-problems-mslr}, albeit with a slightly more involved argument. This is discussed in Section \ref{subsec:2-estimation}.

\paragraph{Robust SLR.} Our hypothesis testing formulation for $\pr{rslr}(n, k, d, \tau, \epsilon)$ has hypotheses given by
\begin{align*}
H_0 : \left\{ (X_i, y_i) \right\}_{i \in [n]} &\sim \left( \mN(0, I_d) \otimes \mN\left(0, 1 + \tau^2\right) \right)^{\otimes n} \\
H_1 : \left\{ (X_i, y_i) \right\}_{i \in [n]} &\sim_{\textnormal{i.i.d.}} \pr{mix}_{\epsilon}\left( \pr{lr}_d(\tau v), \mD_O \right)
\end{align*}
where $\mD_O$ is any adversarially chosen outlier distribution on $\mathbb{R}^d \times \mathbb{R}$, where $v \in \mathbb{R}^d$ is a random $k$-sparse unit vector chosen uniformly at random from all such vectors with entries in $\{0, 1/\sqrt{k}\}$. As with the other formulations of SLR, we defer discussing the implications of lower bounds in this formulation for the estimation task described in Section \ref{subsec:1-problems-robust-slr} to Section \ref{subsec:2-estimation}.

\paragraph{Tensor PCA.} Let $\pr{tpca}^s_D(n, \theta)$ denote the distribution on order $s$ tensors $T \in \mathbb{R}^{n^{\otimes s}}$ with dimensions all equal to $n$ given by $T = v^{\otimes s} + G$ where $G \sim \mN(0, 1)^{\otimes n^{\otimes s}}$ and $v \in \{-1, 1\}^n$ is chosen independently and uniformly at random. As already introduced in Section \ref{subsec:1-problems-tpca}, our hypothesis testing formulation for $\pr{tpca}^s(n, \theta)$ is given by
$$H_0: T \sim \mN(0, 1)^{\otimes n^{\otimes s}} \quad\text{ and }\quad H_1: T \sim \pr{tpca}^s_D(n, \theta)$$
Unlike the other problems we consider, our reductions only show computational lower bounds for blackboxes solving this hypothesis testing problem with a low false positive probability. As we will show in Section \ref{sec:3-tensor}, this implies a lower bound for the canonical estimation formulation for tensor PCA.

\paragraph{Generalized Learning Sparse Mixtures.} Let $\{\mP_{\mu}\}_{\mu \in \mathbb{R}}$ and $\mQ$ be distributions on an arbitrary measurable space $(\mathcal{X}, \mathcal{B})$ and let $\mD$ be a mixture distribution on $\mathbb{R}$. Let $\pr{glsm}_D(n, S, d, \{\mP_{\mu}\}_{\mu \in \mathbb{R}}, \mQ, \mD)$ denote the distribution over $X_1, X_2, \dots, X_n \in \mathcal{X}^d$ introduced in Section \ref{subsec:1-problems-universality} and let $\pr{glsm}_D(n, k, d, \{\mP_{\mu}\}_{\mu \in \mathbb{R}}, \mQ, \mD)$ denote the mixture over these distributions induced by sampling $S$ uniformly at random from the family of $k$-subsets of $[n]$. Our general sparse mixtures detection problem $\pr{glsm}(n, S, d, \{\mP_{\mu}\}_{\mu \in \mathbb{R}}, \mQ, \mD)$ is the following simple vs. simple hypothesis testing formulation
$$H_0 : (X_1, X_2, \dots, X_n) \sim_{\textnormal{i.i.d.}} \mQ^{\otimes d} \quad \text{and} \quad H_1 : (X_1, X_2, \dots, X_n) \sim \pr{glsm}_D\left(n, k, d, \{\mP_{\mu}\}_{\mu \in \mathbb{R}}, \mQ, \mD\right)$$
Lower bounds for this formulation directly imply lower bounds for algorithms that return an estimate $\hat{S}$ of $S$ given samples from $\pr{glsm}_D(n, S, d, \{\mP_{\mu}\}_{\mu \in \mathbb{R}}, \mQ, \mD)$ with $|\hat{S} \Delta S| < k/2$ with probability $1/2 + \Omega(1)$ for all $|S| \le k$. Note that under $H_0$, such an algorithm would output some set $\hat{S}$ of size less than $k/2$ and, under $H_1$, it would output a set of size greater than $k/2$, each with probability $1/2 + \Omega(1)$. Thus thresholding $|\hat{S}|$ at $k/2$ solves this detection formulation in the sense of Section \ref{sec:1-problems}.

\subsection{Notation}
\label{subsec:2-notation}

In this section, we establish notation that will be used repeatedly throughout this paper. Some of these definitions are repeated later upon use for convenience. Let $\mL(X)$ denote the distribution law of a random variable $X$ and given two laws $\mL_1$ and $\mL_2$, let $\mL_1 + \mL_2$ denote $\mL(X + Y)$ where $X \sim \mL_1$ and $Y \sim \mL_2$ are independent. Given a distribution $\mathcal{P}$, let $\mathcal{P}^{\otimes n}$ denote the distribution of $(X_1, X_2, \dots, X_n)$ where the $X_i$ are i.i.d. according to $\mathcal{P}$. Similarly, let $\mathcal{P}^{\otimes m \times n}$ denote the distribution on $\mathbb{R}^{m \times n}$ with i.i.d. entries distributed as $\mathcal{P}$. We let $\mathbb{R}^{n^{\otimes s}}$ denote the set of all order $s$ tensors with dimensions all $n$ in size that contain $n^s$ entries. The distribution $\mP^{\otimes n^{\otimes s}}$ denotes a tensor of these dimensions with entries independently sampled from $\mP$. We say that two parameters $a$ and $b$ are polynomial in one another if there is a constant $C > 0$ such that $a^{1/C} \le b \le a^C$ as $a \to \infty$. In this paper, we adopt the standard asymptotic notation $O(\cdot), \Omega(\cdot), o(\cdot), \omega(\cdot)$ and $\Theta(\cdot)$. We let $a \asymp b$, $a \lesssim b$ and $a \gtrsim b$ be shorthands for $a = \Theta(b), a = O(b)$ and $a = \Omega(b)$, respectively. In all problems that we consider, our main focus is on the polynomial order of growth at computational barriers, usually in terms of a natural parameter $n$. Given a natural parameter $n$ that will usually be clear from context, we let $a = \tilde{O}(b)$ be a shorthand for $a = O\left(b \cdot (\log n)^c \right)$ for some constant $c > 0$, and define $\tilde{\Omega}(\cdot), \tilde{o}(\cdot), \tilde{\omega}(\cdot)$ and $\tilde{\Theta}(\cdot)$ analogously. Oftentimes, it will be true that $b$ is polynomial in $n$, in which case $n$ can be replaced by $b$ in the definition above.

Given a finite or measurable set $\mathcal{X}$, let $\text{Unif}[\mathcal{X}]$ denote the uniform distribution on $\mathcal{X}$. Let $\text{Rad}$ be shorthand for $\text{Unif}[\{-1, 1\}]$, corresponding to the special case of a Rademacher random variable. Let $\TV$, $\KL$ and $\chi^2$ denote total variation distance, KL divergence and $\chi^2$ divergence, respectively. Let $\mN(\mu, \Sigma)$ denote a multivariate normal random vector with mean $\mu \in \mathbb{R}^d$ and covariance matrix $\Sigma$, where $\Sigma$ is a $d \times d$ positive semidefinite matrix, and let $\text{Bern}(p)$ denote the Bernoulli distribution with probability $p$. Let $[n] = \{1, 2, \dots, n\}$ and $\mG_n$ be the set of simple graphs on $n$ vertices. Let $\mG(n, p)$ denote the Erd\H{o}s-R\'{e}nyi distribution over $n$-vertex graphs where each edge is included independently with probability $p$. Let $\mG_B(m, n, p)$ denote the Erd\H{o}s-R\'{e}nyi distribution over $(m + n)$-vertex bipartite graphs with $m$ left vertices, $n$ right vertices and such that each of the $mn$ possible edges included independently with probability $p$. Throughout this paper, we will refer to bipartite graphs with $m$ left vertices and $n$ right vertices and matrices in $\{0, 1\}^{m \times n}$ interchangeably. Let $\mathbf{1}_S$ denote the vector $v \in \mathbb{R}^n$ with $v_i = 1$ if $i \in S$ and $v_i = 0$ if $i \not \in S$ where $S \subseteq [n]$. Let $\pr{mix}_{\epsilon}(\mD_1, \mD_2)$ denote the $\epsilon$-mixture distribution formed by sampling $\mD_1$ with probability $(1 - \epsilon)$ and $\mD_2$ with probability $\epsilon$. Given a partition $E$ of $[N]$ with $k$ parts, let $\mU_N(E)$ denote the uniform distribution over all $k$-subsets of $[N]$ containing exactly one element from each part of $E$.

Given a matrix $M \in \mathbb{R}^{n \times n}$, the matrix $M_{S, T} \in \mathbb{R}^{k \times k}$ where $S, T$ are $k$-subsets of $[n]$ refers to the minor of $M$ restricted to the row indices in $S$ and column indices in $T$. Furthermore, $(M_{S, T})_{i, j} = M_{\sigma_S(i), \sigma_T(j)}$ where $\sigma_S : [k] \to S$ is the unique order-preserving bijection and $\sigma_T$ is analogously defined. Given an index set $I$, subset $S \subseteq I$ and pair of distributions $(\mP, \mQ)$, let $\mathcal{M}_I(S, \mP, \mQ)$ denote the distribution of a collection of independent random variables $(X_i : i \in I)$ with $X_i \sim \mP$ if $i \in S$ and $X_i \sim \mQ$ if $i \not \in S$. When $S$ is a random set, this $\mathcal{M}_I(S, \mP, \mQ)$ denotes a mixture over the randomness of $S$ e.g. $\mathcal{M}_{[N]}(\mU_N(E), \mP, \mQ)$ denotes a mixture of $\mathcal{M}_{[N]}(S, \mP, \mQ)$ over $S \sim \mU_N(E)$. Generally, given an index set $I$ and $|I|$ distributions $\mP_1, \mP_2, \dots, \mP_{|I|}$, let $\mathcal{M}_I(\mP_i : i \in I)$ denote the distribution of independent random variables $(X_i : i \in I)$ with $X_i \sim \mP_i$ for each $i \in I$. The planted Bernoulli distribution $\pr{pb}(n, i, p, q)$ is over $V \in \{0, 1\}^n$ with independent entries satisfying that $V_j \sim \text{Bern}(q)$ unless $j = i$, in which case $V_i \sim \text{Bern}(p)$. In other words, $\pr{pb}(n, i, p, q)$ is a shorthand for $\mathcal{M}_{[n]}\left(\{ i \}, \text{Bern}(p), \text{Bern}(q)\right)$. Similarly, the planted dense subgraph distribution $\mG(n, S, p, q)$ can be written as $\mathcal{M}_I\left(\binom{S}{2}, \text{Bern}(p), \text{Bern}(q) \right)$ where $I = \binom{[n]}{2}$.

\section{Rejection Kernels and Reduction Preprocessing}
\label{sec:2-rejection-kernels}

In this section, we present several average-case reduction primitives that will serve as the key subroutines and preprocessing steps in our reductions. These include pre-existing subroutines from the rejection kernels framework introduced in \cite{brennan2018reducibility, brennan2019universality, brennan2019optimal}, such as univariate rejection kernels from binary inputs and $\pr{Gaussianize}$. We introduce the primitive $\pr{To-}k\textsc{-Partite-Submatrix}$, which is a generalization of $\pr{To-Submatrix}$ from \cite{brennan2019universality} that maps from the $k$-partite variant of planted dense subgraph to Bernoulli matrices, by filling in the missing diagonal and symmetrizing. We also introduce a new variant of rejection kernels called symmetric 3-ary rejection kernels that will be crucial in our reductions showing universality of lower bounds for sparse mixtures.

\subsection{Gaussian Rejection Kernels}
\label{subsec:2-gaussian}

\begin{figure}[t!]
\begin{algbox}
\textbf{Algorithm} $\textsc{rk}_G(\mu, B)$

\vspace{2mm}

\textit{Parameters}: Input $B \in \{0, 1\}$, Bernoulli probabilities $0 < q < p \le 1$, Gaussian mean $\mu$, number of iterations $N$, let $\varphi_\mu(x) = \frac{1}{\sqrt{2\pi}} \cdot \exp\left(- \frac{1}{2}(x - \mu)^2 \right)$ denote the density of $\mN(\mu, 1)$
\begin{enumerate}
\item Initialize $z \gets 0$.
\item Until $z$ is set or $N$ iterations have elapsed:
\begin{enumerate}
\item[(1)] Sample $z' \sim \mN(0, 1)$ independently.
\item[(2)] If $B = 0$, if the condition
$$p \cdot \varphi_0(z') \ge q \cdot \varphi_{\mu}(z')$$
holds, then set $z \gets z'$ with probability $1 - \frac{q \cdot \varphi_\mu(z')}{p \cdot \varphi_0(z')}$.
\item[(3)] If $B = 1$, if the condition
$$(1 - q) \cdot \varphi_\mu(z' + \mu) \ge (1 - p) \cdot \varphi_0(z' + \mu)$$
holds, then set $z \gets z' + \mu$ with probability $1 - \frac{(1 - p) \cdot \varphi_0(z' + \mu)}{(1 - q) \cdot \varphi_\mu(z' + \mu)}$.
\end{enumerate}
\item Output $z$.
\end{enumerate}
\vspace{1mm}
\textbf{Algorithm} $\textsc{Gaussianize}$

\vspace{2mm}

\textit{Parameters}: Collection of variables $X_i \in \{0, 1\}$ for $i \in I$ where $I$ is some index set with $|I| = n$, rejection kernel parameter $R_{\pr{rk}}$, Bernoulli probabilities $0 < q < p \le 1$ with $p - q = R_{\pr{rk}}^{-O(1)}$ and $\min(q, 1 - q) = \Omega(1)$ and a target means $0 \le \mu_{i} \le \tau$ for each $i \in I$ where $\tau > 0$ is a parameter
\begin{enumerate}
\item Form the collection of variables $Y \in \mathbb{R}^{I}$ by setting
$$Y_i \gets \textsc{rk}_{G}(\mu_i, X_i)$$
for each $i \in I$ where each $\textsc{rk}_{G}$ is run with parameter $R_{\pr{rk}}$ and $N_{\text{it}} = \lceil 6\delta^{-1} \log R_{\pr{rk}} \rceil$ iterations where $\delta = \min \left\{ \log \left( \frac{p}{q} \right), \log \left( \frac{1 - q}{1 - p} \right) \right\}$.
\item Output the collection of variables $(Y_i : i \in I)$.
\end{enumerate}
\vspace{1mm}
\end{algbox}
\caption{Gaussian instantiation of the rejection kernel algorithm from \cite{brennan2018reducibility} and the reduction $\textsc{Gaussianize}$ for mapping from Bernoulli to Gaussian planted problems from \cite{brennan2019optimal}.}
\label{fig:rej-kernel}
\end{figure}

Rejection kernels are a framework in \cite{brennan2018reducibility, brennan2019universality, brennan2019optimal} for algorithmic changes of measure based on rejection sampling. Related reduction primitives for changes of measure to Gaussians and binomial random variables appeared earlier in \cite{ma2015computational, hajek2015computational}. Rejection kernels mapping a pair of Bernoulli distributions to a target pair of scalar distributions were introduced in \cite{brennan2018reducibility}. These were extended to arbitrary high-dimensional target distributions and applied to obtain universality results for submatrix detection in \cite{brennan2019universality}. A surprising and key feature of both of these rejection kernels is that they are not lossy in mapping one computational barrier to another. For instance, in \cite{brennan2019universality}, multivariate rejection kernels were applied to increase the relative size $k$ of the planted submatrix, faithfully mapping instances tight to the computational barrier at lower $k$ to tight instances at higher $k$. This feature is also true of the scalar rejection kernels applied in \cite{brennan2018reducibility}.

In this work, we will only need a subset of prior results on rejection kernels. In this section, we give an overview of the key guarantees for Gaussian rejection kernels with binary inputs from \cite{brennan2018reducibility} and for $\textsc{Gaussianize}$ from \cite{brennan2019optimal}. We will also need a new ternary input variant of rejection kernels that will be introduced in Section \ref{subsec:srk}. We begin by introducing the Gaussian rejection kernel $\pr{rk}_G(\mu, B)$ which maps $B \in \{0, 1\}$ to a real valued output and is parameterized by some $0 < q < p \le 1$. The map $\pr{rk}_G(\mu, B)$ transforms two Bernoulli inputs approximately into Gaussians. Specifically, it satisfies the two Markov transition properties
$$\pr{rk}_G(\mu, B) \approx \mN(0, 1) \quad \text{if } B \sim \text{Bern}(q) \quad \quad \text{and} \quad \quad \pr{rk}_G(\mu, B) \approx \mN(\mu, 1) \quad \text{if } B \sim \text{Bern}(p)$$
where $\pr{rk}_G(\mu, B)$ can be computed in $\text{poly}(n)$ time, the $\approx$ above are up to $O_n(n^{-3})$ total variation distance and $\mu = \Theta(1/\sqrt{\log n})$. The maps $\pr{rk}_G(\mu, B)$ can be implemented with the rejection sampling scheme shown in Figure \ref{fig:rej-kernel}. The total variation guarantees for Gaussian rejection kernels are captured formally in the following theorem.

\begin{lemma}[Gaussian Rejection Kernels -- Lemma 5.4 in \cite{brennan2018reducibility}] \label{lem:5c}
Let $R_{\pr{rk}}$ be a parameter and suppose that $p = p(R_{\pr{rk}})$ and $q = q(R_{\pr{rk}})$ satisfy that $0 < q < p \le 1$, $\min(q, 1 - q) = \Omega(1)$ and $p - q \ge R_{\pr{rk}}^{-O(1)}$. Let $\delta = \min \left\{ \log \left( \frac{p}{q} \right), \log \left( \frac{1 - q}{1 - p} \right) \right\}$. Suppose that $\mu = \mu(R_{\pr{rk}}) \in (0, 1)$ satisfies that
$$\mu \le \frac{\delta}{2 \sqrt{6\log R_{\pr{rk}} + 2\log (p-q)^{-1}}}$$
Then the map $\textsc{rk}_{\text{G}}$ with $N = \left\lceil 6\delta^{-1} \log R_{\pr{rk}} \right\rceil$ iterations can be computed in $\text{poly}(R_{\pr{rk}})$ time and satisfies
$$\TV\left(\textsc{rk}_{\text{G}}(\mu, \textnormal{Bern}(p)), \mN(\mu, 1) \right) = O\left(R_{\pr{rk}}^{-3}\right) \quad \text{and} \quad \TV\left(\textsc{rk}_{\text{G}}(\mu, \textnormal{Bern}(q)), \mN(0, 1) \right) = O\left(R_{\pr{rk}}^{-3}\right)$$
\end{lemma}

The proof of this lemma consists of showing that the distributions of the outputs $\pr{rk}_G(\mu, \text{Bern}(p))$ and $\pr{rk}_G(\mu, \text{Bern}(q))$ are close to $\mN(\mu, 1)$ and $\mN(0, 1)$ when conditioned to lie in the set of $x$ with $\frac{1 - p}{1 - q} \le \frac{\varphi_\mu(x)}{\varphi_0(x)} \le \frac{p}{q}$ and then showing that this event occurs with probability close to one. The original framework in \cite{brennan2018reducibility} mapped binary inputs to more general pairs of target distributions than $\mN(\mu, 1)$ and $\mN(0, 1)$, however we will only require binary-input rejection kernels in the Gaussian. A multivariate extension of this framework appeared in \cite{brennan2019universality}.

Given an index set $I$, subset $S \subseteq I$ and pair of distributions $(\mP, \mQ)$, let $\mathcal{M}_I(S, \mP, \mQ)$ denote the distribution of a collection of independent random variables $(X_i : i \in I)$ with $X_i \sim \mP$ if $i \in S$ and $X_i \sim \mQ$ if $i \not \in S$. More generally, given an index set $I$ and $|I|$ distributions $\mP_1, \mP_2, \dots, \mP_{|I|}$, let $\mathcal{M}_I(\mP_i : i \in I)$ denote the distribution of independent random variables $(X_i : i \in I)$ with $X_i \sim \mP_i$ for each $i \in I$. For example, a planted clique in $\mG(n, 1/2)$ on the set $S \subseteq [n]$ can be written as $\mathcal{M}_I\left(\binom{S}{2}, \text{Bern}(1), \text{Bern}(1/2) \right)$ where $I = \binom{[n]}{2}$.

We now review the guarantees for the subroutine $\textsc{Gaussianize}$. The variant presented here is restated from \cite{brennan2019optimal} to be over a general index set $I$ rather than matrices, and with the rejection kernel parameter $R_{\pr{rk}}$ decoupled from the size $n$ of $I$, as shown in Figure \ref{fig:rej-kernel}. $\textsc{Gaussianize}$ maps a set of planted Bernoulli random variables to a set of independent Gaussian random variables with corresponding planted means. The procedure applies a Gaussian rejection kernel entrywise and its total variation guarantees follow by a simple application of the tensorization property of $\TV$ from Fact \ref{tvfacts}.

\begin{lemma}[Gaussianization -- Lemma 4.5 in \cite{brennan2019optimal}] \label{lem:gaussianize}
Let $I$ be an index set with $|I| = n$ and let $R_{\pr{rk}}$, $0 < q < p \le 1$ and $\delta$ be as in Lemma \ref{lem:5c}. Let $\mu_i$ be such that $0 \le \mu_i \le \tau$ for each $i \in I$ where the parameter $\tau > 0$ satisfies that
$$\tau \le \frac{\delta}{2 \sqrt{6\log R_{\pr{rk}} + 2\log (P - Q)^{-1}}}$$
The algorithm $\mathcal{A} = \textsc{Gaussianize}$ runs in $\textnormal{poly}(n, R_{\pr{rk}})$ time and satisfies that
\begin{align*}
\TV\left( \mathcal{A}(\mathcal{M}_I(S, \textnormal{Bern}(P), \textnormal{Bern}(Q))), \, \mathcal{M}_I\left( \mN(\mu_i \cdot \mathbf{1}(i \in S), 1) : i \in I \right) \right) &= O\left(n \cdot R_{\pr{rk}}^{-3}\right)
\end{align*}
for all subsets $S \subseteq I$.
\end{lemma}

%

\subsection{Cloning and Planting Diagonals}
\label{subsec:2-planting-diagonals}

We begin by reviewing the subroutine $\textsc{Graph-Clone}$, shown in Figure \ref{fig:clone}, which was introduced in \cite{brennan2019universality} and produces several independent samples from a planted subgraph problem given a single sample. Its properties as a Markov kernel are stated in the next lemma, which is proven by showing the two explicit expressions for $\bP[x^{ij} = v]$ in Step 1 define valid probability distributions and then explicitly writing the mass functions of $\mathcal{A}\left( \mG(n, q) \right)$ and $\mathcal{A}\left( \mG(n, S, p, q) \right)$.

\begin{figure}[t!]
\begin{algbox}
\textbf{Algorithm} \textsc{Graph-Clone}

\vspace{1mm}

\textit{Inputs}: Graph $G \in \mG_n$, the number of copies $t$, parameters $0 < q < p \le 1$ and $0 < Q < P \le 1$ satisfying $\frac{1 - p}{1 - q} \le \left( \frac{1 - P}{1 - Q} \right)^t$ and $\left( \frac{P}{Q} \right)^t \le \frac{p}{q}$

\begin{enumerate}
\item Generate $x^{ij} \in \{0, 1\}^t$ for each $1 \le i < j \le n$ such that:
\begin{itemize}
\item If $\{i, j \} \in E(G)$, sample $x^{ij}$ from the distribution on $\{0, 1\}^t$ with
$$\bP[x^{ij} = v] = \frac{1}{p - q} \left[ (1 - q) \cdot P^{|v|_1} (1 - P)^{t - |v|_1} - (1 - p) \cdot Q^{|v|_1} (1 - Q)^{t - |v|_1} \right]$$
\item If $\{i, j \} \not \in E(G)$, sample $x^{ij}$ from the distribution on $\{0, 1\}^t$ with
$$\bP[x^{ij} = v] = \frac{1}{p - q} \left[ p \cdot Q^{|v|_1} (1 - Q)^{t - |v|_1} - q \cdot P^{|v|_1} (1 - P)^{t - |v|_1} \right]$$
\end{itemize}
\item Output the graphs $(G_1, G_2, \dots, G_t)$ where $\{i, j\} \in E(G_k)$ if and only if $x^{ij}_k = 1$.
\end{enumerate}
\vspace{1mm}

\end{algbox}
\caption{Subroutine $\textsc{Graph-Clone}$ for producing independent samples from planted graph problems from \cite{brennan2019universality}.}
\label{fig:clone}
\end{figure}

\begin{lemma}[Graph Cloning -- Lemma 5.2 in \cite{brennan2019universality}] \label{lem:graphcloning}
Let $t \in \mathbb{N}$, $0 < q < p \le 1$ and $0 < Q < P \le 1$ satisfy that
$$\frac{1 - p}{1 - q} \le \left( \frac{1 - P}{1 - Q} \right)^t \quad \text{and} \quad \left( \frac{P}{Q} \right)^t \le \frac{p}{q}$$
Then the algorithm $\mathcal{A} = \textsc{Graph-Clone}$ runs in $\textnormal{poly}(t, n)$ time and satisfies that for each $S \subseteq [n]$,
$$\mathcal{A}\left( \mG(n, q) \right) \sim \mG(n, Q)^{\otimes t} \quad \text{and} \quad \mathcal{A}\left( \mG(n, S, p, q) \right) \sim \mG(n, S, P, Q)^{\otimes t}$$
\end{lemma}

Graph cloning more generally produces a method to clone a set of Bernoulli random variables indexed by a general index set $I$ instead of the possible edges of a graph on the vertex set $[n]$. The guarantees for this subroutine are stated in the following lemma. We remark that both of these lemmas will always be applied with $t = O(1)$, resulting in a constant loss in signal strength.

\begin{lemma}[Bernoulli Cloning] \label{lem:bern-clone}
Let $I$ be an index set with $|I| = n$, let $t \in \mathbb{N}$, $0 < q < p \le 1$ and $0 < Q < P \le 1$ satisfy that
$$\frac{1 - p}{1 - q} \le \left( \frac{1 - P}{1 - Q} \right)^t \quad \text{and} \quad \left( \frac{P}{Q} \right)^t \le \frac{p}{q}$$
There is an algorithm $\mathcal{A} = \textsc{Bernoulli-Clone}$ that runs in $\textnormal{poly}(t, n)$ time and satisfying
\begin{align*}
&\mathcal{A}\left( \mathcal{M}_I(\textnormal{Bern}(q)) \right) \sim \mathcal{M}_I(\textnormal{Bern}(Q))^{\otimes t} \quad \text{and} \\
&\mathcal{A}\left( \mathcal{M}_I(S, \textnormal{Bern}(p), \textnormal{Bern}(q)) \right) \sim \mathcal{M}_I(S, \textnormal{Bern}(P), \textnormal{Bern}(Q))^{\otimes t}
\end{align*}
for each $S \subseteq I$.
\end{lemma}

\begin{figure}[t!]
\begin{algbox}
\textbf{Algorithm} $\pr{To-}k\textsc{-Partite-Submatrix}$

\vspace{1mm}

\textit{Inputs}: $k$\pr{-pds} instance $G \in \mG_N$ with clique size $k$ that divides $N$ and partition $E$ of $[N]$, edge probabilities $0 < q < p \le 1$ with $q = N^{-O(1)}$ and target dimension $n \ge \left(\frac{p}{Q} + 1 \right)N$ where $Q = 1 - \sqrt{(1 - p)(1 - q)} + \mathbf{1}_{\{p = 1\}} \left( \sqrt{q} - 1 \right)$ and $k$ divides $n$
\begin{enumerate}
\item Apply $\textsc{Graph-Clone}$ to $G$ with edge probabilities $P = p$ and $Q = 1 - \sqrt{(1 - p)(1 - q)} + \mathbf{1}_{\{p = 1\}} \left( \sqrt{q} - 1 \right)$ and $t = 2$ clones to obtain $(G_1, G_2)$.
\item Let $F$ be a partition of $[n]$ with $[n] = F_1 \cup F_2 \cup \cdots \cup F_k$ and $|F_i| = n/k$. Form the matrix $M_{\text{PD}} \in \{0, 1\}^{n \times n}$ as follows:
\begin{enumerate}
\item[(1)] For each $t \in [k]$, sample $s_1^t \sim \text{Bin}(N/k, p)$ and $s_2^t \sim \text{Bin}(n/k, Q)$ and let $S_t$ be a subset of $F_t$ with $|S_t| = N/k$ selected uniformly at random. Sample $T_1^t \subseteq S_t$ and $T_2^t \subseteq F_t \backslash S_t$ with $|T_1^t| = s_1^t$ and $|T_2^t| = \max\{s_2^t - s_1^t, 0 \}$ uniformly at random.
\item[(2)] Now form the matrix $M_{\text{PD}}$ such that its $(i, j)$th entry is
$$(M_{\text{PD}})_{ij} = \left\{ \begin{array}{ll} \mathbf{1}_{\{\pi_t(i), \pi_t(j)\} \in E(G_1)} & \text{if } i < j \text{ and } i, j \in S_t \\ \mathbf{1}_{\{\pi_t(i), \pi_t(j)\} \in E(G_2)} & \text{if } i > j \text{ and } i, j \in S_t \\ \mathbf{1}_{\{ i \in T_1^t \}} & \text{if } i = j \text{ and } i, j \in S_t \\ \mathbf{1}_{\{i \in T_2^t\}} & \text{if } i = j \text{ and } i, j \in F_t \backslash S_t \\ \sim_{\text{i.i.d.}} \text{Bern}(Q) & \text{if } i \neq j \text{ and } (i, j) \not \in S_t^2 \text{ for a } t \in [k] \end{array} \right.$$
where $\pi_t : S_t \to E_t$ is a bijection chosen uniformly at random.
\end{enumerate}
\item Output the matrix $M_{\text{PD}}$ and the partition $F$.
\end{enumerate}
\vspace{1mm}

\end{algbox}
\caption{Subroutine $\pr{To-}k\textsc{-Partite-Submatrix}$ for mapping from an instance of $k$-partite planted dense subgraph to a $k$-partite Bernoulli submatrix problem.}
\label{fig:tosubmatrix}
\end{figure}

We now introduce the procedure $\pr{To-}k\textsc{-Partite-Submatrix}$, which is shown in Figure \ref{fig:tosubmatrix} and will be crucial in our reductions to dense variants of the stochastic block model. This reduction clones the upper half of the adjacency matrix of the input graph problem to produce an independent lower half and plants diagonal entries while randomly embedding into a larger matrix to hide the diagonal entries in total variation. $\pr{To-}k\textsc{-Partite-Submatrix}$ is similar to $\textsc{To-Submatrix}$ in \cite{brennan2019universality} and $\textsc{To-Bernoulli-Submatrix}$ in \cite{brennan2019optimal} but ensures that the random embedding step accounts for the $k$-partite promise of the input $k$\pr{-pds} instance. Completing the missing diagonal entries in the adjacency matrix will be crucial to apply one of our main techniques, Bernoulli rotations, which will be introduced in the next section.

The next lemma states the total variation guarantees of $\pr{To-}k\textsc{-Partite-Submatrix}$ and is a $k$-partite variant of Theorem 6.1 in \cite{brennan2019universality}. Although technically more subtle than the analysis of $\textsc{To-Submatrix}$ in \cite{brennan2019universality}, this proof is tangential to our main reduction techniques and deferred to Appendix \ref{subsec:appendix-2-k-partite}. Given a partition $E$ of $[N]$ with $k$ parts, let $\mU_N(E)$ denote the uniform distribution over all $k$-subsets of $[N]$ containing exactly one element from each part of $E$.

\begin{lemma}[Reduction to $k$-Partite Bernoulli Submatrix Problems] \label{lem:submatrix}
Let $0 < q < p \le 1$ and $Q = 1 - \sqrt{(1 - p)(1 - q)} + \mathbf{1}_{\{p = 1\}} \left( \sqrt{q} - 1 \right)$. Suppose that $n$ and $N$ are such that
$$n \ge \left( \frac{p}{Q} + 1 \right) N \quad \text{and} \quad k \le QN/4$$
Also suppose that $q = N^{-O(1)}$ and both $N$ and $n$ are divisible by $k$. Let $E = (E_1, E_2, \dots, E_k)$ and $F = (F_1, F_2, \dots, F_k)$ be partitions of $[N]$ and $[n]$, respectively. Then it follows that the algorithm $\mathcal{A} = \pr{To-}k\textsc{-Partite-Submatrix}$ runs in $\textnormal{poly}(N)$ time and satisfies
\begin{align*}
\TV\left( \mathcal{A}(\mG(N, \mU_N(E), p, q)), \, \mathcal{M}_{[n] \times [n]} \left(\mU_n(F), \textnormal{Bern}(p), \textnormal{Bern}(Q) \right) \right) &\le 4k \cdot \exp \left( - \frac{Q^2N^2}{48pkn} \right) + \sqrt{\frac{C_Q k^2}{2n}} \\
\TV\left( \mathcal{A}(\mG(N, q)), \, \textnormal{Bern}\left( Q \right)^{\otimes n \times n} \right) &\le 4k \cdot \exp \left( - \frac{Q^2N^2}{48pkn} \right)
\end{align*}
where $C_Q = \max \left\{ \frac{Q}{1 - Q}, \frac{1 - Q}{Q} \right\}$.
\end{lemma}

For completeness, we give an intuitive summary of the technical subtleties arising in the proof of this lemma. After applying $\textsc{Graph-Clone}$, the adjacency matrix of the input graph $G$ is still missing its diagonal entries. The main difficulty in producing these diagonal entries is to ensure that entries corresponding to vertices in the planted subgraph are properly sampled from $\text{Bern}(p)$. To do this, we randomly embed the original $N \times N$ adjacency matrix in a larger $n \times n$ matrix with i.i.d. entries from $\text{Bern}(Q)$ and sample all diagonal entries corresponding to entries of the original matrix from $\text{Bern}(p)$. The diagonal entries in the new $n - N$ columns are chosen so that the supports on the diagonals within each $F_t$ each have size $\text{Bin}(n/k, Q)$. Even though this causes the sizes of the supports on the diagonals in each $F_t$ to have the same distribution under both $H_0$ and $H_1$, the randomness of the embedding and the fact that $k = o(\sqrt{n})$ ensures that this is hidden in total variation.

\subsection{Symmetric 3-ary Rejection Kernels}
\label{subsec:srk}

\begin{figure}[t!]
\begin{algbox}
\textbf{Algorithm} \textsc{3-srk}$(B, \mP_+, \mP_-, \mQ)$

\vspace{2mm}

\textit{Parameters}: Input $B \in \{-1, 0, 1\}$, number of iterations $N$, parameters $a \in (0, 1)$ and sufficiently small nonzero $\mu_1, \mu_2 \in \mathbb{R}$, distributions $\mP_+, \mP_-$ and $\mQ$ over a measurable space $(X, \mathcal{B})$ such that $(\mP_+, \mQ)$ and $(\mP_-, \mQ)$ are computable pairs
\begin{enumerate}
\item Initialize $z$ arbitrarily in the support of $\mQ$.
\item Until $z$ is set or $N$ iterations have elapsed:
\begin{enumerate}
\item[(1)] Sample $z' \sim \mQ$ independently and compute the two quantities
$$\mL_1(z') = \frac{d\mP_+}{d\mQ} (z') - \frac{d\mP_-}{d\mQ} (z') \quad \text{and} \quad \mL_2(z') = \frac{d\mP_+}{d\mQ} (z') + \frac{d\mP_-}{d\mQ} (z') - 2$$
\item[(2)] Proceed to the next iteration if it does not hold that
$$2|\mu_1| \ge \left| \mL_1(z') \right| \quad \text{and} \quad \frac{2|\mu_2|}{\max\{a, 1 - a\}} \ge |\mL_2(z')|$$
\item[(3)] Set $z \gets z'$ with probability $P_A(x, B)$ where
$$P_A(x, B) = \frac{1}{2} \cdot \left\{ \begin{array}{ll} 1+ \frac{a}{4\mu_2} \cdot \mL_2(z') + \frac{1}{4\mu_1} \cdot \mL_1(z') &\text{if } B = 1 \\ 1 - \frac{1 - a}{4\mu_2} \cdot \mL_2(z') &\text{if } B = 0 \\ 1+ \frac{1}{4\mu_2} \cdot \mL_2(z') - \frac{a}{4\mu_1} \cdot \mL_1(z') &\text{if } B = -1 \end{array} \right.$$
\end{enumerate}
\item Output $z$.
\end{enumerate}
\vspace{1mm}
\end{algbox}
\caption{3-ary symmetric rejection kernel algorithm.}
\label{fig:srej-kernel}
\end{figure}

In this section, we introduce symmetric 3-ary rejection kernels, which will be the key gadget in our reduction showing universality of lower bounds for learning sparse mixtures in Section \ref{sec:universality}. In order to map to universal formulations of sparse mixtures, it is crucial to produce a nontrivial instance of a sparse mixture with multiple planted distributions. Since previous rejection kernels all begin with binary inputs, they do not have enough degrees of freedom to map to three output distributions. The symmetric 3-ary rejection kernels $3\pr{-srk}$ introduced in this section overcome this issue by mapping from distributions supported on $\{-1, 0, 1\}$ to three output distributions $\mP_+, \mP_-$ and $\mQ$. In order to produce clean total variation guarantees, these rejection kernels also exploit symmetry in their three input distributions on $\{-1, 0, 1\}$.

Let $\text{Tern}(a, \mu_1, \mu_2)$ where $a \in (0, 1)$ and $\mu_1, \mu_2 \in \mathbb{R}$ denote the probability distribution on $\{-1, 0, 1\}$ such that if $B \sim \text{Tern}(a, \mu_1, \mu_2)$ then
$$\bP[X = -1] = \frac{1 - a}{2} - \mu_1 + \mu_2, \quad \bP[X = 0] = a - 2\mu_2, \quad \bP[X = 1] = \frac{1 - a}{2} + \mu_1 + \mu_2$$
if all three of these probabilities are nonnegative. The map $3\pr{-srk}(B)$, shown in Figure \ref{fig:srej-kernel}, sends an input $B \in \{-1, 0, 1\}$ to a set $X$ simultaneously satisfying three Markov transition properties:
\begin{enumerate}
\item if $B \sim \text{Tern}(a, \mu_1, \mu_2)$, then $3\textsc{-srk}(B)$ is close to $\mP_+$ in total variation;
\item if $B \sim \text{Tern}(a, -\mu_1, \mu_2)$, then $3\textsc{-srk}(B)$ is close to $\mQ$ in total variation; and
\item if $B \sim \text{Tern}(a, 0, 0)$, then $3\textsc{-srk}(B)$ is close to $\mP_-$ in total variation.
\end{enumerate}
In order to state our main results for $3\pr{-srk}(B)$, we will need the notion of computable pairs from \cite{brennan2019universality}. The definition below is that given in \cite{brennan2019universality}, without the assumption of finiteness of KL divergences. This assumption was convenient for the Chernoff exponent analysis needed for multivariate rejection kernels in \cite{brennan2019universality}. Since our rejection kernels are univariate, we will be able to state our universality conditions directly in terms of tail bounds rather than Chernoff exponents.

\begin{definition}[Relaxed Computable Pair \cite{brennan2019universality}] \label{def:computable}
Define a pair of sequences of distributions $(\mP, \mQ)$ over a measurable space $(X, \mathcal{B})$ where $\mP = (\mP_n)$ and $\mQ = (\mQ_n)$ to be computable if:
\begin{enumerate}
\item there is an oracle producing a sample from $\mQ_n$ in $\textnormal{poly}(n)$ time;
\item for all $n$, $\mP_n$ and $\mQ_n$ are mutually absolutely continuous and the likelihood ratio satisfies
$$\bE_{x \sim \mQ_n} \left[\frac{d\mP_n}{d\mQ_n}(x) \right] = \bE_{x \sim \mP_n}\left[\left( \frac{d\mP_n}{d\mQ_n}(x) \right)^{-1} \right] = 1$$
where $\frac{d\mP_n}{d\mQ_n}$ is the Radon-Nikodym derivative; and
\item there is an oracle computing $\frac{d\mP_n}{d\mQ_n} (x)$ in $\textnormal{poly}(n)$ time for each $x \in X$.
\end{enumerate}
\end{definition}

We remark that the second condition above always holds for discrete distributions and generally for most well-behaved distributions $\mP$ and $\mQ$. We now state our main total variation guarantees for $3\pr{-srk}$. The proof of the next lemma follows a similar structure to the analysis of rejection sampling as in Lemma 5.1 of \cite{brennan2018reducibility} and Lemma 5.1 of \cite{brennan2019universality}. However, the bounds that we obtain are different than those in \cite{brennan2018reducibility, brennan2019universality} because of the symmetry of the three input $\text{Tern}$ distributions. The proof of this lemma is deferred to Appendix \ref{subsec:appendix-3-ary}.

\begin{lemma}[Symmetric 3-ary Rejection Kernels] \label{lem:srk}
Let $a \in (0, 1)$ and $\mu_1, \mu_2 \in \mathbb{R}$ be nonzero and such that $\textnormal{Tern}(a, \mu_1, \mu_2)$ is well-defined. Let $\mP_+, \mP_-$ and $\mQ$ be distributions over a measurable space $(X, \mathcal{B})$ such that $(\mP_+, \mQ)$ and $(\mP_-, \mQ)$ are computable pairs with respect to a parameter $n$. Let $S \subseteq X$ be the set
$$S = \left\{x \in X : 2|\mu_1| \ge \left| \frac{d\mP_+}{d\mQ} (x) - \frac{d\mP_-}{d\mQ} (x) \right| \quad \textnormal{and} \quad \frac{2|\mu_2|}{\max\{a, 1 - a\}} \ge \left|\frac{d\mP_+}{d\mQ} (x) + \frac{d\mP_-}{d\mQ} (x) - 2 \right| \right\}$$
Given a positive integer $N$, then the algorithm $3\textsc{-srk} : \{-1, 0, 1\} \to X$ can be computed in $\textnormal{poly}(n, N)$ time and satisfies that
$$\left. \begin{array}{r} \TV\left( 3\textsc{-srk}(\textnormal{Tern}(a, \mu_1, \mu_2)), \mP_+ \right) \\ \TV\left( 3\textsc{-srk}(\textnormal{Tern}(a, -\mu_1, \mu_2)), \mP_- \right) \\\TV\left( 3\textsc{-srk}(\textnormal{Tern}(a, 0, 0)), \mQ \right) \end{array} \right\} \le 2\delta \left(1 + |\mu_1|^{-1} + |\mu_2|^{-1} \right) + \left( \frac{1}{2} + \delta \left( 1 + |\mu_1|^{-1} + |\mu_2|^{-1} \right) \right)^N$$
where $\delta > 0$ is such that $\bP_{X \sim \mP_+}[X \not \in S]$, $\bP_{X \sim \mP_-}[X \not \in S]$ and $\bP_{X \sim \mQ}[X \not \in S]$ are upper bounded by $\delta$.
\end{lemma}

\section{Dense Bernoulli Rotations}
\label{sec:2-bernoulli-rotations}

In this section, we formally introduce dense Bernoulli rotations and constructions for their design matrices and tensors, which will play an essential role in all of our reductions. For an overview of the main high level ideas underlying these techniques, see Sections \ref{subsec:1-tech-dbr} and \ref{subsec:1-tech-design-matrices}. As mentioned in Sections \ref{subsec:1-tech-dbr}, dense Bernoulli rotations map $\pr{pb}(T, i, p, q)$ to $\mN\left( \mu \lambda^{-1} \cdot A_{i}, I_m\right)$  for each $i \in [T]$ and $\textnormal{Bern}(q)^{\otimes T}$ to $\mN\left( 0, I_m\right)$ approximately in total variation, where $\mu = \tilde{\Theta}(1)$, the vectors $A_1, A_2, \dots, A_T \in \mathbb{R}^m$ are for us to design and $\lambda$ is an upper bound on the singular values of the matrix with columns $A_i$.

Simplifying some technical details, our reduction to $\pr{rsme}$ in Section \ref{subsec:3-rsme-reduction} roughly proceeds as follows: (1) its input is a $k\pr{-bpc}$ instance with parts of size $M$ and $N$ and biclique dimensions $k = k_M$ and $k_N$; (2) it applies dense Bernoulli rotations with $p = 1$ and $q = 1/2$ to the $Mk_N$ vectors of length $T = N/k_N$ representing the adjacency patterns in $\{0, 1\}^{N/k_N}$ between each of the $M$ left vertices and each part in the partition of the right vertices; and (3) it pads the resulting matrix with standard normals so that it has $d$ rows. Under $H_1$, the result is a $d \times k_N m$ matrix $\mathbf{1}_S u^\top + \mN(0, 1)^{\otimes d \times k_N m}$ where $S$ is the left vertex set of the biclique and $u$ consists of scaled concatenations of the $A_i$. We design the adversary so that the target data matrix $D$ in $\pr{rsme}$ is roughly of the form
$$D_{ij} \sim \left\{ \begin{array}{ll} \mN\left(\tau k^{-1/2}, 1\right) &\text{if } i \in S \text{ and } j \text{ is not corrupted} \\ \mN\left(\epsilon^{-1}(1 - \epsilon) \tau k^{-1/2}, 1\right) &\text{if } i \in S \text{ and } j \text{ is corrupted} \\ \mN(0, 1) &\text{otherwise} \end{array} \right.$$
for each $i \in [d]$ and $j \in [n]$ where $n = k_N m$. Matching the two distributions above, we arrive at the following desiderata for the $A_i$.
\begin{itemize}
\item We would like each $\lambda^{-1} A_i$ to consist of $(1 - \epsilon') m$ entries equal to $\tau k^{-1/2}$ and $\epsilon' m$ entries equal to $\epsilon'^{-1}(1 - \epsilon') \tau k^{-1/2}$ where $\tau$ is just below the desired computational barrier $\tau = \tilde{\Theta}(k^{1/2} \epsilon^{1/2} n^{-1/4})$ and $\epsilon' \le \epsilon$ where $\epsilon' = \Theta(\epsilon)$.
\item Now observe that the norm of any such $\lambda^{-1} A_i$ is $\Theta\left( \tau \epsilon^{-1/2} m^{1/2} k^{-1/2} \right)$ which is just below a norm of $\tilde{\Theta}(m^{1/2} n^{-1/4})$ at the computational barrier for $\pr{rsme}$. Note that the normalization by $\lambda^{-1}$ ensures that each $\lambda^{-1} A_i$ has $\ell_2$ norm at most $1$. To be as close to the computational barrier as possible, it is necessary that $m^{1/2} n^{-1/4} = \tilde{\Theta}(1)$ which rearranges to $m = \tilde{\Theta}(k_N)$ since $n = k_N m$.
\item When the input is an instance of $k\pr{-bpc}$ nearly at its computational barrier, we have that $N  = \tilde{\Theta}(k_N^2)$ and thus our necessary condition above implies that $m = \tilde{\Theta}(N/k_N) = \tilde{\Theta}(T)$, and hence that $A$ is nearly square. Furthermore, if we take the $A_i$ to be unit vectors, our desiderata that the $\lambda^{-1} A_i$ have norm $\tilde{\Theta}(m^{1/2} n^{-1/4})$ reduces to $\lambda = \tilde{\Theta}(1)$.
\end{itemize}
Summarizing this discussion, we arrive at exactly the three conditions outlined in Section \ref{subsec:1-tech-design-matrices}. We remark that while these desiderata are tailored to $\pr{rsme}$, they will also turn out to be related to the desired properties of $A$ in our other reductions. We now formally introduce dense Bernoulli rotations.

\subsection{Mapping Planted Bits to Spiked Gaussian Tensors}
\label{subsec:2-planted-bits}

%

\begin{figure}[t!]
\begin{algbox}
\textbf{Algorithm} \textsc{Bern-Rotations}

\vspace{1mm}

\textit{Inputs}: Vector $V \in \{0, 1\}^n$, rejection kernel parameter $R_{\pr{rk}}$, Bernoulli probability parameters $0 < q < p \le 1$, output dimension $m$, an $m \times n$ matrix $A$ with singular values all at most $\lambda > 0$, intermediate mean parameter $\mu > 0$

\begin{enumerate}
\item Form $V_1 \in \{0, 1\}^n$ by applying $\pr{Gaussianize}$ to the entries in the vector $V$ with rejection kernel parameter $R_{\pr{rk}}$, Bernoulli probabilities $q$ and $p$ and target mean parameters all equal to $\mu$.
\item Sample a vector $U \sim \mN(0, 1)^{\otimes m}$ and let $\left(I_m - \lambda^{-2} \cdot AA^\top\right)^{1/2}$ be the positive semidefinite square root of $I_m - \lambda^{-2} \cdot AA^\top$. Now form the vector
$$V_2 = \lambda^{-1} \cdot AV_1 +\left(I_m - \lambda^{-2} \cdot AA^\top\right)^{1/2}U$$
\item Output the vector $V_2$.
\end{enumerate}
\vspace{1mm}

\textbf{Algorithm} \textsc{Tensor-Bern-Rotations}

\vspace{1mm}

\textit{Inputs}: Order $s$ tensor $T \in \mathcal{T}_{s, n}(\{0, 1\})$, rejection kernel parameter $R_{\pr{rk}}$, Bernoulli probability parameters $0 < q < p \le 1$, output dimension $m$, an $m \times n$ matrices $A_1, A_2, \dots, A_s$ with singular values less than or equal to $\lambda_1, \lambda_2, \dots, \lambda_s > 0$, respectively, mean parameter $\mu > 0$

\begin{enumerate}
\item Flatten $T$ into the vector $V_1 \in \{0, 1\}^{n^s}$, form the Kronecker product $A = A_1 \otimes A_2 \otimes \cdots \otimes A_s$ and set $\lambda = \lambda_1 \lambda_2 \cdots \lambda_s$.
\item Let $V_2$ be the output of $\pr{Bern-Rotations}$ applied to $V_1$ with parameters $R_{\pr{rk}}$, $0 < q < p \le 1, A, \lambda, \mu$ and output dimension $m^s$.
\item Rearrange the entries of $V_2$ into a tensor $T_1 \in \mathcal{T}_{s, m}(\mathbb{R})$ and output $T_1$.
\end{enumerate}
\vspace{1mm}

\end{algbox}
\caption{Subroutines $\textsc{Bern-Rotations}$ and $\textsc{Tensor-Bern-Rotations}$ for producing spiked Gaussian vectors and tensors, respectively, from the planted bits distribution.}
\label{fig:bern-rotations}
\end{figure}

Let $\pr{pb}(n, i, p, q)$ and $\pr{pb}(S, i, p, q)$ denote the planted bit distributions defined in Sections \ref{subsec:1-tech-dbr} and \ref{subsec:2-notation}. The procedures $\textsc{Bern-Rotations}$ and its derivative $\textsc{Tensor-Bern-Rotations}$ are shown in Figure \ref{fig:bern-rotations}. Recall that the subroutine $\textsc{Gaussianize}$ was introduced in Figure \ref{fig:rej-kernel}. Note that positive semidefinite square roots of $n \times n$ matrices can be computed in $\text{poly}(n)$ time. The two key Markov transition properties for these procedures that will be used throughout the paper are as follows.

\begin{lemma}[Dense Bernoulli Rotations] \label{lem:bern-rotations}
Let $m$ and $n$ be positive integers and let $A \in \mathbb{R}^{m \times n}$ be a matrix with singular values all at most $\lambda > 0$. Let $R_{\pr{rk}}$, $0 < q < p \le 1$ and $\mu$ be as in Lemma \ref{lem:5c}. Let $\mathcal{A}$ denote $\textsc{Bern-Rotations}$ applied with rejection kernel parameter $R_{\pr{rk}}$, Bernoulli probability parameters $0 < q < p \le 1$, output dimension $m$, matrix $A$ with singular value upper bound $\lambda$ and mean parameter $\mu$. Then $\mathcal{A}$ runs in $\textnormal{poly}(n, R_{\pr{rk}})$ time and it holds that
\begin{align*}
\TV\left( \mathcal{A}\left( \pr{pb}(n, i, p, q) \right), \, \mN\left( \mu \lambda^{-1} \cdot A_{\cdot, i}, I_m\right) \right) &= O\left(n \cdot R_{\pr{rk}}^{-3}\right) \\
\TV\left( \mathcal{A}\left( \textnormal{Bern}(q)^{\otimes n} \right), \, \mN\left( 0, I_m\right) \right) &= O\left(n \cdot R_{\pr{rk}}^{-3}\right)
\end{align*}
for all $i \in [n]$, where $A_{\cdot, i}$ denotes the $i$th column of $A$.
\end{lemma}

\begin{proof}
Let $\mathcal{A}_1$ denote the first step of $\mathcal{A} = \pr{Bern-Rotations}$ with input $V$ and output $V_1$, and let $\mathcal{A}_2$ denote the second step of $\mathcal{A}$ with input $V_1$ and output $V_2$. Fix some index $i \in [n]$. Now Lemma \ref{lem:gaussianize} implies
\begin{equation} \label{eqn:gaussianize1}
\TV\left( \mathcal{A}_1\left( \pr{pb}(n, i, p, q) \right), \, \mN\left( \mu \cdot e_i, I_n \right) \right) = O\left(n \cdot R_{\pr{rk}}^{-3}\right)
\end{equation}
where $e_i \in \mathbb{R}^n$ is the $i$th canonical basis vector. Suppose that $V_1 \sim \mN\left( \mu \cdot e_i, I_n \right)$ and let $V_1 = \mu \cdot e_i + W$ where $W \sim \mN(0, I_n)$. Note that the entries of $AW$ are jointly Gaussian and $\text{Cov}(AW) = AA^\top$. Therefore, we have that
$$AV_1 = \mu \cdot A_{\cdot, i} + AW \sim \mN\left( \mu \cdot A_{\cdot, i}, AA^\top \right)$$
If $U \sim \mN(0, 1)^{\otimes m}$ is independent of $W$, then the entries of $AW + \left(\lambda^2 \cdot I_m - \cdot AA^\top\right)^{1/2}U$ are jointly Gaussian. Furthermore, since both terms are mean zero and independent the covariance matrix of this vector is given by
\begin{align*}
\text{Cov}\left( AW + \left(\lambda^2 \cdot I_m - AA^\top\right)^{1/2}U \right) &= \text{Cov}\left( AW \right) + \text{Cov}\left( \left(\lambda^2 \cdot I_m - AA^\top\right)^{1/2}U \right) \\
&= AA^\top + (\lambda^2 \cdot I_m - AA^\top) = \lambda^2 \cdot I_m
\end{align*}
Therefore it follows that $AW + \left(\lambda^2 \cdot I_m - AA^\top\right)^{1/2}U \sim \mN(0, \lambda^2 \cdot I_m)$ and furthermore that
$$V_2 = \lambda^{-1} \cdot AV_1 +\left(I_m - \lambda^{-2} \cdot AA^\top\right)^{1/2}U \sim \mN\left( \mu \lambda^{-1} \cdot A_{\cdot, i}, I_m\right)$$
Where $V_2 \sim \mathcal{A}_2\left( \mN\left( \mu \cdot e_i, I_n \right) \right)$. Now applying $\mathcal{A}_2$ to both distributions in Equation (\ref{eqn:gaussianize1}) and the data-processing inequality prove that $\TV\left( \mathcal{A}\left( \pr{pb}(n, i, p, q) \right), \, \mN\left( \mu \lambda^{-1} \cdot A_{\cdot, i}, I_m\right) \right) = O\left(n \cdot R_{\pr{rk}}^{-3}\right)$. This argument analyzing $\mathcal{A}_2$ applied with $\mu = 0$ yields that $\mathcal{A}_2\left( \mN(0, I_n) \right) \sim \mN(0, I_m)$. Combining this with
$$\TV\left( \mathcal{A}_1\left( \textnormal{Bern}(q)^{\otimes n} \right), \, \mN\left( 0, I_n \right) \right) = O\left(n \cdot R_{\pr{rk}}^{-3}\right)$$
from Lemma \ref{lem:gaussianize} now yields the bound $\TV\left( \mathcal{A}\left( \textnormal{Bern}(q)^{\otimes n} \right), \, \mN\left( 0, I_n\right) \right) = O\left(n \cdot R_{\pr{rk}}^{-3}\right)$, which completes the proof of the lemma. 
\end{proof}

\begin{corollary}[Tensor Bernoulli Rotations] \label{cor:tensor-bern-rotations}
Let $s, m$ and $n$ be positive integers, let $A_1, A_2, \dots, A_s \in \mathbb{R}^{m \times n}$ be matrices with singular values less than or equal to $\lambda_1, \lambda_2, \dots, \lambda_s > 0$, respectively. Let $R_{\pr{rk}}$, $0 < q < p \le 1$ and $\mu$ be as in Lemma \ref{lem:5c}. Let $\mathcal{A}$ denote $\textsc{Tensor-Bern-Rotations}$ applied with parameters $0 < q < p \le 1$, output dimension $m$, matrix $A = A_1 \otimes A_2 \otimes \cdots \otimes A_s$ with singular value upper bound $\lambda = \lambda_1 \lambda_2 \cdots \lambda_s$ and mean parameter $\mu$. If $s$ is a constant, then $\mathcal{A}$ runs in $\textnormal{poly}(n, R_{\pr{rk}})$ time and it holds that for each $e \in [n]^s$,
\begin{align*}
\TV\left( \mathcal{A}\left( \pr{pb}_s(n, e, p, q) \right), \, \mN\left( \mu (\lambda_1\lambda_2 \cdots \lambda_s)^{-1} \cdot A_{\cdot, e_1} \otimes A_{\cdot, e_2} \otimes \cdots \otimes A_{\cdot, e_s}, I_m^{\otimes s} \right) \right) &= O\left( n^s \cdot R_{\pr{rk}}^{-3} \right) \\
\TV\left( \mathcal{A}\left( \textnormal{Bern}(q)^{\otimes n^{\otimes s}} \right), \, \mN\left( 0, I_m^{\otimes s} \right) \right) &= O\left( n^s \cdot R_{\pr{rk}}^{-3} \right)
\end{align*}
where $A_{\cdot, i}$ denotes the $i$th column of $A$.
\end{corollary}

\begin{proof}
Let $\sigma_i^j$ for $1 \le i \le r_j$ be the nonzero singular values of $A_j$ for each $1 \le j \le s$. Then the nonzero singular values of the Kronecker product $A = A_1 \otimes A_2 \otimes \cdots \otimes A_s$ are all of the products $\sigma_{i_1}^1 \sigma_{i_2}^2 \cdots \sigma_{i_s}^s$ for all $(i_1, i_2, \dots, i_s)$ with $1 \le i_j \le r_j$ for each $1 \le j \le s$. Thus if $\sigma_i^j \le \lambda_j$ for each $1 \le j \le s$, then $\lambda = \lambda_1 \lambda_2 \cdots \lambda_s$ is an upper bound on the singular values of $A$. The corollary now follows by applying Lemma \ref{lem:bern-rotations} with parameters $p, q, \mu$ and $\lambda$, matrix $A$, output dimension $m^s$ and input dimension $n^s$.
\end{proof}

%

\subsection{$\mathbb{F}_r^t$ Design Matrices}
\label{subsec:2-design-matrices}


In this section, we introduce a family of matrices $K_{r, t}$ that plays a key role in constructing the matrices $A$ in our applications of dense Bernoulli rotations. Throughout this section, $r$ will denote a prime number and $t$ will denote a fixed positive integer. As outlined in the beginning of this section and in Section \ref{subsec:1-tech-design-matrices}, there are three desiderata of the matrices $K_{r, t}$ that are needed for our applications of dense Bernoulli rotations. In the context of $K_{r, t}$, these three properties are:
\begin{enumerate}
\item The rows of $K_{r, t}$ are unit vectors and close to orthogonal in the sense that the largest singular value of $K_{r, t}$ is bounded above by a constant.
\item The matrices $K_{r, t}$ both contain exactly two distinct real values as entries.
\item The matrices $K_{r, t}$ contain a fraction of approximately $1/r$ negative entries per column.
\end{enumerate}
The matrices $K_{r, t}$ are constructed based on the incidence structure of the points in $\mathbb{F}_r^t$ with the Grassmanian of hyperplanes in $\mathbb{F}_r^t$ and their affine shifts. The construction of $K_{r, t}$ is motivated by the projective geometry codes and their applications to constructing 2-block designs. We remark that a classic trick counting the number of ordered $d$-tuples of linearly independent vectors in $\mathbb{F}_r^t$ shows that the number of $d$-dimensional subspaces of $\mathbb{F}_r^t$ is
$$|\text{Gr}(d, \mathbb{F}_r^t)| = \frac{(r^t - 1)(r^t - r) \cdots (r^t - r^{d - 1})}{(r^d - 1)(r^d - r) \cdots (r^d - r^{d - 1})}$$
This implies that the number of hyperplanes in $\mathbb{F}_r^t$ is $\ell = \frac{r^t - 1}{r - 1}$. We now give the definition of the matrix $K_{r, t}$ as a weighted incidence matrix between the points of $\mathbb{F}_r^t$ and affine shifts of the hyperplanes in the Grassmanian $\text{Gr}(t - 1, \mathbb{F}_r^t)$.

\begin{definition}[Design Matrices $K_{r, t}$] \label{defn:Krt}
Let $P_1, P_2, \dots, P_{r^t}$ be an enumeration of the points in $\mathbb{F}_r^t$ and $V_1, V_2, \dots, V_\ell$, where $\ell = \frac{r^t - 1}{r - 1}$, be an enumeration of the hyperplanes in $\mathbb{F}_r^t$. For each $V_i$, let $u_i \neq 0$ denote a vector in $\mathbb{F}_r^t$ not contained in $V_i$. Define $K_{r, t} \in \mathbb{R}^{r\ell \times r^t}$ to be the matrix with the following entries
$$(K_{r, t})_{r(i-1) + a + 1, j} = \frac{1}{\sqrt{r^t(r - 1)}} \cdot \left\{ \begin{matrix} 1 & \textnormal{if } P_j \not \in V_i + au_i \\ 1 - r & \textnormal{if } P_j \in V_i + au_i \end{matrix} \right.$$
for each $a \in \{0, 1, \dots, r - 1\}$ where $V_i + v$ denotes the affine shift of $V_i$ by $v$.
\end{definition}

We now establish the key properties of $K_{r, t}$ in the following simple lemma. Note that the lemma implies that the submatrix consisting of the rows of $K_{r, t}$ corresponding to hyperplanes in $\mathbb{F}_r^t$ has rows that are exactly orthogonal. However, the additional rows of $K_{r, t}$ corresponding to affine shifts of these hyperplanes will prove crucial in preserving \emph{tightness to algorithms} in our average-case reductions. As established in the subsequent lemma, these additional rows only mildly perturb the largest singular value of the matrix.

\begin{lemma}[Sub-orthogonality of $K_{r,t}$] \label{lem:suborthogonalmatrices}
If $r \ge 2$ is prime, then $K_{r, t}$ satisfies that:
\begin{enumerate}
\item for each $1 \le i \le kr\ell$, it holds that $\|(K_{r, t})_i\|_2 = 1$;
\item the inner product between the rows $(K_{r, t})_i$ and $(K_{r, t})_j$ where $i \neq j$ are given by
$$\langle (K_{r, t})_i, (K_{r, t})_j \rangle = \left\{ \begin{array}{ll} -(r - 1)^{-1} & \textnormal{if } \lfloor (i-1)/r\rfloor = \lfloor (j-1)/r\rfloor \\ 0 & \textnormal{otherwise} \end{array} \right.$$
\item each column of $K_{r, t}$ contains exactly $\frac{r^t - 1}{r - 1}$ entries equal to $\frac{1 - r}{\sqrt{r^t(r - 1)}}$.
\end{enumerate}
\end{lemma}

\begin{proof}
Let $r_i$ denote the $i$th row $(K_{r, t})_i$ of $K_{r, t}$. Fix a pair $1 \le i < j \le r\ell$ and let $1 \le i' \le j' \le \ell$ and $a, b \in \{0, 1, \dots, r - 1\}$ be such that $i = r(i' - 1) + a$ and $j = r(j' - 1) + b$. The affine subspaces of $\mathbb{F}_r^t$ corresponding to $r_i$ and $r_j$ are then $A_i = V_{i'} + au_{i'}$ and $A_j = V_{j'} + bu_{j'}$. Observe that
$$\| r_i \|_2^2 = (r^t - |A_i|) \cdot \frac{1}{r^t(r - 1)} + |A_i| \cdot \frac{(1 - r)^2}{r^t(r - 1)} = 1$$
Similarly, we have that
$$\langle r_i, r_j \rangle = (r^t - |A_i \cup A_j|) \cdot \frac{1}{r^t(r - 1)} + (|A_i \cup A_j| - |A_i \cap A_j|) \cdot \frac{1 - r}{r^t(r - 1)} + |A_i \cap A_j| \cdot \frac{(1 - r)^2}{r^t(r - 1)}$$
for each $1 \le i, j \le r\ell$. Since the size of a subspace is invariant under affine shifts, we have that $|A_i| = |V_{i'}| = |A_j| = |V_{j'}| = r^{t - 1}$. Furthermore, since $A_i \cap A_j$ is the intersection of two affine shifts of subspaces of dimension $t - 1$ of $\mathbb{F}_r^t$, it follows that $A_i \cap A_j$ is either empty, an affine shift of a $(t - 2)$-dimensional subspace or equal to both $A_i$ and $A_j$. Note that if $i \neq j$, then $A_i$ and $A_j$ are distinct. We remark that when $t = 1$, each $A_i$ is an affine shift of the trivial hyperplane $\{0\}$ and thus is a singleton. Now note that the intersection $A_i \cap A_j$ is only empty if $A_i$ and $A_j$ are affine shifts of one another which occurs if and only if $\lfloor (i-1)/r\rfloor = i' = j' = \lfloor (j-1)/r\rfloor$. In this case, it follows that $|A_i \cup A_j| = |A_i| + |A_j| = 2r^{t - 1}$. In this case, we have
\begin{align*}
\langle r_i, r_j \rangle &= (r^t - 2r^{t - 1}) \cdot \frac{1}{r^t(r - 1)} + 2r^{t - 1} \cdot \frac{1 - r}{r^t(r - 1)} = - (r - 1)^{-1}
\end{align*}
If $i' \neq j'$, then $A_i \cap A_j$ is the affine shift of a $(t - 2)$-dimensional subspace which implies that $|A_i \cap A_j| = r^{t - 1}$. Furthermore, $|A_i \cup A_j| = |A_i| + |A_j| - |A_i \cap A_j| = 2r^{t - 1} - r^{t - 2}$. In this case, we have that
$$\langle r_i, r_j \rangle = (r - 1)^2 \cdot \frac{1}{r^2(r - 1)} - 2(r - 1) \cdot \frac{1}{r^2} + \frac{r - 1}{r^2} = 0$$
This completes the proof of (2). We remark that this last case never occurs if $t = 1$. Now note that any point is in exactly one affine shift of each $V_i$. Therefore each column contains exactly $\ell$ negative entries, which proves (3).
\end{proof}

The next lemma uses the computation of $\langle (K_{r, t})_i, (K_{r, t})_j \rangle$ above to compute the singular values of $K_{r, t}$.

\begin{lemma} \label{lem:Krtsv}
The nonzero singular values of $K_{r, t}$ are $\sqrt{1 + (r - 1)^{-1}}$ with multiplicity $(r - 1)\ell$.
\end{lemma}

\begin{proof}
Lemma \ref{lem:suborthogonalmatrices} shows that $(K_{r, t})(K_{r, t})^\top$ is block-diagonal with $\ell$ blocks of dimension $r \times r$. Furthermore, each block is of the form $\left(1 + (r - 1)^{-1} \right) I_r - (r - 1)^{-1} \mathbf{1} \mathbf{1}^\top$. The eigenvalues of each of these blocks are $1 + (r - 1)^{-1}$ with multiplicity $r - 1$ and $0$ with multiplicity $1$. Thus the eigenvalues of $(K_{r, t})(K_{r, t})^\top$ are $1 + (r - 1)^{-1}$ and $0$ with multiplicities $(r - 1)\ell$ and $\ell$, respectively, implying the result.
\end{proof}

\subsection{$\mathbb{F}_r^t$ Design Tensors}
\label{subsec:2-design-tensors}

In this section, we introduce a family of tensors $T_{r, t}^{(V_i, V_j, L)}$ that will be used in $\pr{Tensor-Bern-Rotations}$ in the matrix case with $s = 2$ to map to hidden partition models in Section \ref{sec:3-hidden-partition}. An overview of how these tensors will be used in dense Bernoulli rotations was given in Section \ref{subsec:1-tech-design-matrices}. Similar to the previous section, the $T_{r, t}^{(V_i, V_j, L)}$ are constructed to have the following properties:
\begin{enumerate}
\item Given a pair of hyperplanes $(V_i, V_j)$ and a linear function $L : \mathbb{F}_r \to \mathbb{F}_r$, the slice $T_{r, t}^{(V_i, V_j, L)}$ of the constructed tensor is an $r^t \times r^t$ matrix with Frobenius norm $\left\| T_{r, t}^{(V_i, V_j, L)} \right\|_F = 1$.
\item These slices are approximately orthogonal in the sense that the Gram matrix with entries given by the matrix inner products $\text{Tr}\left( T_{r, t}^{(V_i, V_j, L)} \cdot T_{r, t}^{(V_{i'}, V_{j'}, L')} \right)$ has a bounded spectral norm.
\item Each slice $T_{r, t}^{(V_i, V_j, L)}$ contains two distinct entries and is an average signed adjacency matrix of a hidden partition model i.e. has these two entries arranged into an $r$-block community structure.
\item Matrices formed by specific concatentations of $T_{r, t}^{(V_i, V_j, L)}$ into larger matrices remain the average signed adjacency matrices of hidden partition models. This will be made precise in Lemma \ref{lem:comm-align-tensors} and will be important in our reduction from $k\pr{-pc}$.
\end{enumerate}
The construction of the family of tensors $T_{r, t}^{(V_i, V_j, L)}$ is another construction using the incidence geometry of $\mathbb{F}_r^t$, but is more involved than the two constructions in the previous section. Throughout this section, we let $V_1, V_2, \dots, V_\ell$ and $P_1, P_2, \dots, P_{r^t}$ be an enumeration of the hyperplanes and points of $\mathbb{F}_r^t$ as in Definition \ref{defn:Krt}. Furthermore, for each $V_i$, we fix a particular point $u_i \neq 0$ of $\mathbb{F}_r^t$ not contained in $V_i$. In order to introduce the family $T_{r, t}^{(V_i, V_j, L)}$, we first define the following important class of bipartite graphs.

\begin{definition}[Affine Block Graphs $G_{r, t}$] \label{defn:Grt}
For each $1 \le i \le \ell$, let $A^i_0 \cup A_1^i \cup \cdots \cup A_{r - 1}^i$ be the partition of $\mathbb{F}_r^t$ given by the affine shifts $A^i_x = (V_i + xu_i)$ for each $x \in \mathbb{F}_r$. Given two hyperplanes $V_i, V_j$ and linear function $L : \mathbb{F}_r \to \mathbb{F}_r$, define the bipartite graph $G_{r, t}(V_i, V_j, L)$ with two parts of size $r^t$, each indexed by points in $\mathbb{F}_r^t$, as follows:
\begin{enumerate}
\item all of the edges between the points with indices in $A^i_x$ in the left part of $G_{r, t}(V_i, V_j, L)$ and the points with indices in $A^j_y$ in the right part are present if $L(x) = y$; and
\item none of the edges between the points of $A^i_x$ on the left and $A^j_y$ on the right are present if $L(x) \neq y$.
\end{enumerate}
\end{definition}

We now define the slices of the tensor $T_{r, t}$ to be weighted adjacency matrices of the bipartite graphs $G_{r, t}(V_i, V_j, L)$ as in the following definition.

\begin{definition}[Design Tensors $T_{r, t}$] \label{defn:Trt}
For any two hyperplanes $V_i, V_j$ and linear function $L : \mathbb{F}_r \to \mathbb{F}_r$, define the $r^t \times r^t$ matrix $T_{r, t}^{(V_i, V_j, L)}$ to have entries given by
$$\left( T_{r, t}^{(V_i, V_j, L)} \right)_{k,l} = \frac{1}{r^t \sqrt{r - 1}} \cdot \left\{ \begin{array}{ll} r - 1 & \textnormal{if } ( P_k, P_l ) \in E\left( G_{r, t}(V_i, V_j, L) \right) \\ -1 & \textnormal{otherwise} \end{array} \right.$$
for each $1 \le k, l \le r^t$.
\end{definition}

The next two lemmas establish that the tensor $T_{r, t}$ satisfies the four desiderata discussed above, which will be crucial in our reduction to hidden partition models.

\begin{lemma}[Sub-orthogonality of $T_{r,t}$] \label{lem:suborthogonaltensors}
If $r \ge 2$ is prime, then $T_{r, t}$ satisfies that:
\begin{enumerate}
\item for each $1 \le i, j \le r^t$ and linear function $L$, it holds that $\left\| T_{r, t}^{(V_i, V_j, L)} \right\|_F = 1$;
\item the inner product between the slices $T_{r, t}^{(V_i, V_j, L)}$ and $T_{r, t}^{(V_{i'}, V_{j'}, L')}$ where $(V_i, V_j, L) \neq (V_{i'}, V_{j'}, L')$ is
$$\textnormal{Tr}\left( T_{r, t}^{(V_i, V_j, L)} \cdot T_{r, t}^{(V_{i'}, V_{j'}, L')} \right) = \left\{ \begin{array}{ll} -(r - 1)^{-1} & \textnormal{if } (V_i, V_j) = (V_{i'}, V_{j'}) \textnormal{ and } L = L' + a \textnormal{ for some } a \neq 0 \\ 0 & \textnormal{if } (V_i, V_j) \neq (V_{i'}, V_{j'}) \textnormal{ or } L \neq L' + a \textnormal{ for all } a \in \mathbb{F}_r \end{array} \right.$$
\end{enumerate}
\end{lemma}

\begin{proof}
Fix two triples $(V_i, V_j, L)$ and $(V_{i'}, V_{j'}, L')$ and let $G_1 = G_{r, t}(V_i, V_j, L)$ and $G_2 = G_{r, t}(V_{i'}, V_{j'}, L')$. Now observe that
\begin{align}
\textnormal{Tr}\left( T_{r, t}^{(V_i, V_j, L)} \cdot T_{r, t}^{(V_{i'}, V_{j'}, L')} \right) &= \frac{1}{r^{2t}(r - 1)} \cdot (r - 1)^2 \cdot |E(G_1) \cap E(G_2)| \nonumber \\
&\quad \quad - \frac{1}{r^{2t}(r - 1)} \cdot (r - 1) \cdot \left( |E(G_1) \cup E(G_2)| - |E(G_1) \cap E(G_2)| \right) \nonumber \\
&\quad \quad + \frac{1}{r^{2t}(r - 1)} \cdot \left( r^{2t} - |E(G_1) \cup E(G_2)| \right) \label{eqn:inner-matrix}
\end{align}
Now note that since $L$ is a function, there are exactly $r$ pairs $(x, y) \in \mathbb{F}_r^2$ such that $L(x) = y$ and thus exactly $r$ pairs of left and right sets $(A^i_x, A^j_y)$ that are completely connected by edges in $G_1$. This implies that there are $|E(G_1)| = |E(G_2)| = r^{2t - 1}$ edges in both $G_1$ and $G_2$. We now will show that
\begin{equation} \label{eqn:int-sizes}
|E(G_1) \cap E(G_2)| = \left\{ \begin{array}{ll} r^{2t - 1} & \textnormal{if } (V_i, V_j, L) = (V_{i'}, V_{j'}, L') \\ r^{2t - 2} & \textnormal{if } (V_i, V_j) \neq (V_{i'}, V_{j'}) \textnormal{ or } L \neq L' + a \textnormal{ for all } a \in \mathbb{F}_r \\ 0 & \textnormal{if } (V_i, V_j) = (V_{i'}, V_{j'}) \textnormal{ and } L = L' + a \textnormal{ for some } a \neq 0 \end{array} \right.
\end{equation}
We remark that, as in the proof of Lemma \ref{lem:suborthogonalmatrices}, it is never true that $(V_i, V_j) \neq (V_{i'}, V_{j'})$ if $t = 1$. The first case follows immediately from the fact that $|E(G_1)| = r^{2t - 1}$. Now consider the case in which $V_i \neq V_{i'}$ and $V_j \neq V_{j'}$. As in the proof of Lemma \ref{lem:suborthogonalmatrices}, any pair of affine spaces $A^{i}_x$ and $A^{i'}_{x'}$ either intersects in an affine space of dimension $t - 2$, an affine space of dimension $t - 1$ if $A^{i}_x = A^{i'}_{x'}$ are equal and in the empty set if $A^{i}_x$ and $A^{i'}_{x'}$ are affine shifts of one another. Since $V_i \neq V_{i'}$, only the first of these three options is possible. Therefore, for all $x, x', y, y' \in \mathbb{F}_r$, it follows that $(A_x^i \times A_y^j) \cap (A_{x'}^{i'} \times A_{y'}^{j'}) = (A_x^i \cap A_{x'}^{i'}) \times (A_y^j \times A_{y'}^{j'})$ has size $r^{2t - 4}$ since both $A_x^i \cap A_{x'}^{i'}$ and $A_y^j \times A_{y'}^{j'}$ are affine spaces of dimension $t - 2$. Now observe that
$$|E(G_1) \cap E(G_2)| = \sum_{L(x) = y} \sum_{L'(x') = y'} \left| \left(A_x^i \times A_y^j\right) \cap \left(A_{x'}^{i'} \times A_{y'}^{j'}\right) \right| = r^2 \cdot r^{2t - 4} = r^{2t - 2}$$
since there are exactly $r$ pairs $(x, y)$ with $L(x) = y$. Now suppose that $V_i = V_{i'}$ and $V_j \neq V_{j'}$. In this case, we have that $A_x^i \cap A_{x'}^{i'}$ is empty if $x \neq x'$ and otherwise has size $|A_x^i| = r^{t - 1}$. Thus it follows that
$$\left| \left(A_x^i \times A_y^j\right) \cap \left(A_{x'}^{i'} \times A_{y'}^{j'}\right) \right| = \left\{ \begin{array}{ll} r^{2t - 3} & \textnormal{if } x = x' \\ 0 &\textnormal{otherwise} \end{array} \right.$$
This implies that
$$|E(G_1) \cap E(G_2)| = \sum_{L(x) = y} \sum_{L'(x') = y'} \left| \left(A_x^i \times A_y^j\right) \cap \left(A_{x'}^{i'} \times A_{y'}^{j'}\right) \right| = r \cdot r^{2t - 3} = r^{2t - 2}$$
since for each fixed $x = x'$, there is a unique pair $(y, y')$ with $L(x) = y$ and $L(x') = y'$. The case in which $V_i \neq V_{i'}$ and $V_j = V_{j'}$ is handled by a symmetric argument. Now suppose that $(V_i, V_j) = (V_{i'}, V_{j'})$. It follows that $(A_x^i \times A_y^j) \cap (A_{x'}^{i'} \times A_{y'}^{j'})$ has size $r^{2t - 2}$ if $x = x'$ and $y = y'$, and is empty otherwise. The formula above therefore implies that $|E(G_1) \cap E(G_2)|$ is $r^{2t - 2}$ times the number of solutions to $L(x) = L'(x)$. Since $L - L'$ is linear, the number of solutions is $0$ if $L - L'$ is constant and not equal to zero, $1$ if $L - L'$ is not constant or $r$ if $L = L'$. This completes the proof of Equation (\ref{eqn:int-sizes}). Now observe that $|E(G_1) \cup E(G_2)| = |E(G_1)| + |E(G_2)| - |E(G_1) \cap E(G_2)| = 2r^{2t - 1} - |E(G_1) \cap E(G_2)|$. Substituting this expression for $|E(G_1) \cup E(G_2)|$ into Equation (\ref{eqn:inner-matrix}) yields that
$$\textnormal{Tr}\left( T_{r, t}^{(V_i, V_j, L)} \cdot T_{r, t}^{(V_{i'}, V_{j'}, L')} \right) = \frac{r^2}{r^{2t}(r - 1)} \cdot |E(G_1) \cap E(G_2)| - \frac{1}{r - 1}$$
Combining this with the different cases of Equation (\ref{eqn:int-sizes}) shows part (2) of the lemma. Part (1) of the lemma follows from this computation and fact that
$$\left\| T_{r, t}^{(V_i, V_j, L)} \right\|_F^2 = \textnormal{Tr}\left( \left( T_{r, t}^{(V_i, V_j, L)} \right)^2 \right)$$
This completes the proof of the lemma.
\end{proof}

We now define an unfolded matrix variant of the tensor $T_{r, t}$ that will be used in our applications of $\pr{Tensor-Bern-Rotations}$ to map to hidden partition models. The row indexing in $M_{r, t}$ will be important and related to the community alignment property of $T_{r, t}$ that will be established in Lemma \ref{lem:comm-align-tensors}.

\begin{definition}[Unfolded Matrix $M_{r,t}$] \label{defn:unfolded-Trt}
Let $M_{r, t}$ be an $(r - 1)^2 \ell^2 \times r^{2t}$ matrix with entries given by
$$\left( M_{r, t} \right)_{a(r - 1)\ell^2 + i'(r - 1)\ell + b\ell + j' + 1, ir^{t} + j + 1} = \left( T_{r, t}^{(V_{i' + 1}, V_{j' + 1}, L_{a + 1,b + 1})} \right)_{i, j}$$
for each $0 \le i', j' \le (r - 1)\ell - 1$, $0 \le a, b \le r - 2$ and $0 \le i, j \le r^t - 1$, where $L_{c, d} : \mathbb{F}_r \to \mathbb{F}_r$ denotes the linear function given by $L_{c, d}(x) = cx + d$.
\end{definition}

The next lemma is similar to Lemma \ref{lem:Krtsv} and deduces the singular values of $M_{r, t}$ from Lemma \ref{lem:suborthogonaltensors}. The proof is very similar to that of Lemma \ref{lem:Krtsv}.

\begin{lemma}[Singular Values of $M_{r, t}$] \label{lem:Mrtsv}
The nonzero singular values of $M_{r, t}$ are $\sqrt{1 + (r - 1)^{-1}}$ with multiplicity $(r - 1)(r - 2)\ell^2$ and $(r - 1)^{-1/2}$ with multiplicity $(r - 1)\ell^2$.
\end{lemma}

\begin{proof}
Observe that the rows of $M_{r, t}$ are formed by vectorizing the slices of $T_{r, t}$. Thus Lemma \ref{lem:suborthogonaltensors} implies that $(M_{r, t})(M_{r, t})^\top$ is block-diagonal with $(r - 1)\ell^2$ blocks of dimension $(r - 1) \times (r - 1)$, where each block corresponds to slices with indices $(V_i, V_j, L_{c, d})$ where $i, j$ and $c$ are fixed on over each block while $d$ ranges over $\{1, 2, \dots, r - 1\}$. Furthermore, each block is of the form $\left(1 + (r - 1)^{-1} \right) I_{r - 1} - (r - 1)^{-1} \mathbf{1} \mathbf{1}^\top$. The eigenvalues of each of these blocks are $1 + (r - 1)^{-1}$ with multiplicity $r - 2$ and $(r - 1)^{-1}$ with multiplicity $1$. Thus the eigenvalues of $(M_{r, t})(M_{r, t})^\top$ are $1 + (r - 1)^{-1}$ and $(r - 1)^{-1}$ with multiplicities $(r - 1)(r - 2)\ell^2$ and $(r - 1)\ell^2$, respectively, which implies the result.
\end{proof}

Given $m^2$ matrices $M^{1,1}, M^{1,2}, \dots, M^{k,k} \in \mathbb{R}^{n \times n}$, let $\mathcal{C}\left(M^{1,1}, M^{1,2}, \dots, M^{k,k}\right)$ denote the matrix $X \in \mathbb{R}^{kn \times kn}$ formed by concatenating the $M^{i,j}$ with
$$X_{an + b + 1, cn + d + 1} = M^{a + 1, c + 1}_{b+1, d+1} \quad \text{for all } 0 \le a, c \le k - 1 \text{ and } 0 \le b, d \le n - 1$$
We refer to a matrix $M \in \mathbb{R}^{n \times n}$ as a $k$-block matrix for some $k$ that divides $n$ if there are two values $x_1, x_2 \in \mathbb{R}$ and two partitions $[n] = E_1 \cup E_2 \cup \cdots \cup E_k = F_1 \cup F_2 \cup \cdots \cup F_k$ both into parts of size $n/k$ such that
$$M_{ij} = \left\{ \begin{array}{ll} x_1 & \text{if } (i, j) \in E_h \times F_h \text{ for some } 1 \le h \le k \\ x_2 &\text{otherwise} \end{array} \right.$$
The next lemma shows an alignment property of different slices of $T_{r, t}$ that will be crucial in stitching together the local applications of $\pr{Tensor-Bern-Rotations}$ with $M_{r, t}$ in our reduction to hidden partition models. This lemma will use indexing the in $M_{r, t}$ and the role of linear functions $L$ in defining the affine block graphs $G_{r, t}$.

\begin{lemma}[Community Alignment in $T_{r, t}$] \label{lem:comm-align-tensors}
Let $1 \le s_1, s_2, \dots, s_k \le (r - 1)\ell$ be arbitrary indices and
$$M^{i, j} = T_{r, t}^{(V_{i'}, V_{j'}, L)} \quad \textnormal{for each } 1 \le i, j \le k$$
where $i'$ and $j'$ are the unique $1 \le i', j' \le \ell$ such that $i' \equiv s_i \pmod{\ell}$ and $j' \equiv s_j \pmod{\ell}$ and $L(x) = ax + b$ where $a = \lceil s_i/\ell \rceil$ and $b = \lceil s_j/\ell \rceil$. Then it follows that $\mathcal{C}\left(M^{1,1}, M^{1,2}, \dots, M^{k,k}\right)$ is an $r$-block matrix.
\end{lemma}

\begin{proof}
Let $t_i = i'$ be the unique $1 \le i' \le \ell$ such that $i' \equiv s_i \pmod{\ell}$ and let $a_i = \lceil s_i/\ell \rceil \in \{1, 2, \dots,  r - 1\}$ for each $1 \le i \le \ell$. Furthermore, let $L_{ij}(x) = a_i x + a_j$ for $1 \le i, j \le k$ and, for each $x \in \mathbb{R}$ and $1 \le i \le \ell$, let $A^i_x$ be the affine spaces as in Definition \ref{defn:Grt}. Note that since $0 < a_i < r$, it follows that each $L_{ij}$ is a non-constant and hence invertible linear function. Given a subset $S \subseteq \mathbb{F}_r^t$ and some $s \in \mathbb{N}$, let $I(s, S)$ denote the set of indices $I(s, S) = \{ s + i : P_i \in S\}$.

Now define the partition $[kr^t] = E_0 \cup E_2 \cup \cdots \cup E_{r-1}$ as follows
$$E_i = \bigcup_{j = 1}^k I\left((j - 1)r^t, A^{t_j}_{x_{ij}}\right) \quad \text{where } x_{ij} = L_{j1}^{-1}(L_{11}(i))$$
and similarly define the partition $[kr^t] = F_0 \cup F_2 \cup \cdots \cup F_{r-1}$ as follows
$$F_i = \bigcup_{j = 1}^k I\left((j - 1)r^t, A^{t_j}_{y_{ij}}\right) \quad \text{where } y_{ij} = L_{1j}(i)$$
Let $X \in \mathbb{R}^{kr^t \times kr^t}$ denote the matrix $X = \mathcal{C}\left(M^{1,1}, M^{1,2}, \dots, M^{k,k}\right)$. We will show that
\begin{equation} \label{eqn:blockstruct}
X_{a, b} = \frac{r - 1}{r^t\sqrt{r - 1}} \quad \text{if } (a, b) \in E_i \times F_i \text{ for some } 0 \le i \le r - 1
\end{equation}
Suppose that $(a, b) \in E_i \times F_i$ and specifically that $(j_a - 1)r^t  + 1 \le a \le j_a r^t$ and $(j_b - 1)r^t  + 1 \le b \le j_b r^t$ for some $1 \le j_a, j_b \le k$. The definitions of $E_i$ and $F_i$ imply that $z_a \in A_{x_{ij_a}}^{t_{j_a}}$ where $z_a = P_{a - (j_a - 1)r^t}$ and $z_b \in A_{y_{ij_b}}^{t_{j_b}}$ where $z_b = P_{b - (j_b - 1)r^t}$. Note that
$$X_{a, b} = M^{j_a, j_b}_{a - (j_a - 1)r^t, b - (j_b - 1)r^t}$$
by the definition of $\mathcal{C}$. Therefore by Definition \ref{defn:Trt}, it suffices to show that $(z_a, z_b)$ is an edge of the bipartite graph $G_{r, t}(V_{t_{j_a}}, V_{t_{j_b}}, L_{j_a j_b})$ for all such $(a, b)$ to establish (\ref{eqn:blockstruct}). By Definition \ref{defn:Grt}, $(z_a, z_b)$ is an edge if and only if $L_{j_a j_b}(x_{ij_a}) = y_{ij_b}$. Observe that the definitions of $x_{ij_a}$ and $y_{ij_b}$ yield that
\begin{align}
a_{j_a} x_{ij_a} + a_1 &= L_{j_a 1}(x_{ij_a}) = L_{11}(i) = a_1 \cdot i + a_1 \label{eqn:linear-cons} \\
y_{ij_b} &= L_{1j_b}(i) = a_1 \cdot i + a_{j_b} \nonumber \\
L_{j_a j_b}(x) &= a_{j_a} x + a_{j_b} \nonumber
\end{align}
Adding $a_{j_b} - a_1$ to both sides of Equation (\ref{eqn:linear-cons}) therefore yields that
$$L_{j_a j_b}(x_{ij_a}) = a_{j_a} x_{ij_a} + a_{j_b} = a_1 \cdot i + a_{j_b} = y_{ij_b}$$
which completes the proof of (\ref{eqn:blockstruct}). Now note that each $M^{i, j}$ contains exactly $r^{2t - 1}$ entries equal to $(r - 1)/r^t\sqrt{r - 1}$ and thus $X$ contains exactly $k^2 r^{2t - 1}$ such entries. The definitions of $E_i$ and $F_i$ imply that they each contain exactly $kr^{t - 1}$ elements. Thus $\cup_{i = 0}^{r - 1} E_i \times F_i$ contains $k^2 r^{2t - 1}$ elements. Therefore (\ref{eqn:blockstruct}) also implies that $X_{a, b} = - 1/r^t\sqrt{r - 1}$ for all $(a, b) \not \in \cup_{i = 0}^{r - 1} E_i \times F_i$. This proves that $X$ is an $r$-block matrix and completes the proof of the lemma.
\end{proof}

The community alignment property shown in this lemma is directly related to the indexing of rows in $M_{r, t}$. More precisely, the above lemma implies that for any subset $S \subseteq [(r - 1)\ell]$, the rows of $M_{r, t}$ indexed by elements in the support of $\mathbf{1}_S \otimes \mathbf{1}_S$ can be arranged as sub-matrices of an $|S| r^t \times |S| r^t$ matrix that is an $r$-block matrix. This property will be crucial in our reduction from $k\pr{-pc}$ and $k\pr{-pds}$ to hidden partition models in Section \ref{sec:3-hidden-partition}.

\subsection{A Random Matrix Alternative to $K_{r, t}$}
\label{subsec:2-Rne}

In this section, we introduce the random matrix analogue $R_{n, \epsilon}$ of $K_{r, t}$ defined below. Rather than have all independent entries, $R_{n, \epsilon}$ is constrained to be symmetric. This ends up being technically convenient, as it suffices to bound the eigenvalues of $R_{n, \epsilon}$ in order to upper bound its largest singular value. This symmetry also yields a direct connection between the eigenvalues of $R_{n, \epsilon}$ and the eigenvalues of sparse random graphs, which have been studied extensively.

\begin{definition}[Random Design Matrix $R_{n, \epsilon}$] \label{defn:Rne}
If $\epsilon \in (0, 1/2]$, let $R_{n, \epsilon} \in \mathbb{R}^{n \times n}$ denote the random symmetric matrix with independent entries sampled as follows
$$(R_{n, \epsilon})_{ij} = (R_{n, \epsilon})_{ji} \sim \left\{ \begin{array}{ll} -\sqrt{\frac{1 - \epsilon}{\epsilon n}} & \textnormal{with prob. } \epsilon \\ \sqrt{\frac{\epsilon}{(1 - \epsilon)n}} & \textnormal{with prob. } 1 - \epsilon \end{array} \right.$$
for all $1 \le i < j \le n$, and $(R_{n, \epsilon})_{ii} = \sqrt{\frac{\epsilon}{(1 - \epsilon)n}}$ for each $1 \le i \le n$.
\end{definition}

We now establish the key properties of the matrix $R_{n, \epsilon}$. Consider the graph $G$ where $\{i, j\} \in E(G)$ if and only if $(R_{n, \epsilon})_{ij}$ is negative. By definition, we have that $G$ is an $\epsilon$-sparse Erd\H{o}s-R\'{e}nyi graph with $G \sim \mG(n, \epsilon)$. Furthermore, if $A$ is the adjacency matrix of $G$, a direct calculation yields that $R_{n, \epsilon}$ can be expressed as
\begin{equation} \label{eqn:Rne-decomp}
R_{n, \epsilon} = \sqrt{\frac{\epsilon}{(1 - \epsilon)n}} \cdot I_n + \frac{1}{\sqrt{\epsilon(1 - \epsilon)n}} \cdot \left( \bE[A] - A \right)
\end{equation}
A line of work has given high probability upper bounds on the largest eigenvalue of $\bE[A] - A$ in order to study concentration of sparse Erd\H{o}s-R\'{e}nyi graphs in the spectral norm of their adjacency matrices \cite{furedi1981eigenvalues, vu2005spectral, feige2005spectral, lu2013spectra, bandeira2016sharp, le2017concentration}. As outlined in \cite{le2017concentration}, the works \cite{furedi1981eigenvalues, vu2005spectral, lu2013spectra} apply Wigner's trace method to obtain spectral concentration results for general random matrices that, in this context, imply with high probability that
$$\left\| \bE[A] - A \right\| = 2\sqrt{d} \left( 1 + o_n(1) \right) \quad \text{for } d \gg (\log n)^4$$
where $d = \epsilon n$ and $\| \cdot \|$ denotes the spectral norm on $n \times n$ symmetric matrices. In \cite{feige2005spectral, bandeira2016sharp, le2017concentration}, it is shown that this requirement on $d$ can be relaxed and that it holds with high probability that
$$\left\| \bE[A] - A \right\| = O_n(\sqrt{d}) \quad \text{for } d = \Omega_n(\log n)$$
These results, the fact that $R_{n, \epsilon}$ is symmetric and the above expression for $R_{n, \epsilon}$ in terms of $A$ are enough to establish our main desired properties of $R_{n, \epsilon}$, which are stated formally in the following lemma.

\begin{lemma}[Key Properties of $R_{n, \epsilon}$] \label{lem:Rne}
If $\epsilon \in (0, 1/2]$ satisfies that $\epsilon n = \omega_n(\log n)$, there is a constant $C > 0$ such that the random matrix $R_{n, \epsilon}$ satisfies the following two conditions with probability $1 - o_n(1)$:
\begin{enumerate}
\item the largest singular value of $R_{n, \epsilon}$ is at most $C$; and
\item every column of $R_{n, \epsilon}$ contains between $\epsilon n - C\sqrt{\epsilon n \log n}$ and $\epsilon n + C\sqrt{\epsilon n \log n}$ negative entries.
\end{enumerate}
\end{lemma}

\begin{proof}
The number of negative entries in the $i$th column of $R_{n, \epsilon}$ is distributed as $\text{Bin}(n - 1, \epsilon)$. A standard Chernoff bound for the binomial distribution yields that if $X \sim \text{Bin}(n - 1, \epsilon)$, then
$$\bP\left[ |X - (n - 1)\epsilon | \ge \delta (n - 1)\epsilon \right] \le 2 \exp\left( - \frac{\delta^2 (n - 1)\epsilon}{3} \right)$$
for all $\delta \in (0, 1)$. Setting $\delta = C' \sqrt{n^{-1} \epsilon^{-1} \log n}$ for a sufficiently large constant $C' > 0$ and taking a union bound over all columns $i$ implies that property (2) in the lemma statement holds with probability $1 - o_n(1)$. We now apply Theorem 1.1 in \cite{le2017concentration} as in the first example in Section 1.4, where the graph is not modified. Since $\epsilon n = \omega_n(\log n)$, this yields with probability $1 - o_n(1)$ that
$$\left\| \bE[A] - A \right\| \le C''\sqrt{d}$$
for some constant $C'' > 0$, where $A$ and $d$ are as defined above. The decomposition of $R_{n, \epsilon}$ in Equation (\ref{eqn:Rne-decomp}) now implies that with probability $1 - o_n(1)$
$$\| R_{n, \epsilon} \| \le \sqrt{\frac{\epsilon}{(1 - \epsilon)n}} + \frac{1}{\sqrt{\epsilon(1 - \epsilon)n}} \cdot C'' \sqrt{d} = O_n(1)$$
since $\epsilon \in (0, 1/2]$ and $d = \epsilon n$. This establishes that property (1) holds with probability $1 - o_n(1)$. A union bound over (1) and (2) now completes the proof of the lemma.
\end{proof}

While $R_{n, \epsilon}$ and $K_{r, t}$ satisfy similar conditions needed by our reductions, they also have differences that will dictate when one is used over the other. The following highlights several key points in comparing these two matrices.
\begin{itemize}
\item $R_{n, \epsilon}$ and $K_{r, t}$ are analogous when $n = r^t$ and $\epsilon = 1/r$. In this case, both matrices contain the same two values $1/\sqrt{r^t(r - 1)}$ and $-\sqrt{(r - 1)/r^t}$. The rows of $K_{r, t}$ are unit vectors and the rows of $R_{n, \epsilon}$ are approximately unit vectors -- property (2) in Lemma \ref{lem:Rne} implies that the norm of each row is $1 \pm O_n(\sqrt{(\epsilon n)^{-1} \log n})$. Like $K_{r, t}$, Lemma \ref{lem:Rne} implies that $R_{n, \epsilon}$ is also approximately orthogonal with largest singular value bounded above by a constant.
\item While $K_{r, t}$ has exactly a $(1/r)$-fraction of entries in each column that are negative, $R_{n, \epsilon}$ only has \textit{approximately} an $\epsilon$-fraction of entries in each of its columns that are negative. For some of our reductions, such as our reductions to $\pr{rsme}$ and $\pr{rslr}$, having approximately an $\epsilon$-fraction of its entries equal to the negative value in Definition \ref{defn:Rne} is sufficient. In our reductions to $\pr{isbm}$, $\pr{ghpm}$, $\pr{bhpm}$ and $\pr{semi-cr}$, it will be important that $K_{r, t}$ contains exactly $(1/r)$-fraction of negative entries per column. The approximate guarantee of $R_{n, \epsilon}$ would correspond to only showing lower bounds against algorithms that are \textit{adaptive} and do not need to know the sizes of the hidden communities.
\item As is mentioned in Section \ref{subsec:1-problems-sbm} and will be discussed in Section \ref{sec:3-robust-and-supervised}, our applications of dense Bernoulli rotations with $K_{r, t}$ will generally be tight when a natural parameter $n$ in our problems satisfies that $\sqrt{n} = \tilde{\Theta}(r^t)$. This imposes a number theoretic condition (\pr{t}) on the pair $(n, r)$, arising from the fact that $t$ must be an integer, which generally remains a condition in the computational lower bounds we show for $\pr{isbm}$, $\pr{ghpm}$ and $\pr{bhpm}$. In contrast, $R_{n, \epsilon}$ places no number-theoretic constraints on $n$ and $\epsilon$, which can be arbitrary, and thus the condition (\pr{t}) can be removed from our computational lower bounds for $\pr{rsme}$ and $\pr{rslr}$. We remark that when $r = n^{o(1)}$, which often is the regime of interest in problems such as $\pr{rsme}$, then the condition \pr{(t)} is trivial and places no further constraints on $(n, r)$ as will be shown in Lemma \ref{lem:propT}.
\item $R_{n, \epsilon}$ is random while $K_{r, t}$ is fixed. In our reductions, it is often important that the same design matrix is used throughout multiple applications of dense Bernoulli rotations. Since $R_{n, \epsilon}$ is a random matrix, this requires generating a single instance of $R_{n, \epsilon}$ and using this one instance throughout our reductions. In each of our reductions, we will rejection sample $R_{n, \epsilon}$ until it satisfies the two criteria in Lemma \ref{lem:Rne} for a maximum of $O((\log n)^2)$ rounds, and then use the resulting matrix throughout all applications of dense Bernoulli rotations in that reduction. The probability bounds in Lemma \ref{lem:Rne} imply that the probability no sample from $R_{n, \epsilon}$ satisfying these criteria is found is $n^{-\omega_n(1)}$. This is a failure mode for our reductions and contributes a negligible $n^{-\omega_n(1)}$ to the total variation distance between the output of our reductions and their target distributions.
\item For some of our reductions, applying dense Bernoulli rotations with either of the two matrices $R_{n, \epsilon}$ or $K_{r, t}$ yields the same guarantees. This is the case for our reductions to $\pr{mslr}$, $\pr{glsm}$ and $\pr{tpca}$, where $r = 2$ and the condition (\pr{t}) is trivial and mapping to columns with approximately half of their entries negative is sufficient. As mentioned above, this is also the case when $r \asymp \epsilon^{-1} = n^{o(1)}$ in $\pr{rsme}$.
\item Some differences between $R_{n, \epsilon}$ and $K_{r, t}$ that are unimportant for our reductions include that $R_{n, \epsilon}$ is exactly square while $K_{r, t}$ is only approximately square and that $R_{n, \epsilon}$ is symmetric while $K_{r, t}$ is not.
\end{itemize}
For consistency, the pseudocode and analysis for all of our reductions are written with $K_{r, t}$ rather than $R_{n, \epsilon}$. Modifying our reductions to use $R_{n, \epsilon}$ is straightforward and consists of adding the rejection sampling step to sample $R_{n, \epsilon}$ discussed above. In Sections \ref{subsec:3-rsme-reduction}, \ref{subsec:2-mixtures-slr} and \ref{sec:3-robust-and-supervised}, we discuss in more detail how to make these modifications to our reductions to $\pr{rsme}$ and $\pr{rslr}$ and the computational lower bounds they yield.

There are several reasons why we present our reductions with $K_{r, t}$ rather than $R_{n, \epsilon}$. The analysis of $K_{r, t}$ in Section \ref{subsec:2-design-matrices} is simple and self-contained while the analysis of $R_{n, \epsilon}$ requires fairly involved results from random matrix theory. The construction of $K_{r, t}$ naturally extends to $T_{r, t}$ while a random tensor analogue $T_{r, t}$ seems as though it would be prohibitively difficult to analyze. Reductions with $K_{r, t}$ give an explicit encoding of cliques into the hidden structure of our target problems as discussed in Section \ref{subsec:1-tech-encoding}, yielding slightly stronger and more explicit computational lower bounds in this sense.

\section{Negatively Correlated Sparse PCA}
\label{sec:2-neg-spca}

This section is devoted to giving a reduction from bipartite planted dense subgraph to negatively correlated sparse PCA, the high level of which was outlined in Section \ref{subsec:1-tech-inverse-wishart}. This reduction will be used in the next section as a crucial subroutine in reductions to establish conjectured statistical-computational gaps for two supervised problems: mixtures of sparse linear regressions and robust sparse linear regression. The analysis of this reduction relies on a result on the convergence of the Wishart distribution and its inverse. This result is proven in the second half of this section.

\subsection{Reducing to Negative Sparse PCA}
\label{subsec:2-neg-spca-reduction}

\begin{figure}[t!]
\begin{algbox}
\textbf{Algorithm} $\chi^2\textsc{-Random-Rotation}$

\vspace{1mm}

\textit{Inputs}: Matrix $M \in \{0, 1\}^{m \times n}$, Bernoulli probabilities $0 < q < p \le 1$, planted subset size $k_n$ that divides $n$ and a parameter $\tau > 0$
\begin{enumerate}
\item Sample $r_1, r_2, \dots, r_n \sim_{\text{i.i.d.}} \sqrt{\chi^2(n/k_n)}$ and truncate the $r_j$ with $r_j \gets \min\left\{ r_j, 2 \sqrt{n/k_n} \right\}$ for each $j \in [n]$.
\item Compute $M'$ by applying $\textsc{Gaussianize}$ to $M$ with Bernoulli probabilities $p$ and $q$, rejection kernel parameter $R_{\pr{rk}} = mn$, parameter $\tau$ and target mean values $\mu_{ij} = \frac{1}{2} \tau \cdot r_j \cdot \sqrt{k_n/n}$ for each $i \in [m]$ and $j \in [n]$.
\item Sample an orthogonal matrix $R \in \mathbb{R}^{n\times n}$ from the Haar measure on the orthogonal group $\mathcal{O}_n$ and output the columns of the matrix $M'R$.
\end{enumerate}
\vspace{1mm}

\textbf{Algorithm} $\textsc{bpds-to-neg-spca}$

\vspace{1mm}

\textit{Inputs}: Matrix $M \in \{0, 1\}^{m \times n}$, Bernoulli probabilities $0 < q < p \le 1$, planted subset size $k_n$ that divides $n$ and a parameter $\tau > 0$, target dimension $d \ge m$
\begin{enumerate}
\item Compute $X = (X_1, X_2, \dots, X_n)$ where $X_i \in \mathbb{R}^m$ as the columns of the matrix output by $\chi^2\textsc{-Random-Rotation}$ applied to $M$ with parameters $p, q, k_n$ and $\tau$.
\item Compute $\hat{\Sigma} = \sum_{i = 1}^n X_i X_i^\top$ and let $R \in \mathbb{R}^{m \times n}$ be the top $m$ rows of an orthogonal matrix sampled from the Haar measure on the orthogonal group $\mathcal{O}_n$ and compute the matrix
$$M' = \sqrt{n(n - m - 1)} \cdot \hat{\Sigma}^{-1/2} R$$
where $\hat{\Sigma}^{-1/2}$ is the positive semidefinite square root of the inverse of $\hat{\Sigma}$.
\item Output the columns of the $d \times n$ matrix with upper left $m \times n$ submatrix $M'$ and all remaining entries sampled i.i.d. from $\mN(0, 1)$.
\end{enumerate}
\vspace{1mm}

\end{algbox}
\caption{Subroutine $\chi^2\textsc{-Random-Rotation}$ for random rotations to instances of sparse PCA from \cite{brennan2019optimal} and our reduction from bipartite planted dense subgraph to negative sparse PCA.}
\label{fig:neg-spca-reduction}
\end{figure}

In this section, we give our reduction $\textsc{bpds-to-neg-spca}$ from bipartite planted dense subgraph to negatively correlated sparse PCA, which is shown in Figure \ref{fig:neg-spca-reduction}. This reduction is described with the input bipartite graph as its adjacency matrix of Bernoulli random variables. A key subroutine in this reduction is the procedure $\chi^2\textsc{-Random-Rotation}$ from \cite{brennan2019optimal}, which is also shown in Figure \ref{fig:neg-spca-reduction}. The lemma below provides total variation guarantees for $\chi^2\textsc{-Random-Rotation}$ and is adapted from Lemma 4.6 from \cite{brennan2019optimal} to be in our notation and apply to the generalized case where the input matrix $M$ is rectangular instead of square.

This lemma can be proven with an identical argument to that in Lemma 4.6 from \cite{brennan2019optimal}, with the following adjustment of parameters to the rectangular case. The first two steps of $\chi^2\textsc{-Random-Rotation}$ maps $\mathcal{M}_{[m] \times [n]}(S \times T, p, q)$ approximately to
$$\frac{\tau}{2} \sqrt{\frac{k_n}{n}} \cdot \mathbf{1}_S u_T^\top + \mN(0, 1)^{\otimes m \times n}$$
where $u_T$ is the vector with $(u_T)_i = r_i$ if $i \in T$ and $(u_T)_i = 0$ otherwise. The argument in Lemma 4.6 from \cite{brennan2019optimal} shows that the final step of $\chi^2\textsc{-Random-Rotation}$ maps this distribution approximately to
$$\frac{\tau}{2} \sqrt{\frac{k_n}{n}} \cdot \mathbf{1}_S w^\top + \mN(0, 1)^{\otimes m \times n}$$
where $w \sim \mN(0, I_n)$. Now observe that the entries of this matrix are zero mean and jointly Gaussian. Furthermore, the columns are independent and have covariance matrix $I_m + \frac{\tau^2 k_n |S|}{4n} \cdot v_S v_S^\top$ where $v_S = |S|^{-1/2} \cdot \mathbf{1}_S$. Summarizing the result of this argument, we have the following lemma. 

\begin{lemma}[$\chi^2$ Random Rotations -- Adapted from Lemma 4.6 in \cite{brennan2019optimal}] \label{lem:randomrotations}
Given parameters $m, n$, let $0 < q < p \le 1$ be such that $p - q = (mn)^{-O(1)}$ and $\min(q, 1 - q) = \Omega(1)$, let $k_n \le n$ be such that $k_n$ divides $n$ and let $\tau > 0$ be such that
$$\tau \le \frac{\delta}{2 \sqrt{6\log (mn) + 2\log (p - q)^{-1}}} \quad \text{where} \quad \delta = \min \left\{ \log \left( \frac{p}{q} \right), \log \left( \frac{1 - q}{1 - p} \right) \right\}$$
The algorithm $\mathcal{A} = \chi^2\textsc{-Random-Rotation}$ runs in $\textnormal{poly}(m, n)$ time and satisfies that
\begin{align*}
\TV\left( \mathcal{A}\left(\mathcal{M}_{[m] \times [n]}(S \times T, p, q)\right), \, \mN\left(0, I_m + \frac{\tau^2 k_n |S|}{4n} \cdot v_S v_S^\top \right)^{\otimes n} \right) &\le O\left((mn)^{-1}\right) + k_n(4e^{-3})^{n/2k_n} \\
\TV\left( \mathcal{A}\left( \textnormal{Bern}(q)^{\otimes m \times n}\right), \, \mN(0, 1)^{\otimes m \times n} \right) &= O\left((mn)^{-1}\right)
\end{align*}
where $v_S = \frac{1}{\sqrt{|S|}} \cdot \mathbf{1}_S \in \mathbb{R}^m$ for all subsets $S \subseteq [m]$ and $T \subseteq [n]$ with $|T| = k_n$.
\end{lemma}

Throughout the remainder of this section, we will need to use properties of the Wishart and inverse Wishart distributions. These distributions on random matrices are defined as follows.

\begin{definition}[Wishart Distribution] \label{defn:wishart}
Let $n$ and $d$ be positive integers and $\Sigma \in \mathbb{R}^{d \times d}$ be a positive semidefinite matrix. The Wishart distribution $\mathcal{W}_d(n, \Sigma)$ is the distribution of the matrix $\hat{\Sigma} = \sum_{i = 1}^n X_i X_i^\top$ where $X_1, X_2, \dots, X_n \sim_{\textnormal{i.i.d.}} \mN(0, \Sigma)$.
\end{definition}

\begin{definition}[Inverse Wishart Distribution] \label{defn:inverted-wishart}
Let $n, d$ and $\Sigma$ be as in Definition \ref{defn:wishart}. The inverted Wishart distribution $\mathcal{W}^{-1}_d(n, \Sigma)$ is the distribution of $\hat{\Sigma}^{-1}$ where $\hat{\Sigma} \sim \mathcal{W}_d(n, \Sigma)$.
\end{definition}

In order to analyze $\textsc{bpds-to-neg-spca}$, we also will need the following observation from \cite{brennan2019optimal}. This is a simple consequence of the fact that the distribution $\mN(0, I_n)$ is isotropic and thus invariant under multiplication by elements of the orthogonal group $\mO_n$.

\begin{lemma}[Lemma 6.5 in \cite{brennan2019optimal}] \label{lem:invariance}
Suppose that $n \ge d$ and let $\Sigma \in \mathbb{R}^{d \times d}$ be a fixed positive definite matrix and let $\Sigma_e \sim \mathcal{W}_d(n, \Sigma)$. Let $R \in \mathbb{R}^{d \times n}$ be the matrix consisting of the first $d$ rows of an $n \times n$ matrix chosen randomly and independently of $\Sigma_e$ from the Haar measure $\mu_{\mathcal{O}_n}$ on $\mathcal{O}_n$. Let $(Y_1, Y_2, \dots, Y_n)$ be the $n$ columns of $\Sigma_e^{1/2} R$, then $Y_1, Y_2, \dots, Y_n \sim_{\textnormal{i.i.d.}} \mN(0, \Sigma)$.
\end{lemma}

We now will state and prove the main total variation guarantees for $\textsc{bpds-to-neg-spca}$ in the theorem below. The proof of the theorem below crucially relies on the upper bound in Theorem \ref{thm:inverse-wishart} on the KL divergence between Wishart matrices and their inverses. Proving this KL divergence bound is the focus of the next subsection.

\begin{theorem}[Reduction to Negative Sparse PCA] \label{thm:neg-spca}
Let $m, n, p, q, k_n$ and $\tau$ be as in Lemma \ref{lem:randomrotations} and suppose that $d \ge m$ and $n \gg m^3$ as $n \to \infty$. Fix any subset $S \subseteq [m]$ and let $\theta_S$ be given by
$$\theta_S = \frac{\tau^2 k_n |S|}{4n + \tau^2 k_n |S|}$$
Then algorithm $\mathcal{A} = \pr{bpds-to-neg-spca}$ runs in $\textnormal{poly}(m, n)$ time and satisfies that
\begin{align*}
\TV\left( \mathcal{A}\left(\mathcal{M}_{[m] \times [n]}(S \times T, p, q)\right), \, \mN\left(0, I_d - \theta_S v_S v_S^\top \right)^{\otimes n} \right) &\le O\left(m^{3/2} n^{-1/2} \right) + k_n(4e^{-3})^{n/2k_n} \\
\TV\left( \mathcal{A}\left( \textnormal{Bern}(q)^{\otimes m \times n}\right), \, \mN(0, 1)^{\otimes d \times n} \right) &= O\left(m^{3/2} n^{-1/2}\right)
\end{align*}
where $v_S = \frac{1}{\sqrt{|S|}} \cdot \mathbf{1}_S \in \mathbb{R}^d$ for all subsets $S \subseteq [m]$ and $T \subseteq [n]$ with $|T| = k_n$.
\end{theorem}

\begin{proof}
Let $\mathcal{A}_{\text{1}}$ denote the application of $\chi^2\textsc{-Random-Rotation}$ with input $M$ and output $X$ in Step 1 of $\mathcal{A}$. Let $\mathcal{A}_{\text{2a}}$ denote the Markov transition with input $X$ and output $n(n - m - 1) \cdot \hat{\Sigma}^{-1}$, as defined in Step 2 of $\mathcal{A}$, and let $\mathcal{A}_{\text{2b-3}}$ denote the Markov transition with input $Y = n(n - m - 1) \cdot \hat{\Sigma}^{-1}$ and output $Z$ formed by padding $Y^{1/2} R$ with i.i.d. $\mN(0, 1)$ random variables to be $d \times n$ i.e. the output of $\mathcal{A}$. Furthermore, let $\mathcal{A}_{\text{2-3}} = \mathcal{A}_{\text{2b-3}} \circ \mathcal{A}_{\text{2a}}$ denote Steps 2 and 3 with input $X$ and output $Z$.

Now fix some positive semidefinite matrix $\Sigma \in \mathbb{R}^{m \times m}$ and observe that if $A = \sum_{i = 1}^n Z_i Z_i^\top \sim \mathcal{W}_m(n, I_m)$ where $Z_1, Z_2, \dots, Z_n \sim_{\text{i.i.d.}} \mN(0, I_m)$, then it also follows that
$$\Sigma^{1/2} A \Sigma^{1/2} = \sum_{i = 1}^n \left(\Sigma^{1/2}Z_i\right) \left(\Sigma^{1/2}Z_i\right)^\top \sim \mathcal{W}_m(n, \Sigma)$$
since $\Sigma^{1/2}Z_i \sim \mN(0, \Sigma)$. Now observe that $(\Sigma^{1/2} A \Sigma^{1/2})^{-1} = \Sigma^{-1/2} A^{-1} \Sigma^{-1/2}$ and thus if $B \sim \mathcal{W}^{-1}_m(n, I_m)$ then $\Sigma^{-1/2} B \Sigma^{-1/2} \sim \mathcal{W}^{-1}_m(n, \Sigma)$. Let $\beta^{-1} = n(n - m - 1)$ and $C \sim \mathcal{W}^{-1}_m(n, \beta \cdot I_m)$. Therefore we have by the data processing inequality for total variation in Fact \ref{tvfacts} that
\begin{align*}
\TV\left( \mathcal{W}_m(n, \Sigma), \, \mathcal{W}^{-1}_m\left(n, \beta \cdot \Sigma^{-1}\right) \right) &= \TV\left( \mL\left( \Sigma^{1/2} A \Sigma^{1/2} \right), \, \mL\left( \Sigma^{1/2} C \Sigma^{1/2} \right) \right) \\
&\le \TV\left( \mL\left( A \right), \, \mL\left( C \right) \right) \\
&\le \sqrt{\frac{1}{2} \cdot \KL\left( \mathcal{W}_m(n, I_m) \, \Big\| \, \mathcal{W}^{-1}_m(n, \beta \cdot I_m) \right)} \\
&= O\left( m^{3/2} n^{-1/2} \right)
\end{align*}
where the last inequality follows from the fact that $n \gg m^3$, Theorem \ref{thm:inverse-wishart} and Pinsker's inequality.

Suppose that $X \sim \mN\left(0, I_m + \theta_S' v_S v_S^\top \right)^{\otimes n}$ where $\theta_S' = \frac{\tau^2 k_n |S|}{4n}$. Then we have that the output $Y$ of $\mathcal{A}_{\text{2a}}$ satisfies $Y = n(n - m - 1) \cdot \hat{\Sigma}^{-1} \sim \mathcal{W}^{-1}_m\left(n, \beta \cdot \Sigma^{-1}\right)$ where
$$\Sigma = \left( I_m + \theta_S' v_S v_S^\top \right)^{-1} = I_m - \frac{\theta_S'}{1 + \theta_S'} \cdot v_S^\top v_S^\top = I_m - \theta_S v_S v_S^\top$$
Therefore it follows from the inequality above that
$$\TV\left( \mathcal{A}_{\text{2a}}\left( \mN\left(0, I_m + \theta_S' v_S v_S^\top \right)^{\otimes n} \right) , \, \mathcal{W}_m\left(n, I_m - \theta_S v_S v_S^\top\right) \right) = O\left( m^{3/2} n^{-1/2} \right)$$
Similarly, if $X \sim \mN\left(0, I_m \right)^{\otimes n}$ then we have that
$$\TV\left( \mathcal{A}_{\text{2a}}\left( \mN\left(0, I_m \right)^{\otimes n} \right) , \, \mathcal{W}_m\left(n, I_m \right) \right) = O\left( m^{3/2} n^{-1/2} \right)$$
applying the same argument with $\Sigma = I_m$. Now note that if $Y \sim \mathcal{W}_m\left(n, I_m - \theta_S v_S v_S^\top\right)$ then Lemma \ref{lem:invariance} implies that $\mathcal{A}_{\text{2b-3}}$ produces $Z \sim \mN\left(0, I_d - \theta_S v_S v_S^\top \right)^{\otimes n}$. Similarly, it follows that if $Y \sim \mathcal{W}_m(n, I_m)$ then Lemma \ref{lem:invariance} implies that $Z \sim \mN\left(0, I_d \right)^{\otimes n}$.

We now will use Lemma \ref{lem:tvacc} applied to the steps $\mathcal{A}_i$ above and the following sequence of distributions
\allowdisplaybreaks
\begin{align*}
\mathcal{P}_0 &= \mathcal{M}_{[m] \times [n]}(S \times T, p, q) \\
\mathcal{P}_1 &= \mN\left(0, I_m + \theta'_S v_S v_S^\top \right)^{\otimes n}\\
\mathcal{P}_{\text{2a}} &= \mathcal{W}_m\left(n, I_m - \theta_S v_S v_S^\top\right) \\
\mathcal{P}_{\text{2b-3}} &= \mN\left(0, I_d - \theta_S v_S v_S^\top \right)^{\otimes n}
\end{align*}
As in the statement of Lemma \ref{lem:tvacc}, let $\epsilon_i$ be any real numbers satisfying $\TV\left( \mathcal{A}_i(\mP_{i-1}), \mP_i \right) \le \epsilon_i$ for each step $i$. A direct application of Lemma \ref{lem:randomrotations}, shows that we can take $\epsilon_1 = O(m^{-1} n^{-1}) + k(4e^{-3})^{n/2k}$. The arguments above show we can take $\epsilon_{\text{2a}} = O(m^{3/2} n^{-1/2})$ and $\epsilon_{\text{2b-3}} = 0$. Lemma \ref{lem:tvacc} now implies the first bound in the theorem statement. The second bound follows from an analogous argument for the distributions
$$\mathcal{P}_0 = \text{Bern}(q)^{\otimes m \times n}, \quad \mathcal{P}_1 = \mN\left(0, I_m \right)^{\otimes n}, \quad \mathcal{P}_{\text{2a}} = \mathcal{W}_m\left(n, I_m \right) \quad \text{and} \quad \mathcal{P}_{\text{2b-3}} = \mN\left(0, I_d \right)^{\otimes n}$$
with $\epsilon_1 = O(m^{-1} n^{-1})$, $\epsilon_{\text{2a}} = O(m^{3/2} n^{-1/2})$ and $\epsilon_{\text{2b-3}} = 0$. This completes the proof of the theorem.
\end{proof}

\subsection{Comparing Wishart and Inverse Wishart}
\label{subsec:2-inverse-wishart}

This section is devoted to proving the upper bound on the KL divergence between Wishart matrices and their inverses in Theorem \ref{thm:inverse-wishart} used in the proof of Theorem \ref{thm:neg-spca}. As noted in the previous subsection, the next theorem also implies total variation convergence between Wishart and inverse Wishart when $n \gg d^3$ by Pinsker's inequality. This theorem is related to a line of recent research examining the total variation convergence between ensembles of random matrices in the regime where $n \gg d$. A number of recent papers have investigated the total variation convergence between the fluctuations of the Wishart and Gaussian orthogonal ensembles, also showing these converge when $n \gg d^3$ \cite{jiang2015approximation, bubeck2016testing, bubeck2016entropic, racz2019smooth}, convergence with other matrix ensembles at intermediate asymptotic scales of $d \ll n \ll d^3$ \cite{chetelat2019middle} and applications of these results to random geometric graphs \cite{bubeck2016testing, eldan2016information, brennan2019phase}.

Let $\Gamma_d(x)$ and $\psi_d(x)$ denote the multivariate gamma and digamma functions given by
$$\Gamma_d(a) = \pi^{d(d-1)/4} \cdot \prod_{i = 1}^d \Gamma\left( a - \frac{i - 1}{2} \right) \quad \text{and} \quad \psi_d(a) = \frac{\partial \log \Gamma_d(a)}{\partial a} = \sum_{i = 1}^d \psi\left( a - \frac{i - 1}{2} \right)$$
where $\Gamma(z)$ and $\psi = \Gamma'(z)/\Gamma(z)$ denote the ordinary gamma and digamma functions. We will need several approximations to the log-gamma and digamma functions to prove our desired bound on KL divergence. The classical Stirling series for the log-gamma function is
$$\log \Gamma(z) \sim \frac{1}{2} \log(2\pi) + \left( z - \frac{1}{2} \right) \log z - z + \sum_{k = 1}^\infty \frac{B_{2k}}{2k(2k-1)z^{2k - 1}}$$
where $B_m$ denotes the $m$th Bernoulli number. While this series does not converge absolutely for any $z$ because of the growth rate of the coefficients $B_{2k}$, its partial sums are increasingly accurate. More precisely, we have the following series approximation to the log-gamma function (see e.g. pg. 67 of \cite{remmert2013classical}) up to second order
$$\log \Gamma(z) = \frac{1}{2} \log(2\pi) + \left( z - \frac{1}{2} \right) \log z - z + \frac{1}{12z} + O(z^{-3})$$
as $z \to \infty$. A similar series expansion exists for the digamma function, given by
$$\psi(z) \sim \log z - \frac{1}{2z} - \sum_{k = 1}^\infty \frac{B_{2k}}{2kz^{2k}}$$
This series also exhibits the phenomenon that, while not converge absolutely for any $z$, its partial sums are increasingly accurate. We have the following third order expansion of $\psi(z)$ given by
$$\psi(z) = \log z - \frac{1}{2z} - \frac{1}{12z^2} + \frac {2}{z^{2}} \int _{0}^{\infty }\frac {t^{3}}{(t^{2}+z^{2})(e^{2\pi t}-1)} dt = \log z - \frac{1}{2z} - \frac{1}{12z^2} + O(z^{-4})$$
as $z \to \infty$. We now state and prove the main theorem of this section.

\begin{theorem}[Comparing Wishart and Inverse Wishart] \label{thm:inverse-wishart}
Let $n \ge d + 1$ and $m \ge d$ be positive integers such that $n = \Theta(m)$, $|m - n| = o(n)$ and $n - d = \Omega(n)$ as $m, n, d \to \infty$, and let $\beta = \frac{1}{m(n - d - 1)}$. Then
\begin{align*}
\KL\left( \mathcal{W}_d(n, I_d) \, \Big\| \, \mathcal{W}^{-1}_d(m, \beta \cdot I_d) \right) &= \frac{d^3}{6n} + \frac{s^2d(d + 1)}{8n^2} - \frac{5sd^3}{24n^2} + \frac{sd^3}{12mn} \\
&\quad \quad + O\left( d^2 n^{-3} |s|^3 + d^4 n^{-2} + d^2 n^{-1} \right)
\end{align*}
where $s = n - m$. In particular, when $m = n$ and $n \gg d^3$ it follows that
$$\KL\left( \mathcal{W}_d(n, I_d) \, \Big\| \, \mathcal{W}^{-1}_d(n, \beta \cdot I_d) \right) = o(1)$$
\end{theorem}

\begin{proof}
Note that the given conditions also imply that $m - d = \Omega(m)$. Let $X \sim \mathcal{W}_d(n, I_d)$ and $Y \sim \mathcal{W}^{-1}_d(m, \beta \cdot I_d)$. Throughout this section, $A \in \mathbb{R}^{d \times d}$ will denote a positive semidefinite matrix. It is well known that the Wishart distribution $\mathcal{W}_d(n, I_d)$ is absolutely continuous with respect to the Lebesgue measure on the cone $\mathcal{C}^{\text{PSD}}_d$ of positive semidefinite matrices in $\mathbb{R}^{d \times d}$ \cite{wishart1928generalised}. Furthermore the density of $X$ with respect to the Lebesgue measure can be written as
$$f_X(A) = \frac{1}{2^{nd/2} \cdot \Gamma_d\left( \frac{n}{2} \right)} \cdot |A|^{(n - d - 1)/2} \cdot \exp\left( - \frac{1}{2} \text{Tr}(A) \right)$$
A change of variables from $A \to \beta^{-1} \cdot A^{-1}$ shows that the distribution $\mathcal{W}_d^{-1}(m, \beta \cdot I_d)$ is also absolutely continuous with respect to the Lebesgue measure on $\mathcal{C}^{\text{PSD}}_d$. It is well-known (see e.g. \cite{gelman2013bayesian}) that the density of $Y$ can be written as
$$f_Y(A) = \frac{\beta^{-md/2}}{2^{md/2} \cdot \Gamma_d\left( \frac{m}{2} \right)} \cdot |A|^{-(m + d + 1)/2} \cdot \exp\left( - \frac{\beta^{-1}}{2} \cdot \text{Tr}\left(A^{-1}\right) \right)$$
Now note that
\begin{align*}
\log f_X(A) - \log f_Y(A) &= \frac{(m - n)d}{2} \cdot \log 2 + \log \Gamma_d \left( \frac{m}{2} \right) - \log \Gamma_d \left( \frac{n}{2} \right) + \frac{md}{2} \cdot \log \beta \\
&\quad \quad + \frac{m + n}{2} \cdot \log |A| - \frac{1}{2} \text{Tr}(A) + \frac{\beta^{-1}}{2} \cdot \text{Tr}\left(A^{-1}\right)
\end{align*}
The expectation of $\log |A|$ where $A \sim \mathcal{W}_d(n, I_d)$ is well known (e.g. see pg. 693 of \cite{bishop2006pattern}) to be equal to
$$\bE_{A \sim \mathcal{W}_d(n, I_d)} \left[ \log |A| \right] = \psi_d\left( \frac{n}{2} \right) + d \log 2$$
Furthermore, it is well known (e.g. see pg. 85 \cite{mardia1979multivariate}) that the mean of $A^{-1}$ if $A \sim \mathcal{W}_d(n, I_d)$ is
$$\bE_{A \sim \mathcal{W}_d(n, I_d)} \left[ A^{-1} \right] = \frac{I_d}{n - d - 1}$$
Therefore we have that $\bE_{A \sim \mathcal{W}_d(n, I_d)} \left[ \text{Tr}\left(A^{-1}\right) \right] = d/(n - d - 1)$. Similarly, we have that $\bE_{A \sim \mathcal{W}_d(n, I_d)} \left[ A \right] = n \cdot I_d$ and thus $\bE_{A \sim \mathcal{W}_d(n, I_d)} \left[ \text{Tr}(A) \right] = nd$. Combining these identities yields that
\begin{align}
\KL\left( \mathcal{W}_d(n, I_d) \, \Big\| \, \mathcal{W}^{-1}_d(m, \beta \cdot I_d) \right) &= \bE_{A \sim \mathcal{W}_d(n, I_d)} \left[ \log f_X(A) - \log f_Y(A) \right] \nonumber \\
&= \frac{(m - n)d}{2} \cdot \log 2 + \log \Gamma_d \left( \frac{m}{2} \right) - \log \Gamma_d \left( \frac{n}{2} \right) + \frac{md}{2} \cdot \log \beta \nonumber \\
&\quad \quad + \frac{m + n}{2} \cdot \left( \psi_d\left( \frac{n}{2} \right) + d \log 2 \right) - \frac{nd}{2} + \frac{\beta^{-1}d}{2(n - d - 1)} \label{eqn:klequation}
\end{align}
We now use the series approximations for $\Gamma(z)$ and $\psi(z)$ mentioned above to approximate each of these terms. Note that since $m - d = \Omega(m)$, we have that
\allowdisplaybreaks
\begin{align*}
\log \Gamma_d \left( \frac{m}{2} \right) &= \frac{d(d - 1)}{4} \log \pi + \sum_{i = 1}^d \log \Gamma\left( \frac{m - i + 1}{2} \right) \\
&= \frac{d(d - 1)}{4} \log \pi + \sum_{i = 1}^d \left( \frac{1}{2} \log(2\pi) + \left( \frac{m - i}{2} \right) \log \left( \frac{m - i + 1}{2} \right) \right. \\
&\quad \quad \quad \quad \quad \quad \quad \quad \quad \quad \quad \quad \left. - \left( \frac{m - i + 1}{2} \right) + \frac{1}{6(m - i + 1)} + O(m^{-3}) \right) \\
&= \frac{d(d - 1)}{4} \log \pi + \frac{d}{2} \log(2\pi) - \frac{dm}{2} + \frac{d(d - 1)}{4} + O(dm^{-3}) \\
&\quad \quad + \sum_{i = 1}^d \left( \left( \frac{m - i}{2} \right) \log \left( \frac{m}{2} \right) + \left( \frac{m - i}{2} \right) \log \left( 1 - \frac{i - 1}{m} \right) + \frac{1}{6(m - i + 1)} \right)
\end{align*}
using the fact that $\sum_{i = 1}^d (i - 1) = d(d - 1)/2$. Let $H_n$ denote the harmonic series $H_n = \sum_{i =1 }^n 1/i$. Using the well-known fact that $\psi(n + 1) = H_n - \gamma$ where $\gamma$ is the Euler-Mascheroni constant, we have that
\allowdisplaybreaks
\begin{align*}
\sum_{i = 1}^d \frac{1}{m - i + 1} &= H_m - H_{m - d} \\
&= \log(m + 1) - \log(m - d + 1) + O(m^{-1}) \\
&= \frac{d}{m + 1} + \frac{d^2}{2(m + 1)^2} + O(d^3 m^{-3}) + O(m^{-1}) \\
&=  O\left(dm^{-1}\right)
\end{align*}
where the second last estimate follows applying the Taylor approximation $\log(1 - x) = - x - \frac{1}{2} x^2 + O(x^{3})$ for $x = \frac{d}{m + 1} \in (0, 1)$. Applying this Taylor approximation again, we have that
\allowdisplaybreaks
\begin{align*}
&\sum_{i = 1}^d \left( \frac{m - i}{2} \right) \log \left( 1 - \frac{i - 1}{m} \right) \\
&\quad \quad = - \frac{1}{2} \sum_{i = 1}^d \left( \frac{(m - i)(i - 1)}{m} + \frac{(m - i)(i - 1)^2}{2m^2} + O\left(i^3m^{-2}\right) \right) \\
&\quad \quad = O(d^4 m^{-2}) - \frac{1}{2} \sum_{i = 1}^d \left( \frac{(m - 1)(i - 1)}{m} - \frac{(i - 1)^2}{m} + \frac{(m - 1)(i - 1)^2}{2m^2} - \frac{(i - 1)^3}{2m^2}\right) \\
&\quad \quad = O(d^4 m^{-2}) - \frac{(m - 1)d(d - 1)}{4m} + \frac{d(d - 1)(2d - 1)}{12m} - \frac{(m - 1)d(d - 1)(2d - 1)}{24m^2} + \frac{d^2(d - 1)^2}{16m^2} \\
&\quad \quad = O(d^4 m^{-2}) - \frac{d(d - 1)}{4} + \frac{d(d - 1)(2d + 5)}{24m}
\end{align*}
using the identities $\sum_{i = 1}^d (i - 1)^2 = d(d - 1)(2d - 1)/6$ and $\sum_{i = 1}^d (i - 1)^3 = d^2(d - 1)^2/4$. Combining all of these approximations and simplifying using the fact that $m - d = \Omega(m)$ yields that
\begin{align*}
\log \Gamma_d \left( \frac{m}{2} \right) &= \frac{d(d - 1)}{4} \log \pi + \frac{d}{2} \log(2\pi) - \frac{dm}{2} + \frac{dm}{2} \log \left( \frac{m}{2} \right) - \frac{d(d + 1)}{4} \log \left( \frac{m}{2} \right) \\
&\quad \quad + \frac{d(d - 1)(2d + 5)}{24m} + O\left(d^4 m^{-2} + dm^{-1} \right)
\end{align*}
as $m, d \to \infty$ and $m - d = \Omega(m)$. 
An analogous estimate is also true for $\log \Gamma_d \left( \frac{n}{2} \right)$. Similar approximations now yield since $n - d = \Omega(n)$, we have that
\begin{align*}
\psi_d\left( \frac{n}{2} \right) &= \sum_{i = 1}^d \left( \log \left( \frac{n - i + 1}{2} \right) - \frac{1}{n - i + 1} + O(n^{-2}) \right) \\
&= d \log \left( \frac{n}{2} \right) + \sum_{i = 1}^d \log \left( 1 - \frac{i - 1}{n} \right) - H_{n} + H_{n - d} + O(dn^{-2}) \\
&= d \log \left( \frac{n}{2} \right) - \sum_{i = 1}^d \left( \frac{i - 1}{n} + \frac{(i - 1)^2}{2n^2} + O\left(i^3 n^{-3}\right) \right) - \frac{d}{n + 1} - \frac{d^2}{2(n + 1)^2} \\
&\quad \quad + O\left(d^3 n^{-3} + dn^{-2} \right) \\
&= d \log \left( \frac{n}{2} \right) - \frac{d(d - 1)}{2n} - \frac{d(d - 1)(2d - 1)}{12n^2} - \frac{d}{n + 1} + O\left(d^4 n^{-3} + d^2 n^{-2} \right)
\end{align*}
Here we have expanded $\psi(n + 1) = H_n - \gamma$ to an additional order with the approximation
\begin{align*}
H_{n} - H_{n - d} &= \log(n + 1) - \log(n - d + 1) - \frac{1}{2(n + 1)} + \frac{1}{2(n - d + 1)} + O(n^{-2}) \\
&= \frac{d}{n + 1} + \frac{d^2}{2(n + 1)^2} + O(dn^{-2})
\end{align*}
Combining all of these estimates and simplifying with $\beta^{-1} = m(n - d - 1)$ now yields that
\allowdisplaybreaks
\begin{align*}
&\KL\left( \mathcal{W}_d(n, I_d) \, \Big\| \, \mathcal{W}^{-1}_d(m, \beta \cdot I_d) \right) \\
&\quad \quad = md \log 2 + \frac{md}{2} \log \beta - \frac{nd}{2} + \frac{\beta^{-1}d}{2(n - d - 1)} + \log \Gamma_d \left( \frac{m}{2} \right) - \log \Gamma_d \left( \frac{n}{2} \right) + \frac{m + n}{2} \cdot \psi_d\left( \frac{n}{2} \right) \\
&\quad \quad = md \log 2 + \frac{md}{2} \log \beta - \frac{nd}{2} + \frac{\beta^{-1}d}{2(n - d - 1)} - \frac{d(m - n)}{2} + \frac{dm}{2} \log \left( \frac{m}{2} \right) - \frac{dn}{2} \log \left( \frac{n}{2} \right) \\
&\quad \quad \quad \quad - \frac{d(d + 1)}{4} \log \left( \frac{m}{n} \right) + \frac{d(d - 1)(2d + 5)}{24} \cdot (m^{-1} - n^{-1}) + \frac{(m + n)d}{2} \log \left( \frac{n}{2} \right) \\
&\quad \quad \quad \quad - \frac{(m + n)d(d - 1)}{4n} - \frac{(m + n)d(d - 1)(2d - 1)}{24n^2} - \frac{(m + n)d}{2(n + 1)} + O\left( d^4 n^{-2} + d^2 n^{-1} \right) \\
&\quad \quad = - \frac{(m + n)d(d - 1)}{4n} - \frac{(m + n)d}{2(n + 1)} - \frac{d(d + 1)}{4} \log \left( \frac{m}{n} \right) - \frac{dm}{2} \log \left( 1 - \frac{d + 1}{n} \right) \\
&\quad \quad \quad \quad - \frac{(m + n)d(d - 1)(2d - 1)}{24n^2} + \frac{(n - m)d(d - 1)(2d + 5)}{24mn} + O\left( d^4 n^{-2} + d^2 n^{-1} \right) \\
&\quad \quad = - \frac{(m + n)d(d + 1)}{4n} + \frac{d(d + 1)}{4} \left( \frac{n - m}{n} + \frac{(n - m)^2}{2n^2} + O\left(n^{-3} |s|^3\right) \right) \\
&\quad \quad \quad \quad + \frac{dm}{2} \left( \frac{d + 1}{n} + \frac{(d + 1)^2}{2n^2} + O(d^3 n^{-3}) \right) - \frac{(m + n)d(d - 1)(2d - 1)}{24n^2} \\
&\quad \quad \quad \quad + \frac{sd(d - 1)(2d + 5)}{24mn} + O\left( d^4 n^{-2} + d^2 n^{-1} \right) \\
&\quad \quad = \frac{d^3}{6n} + \frac{s^2d(d + 1)}{8n^2} - \frac{5sd^3}{24n^2} + \frac{sd^3}{12mn} + O\left( d^2 n^{-3} |s|^3 + d^4 n^{-2} + d^2 n^{-1} \right)
\end{align*}
In the fourth equality, we used the fact that $1/(n + 1) = 1/n + O(n^{-2})$, that $s = n - m = o(n)$ and the Taylor approximation $\log(1 - x) = - x - \frac{1}{2} x^2 + O(x^{3})$ for $|x| < 1$. The last line follows from absorbing small terms into the error term. The second part of the theorem statement follows immediately from substituting $m = n$ and $s = 0$ into the bound above and noting that the dominant term is $d^3/6n$ when $n \gg d^3$.
\end{proof}

We now make two remarks on the theorem above. The first motivates the choice of the parameter $\beta$ to satisfy $\beta^{-1} = m(n - d - 1)$. Note that the KL divergence in Equation (\ref{eqn:klequation}) depends on $\beta$ through the terms
$$\frac{md}{2} \log \beta + \frac{\beta^{-1}d}{2(n - d - 1)}$$
which is minimized at the stationary point $\beta^{-1} = m(n - d - 1)$. Thus the KL divergence in Equation (\ref{eqn:klequation}) is minimized for a fixed pair $(m, n)$ at this value of $\beta$. We also remark that the distributions $\mathcal{W}_d(n, I_d)$ and $\mathcal{W}^{-1}_d(m, \beta \cdot I_d)$ only converge in KL divergence if $d \gg n^3$ as the expression in Theorem \ref{thm:inverse-wishart} is easily seen to not converge to zero if $d = O(n^3)$.

\section{Negative Correlations, Sparse Mixtures and Supervised Problems}
\label{sec:2-supervised}

In the first part of this section, we introduce and give a reduction to the intermediate problem imbalanced sparse Gaussian mixtures, as outlined in Section \ref{subsec:1-tech-design-matrices} and the beginning of Section \ref{sec:2-bernoulli-rotations}. This reduction is then used in the second part of this section, along with the reduction to negative sparse PCA in the previous section, as a subroutine in a reduction to robust sparse linear regression and mixtures of sparse linear regressions, as outlined in Section \ref{subsec:1-tech-decomposing}. Our reduction to imbalanced sparse Gaussian mixtures will also be used in Section \ref{sec:3-robust-and-supervised} to show computational lower bounds for robust sparse mean estimation.

\subsection{Reduction to Imbalanced Sparse Gaussian Mixtures}
\label{subsec:3-rsme-reduction}

\begin{figure}[t!]
\begin{algbox}
\textbf{Algorithm} $k$\textsc{-bpds-to-isgm}

\vspace{1mm}

\textit{Inputs}: Matrix $M \in \{0, 1\}^{m \times n}$, dense subgraph dimensions $k_m$ and $k_n$ where $k_n$ divides $n$ and the following parameters
\begin{itemize}
\item partition $F$ of $[n]$ into $k_n$ parts of size $n/k_n$, edge probabilities $0 < q < p \le 1$ and a slow growing function $w(n) = \omega(1)$
\item target $\pr{isgm}$ parameters $(N, d, \mu, \epsilon)$ satisfying that $\epsilon = 1/r$ for some prime number $r$,
$$wN \le k_nr\ell, \quad m \le d, \quad n \le k_nr^t \le \textnormal{poly}(n) \quad \text{and} \quad \mu \le \frac{c}{\sqrt{r^t(r - 1) \log(k_nmr^t)}}$$
for some $t \in \mathbb{N}$, a sufficiently small constant $c > 0$ and where $\ell = \frac{r^t - 1}{r - 1}$
\end{itemize}

\begin{enumerate}
\item \textit{Pad}: Form $M_{\text{PD}} \in \{0, 1\}^{m \times k_nr^t}$ by adding $k_nr^t - n$ new columns sampled i.i.d. from $\text{Bern}(q)^{\otimes m}$ to the right end of $M$. Let $F'$ be the partition formed by letting $F'_i$ be $F_i$ with exactly $r^t - n/k_n$ of the new columns.
\item \textit{Bernoulli Rotations}: Fix a partition $[k_nr\ell] = F_1'' \cup F_2'' \cup \cdots \cup F_{k_n}''$ into $k_n$ parts each of size $r\ell$ and compute the matrix $M_{\text{R}} \in \mathbb{R}^{m \times k_n r\ell}$ as follows:
\begin{enumerate}
\item[(1)] For each row $i$ and part $F_j'$, apply $\pr{Bern-Rotations}$ to the vector $(M_{\text{PD}})_{i, F_j'}$ of entries in row $i$ and in columns from $F_j'$ with matrix parameter $K_{r, t}$, rejection kernel parameter $R_{\pr{rk}} = k_n mr^t$, Bernoulli probabilities $0 < q < p \le 1$, $\lambda = \sqrt{1 + (r - 1)^{-1}}$, mean parameter $\lambda \sqrt{r^t(r - 1)} \cdot \mu$ and output dimension $r\ell$.
\item[(2)] Set the entries of $(M_{\text{R}})_{i, F''_j}$ to be the entries in order of the vector output in (1).
\end{enumerate}
\item \textit{Permute and Output}: Form $X \in \mathbb{R}^{d \times N}$ by choosing $N$ distinct columns of $M_{\text{R}}$ uniformly at random, embedding the resulting matrix as the first $m$ rows of $X$ and sampling the remaining $d - m$ rows of $X$ i.i.d. from $\mN(0, I_N)$. Output the columns $(X_1, X_2, \dots, X_N)$ of $X$.
\end{enumerate}
\vspace{0.5mm}

\end{algbox}
\caption{Reduction from bipartite $k$-partite planted dense subgraph to exactly imbalanced sparse Gaussian mixtures.}
\label{fig:isgmreduction}
\end{figure}

In this section, we give our reduction from $k$\pr{-bpds} to the intermediate problem \pr{isgm}, which we will reduce from in subsequent sections to obtain several of our main computational lower bounds. We present our reduction to $\pr{isgm}$ with dense Bernoulli rotations applied with the design matrix $K_{r, t}$ from Definition \ref{defn:Krt}, and at the end of this section sketch the variant using the random design matrix alternative $R_{n, \epsilon}$ introduced in Section \ref{subsec:2-Rne}. Throughout this section, the input $k$\pr{-bpds} instance will be described by its $m \times n$ adjacency matrix of Bernoulli random variables. The problem \pr{isgm}, imbalanced sparse Gaussian mixtures, is a simple vs. simple hypothesis testing problem defined formally below. A similar distribution was also used in \cite{diakonikolas2017statistical} to construct an instance of robust sparse mean estimation inducing the tight statistical-computational gap in the statistical query model.

\begin{definition}[Imbalanced Sparse Gaussian Mixtures]
Given some $\mu \in \mathbb{R}$ and $\epsilon \in (0, 1)$, let $\mu'$ be such that $\epsilon \cdot \mu' + (1 - \epsilon) \cdot \mu = 0$. For each subset $S \subseteq [d]$, $\pr{isgm}_D(n, S, d, \mu, \epsilon)$ denotes the distribution over $X = (X_1, X_2, \dots, X_n)$ where $X_i \in \mathbb{R}^d$ where
$$X_1, X_2, \dots, X_n \sim_{\textnormal{i.i.d.}} \pr{mix}_{\epsilon}\left( \mN(\mu \cdot \mathbf{1}_S, I_d), \mN(\mu' \cdot \mathbf{1}_S, I_d) \right)$$
\end{definition}

We will use the notation $\pr{isgm}(n, k, d, \mu, \epsilon)$ to refer to the hypothesis testing problem between $H_0: X_1, X_2, \dots, X_n \sim_{\text{i.i.d.}} \mN(0, I_d)$ and an alternative hypothesis $H_1$ sampling the distribution above where $S$ is chosen uniformly at random among all $k$-subsets of $[d]$. Our reduction $k$\textsc{-bpds-to-isgm} is shown in Figure \ref{fig:isgmreduction}. The next theorem encapsulates the total variation guarantees of this reduction. A key parameter is the prime number $r$, which is used to parameterize the design matrices $K_{r, t}$ in the $\pr{Bern-Rotations}$ step.

To show the tightest possible statistical-computational gaps in applications of this theorem, we ideally would want to take $n$ such that $n = \Theta(k_nr^t)$. When $r$ is growing with $N$, this induces number theoretic constraints on our choices of parameters that require careful attention and will be discussed in Section \ref{subsec:3-rsme}. Because of this subtlety, we have kept the statement of our next theorem technically precise and in terms of all of the free parameters of the reduction $k$\pr{-bpds-to-isgm}. Ignoring these number theoretic constraints, the reduction $k$\pr{-bpds-to-isgm} can be interpreted as essentially mapping an instance of $k\pr{-bpds}$ with parameters $(m, n, k_m, k_n, p, q)$ with $k_n = o(\sqrt{n})$, $k_m = o(\sqrt{m})$ and planted row indices $S$ where $|S| = k_m$ to the instance $\pr{isgm}_D(N, S, d, \mu, \epsilon)$ where $\epsilon \in (0, 1)$ is arbitrary and can vary with $n$. The target parameters $N, d$ and $\mu$ satisfy that
$$d = \Omega(m), \quad N = o(n) \quad \text{and} \quad \mu \asymp \frac{1}{\sqrt{\log n}} \cdot \sqrt{\frac{\epsilon k_n}{n}}$$
All of our applications will handle the number theoretic constraints to set parameters so that they nearly satisfy these conditions. The slow-growing function $w(n)$ is so that Step 3 subsamples the produced samples by a large enough factor to enable an application of finite de Finetti's theorem.

We now state our total variation guarantees for $k$\textsc{-bpds-to-isgm}. Given a partition $F$ of $[n]$ with $[n] = F_1 \cup F_2 \cup \cdots \cup F_{k_n}$, let $\mU_n(F)$ denote the distribution of $k_n$-subsets of $[n]$ formed by choosing one member element of each of $F_1, F_2, \dots, F_{k_n}$ uniformly at random. Let $\mU_{n, k_n}$ denote the uniform distribution on $k_n$-subsets of $[n]$.


\begin{theorem}[Reduction from $k$\pr{-bpds} to \pr{isgm}] \label{thm:isgmreduction}
Let $n$ be a parameter, $r = r(n) \ge 2$ be a prime number and $w(n) = \omega(1)$ be a slow-growing function. Fix initial and target parameters as follows:
\begin{itemize}
\item \textnormal{Initial} $k\pr{-bpds}$ \textnormal{Parameters:} vertex counts on each side $m$ and $n$ that are polynomial in one another, dense subgraph dimensions $k_m$ and $k_n$ where $k_n$ divides $n$, edge probabilities $0 < q < p \le 1$ with $\min\{q, 1 - q\} = \Omega(1)$ and $p - q \ge (mn)^{-O(1)}$, and a partition $F$ of $[n]$.
\item \textnormal{Target} $\pr{isgm}$ \textnormal{Parameters:} $(N, d, \mu, \epsilon)$ where $\epsilon = 1/r$ and there is a parameter $t = t(N) \in \mathbb{N}$ with
$$wN \le \frac{k_nr(r^t - 1)}{r - 1}, \quad m \le d \le \textnormal{poly}(n), \quad n \le k_nr^t \le \textnormal{poly}(n)  \quad \textnormal{and}$$
$$0 \le \mu \le \frac{\delta}{2 \sqrt{6\log (k_nmr^t) + 2\log (p - q)^{-1}}} \cdot \frac{1}{\sqrt{r^t(r - 1)(1 + (r - 1)^{-1})}}$$
where $\delta = \min \left\{ \log \left( \frac{p}{q} \right), \log \left( \frac{1 - q}{1 - p} \right) \right\}$.
\end{itemize}
Let $\mathcal{A}(G)$ denote $k$\textsc{-bpds-to-isgm} applied with the parameters above to a bipartite graph $G$ with $m$ left vertices and $n$ right vertices. Then $\mathcal{A}$ runs in $\textnormal{poly}(m, n)$ time and it follows that
\begin{align*}
\TV\left( \mathcal{A}\left( \mathcal{M}_{[m] \times [n]}(S \times T, p, q) \right), \, \pr{isgm}_D(N, S, d, \mu, \epsilon) \right) &= O\left( w^{-1} + k_n^{-2}m^{-2}r^{-2t} \right) \\
\TV\left( \mathcal{A}\left( \textnormal{Bern}(q)^{\otimes m \times n} \right), \, \mN(0, I_d)^{\otimes N} \right) &= O\left( k_n^{-2}m^{-2}r^{-2t} \right)
\end{align*}
for all subsets $S \subseteq [m]$ with $|S| = k_m$ and subsets $T \subseteq [n]$ with $|T| = k_n$ and $|T \cap F_i| = 1$ for each $1 \le i \le k_n$.
\end{theorem}

In the rest of this section, let $\mathcal{A}$ denote the reduction $k\pr{-bpds-to-isgm}$ with input $(M, F)$ where $F$ is a partition of $[n]$ and output $(X_1, X_2, \dots, X_N)$. Let $\text{Hyp}(N, K, n)$ denote a hypergeometric distribution with $n$ draws from a population of size $N$ with $K$ success states. We will also need the upper bound on the total variation between hypergeometric and binomial distributions given by
$$\TV\left( \text{Hyp}(N, K, n), \text{Bin}(n, K/N) \right) \le \frac{4n}{N}$$
This bound is a simple case of finite de Finetti's theorem and is proven in Theorem (4) in \cite{diaconis1980finite}. We now proceed to establish the total variation guarantees for Bernoulli rotations and subsampling as in Steps 2 and 3 of $\mathcal{A}$ in the next two lemmas.

Before proceeding to prove these lemmas, we make a definition that will be used in the next few sections. Suppose that $M$ is a $b \times a$ matrix, $F$ and $F'$ are partitions of $[ka]$ and $[kb]$ into $k$ equally sized parts and $S \subseteq [kb]$ is such that $|S \cap F_i| = 1$ for each $1 \le i \le k$. Then define the vector $v = v_{S, F, F'}(M) \in \mathbb{R}^{kb}$ to be such that the restriction $v_{F'_i}$ to the elements of $F'_i$ is given by
$$v_{F'_i} = M_{\cdot, \sigma_{F_i}(j)} \quad \textnormal{where } j \text{ is the unique element in } S \cap F_i$$
Here, $M_{\cdot, j}$ denotes the $j$th column of $M$ and $\sigma_{F_i}$ denotes the order preserving bijection from $F_i$ to $[b]$. In other words, $v_{S, F, F'}$ is the vector formed by concatenating the columns of $M$ along the partition $F'$, where the elements $S \cap F_i$ select which column appears along each part $F_i'$. In this section, whenever $S \cap F_i$ has size one, we will abuse notation and also use $S \cap F_i$ to denote its unique element.

\begin{lemma}[Bernoulli Rotations for $\pr{isgm}$] \label{lem:isgm-rotations}
Let $F'$ and $F''$ be a fixed partitions of $[k_nr^t]$ and $[k_nr\ell]$ into $k_n$ parts of size $r^t$ and $r\ell$, respectively, and let $S \subseteq [m]$ be a fixed $k_m$-subset. Let $T \subseteq [k_n r^t]$ where $|T \cap F_i'| = 1$ for each $1 \le i \le k_n$. Let $\mathcal{A}_{\textnormal{2}}$ denote Step 2 of $k\pr{-bpds-to-isgm}$ with input $M_{\textnormal{PD}}$ and output $M_{\textnormal{R}}$. Suppose that $p, q$ and $\mu$ are as in Theorem \ref{thm:isgmreduction}, then it follows that
\begin{align*}
&\TV\left( \mathcal{A}_{\textnormal{2}} \left( \mathcal{M}_{[m] \times [k_n r^t]} \left( S \times T, \textnormal{Bern}(p), \textnormal{Bern}(q) \right) \right), \, \mL\left( \mu \sqrt{r^t(r - 1)} \cdot \mathbf{1}_S v_{T, F', F''}(K_{r, t})^\top + \mN(0, 1)^{\otimes m \times k_nr\ell} \right) \right) \\
&\quad \quad = O\left(k_n^{-2}m^{-2}r^{-2t} \right) \\
&\TV\left( \mathcal{A}_{\textnormal{2}} \left(\textnormal{Bern}(q)^{\otimes m \times k_nr^t} \right), \, \mN(0, 1)^{\otimes m \times k_nr\ell} \right) = O\left(k_n^{-2}m^{-2}r^{-2t} \right)
\end{align*}
\end{lemma}

\begin{proof}
First consider the case where $M_{\textnormal{PD}} \sim \mathcal{M}_{[m] \times [k_nr^t]} \left( S \times T, \textnormal{Bern}(p), \textnormal{Bern}(q) \right)$. Observe that the subvectors of $M_{\textnormal{PD}}$ are distributed as
$$(M_{\textnormal{PD}})_{i, F_j'} \sim \left\{ \begin{array}{ll} \pr{pb}\left(F_j', T \cap F_j', p, q\right) &\textnormal{if } i \in S \\ \textnormal{Bern}(q)^{\otimes r^t} &\textnormal{otherwise} \end{array} \right.$$
and are independent. Combining upper bound on the singular values of $K_{r, t}$ in Lemma \ref{lem:Krtsv}, Lemma \ref{lem:bern-rotations} applied with $R_{\pr{rk}} = k_n m r^t$ and the condition on $\mu$ in the statement of Theorem \ref{thm:isgmreduction} implies that
\begin{align*}
\TV\left( (M_{\textnormal{R}})_{i, F''_j}, \, \mN\left( \mu \sqrt{r^t(r - 1)} \cdot (K_{r, t})_{\cdot, T \cap F_j'}, I_{r\ell} \right) \right) &= O\left(k_n^{-3}m^{-3}r^{-2t} \right) \quad \textnormal{if } i \in S\\
\TV\left( (M_{\textnormal{R}})_{i, F''_j}, \, \mN\left( 0, I_{r\ell} \right) \right) &= O\left(k_n^{-3}m^{-3}r^{-2t} \right) \quad \textnormal{otherwise}
\end{align*}
Now observe that the subvectors $(M_{\textnormal{R}})_{i, F''_j}$ are also independent. Therefore the tensorization property of total variation in Fact \ref{tvfacts} implies that $\TV\left( M_{\textnormal{R}}, \mL(Z) \right) = O\left(k_n^{-2}m^{-2}r^{-2t} \right)$ where $Z$ is defined so that its subvectors $Z_{i, F_j''}$ are independent and distributed as
$$Z_{i, F_j''} \sim \left\{ \begin{array}{ll} \mN\left( \mu \sqrt{r^t(r - 1)} \cdot (K_{r, t})_{\cdot, T \cap F_j'}, I_{r\ell} \right) &\textnormal{if } i \in S \\ \mN\left( 0, I_{r\ell} \right) &\textnormal{otherwise} \end{array} \right.$$
Note that the entries of $Z$ are independent Gaussians each with variance $1$. Furthermore, the mean of $Z$ can be verified to be exactly $\mu \sqrt{r^t(r - 1)} \cdot \mathbf{1}_S v_{T, F', F''}(K_{r, t})^\top$. This completes the proof of the first total variation upper bound in the statement of the lemma. The second bound follows from the same argument above applied with $S = \emptyset$.
\end{proof}

\begin{lemma}[Subsampling for $\pr{isgm}$] \label{lem:subsampling}
Let $F', F'', S$ and $T$ be as in Lemma \ref{lem:isgm-rotations}. Let $\mathcal{A}_{\textnormal{3}}$ denote Step 3 of $k\pr{-pds-to-isgm}$ with input $M_{\textnormal{R}}$ and output $(X_1, X_2, \dots, X_N)$. Then
$$\TV\left( \mathcal{A}_{\textnormal{3}} \left( \tau \cdot \mathbf{1}_S v_{T, F', F''}(K_{r, t})^\top + \mN(0, 1)^{\otimes m \times k_n r\ell} \right), \pr{isgm}_D(N, S, d, \mu, \epsilon) \right) \le 4w^{-1}$$
where $\epsilon = 1/r$ and $\mu = \frac{\tau}{\sqrt{r^t(r - 1)}}$. Furthermore, it holds that $\mathcal{A}_{\textnormal{3}} \left( \mN(0, 1)^{\otimes m \times k_n r\ell} \right) \sim \mN(0, I_d)^{\otimes N}$.
\end{lemma}

\begin{proof}
Suppose that $M_{\textnormal{R}} \sim \tau \cdot \mathbf{1}_S K_{T, F', F''}^\top + \mN(0, 1)^{\otimes m \times k_nr\ell}$. For fixed $S, T, F'$ and $F''$, the entries of $M_{\textnormal{R}}$ are independent. Observe that the columns of $M_{\text{R}}$ are independent and either distributed according $\mN(\mu \cdot \mathbf{1}_S, I_m)$ or $\mN(\mu' \cdot \mathbf{1}_S, I_m)$ where $\mu' =  \tau(1 - r)/\sqrt{r^t(r - 1)}$ depending on whether the entry of $v_{T, F', F''}(K_{r, t})$ at the index corresponding to the column is $1/\sqrt{r^t(r - 1)}$ or $(1 - r)/\sqrt{r^t(r - 1)}$.

By Lemma \ref{lem:suborthogonalmatrices}, it follows that each column of $K_{r, t}$ contains exactly $\ell$ entries equal to $(1 - r)/\sqrt{r^t(r - 1)}$. This implies that exactly $k_n(r - 1)\ell$ entries of $v_{T, F', F''}(K_{r, t})$ are equal to $1/\sqrt{r^t(r - 1)}$. Define $\mR_{N}(s)$ to be the distribution on $\mathbb{R}^N$ with a sample $v \sim \mR_{N}(s)$ generated by first choosing an $s$-subset $U$ of $[N]$ uniformly at random and then setting $v_i = 1/\sqrt{r^t(r - 1)}$ if $i \in U$ and $v_i = (1 - r)/\sqrt{r^t(r - 1)}$ if $i \not \in U$. Note that the number of columns distributed as $\mN(\mu \cdot \mathbf{1}_S, I_m)$ in $M_{\text{R}}$ chosen to be in $X$ is distributed according to $\text{Hyp}(k_nr\ell, k_n(r - 1)\ell, N)$. Step 3 of $\mathcal{A}$ therefore ensures that, if $M_{\textnormal{R}}$ is distributed as above, then
$$X \sim \mL\left( \tau \cdot \mathbf{1}_{S} \mR_N(\text{Hyp}(k_n\ell, k_n(r - 1)\ell, N))^\top + \mN(0, 1)^{\otimes d \times N} \right)$$
Observe that the data matrix for a sample from $\pr{isgm}_D(N, S, d, \mu, \epsilon)$ can be expressed similarly as
$$\pr{isgm}_D(N, S, d, \mu, \epsilon) = \mL\left( \tau \cdot \mathbf{1}_{S} \mR_n(\text{Bin}(N, 1 - \epsilon))^\top + \mN(0, 1)^{\otimes d \times N} \right)$$
where again we set $\mu = \tau/\sqrt{r^t(r - 1)}$. The conditioning property of $\TV$ in Fact \ref{tvfacts} now implies that
$$\TV\left( \mL(X), \pr{isgm}_D(N, S, d, \mu, \epsilon) \right) \le \TV\left(\text{Bin}(N, 1 - \epsilon), \text{Hyp}\left(k_nr\ell, k_n(r - 1)\ell, N\right) \right) \le \frac{4N}{k_nr\ell} \le 4w^{-1}$$
The last inequality follows from the application of Theorem (4) in \cite{diaconis1980finite} to hypergeometric distributions above along with the fact that $1 - \epsilon = (k_n(r - 1)\ell)/k_nr\ell$ and $wN \le k_n r\ell$. This completes the proof of the upper bound in the lemma statement. Now consider applying the above argument with $\tau = 0$. It follows that $\mathcal{A}_{\textnormal{3}} \left( \mN(0, 1)^{\otimes m \times k_n r\ell} \right) \sim \mN(0, 1)^{\otimes d \times N} = \mN(0, I_d)^{\otimes N}$, which completes the proof of the lemma.
\end{proof}

We now combine these lemmas to complete the proof of Theorem \ref{thm:isgmreduction}.

\begin{proof}[Proof of Theorem \ref{thm:isgmreduction}]
We apply Lemma \ref{lem:tvacc} to the steps $\mathcal{A}_i$ of $\mathcal{A}$ under each of $H_0$ and $H_1$ to prove Theorem \ref{thm:isgmreduction}. Define the steps of $\mathcal{A}$ to map inputs to outputs as follows
$$(M, F) \xrightarrow{\mathcal{A}_1} (M_{\text{PD}}, F') \xrightarrow{\mathcal{A}_2} (M_{\text{R}}, F'') \xrightarrow{\mathcal{A}_{\text{3}}} (X_1, X_2, \dots, X_N)$$
We first prove the desired result in the case that $H_1$ holds. Consider Lemma \ref{lem:tvacc} applied to the steps $\mathcal{A}_i$ above and the following sequence of distributions
\allowdisplaybreaks
\begin{align*}
\mathcal{P}_0 &= \mathcal{M}_{[m] \times [n]}(S \times T, \textnormal{Bern}(p), \textnormal{Bern}(q)) \\
\mathcal{P}_1 &= \mathcal{M}_{[m] \times [k_nr^t]} \left( S \times T, \textnormal{Bern}(p), \textnormal{Bern}(q) \right) \\
\mathcal{P}_2 &=\mu \sqrt{r^t(r - 1)} \cdot \mathbf{1}_{S} v_{T, F', F''}(K_{r, t})^\top + \mN(0, 1)^{\otimes m \times k_nr\ell} \\
\mathcal{P}_{\text{3}} &= \pr{isgm}_D(N, S, d, \mu, \epsilon)
\end{align*}
As in the statement of Lemma \ref{lem:tvacc}, let $\epsilon_i$ be any real numbers satisfying $\TV\left( \mathcal{A}_i(\mP_{i-1}), \mP_i \right) \le \epsilon_i$ for each step $i$. By construction, the step $\mathcal{A}_1$ is exact and we can take $\epsilon_1 = 0$. Lemma \ref{lem:isgm-rotations} yields that we can take $\epsilon_2 = O\left(k_n^{-2}m^{-2}r^{-2t} \right)$. Applying Lemma \ref{lem:subsampling} yields that we can take $\epsilon_{\text{3}} = 4w^{-1}$. By Lemma \ref{lem:tvacc}, we therefore have that
$$\TV\left( \mathcal{A}\left( \mathcal{M}_{[m] \times [n]}(S \times T, p, q) \right), \, \pr{isgm}_D(N, S, d, \mu, \epsilon) \right) = O\left( w^{-1} + k_n^{-2}m^{-2}r^{-2t} \right)$$
which proves the desired result in the case of $H_1$. Now consider the case that $H_0$ holds and Lemma \ref{lem:tvacc} applied to the steps $\mathcal{A}_i$ and the following sequence of distributions
$$\mathcal{P}_0 = \text{Bern}(Q)^{\otimes m \times n}, \quad \mathcal{P}_1 = \text{Bern}(Q)^{\otimes m \times k_nr^t}, \quad \mathcal{P}_2 = \mN(0, 1)^{\otimes m \times k_nr\ell} \quad \text{and} \quad \mathcal{P}_{\text{3}} = \mN(0, I_d)^{\otimes N}$$
As above, Lemmas \ref{lem:isgm-rotations} and \ref{lem:subsampling} imply that we can take
$$\epsilon_1 = 0, \quad \epsilon_2 = O\left(k_n^{-2}m^{-2}r^{-2t} \right) \quad \text{and} \quad \epsilon_{\text{3}} = 0$$
By Lemma \ref{lem:tvacc}, we therefore have that
$$\TV\left( \mathcal{A}\left( \textnormal{Bern}(q)^{\otimes m \times n} \right), \, \mN(0, I_d)^{\otimes N} \right) = O\left(k_n^{-2}m^{-2}r^{-2t} \right)$$
which completes the proof of the theorem.
\end{proof}

As discussed in Section \ref{subsec:2-Rne}, we can replace $K_{r, t}$ in $k\textsc{-bpds-to-isgm}$ with the random matrix alternative $R_{L, \epsilon}$. More precisely, let $k\textsc{-bpds-to-isgm}_R$ denote the reduction in Figure \ref{fig:isgmreduction} with the following changes:
\begin{itemize}
\item At the beginning of the reduction, rejection sample $R_{L, \epsilon}$ for at most $\Theta((\log L)^2)$ iterations until the criteria of Lemma \ref{lem:Rne} are met, as outlined in Section \ref{subsec:2-Rne}. Let $A \in \mathbb{R}^{L \times L}$ be the resulting matrix or stop the reduction if no such matrix is found. The latter case contributes $L^{-\omega(1)}$ to each of the total variation errors in Corollary \ref{thm:mod-isgmreduction}.
\item The dimensions $r\ell$ and $r^t$ of the matrix $K_{r, t}$ used in $\pr{Bern-Rotations}$ in Step 2 are both replaced throughout the reduction by the parameter $L$. This changes the output dimensions of $M_{\text{PD}}$ and $M_{\text{R}}$ in Steps 1 and 2 to both be $m \times k_n L$.
\item In Step 2, apply $\pr{Bern-Rotations}$ with $A$ instead of $K_{r, t}$ and let $\lambda = C$ where $C$ is the constant in Lemma \ref{lem:Rne}.
\end{itemize}
The reduction $k\textsc{-bpds-to-isgm}_R$ eliminates a number-theoretic constraint in $k\textsc{-bpds-to-isgm}$ arising from the fact the intermediate matrix $M_{\text{R}}$ has a dimension that must be of the form $k_n r^t$ for some integer $t$. In contrast, $k\textsc{-bpds-to-isgm}_R$ only requires that this dimension of $M_{\text{R}}$ be a multiple of $k_n$. This will remove the condition (\pr{t}) from our computational lower bounds for $\pr{rsme}$, which is only restrictive in the very small $\epsilon$ regime of $\epsilon = n^{-\Omega(1)}$. We will deduce this computational lower bound for \pr{rsme} implied by the reduction $k\textsc{-bpds-to-isgm}_R$ formally in Section \ref{subsec:3-rsme}.

The reduction $k\textsc{-bpds-to-isgm}_R$ can be analyzed using an argument identical to the one above, with Lemma \ref{lem:Rne} used in place of Lemma \ref{lem:Krtsv} and accounting for the additional $L^{-\omega(1)}$ total variation error incurred by failing to obtain a $R_{n, \epsilon}$ satisfying the criteria in Lemma \ref{lem:Rne}. Carrying this out yields the following corollary. We remark that the new condition $\epsilon \gg L^{-1} \log L$ in the corollary below will amount to the condition $\epsilon \gg N^{-1/2} \log N$ in our computational lower bounds. This is because, in our applications, we will typically set $N = \tilde{\Theta}(k_n L)$ and $k_n$ to be very close to but slightly smaller than $\sqrt{n} = \tilde{\Theta}(\sqrt{N})$, to ensure that the input $k\pr{-bpds}$ instance is hard. These conditions together with $\epsilon \gg L^{-1} \log L$ amount to the condition on the target parameters given by $\epsilon \gg N^{-1/2} \log N$.

\begin{corollary}[Reduction from $k$\pr{-bpds} to \pr{isgm} with $R_{L, \epsilon}$] \label{thm:mod-isgmreduction}
Let $n$ be a parameter and let $w(n) = \omega(1)$ be a slow-growing function. Fix initial and target parameters as follows:
\begin{itemize}
\item \textnormal{Initial} $k\pr{-bpds}$ \textnormal{Parameters:} $m, n, k_m, k_n, p, q$ and $F$ as in Theorem \ref{thm:isgmreduction}.
\item \textnormal{Target} $\pr{isgm}$ \textnormal{Parameters:} $(N, d, \mu, \epsilon)$ such that there is a parameter $L = L(N) \in \mathbb{N}$ such that $L(N) \to \infty$ and it holds that
$$\max\{wN, n\} \le k_n L \le \textnormal{poly}(n), \quad m \le d \le \textnormal{poly}(n), \quad \frac{w\log L}{L} \le \epsilon \le \frac{1}{2} \quad \textnormal{and}$$
$$0 \le \mu \le \frac{C \delta}{\sqrt{\log (k_nmL) + \log (p - q)^{-1}}} \cdot \sqrt{\frac{\epsilon}{L}}$$
for some sufficiently small constant $C > 0$, where $\delta$ is as in Theorem \ref{thm:isgmreduction}.
\end{itemize}
If $\mathcal{A}$ denotes $k\textsc{-bpds-to-isgm}_R$ applied with the parameters above, then $\mathcal{A}$ runs in $\textnormal{poly}(m, n)$ time and
\begin{align*}
\TV\left( \mathcal{A}\left( \mathcal{M}_{[m] \times [n]}(S \times T, p, q) \right), \, \pr{isgm}_D(N, S, d, \mu, \epsilon) \right) &= o(1) \\
\TV\left( \mathcal{A}\left( \textnormal{Bern}(q)^{\otimes m \times n} \right), \, \mN(0, I_d)^{\otimes N} \right) &= o(1)
\end{align*}
for all $k_m$-subsets $S \subseteq [m]$ and $k_n$-subsets $T \subseteq [n]$ with $|T \cap F_i| = 1$ for each $1 \le i \le k_n$.
\end{corollary}

\subsection{Sparse Mixtures of Regressions and Negative Sparse PCA}
\label{subsec:2-mixtures-slr}

\begin{figure}[t!]
\begin{algbox}
\textbf{Algorithm} $k\textsc{-bpds-to-mslr}$

\vspace{1mm}

\textit{Inputs}: Matrix $M \in \{0, 1\}^{m \times n}$, dense subgraph dimensions $k_m$ and $k_n$ where $k_n$ divides $n$ and the following parameters
\begin{itemize}
\item partition $F$, edge probabilities $0 < q < p \le 1$ and $w(n)$ as in Figure \ref{fig:isgmreduction}
\item target $\pr{mslr}$ parameters $(N, d, \gamma, \epsilon)$ and prime $r$ and $t \in \mathbb{N}$ where $N, d, r, t, \ell$ and $\epsilon = 1/r$ are as in Figure \ref{fig:isgmreduction} with the additional requirement that $N \le n$ and where $\gamma \in (0, 1)$ satisfies that
$$\gamma^2 \le c \cdot \min\left\{ \frac{k_m}{r^{t + 1}\log(k_nmr^t) \log N}, \, \frac{k_n k_m}{n \log(mn)} \right\}$$
for a sufficiently small constant $c > 0$.
\end{itemize}
\begin{enumerate}
\item \textit{Clone}: Compute the matrices $M_{\pr{isgm}} \in \{0, 1\}^{m \times n}$ and $M_{\pr{neg-spca}} \in \{0, 1\}^{m \times n}$ by applying $\pr{Bernoulli-Clone}$ with $t = 2$ copies to the entries of the matrix $M$ with input Bernoulli probabilities $p$ and $q$, and output probabilities $p$ and $Q = 1 - \sqrt{(1 - p)(1 - q)} + \mathbf{1}_{\{p = 1\}} \left( \sqrt{q} - 1 \right)$.
\item \textit{Produce} \textsc{isgm} \textit{Instance}: Form $(Z_1, Z_2, \dots, Z_N)$ where $Z_i \in \mathbb{R}^d$ as the output of $k\pr{-bpds-to-isgm}$ applied to the matrix $M_{\pr{isgm}}$ with partition $F$, edge probabilities $0 < Q < p \le 1$, slow-growing function $w$, target $\pr{isgm}$ parameters $(N, d, \mu, \epsilon)$ and $\mu > 0$ given by
$$\mu = 4 \gamma \cdot \sqrt{\frac{\log N}{k_m}}$$
\item \textit{Produce} \textsc{neg-spca} \textit{Instance}: Form $(W_1, W_2, \dots, W_n)$ where $W_i \in \mathbb{R}^d$ as the output of $\pr{bpds-to-neg-spca}$ applied to the matrix $M_{\pr{neg-spca}}$ with edge probabilities $0 < Q < p \le 1$, target dimension $d$ and parameter $\tau > 0$ satisfying that
$$\tau^2 = \frac{8n \gamma^2}{k_n k_m(1 - \gamma^2)}$$
\item \textit{Scale and Label} \textsc{isgm} \textit{Instance}: Generate $y_1, y_2, \dots, y_N \sim_{\text{i.i.d.}} \mN(0, 1 + \gamma^2)$ and truncate each $y_i$ to satisfy $|y_i| \le 2 \sqrt{(1 + \gamma^2) \log N}$. Generate $G_1, G_2, \dots, G_N \sim_{\text{i.i.d.}} \mN(0, I_d)$ and form $(Z_1', Z_2', \dots, Z_N')$ where $Z_i' \in \mathbb{R}^d$ as 
$$Z_i' = \frac{y_i}{4(1 + \gamma^2)} \sqrt{\frac{2}{\log N}} \cdot Z_i + \sqrt{1 - \frac{y_i^2}{4(1 + \gamma^2)^2\log N}} \cdot G_i$$
\item \textit{Merge and Output}: For each $1 \le i \le N$, let $X_i = \frac{1}{\sqrt{2}} \left( Z_i' + W_i \right)$ and output the $N$ labelled pairs $(X_1, y_1), (X_2, y_2), \dots, (X_N, y_N)$.
\end{enumerate}
\vspace{1mm}

\end{algbox}
\caption{Reduction from bipartite planted dense subgraph to mixtures of sparse linear regressions through imbalanced Gaussian mixtures and negative sparse PCA}
\label{fig:mixtures-slr-reduction}
\end{figure}


In this section, we combine the previous two reductions to $\pr{neg-spca}$ and $\pr{isgm}$ with some additional observations to produce a single reduction that will be used to prove two of our main results in Section \ref{subsec:3-slr} -- computational lower bounds for mixtures of SLRs and robust SLR. We begin this section by generalizing our definition of the distribution $\pr{mslr}_D(n, S, d, \gamma, 1/2)$ from Section \ref{subsec:2-formulations} to simultaneously capture the mixtures of SLRs distributions we will reduce to and our adversarial construction for robust SLR.

Recall from Section \ref{subsec:2-formulations} that $\pr{lr}_d(v)$ denotes the distribution of a single sample-label pair $(X, y) \in \mathbb{R}^d \times \mathbb{R}$ given by $y = \langle v, X \rangle + \eta$ where $X \sim \mN(0, I_d)$ and $\eta \sim \mN(0, 1)$. Our generalization of $\pr{mslr}_D$ will be parameterized by $\epsilon \in (0, 1)$. The canonical setup for mixtures of SLRs from Section \ref{subsec:2-formulations} corresponds to setting $\epsilon = 1/2$ and formally is restated in the following definition for convenience.

\begin{definition}[Mixtures of Sparse Linear Regressions with $\epsilon = 1/2$] \label{defn:mslr-balanced}
Let $\gamma \in \mathbb{R}$ be such that $\gamma > 0$. For each subset $S \subseteq [d]$, let $\pr{mslr}_D(n, S, d, \gamma, 1/2)$ denote the distribution over $n$-tuples of independent data-label pairs $(X_1, y_1), (X_2, y_2), \dots, (X_n, y_n)$ where $X_i \in \mathbb{R}^d$ and $y_i \in \mathbb{R}$ are sampled as follows:
\begin{itemize}
\item first sample $n$ independent Rademacher random variables $s_1, s_2, \dots, s_n \sim_{\textnormal{i.i.d.}} \textnormal{Rad}$; and
\item then form data-label pairs $(X_i, y_i) \sim \pr{lr}_d(\gamma s_i v_S)$ for each $1 \le i \le n$.
\end{itemize}
where $v_S \in \mathbb{R}^d$ is the $|S|$-sparse unit vector $v_S = |S|^{-1/2} \cdot \mathbf{1}_S$.
\end{definition}

Our more general formulation when $\epsilon < 1/2$ is described in the definition below. When $\epsilon < 1/2$, the distribution $\pr{mslr}_D(n, S, d, \gamma, \epsilon)$ can always be produced by an adversary in robust SLR. This observation will be discussed in more detail and used in Section \ref{subsec:3-slr} to show computational lower bounds for robust SLR. The reason we have chosen to write these two different distributions under a common notation is that the main reduction of this section, $k\textsc{-bpds-to-mslr}$, will simultaneously map to both mixtures of SLRs and robust SLR. Lower bounds for the mixture problem will be obtained by setting $r = 2$ in the reduction to $\pr{isgm}$ used as a subroutine in $k\textsc{-bpds-to-mslr}$, while lower bounds for robust sparse regression will be obtained by taking $r > 2$. These implications of $k\textsc{-bpds-to-mslr}$ are discussed further in Section \ref{sec:3-robust-and-supervised}.

\begin{definition}[Mixtures of Sparse Linear Regressions with $\epsilon < 1/2$] \label{defn:mslr-imbalanced}
Let $\gamma > 0$, $\epsilon \in (0, 1/2)$ and let $a$ denote $a = \epsilon^{-1}(1 - \epsilon)$. For each subset $S \subseteq [d]$, let $\pr{mslr}_D(n, S, d, \gamma, \epsilon)$ denote the distribution over $n$-tuples of data-label pairs $(X_1, y_1), (X_2, y_2), \dots, (X_n, y_n)$ sampled as follows:
\begin{itemize}
\item the pairs $(b_1, X_1, y_1), (b_2, X_2, y_2), \dots, (b_n, X_n, y_n)$ are i.i.d. and $b_1, b_2, \dots, b_n \sim_{\textnormal{i.i.d.}} \textnormal{Bern}(1 - \epsilon)$;
\item if $b_i = 1$, then $(X_i, y_u) \sim \pr{lr}_d(\gamma v_S)$ where $v_S$ is as in Definition \ref{defn:mslr-balanced}; and
\item if $b_i = 0$, then $(X_i, y_i)$ is jointly Gaussian with mean zero and $(d + 1) \times (d + 1)$ covariance matrix
$$\left[\begin{matrix} \Sigma_{XX} & \Sigma_{Xy} \\ \Sigma_{yX} & \Sigma_{yy} \end{matrix} \right] = \left[\begin{matrix} I_d + \frac{(a^2 - 1)\gamma^2}{1 + \gamma^2} \cdot v_S v_S^\top & -a\gamma \cdot v_S \\ -a\gamma \cdot v_S^\top & 1 + \gamma^2 \end{matrix} \right]$$
\end{itemize}
\end{definition}

The main reduction of this section from $k\pr{-bpds}$ to $\pr{mslr}$ is shown in Figure \ref{fig:mixtures-slr-reduction}. This reduction inherits the number theoretic constraints of our reduction to $\pr{isgm}$ mentioned in the previous section. These will be discussed in more detail when $k\textsc{-bpds-to-mslr}$ is used to deduce computational lower bounds in Section \ref{subsec:3-slr}. The following theorem gives the total variation guarantees for $k\textsc{-bpds-to-mslr}$.


\begin{theorem}[Reduction from $k\pr{-bpds}$ to $\pr{mslr}$] \label{thm:slr-reduction}
Let $n$ be a parameter, $r = r(n) \ge 2$ be a prime number and $w(n) = \omega(1)$ be a slow-growing function. Fix initial and target parameters as follows:
\begin{itemize}
\item \textnormal{Initial} $k\pr{-bpds}$ \textnormal{Parameters:} vertex counts on each side $m$ and $n$ that are polynomial in one another and satisfy the condition that $n \gg m^3$, subgraph dimensions $k_m$ and $k_n$ where $k_n$ divides $n$, constant densities $0 < q < p \le 1$ and a partition $F$ of $[n]$.
\item \textnormal{Target} $\pr{mslr}$ \textnormal{Parameters:} $(N, d, \gamma, \epsilon)$ where $\epsilon = 1/r$ and there is a parameter $t = t(N) \in \mathbb{N}$ with
$$N \le n, \quad wN \le \frac{k_nr(r^t - 1)}{r - 1}, \quad m \le d \le \textnormal{poly}(n), \quad \textnormal{and} \quad n \le k_nr^t \le \textnormal{poly}(n)$$
and where $\gamma \in (0, 1/2)$ satisfies that
$$\gamma^2 \le c \cdot \min\left\{ \frac{k_m}{r^{t + 1}\log(k_nmr^t) \log N}, \, \frac{k_n k_m}{n \log(mn)} \right\}$$
for a sufficiently small constant $c > 0$.
\end{itemize}
Let $\mathcal{A}(G)$ denote $k$\textsc{-bpds-to-mslr} applied with the parameters above to a bipartite graph $G$ with $m$ left vertices and $n$ right vertices. Then $\mathcal{A}$ runs in $\textnormal{poly}(m, n)$ time and it follows that
\begin{align*}
\TV\left( \mathcal{A}\left( \mathcal{M}_{[m] \times [n]}(S \times T, p, q) \right), \, \pr{mslr}_D(N, S, d, \gamma, \epsilon) \right) &= O\left( w^{-1} + k_n^{-2}m^{-2}r^{-2t} + m^{3/2} n^{-1/2} \right) \\
&\quad \quad + O\left( k_n(4e^{-3})^{n/2k_n} + N^{-1} \right) \\
\TV\left( \mathcal{A}\left( \textnormal{Bern}(q)^{\otimes m \times n} \right), \, \left( \mN(0, I_d) \otimes \mN\left(0, 1 + \gamma^2\right) \right)^{\otimes N} \right) &= O\left( k_n^{-2}m^{-2}r^{-2t} + m^{3/2} n^{-1/2} \right)
\end{align*}
for all subsets $S \subseteq [m]$ with $|S| = k_m$ and subsets $T \subseteq [n]$ with $|T| = k_n$ and $|T \cap F_i| = 1$ for each $1 \le i \le k_n$.
\end{theorem}

The proof of this theorem will be broken into several lemmas for clarity. The following four lemmas analyze the approximate Markov transition properties of Steps 4 and 5 of $k$\textsc{-bpds-to-mslr}. The first three lemmas establishes a total variation upper bound in the single sample case. The fourth lemma is a simple consequence of the first two and establishes the Markov transition properties for Steps 4 and 5 together.


\begin{lemma}[Planted Single Sample Labelling] \label{lem:planted-label}
Let $N$ be a parameter, $\gamma, \mu' \in (0, 1)$, $C > 0$ be a constant and $u \in \mathbb{R}^d$ be such that $\| u \|_2 = 1$ and $4C^2 \gamma^2 \le (\mu')^2/\log N $. Define the random variables $(X, y)$ and $(X', y')$ where $X, X' \in \mathbb{R}^d$ and $y, y' \in \mathbb{R}$ as follows:
\begin{itemize}
\item Let $X \sim \mN\left(0, I_d \right)$ and $\eta \sim \mN(0, 1)$ be independent, and define
$$y = \gamma \cdot \langle u, X \rangle + \eta$$
\item Let $y'$ be a sample from $\mN(0, 1 + \gamma^2)$ truncated to satisfy $|y'| \le C\sqrt{(1 + \gamma^2) \log N}$, and let $Z \sim \mN(\mu' \cdot u, I_d)$, $G \sim \mN(0, I_d)$ and $W \sim \mN\left(0, I_d - \frac{2\gamma^2}{1 + \gamma^2} \cdot uu^\top\right)$ be independent. Now let $X'$ be
\begin{equation} \label{eqn:observation-mslr}
X' = \frac{1}{\sqrt{2}} \left( \frac{\gamma \cdot y'\sqrt{2}}{\mu'(1 + \gamma^2)} \cdot Z + \sqrt{1 - 2\left( \frac{\gamma \cdot y'}{\mu'(1 + \gamma^2)} \right)^2} \cdot G + W \right)
\end{equation}
\end{itemize}
Then it follows that, as $N \to \infty$,
$$\TV\left( \mL(X, y), \mL(X', y') \right) = O\left( N^{-C^2/2} \right) $$
\end{lemma}

\begin{proof}
First observe that $4C^2 \gamma^2 \le (\mu')^2/\log N$ implies that since $|y'| \le C \sqrt{(1 + \gamma^2) \log N}$ holds almost surely and $\gamma \in (0, 1)$, it follows that
$$2\left( \frac{\gamma \cdot y'}{\mu'(1 + \gamma^2)} \right)^2 \le 2(1 + \gamma^2) C^2 \gamma^2 (\mu')^{-2} \log N \le 1$$
and hence $X'$ is well-defined almost surely.

Now note that since $y$ is a linear function of $X$ and $\eta$, which are independent Gaussians, it follows that the $d + 1$ entries of $(X, y)$ are jointly Gaussian. Since $\| u \|_2 = 1$, it follows that $\text{Var}(y) = 1 + \gamma^2$ and furthermore $\text{Cov}(y, X) = \bE[Xy] = \gamma \cdot u$. This implies that the covariance matrix of $(X, y)$ is given by
$$\left[\begin{matrix} I_d & \gamma \cdot u \\ \gamma \cdot u^\top & 1 + \gamma^2 \end{matrix} \right]$$
It is well known that $X | y$ is a Gaussian vector with mean and covariance matrix given by
$$\mL(X | y) = \mN\left( \frac{\gamma \cdot y}{1 + \gamma^2} \cdot u, \, I_d - \frac{\gamma^2}{1 + \gamma^2} \cdot uu^\top \right)$$
Now consider $\mL(X' | y')$. Let $Z = \mu' \cdot u + G'$ where $G' \sim \mN(0, I_d)$ and note that
$$X' = \frac{\gamma \cdot y'}{1 + \gamma^2} \cdot u + \frac{\gamma \cdot y'}{\mu'(1 + \gamma^2)} \cdot G' + \frac{1}{\sqrt{2}} \cdot \sqrt{1 - 2\left( \frac{\gamma \cdot y'}{\mu'(1 + \gamma^2)} \right)^2} \cdot G + \frac{1}{\sqrt{2}} \cdot W$$
Note that since $y', G', G$ and $W$ are independent, it follows that all of the entries of the second, third and fourth terms in the expression above are jointly Gaussian conditioned on $y'$. Therefore the entries of $X' | y'$ are also jointly Gaussian. Furthermore the second, third and fourth terms in the expression above for $X'$ have covariance matrices given by
$$\left( \frac{\gamma \cdot y'}{\mu'(1 + \gamma^2)} \right)^2 \cdot I_d, \quad \left( \frac{1}{2} - \left( \frac{\gamma \cdot y'}{\mu'(1 + \gamma^2)} \right)^2 \right) \cdot I_d \quad \textnormal{and} \quad \frac{1}{2} \cdot I_d - \frac{\gamma^2}{1 + \gamma^2} \cdot uu^\top$$
respectively, conditioned on $y'$. Since these three terms are independent conditioned on $y'$, it follows that $X' | y'$ has covariance matrix $I_d - \frac{\gamma^2}{1 + \gamma^2} \cdot uu^\top$ and therefore that
$$\mL(X' | y') = \mN\left( \frac{\gamma \cdot y}{1 + \gamma^2} \cdot u, \, I_d - \frac{\gamma^2}{1 + \gamma^2} \cdot uu^\top \right)$$
and is hence identically distributed to $\mL(X|y)$. Let $\Phi(x) = \frac{1}{\sqrt{2\pi}} \int_{-\infty}^x e^{-x^2/2} dx$ be the CDF of $\mN(0, 1)$. The conditioning property of total variation in Fact \ref{tvfacts} therefore implies that
\begin{align*}
\TV\left( \mL(X, y), \mL(X', y') \right) &\le \TV\left( \mL(y), \mL(y') \right) \\
&= \bP\left[ |y| > c \sqrt{(1 + \gamma^2) \log N} \right] \\
&= 2 \cdot \left( 1 - \Phi\left( C \sqrt{\log N} \right) \right) \\
&= O\left( N^{-C^2/2} \right) 
\end{align*}
where the first equality holds due to the conditioning on an event property of total variation in Fact \ref{tvfacts} and the last upper bound follows from the standard estimate $1 - \Phi(x) \le \frac{1}{\sqrt{2\pi}} \cdot x^{-1} \cdot e^{-x^2/2}$ for $x \ge 1$. This completes the proof of the lemma.
\end{proof}

The next lemma establishes single sample guarantees that will be needed to analyze the case in which $\epsilon < 1/2$. The proof of this lemma is very similar to that of Lemma \ref{lem:planted-label} and is deferred to Appendix \ref{sec:app-label-generation}.

\begin{lemma}[Imbalanced Planted Single Sample Labelling] \label{lem:imbalanced-planted-label}
Let $N, \gamma, \mu', C$ and $u$ be as in Lemma \ref{lem:planted-label} and let $\mu'' \in (0, 1)$. Define the random variables $(X, y)$ and $(X', y')$ as follows:
\begin{itemize}
\item Let $(X, y)$ where $X \in \mathbb{R}^d$ and $y \in \mathbb{R}$ be jointly Gaussian with mean zero and $(d + 1) \times (d + 1)$ covariance matrix given by
$$\left[\begin{matrix} \Sigma_{XX} & \Sigma_{Xy} \\ \Sigma_{yX} & \Sigma_{yy} \end{matrix} \right] = \left[\begin{matrix} I_d + \frac{(a^2 - 1)\gamma^2}{1 + \gamma^2} \cdot uu^\top & a\gamma \cdot u \\ a\gamma \cdot u^\top & 1 + \gamma^2 \end{matrix} \right]$$
\item Let $y', Z, G$ and $W$ be independent where $y', G$ and $W$ are distributed as in Lemma \ref{lem:planted-label} and $Z \sim \mN(\mu'' \cdot u, I_d)$. Let $X'$ be defined by Equation (\ref{eqn:observation-mslr}) as in Lemma \ref{lem:planted-label}.
\end{itemize}
Then it follows that, as $N \to \infty$,
$$\TV\left( \mL(X, y), \mL(X', y') \right) = O\left( N^{-C^2/2} \right)$$
\end{lemma}

We now state a similar lemma analyzing a single sample in Step 4 of $k$\textsc{-bpds-to-mslr} in the case where $X$ and $W$ are not planted. Its proof is also deferred to Appendix \ref{sec:app-label-generation}.

\begin{lemma}[Unplanted Single Sample Labelling] \label{lem:unplanted-label}
Let $N, \gamma, \mu', C$ and $u$ be as in Lemma \ref{lem:planted-label}. Suppose that $y'$ is a sample from $\mN(0, 1 + \gamma^2)$ truncated to satisfy $|y'| \le C\sqrt{(1 + \gamma^2) \log N}$ and $Z, G, W \sim_{\textnormal{i.i.d.}} \mN(0, I_d)$ are independent. Let $X'$ be defined by Equation (\ref{eqn:observation-mslr}) as in Lemma \ref{lem:planted-label}. Then, as $N \to \infty$,
$$\TV\left( \mL(X', y'), \mN(0, I_d) \otimes \mN(0, 1 + \gamma^2) \right) = O\left( N^{-C^2/2} \right) $$
\end{lemma}

Combining these three lemmas, we now can analyze Step 4 and Step 5 of $\mathcal{A}$. Let $\mathcal{A}_{\text{4-5}}(Z, W)$ denote Steps 4 and 5 of $\mathcal{A}$ with inputs $Z = (Z_1, Z_2, \dots, Z_N)$ and $W = (W_1, W_2, \dots, W_n)$ and output $\left((X_1, y_1), (X_2, y_2), \dots, (X_N, y_N)\right)$. The next lemma applies the previous two lemmas to establish the Markov transition properties of $\mathcal{A}_{\text{4-5}}$.

\begin{lemma}[Scaling and Labelling $\pr{isgm}$ Instances] \label{lem:isgm-label}
Let $r, N, d, \gamma, \epsilon, m, n, k_n, k_m$ and $S \subseteq [m]$ where $|S| = k_m$ be as in Theorem \ref{thm:slr-reduction} and let $\mu, \gamma, \theta > 0$ be such that
$$\mu = 4 \gamma \cdot \sqrt{\frac{\log N}{k_m}}, \quad \tau^2 = \frac{8n \gamma^2}{k_n k_m(1 - \gamma^2)} \quad \textnormal{and} \quad \theta = \frac{\tau^2 k_n k_m}{4n + \tau^2 k_n k_m}$$
If $Z \sim \pr{isgm}(N, S, d, \mu, \epsilon)$ and $W \sim \mN\left(0, \, I_d - \theta v_S v_S^\top\right)^{\otimes n}$, then
$$\TV\left( \mathcal{A}_{\textnormal{4-5}}(Z, W), \, \pr{mslr}_D(N, S, d, \gamma, \epsilon) \right) = O\left(N^{-1}\right)$$
If $Z \sim \mN(0, I_d)^{\otimes N}$ and $W \sim \mN(0, 1)^{\otimes d \times n}$, then
$$\TV\left( \mathcal{A}_{\textnormal{4-5}}(Z, W), \, \left( \mN(0, I_d) \otimes \mN(0, 1) \right)^{\otimes N} \right) = O\left(N^{-1}\right)$$
\end{lemma}

\begin{proof}
We treat the cases in which $\epsilon = 1/2$ and $\epsilon < 1/2$ as well as the two possible distributions of $(Z, W)$ in the lemma statement separately. We first consider the case where $\epsilon = 1/2$ and $r = 2$ and $Z \sim \pr{isgm}_D(N, S, d, \mu, \epsilon)$ and $W \sim \mN\left(0, \, I_d - \theta v_S v_S^\top\right)^{\otimes n}$. The $Z_i$ are independent and can be generated by first sampling $s_1, s_2, \dots, s_N \sim_{\text{i.i.d.}} \text{Bern}(1/2)$ and then setting
$$Z_i \sim \left\{ \begin{array}{ll} \mN(\mu \sqrt{k_m} \cdot v_S, I_d) &\text{if } s_i = 1 \\ \mN(-\mu \sqrt{k_m} \cdot v_S, I_d) &\text{if } s_i = 0 \end{array} \right.$$
where $v_S = k_m^{-1/2} \cdot \mathbf{1}_S$. Let $\mu' = \mu \sqrt{k_m}$. It can be verified that the settings of $\mu, \gamma$ and $\theta$ above ensure that
$$\frac{\gamma \sqrt{2}}{\mu'(1 + \gamma^2)} = \frac{1}{4(1 + \gamma^2)} \cdot \sqrt{\frac{2}{\log N}} \quad \text{and} \quad \theta = \frac{2\gamma^2}{1 + \gamma^2}$$
Let $X \sim \mN(0, I_d)$ and $\eta \sim \mN(0, 1)$ be independent. Applying Lemma \ref{lem:planted-label} with $\mu' = \mu \sqrt{k_m}$, $C = 2$, $u = v_S$ and $u = -v_S$, the equalities above and the definition of $X_i$ in Figure \ref{fig:mixtures-slr-reduction} now imply that
\begin{align*}
\TV\left( \mL(X_i, y_i | s_i = 1), \mL\left(X, \gamma \cdot \langle v_S, X \rangle + \eta \right) \right) &= O(N^{-2}) \\
\TV\left( \mL(X_i, y_i | s_i = 0), \mL\left(X, -\gamma \cdot \langle v_S, X \rangle + \eta \right) \right) &= O(N^{-2})
\end{align*}
for each $1 \le i \le N$. The conditioning property of total variation from Fact \ref{tvfacts} now implies that if $\mL_1 = \mL\left(X, \gamma \cdot \langle v_S, X \rangle + \eta \right)$ and $\mL_2 = \mL\left(X, -\gamma \cdot \langle v_S, X \rangle + \eta \right)$, then we have that
$$\TV\left( \mL(X_i, y_i), \pr{mix}_{1/2}(\mL_1, \mL_2) \right) = O(N^{-2})$$
For the given distribution on $(Z, W)$, observe that the pairs $(X_i, y_i)$ for $1 \le i \le N$ are independent by construction in $\mathcal{A}$. Thus the tensorization property of total variation from Fact \ref{tvfacts} implies that
$$\TV\left( \mL\left( (X_1, y_1), (X_2, y_2), \dots, (X_N, y_N) \right),\, \pr{mslr}(N, S, d, \gamma, 1/2) \right) = O(N^{-1})$$
where $\pr{mslr}_D(N, S, d, \gamma, 1/2) = \pr{mix}_{1/2}(\mL_1, \mL_2)^{\otimes N}$, which establishes the desired bound when $\epsilon = 1/2$ and for the first distribution of $(Z, W)$. 

The other two cases will follow by nearly identical arguments. Consider the case where $\epsilon$ is arbitrary and if $Z \sim \mN(0, I_d)^{\otimes N}$ and $W \sim \mN(0, 1)^{\otimes d \times n}$, applying Lemma \ref{lem:unplanted-label} with $C = 2$ and $\mu' = \mu \sqrt{k_m}$ yields that
$$\TV\left( \mL(X_i, y_i), \mN(0, I_d) \otimes \mN(0, 1) \right) = O(N^{-2})$$
Applying the tensorization property of total variation from Fact \ref{tvfacts} as above then implies the second bound in the lemma statement. Finally, consider the case in which $\epsilon < 1/2$, $r > 2$ and $(Z, W)$ is still distributed as $Z \sim \pr{isgm}_D(N, S, d, \mu, \epsilon)$ and $W \sim \mN\left(0, \, I_d - \theta v_S v_S^\top\right)^{\otimes n}$. If the $s_i$ are defined as above, then the $Z_i$ are distributed as
$$Z_i \sim \left\{ \begin{array}{ll} \mN\left(\mu \sqrt{k_m} \cdot v_S, I_d\right) &\text{if } s_i = 1 \\ \mN\left(-a\mu \sqrt{k_m} \cdot v_S, I_d\right) &\text{if } s_i = 0 \end{array} \right.$$
where $a = \epsilon^{-1}(1 - \epsilon)$. Now consider applying Lemma \ref{lem:imbalanced-planted-label} with $\mu' = \mu \sqrt{k_m}$, $\mu'' = a \mu' = \mu \epsilon^{-1}(1 - \epsilon)$, $C = 2$ and $u = - v_S$. This yields that
$$\TV\left( \mL(X_i, y_i | s_i = 0), \mL(X, y) \right) = O(N^{-2})$$
where $X$ and $y$ are as in the statement of Lemma \ref{lem:imbalanced-planted-label}. Combining this with the conditioning property of total variation from Fact \ref{tvfacts}, the application of Lemma \ref{lem:planted-label} in the first case above, the tensorization property of total variation from Fact \ref{tvfacts} as in the previous argument and Definition \ref{defn:mslr-imbalanced} yields that
$$\TV\left( \mL\left( (X_1, y_1), (X_2, y_2), \dots, (X_N, y_N) \right), \, \pr{mslr}(N, S, d, \gamma, \epsilon) \right) = O\left(N^{-1}\right)$$
which completes the proof of the lemma.
\end{proof}

With this lemma, the proof of Theorem \ref{thm:slr-reduction} reduces to an application of Lemma \ref{lem:tvacc} through a similar argument to the proof of Theorem \ref{thm:isgmreduction}.

\begin{proof}[Proof of Theorem \ref{thm:slr-reduction}]
Define the steps of $\mathcal{A}$ to map inputs to outputs as follows
$$M \xrightarrow{\mathcal{A}_1} (M_{\pr{isgm}}, M_{\pr{neg-spca}}) \xrightarrow{\mathcal{A}_2} \left(Z, M_{\pr{neg-spca}}\right) \xrightarrow{\mathcal{A}_3} \left(Z, W\right)  \xrightarrow{\mathcal{A}_{\text{4-5}}} \left((X_1, y_1), (X_2, y_2), \dots, (X_N, y_N) \right)$$
where $Z = (Z_1, Z_2, \dots, Z_N)$ and $W = (W_1, W_2, \dots, W_n)$ in Figure \ref{fig:mixtures-slr-reduction}. First note that the condition on $\gamma$ in the theorem statement along with the settings of $\mu$ and $\tau$ in Figure \ref{fig:mixtures-slr-reduction} imply that
\begin{align*}
\tau &\le \frac{\delta}{2 \sqrt{6\log (mn) + 2\log (p - Q)^{-1}}} \quad \text{where} \quad \delta = \min \left\{ \log \left( \frac{p}{Q} \right), \log \left( \frac{1 - Q}{1 - p} \right) \right\} \\
\mu &\le \frac{\delta}{2 \sqrt{6\log (k_nmr^t) + 2\log (p - Q)^{-1}}} \cdot \frac{1}{\sqrt{r^t(r - 1)(1 + (r - 1)^{-1})}}
\end{align*}
for a sufficiently small constant $c > 0$ since $0 < q < p \le 1$ are constants. Let $\theta$ and $v_S$ be as in Lemma \ref{lem:isgm-label}. Consider Lemma \ref{lem:tvacc} applied to the steps $\mathcal{A}_i$ above and the following sequence of distributions
\allowdisplaybreaks
\begin{align*}
\mathcal{P}_0 &= \mathcal{M}_{[m] \times [n]}(S \times T, \textnormal{Bern}(p), \textnormal{Bern}(q)) \\
\mathcal{P}_1 &= \mathcal{M}_{[m] \times [n]}(S \times T, \textnormal{Bern}(p), \textnormal{Bern}(Q)) \otimes \mathcal{M}_{[m] \times [n]}(S \times T, \textnormal{Bern}(p), \textnormal{Bern}(Q)) \\
\mathcal{P}_2 &= \pr{isgm}_D(N, S, d, \mu, \epsilon) \otimes \mathcal{M}_{[m] \times [n]}(S \times T, \textnormal{Bern}(p), \textnormal{Bern}(Q)) \\
\mathcal{P}_{\text{3}} &= \pr{isgm}_D(N, S, d, \mu, \epsilon) \otimes \mN\left(0, \, I_d - \theta v_S v_S^\top\right)^{\otimes n} \\
\mathcal{P}_{\text{4-5}} &= \pr{mslr}_D(N, S, d, \gamma, \epsilon)
\end{align*}
Combining the inequalities above for $\mu$ and $\tau$ with Lemmas \ref{lem:bern-clone} and \ref{lem:isgm-label} and Theorems \ref{thm:isgmreduction} and \ref{thm:neg-spca} implies that we can take
$$\epsilon_1 = 0, \quad \epsilon_2 = O\left( w^{-1} + k_n^{-2}m^{-2}r^{-2t} \right), \quad \epsilon_3 = O\left( m^{3/2} n^{-1/2} + k_n(4e^{-3})^{n/2k_n} \right) \quad \text{and} \quad \epsilon_{\text{4-5}} = O(N^{-1})$$
Applying Lemma \ref{lem:tvacc} now yields the first total variation upper bound in the theorem. Now consider Lemma \ref{lem:tvacc} applied to
\allowdisplaybreaks
\begin{align*}
\mathcal{P}_0 &= \text{Bern}(q)^{\otimes m \times n} \\
\mathcal{P}_1 &= \text{Bern}(Q)^{\otimes m \times n} \otimes \text{Bern}(Q)^{\otimes m \times n} \\
\mathcal{P}_2 &= \mN(0, I_d)^{\otimes N} \otimes \text{Bern}(Q)^{\otimes m \times n} \\
\mathcal{P}_{\text{3}} &= \mN(0, I_d)^{\otimes N} \otimes \mN(0, I_d)^{\otimes n} \\
\mathcal{P}_{\text{4-5}} &= \left( \mN(0, I_d) \otimes \mN(0, 1 + \gamma^2) \right)^{\otimes N}
\end{align*}
By Lemmas \ref{lem:bern-clone} and \ref{lem:isgm-label} and Theorems \ref{thm:isgmreduction} and \ref{thm:neg-spca}, we can take
$$\epsilon_1 = 0, \quad \epsilon_2 = O\left( k_n^{-2}m^{-2}r^{-2t} \right), \quad \epsilon_3 = O\left( m^{3/2} n^{-1/2} \right) \quad \text{and} \quad \epsilon_{\text{4-5}} = O(N^{-1})$$
Applying Lemma \ref{lem:tvacc} now yields the second total variation upper bound in the theorem and completes the proof of the theorem.
\end{proof}

As in the previous section, the random matrix $R_{L, \epsilon}$ can be used in place of $K_{r, t}$ in our reduction $k$\textsc{-bpds-to-mslr}. Specifically, replacing $k\pr{-bpds-to-isgm}$ in Step 2 with $k\pr{-bpds-to-isgm}_R$ and again replacing $r^t$ with the more flexible parameter $L$ yields an alternative reduction $k\textsc{-bpds-to-mslr}_R$. The guarantees below for this modified reduction follow from the same argument as in the proof of Theorem \ref{thm:slr-reduction}, using Corollary \ref{thm:mod-isgmreduction} in place of Theorem \ref{thm:isgmreduction}.

\begin{corollary}[Reduction from $k\pr{-bpds}$ to $\pr{mslr}$ with $R_{L, \epsilon}$] \label{thm:mod-slr-reduction}
Let $n$ be a parameter and let $w(n) = \omega(1)$ be a slow-growing function. Fix initial and target parameters as follows:
\begin{itemize}
\item \textnormal{Initial} $k\pr{-bpds}$ \textnormal{Parameters:} $m, n, k_m, k_n, p, q$ and $F$ as in Theorem \ref{thm:slr-reduction}.
\item \textnormal{Target} $\pr{mslr}$ \textnormal{Parameters:} $(N, d, \gamma, \epsilon)$ and a parameter $L = L(N) \in \mathbb{N}$ such that $N \le n$ and $(N, d, \epsilon, L)$ satisfy the conditions in Corollary \ref{thm:mod-isgmreduction}. Suppose that $\gamma \in (0, 1/2)$ satisfies that
$$\gamma^2 \le c \cdot \min\left\{ \frac{\epsilon k_m}{L \log(k_nmL) \log N}, \, \frac{k_n k_m}{n \log(mn)} \right\}$$
for a sufficiently small constant $c > 0$.
\end{itemize}
If $\mathcal{A}$ denotes $k\textsc{-bpds-to-mslr}_R$ applied with the parameters above, then $\mathcal{A}$ runs in $\textnormal{poly}(m, n)$ time and
\begin{align*}
\TV\left( \mathcal{A}\left( \mathcal{M}_{[m] \times [n]}(S \times T, p, q) \right), \, \pr{mslr}_D(N, S, d, \gamma, \epsilon) \right) &= o(1) \\
\TV\left( \mathcal{A}\left( \textnormal{Bern}(q)^{\otimes m \times n} \right), \, \left( \mN(0, I_d) \otimes \mN\left(0, 1 + \gamma^2\right) \right)^{\otimes N} \right) &= o(1)
\end{align*}
for all $k_m$-subsets $S \subseteq [m]$ and $k_n$-subsets $T \subseteq [n]$ with $|T \cap F_i| = 1$ for each $1 \le i \le k_n$.
\end{corollary}


\section{Completing Tensors from Hypergraphs}
\label{sec:2-hypergraph-planting}


\begin{figure}[t!]
\begin{algbox}
\textbf{Algorithm} \textsc{Advice-Complete-Tensor}

\vspace{1mm}

\textit{Inputs}: \pr{hpds} instance $H \in \mG_n^s$ with edge probabilities $0 < q < p \le 1$, an $(s - 1)$-set of advice vertices $V = \{v_1, v_2, \dots, v_{s - 1}\}$ of $H$

\begin{enumerate}
\item \textit{Clone Hyperedges}: Compute the $(s!)^2$ hypergraphs $H^{\sigma_1, \sigma_2} \in \mG_n^s$ for each pair $\sigma_1, \sigma_2 \in S_s$ by applying $\pr{Bernoulli-Clone}$ with $t = (s!)^2$ to the $\binom{N}{s}$ hyperedge indicators of $H$ with input Bernoulli probabilities $p$ and $q$ and output probabilities $p$ and
$$Q = 1 - (1 - p)^{1 - 1/t}(1 - q)^{1/t} + \mathbf{1}_{\{p = 1\}}\left( q^{1/t} - 1 \right)$$
\item \textit{Form Tensor Entries}: For each $I = (i_1, i_2, \dots, i_s) \in \left( [N] \backslash V \right)^s$, set the $(i_1, i_2, \dots, i_s)$th entry of the tensor $T$ with dimensions $(N - s + 1)^{\otimes s}$ to be the following hyperedge indicator
$$T_{i_1, i_2, \dots, i_s} = \mathbf{1}\left\{ \{v_1, v_2, \dots, v_{s - |P(I)|} \} \cup \{i_1, i_2, \dots, i_s\} \in E\left( H^{\tau_{\textnormal{P}}(I), \tau_{\textnormal{V}}(I)} \right) \right\}$$
where $P(I)$, $\tau_{\textnormal{P}}(I)$ and $\tau_{\textnormal{V}}(I)$ are as in Definition \ref{defn:tuple-stats}.
\item \textit{Output}: Output the order $s$ tensor $T$ with axes indexed by the set $[N] \backslash V$.
\end{enumerate}
\vspace{0.5mm}

\textbf{Algorithm} \textsc{Iterate-and-Reduce}

\vspace{1mm}

\textit{Inputs}: $k$\pr{-hpds} instance $H \in \mG_n^s$ with edge probabilities $0 < q < p \le 1$, partition $E$ of $[n]$ into $k$ equally-sized parts, a one-sided blackbox $\mathcal{B}$ for the corresponding planted tensor problem

\begin{enumerate}
\item For every $(s - 1)$-set of vertices $\{v_1, v_2, \dots, v_{s - 1}\}$ all from different parts of $E$, form the tensor $T$ by applying \textsc{Advice-Complete-Tensor} to $H$ and $\{v_1, v_2, \dots, v_{s - 1}\}$, remove the indices of $T$ that are in the same part of $E$ as at least one of $\{v_1, v_2, \dots, v_{s - 1}\}$ and run the blackbox $\mathcal{B}$ on the resulting tensor $T$.
\item Output $H_1$ if any application if $\mathcal{B}$ in Step 1 output $H_1$.
\end{enumerate}
\vspace{0.5mm}

\end{algbox}
\caption{The first reduction is a subroutine to complete the entries of a planted dense sub-hypergraph problem into a planted tensor problem given an advice set of vertices. The second reduction uses this subroutine to reduce solving a planted dense sub-hypergraph problem to producing a one-sided blackbox solving the planted tensor problem.}
\label{fig:hyp-to-tensors}
\end{figure}


In this section we introduce a key subroutine that will be used in our reduction to tensor PCA in Section \ref{sec:3-tensor}. The starting point for our reduction $k\pr{-hpds-to-tpca}$ is the hypergraph problem $k\pr{-hpds}$. The adjacency tensor of this instance is missing all entries with at least one pair of equal indices. The first procedure \textsc{Advice-Complete-Tensor} in this section gives a method of completing these missing entries and producing an instance of the planted sub-tensor problem, given access to a set of $s - 1$ vertices in the clique, where $s$ is the order of the target tensor. In order to translate this into a reduction, we iterate over all $(s - 1)$-sets of vertices and carry out this reduction for each one, as will be described in more detail later in this section. For the motivation and high-level ideas behind the reductions in this section, we refer to the discussion in Section \ref{subsec:1-tech-completing}.

In order to describe our reduction \textsc{Advice-Complete-Tensor}, we will need the following definition which will be crucial in indexing the missing entries of the tensor.

\begin{definition}[Tuple Statistics] \label{defn:tuple-stats}
Given a tuple $I = (i_1, i_2, \dots, i_s)$ where each $i_j \in U$ for some set $U$, we define the partition $P(I)$ and permutations $\tau_{\textnormal{P}}(I)$ and $\tau_{\textnormal{V}}(I)$ of $[s]$ as follows:
\begin{enumerate}
\item Let $P(I)$ be the unique partition of $[s]$ into nonempty parts $P_1, P_2, \dots, P_t$ where $i_k = i_l$ if and only if $k, l \in P_j$ for some $1 \le j \le t$, and let $|P(I)| = t$.
\item Given the partition $P(I)$, let $\tau_{\textnormal{P}}(I)$ be the permutation of $[s]$ formed by ordering the parts $P_j$ in increasing order of their largest element, and then listing the elements of the parts $P_j$ according to this order, where the elements of each individual part are written in decreasing order.
\item Let $P'_1, P'_2, \dots, P'_t$ be the ordering of the parts of $P(I)$ as defined above and let $v_1, v_2, \dots, v_t$ be such that $v_j = i_k$ for all $k \in P'_j$ or in other words $v_j$ is the common value of $i_k$ of all indices $k$ in the part $P'_j$. The values $v_1, v_2, \dots, v_t$ are by definition distinct and their ordering induces a permutation $\sigma$ on $[t]$. Let $\tau_{\textnormal{V}}(I)$ be the permutation on $[s]$ formed by setting $\left( \tau_{\textnormal{V}}(I) \right)_{[t]} = \sigma$ and extending $\sigma$ to $[s]$ by taking $\left( \tau_{\textnormal{V}}(I) \right)(j) = j$ for all $t < j \le s$.
\end{enumerate}
\end{definition}

Note that $|P(I)|$ is the number of distinct values in $I$ and thus $|P(I)| = |\{i_1, i_2, \dots, i_s\}|$ for each $I$. For example, if $I = (4, 4, 1, 2, 2, 5, 3, 5, 2)$ and $s = 9$, then $P(I)$, $\tau_{\textnormal{P}}(I)$ and $\tau_{\textnormal{V}}(I)$ are
\begin{align*}
P(I) &= \left\{ \{ 1, 2 \}, \{ 3 \}, \{4, 5, 9\}, \{6, 8\}, \{ 7 \} \right\}, \quad \tau_{\textnormal{P}}(I) = (2, 1, 3, 7, 8, 6, 9, 5, 4) \quad \text{and} \\
\tau_{\textnormal{V}}(I) &= (4, 1, 3, 5, 2, 6, 7, 8, 9)
\end{align*}

We now establish the main Markov transition properties of $\textsc{Advice-Complete-Tensor}$. Given a set $X$, let $\mathcal{E}_{X, s}$ be the set $\binom{X}{s}$ of all subsets of $X$ of size $s$.

\begin{lemma}[Completing Tensors with Advice Vertices] \label{lem:completing}
Let $0 < q < p \le 1$ be such that $\min\{q, 1 - q\} = \Omega_N(1)$ and let $s$ be a constant. Let $0 < Q < p$ be given by
$$Q = 1 - (1 - p)^{1 - 1/t}(1 - q)^{1/t} + \mathbf{1}_{\{p = 1\}}\left( q^{1/t} - 1 \right)$$
where $t = (s!)^2$. Let $V$ be an arbitrary $(s - 1)$-subset of $[N]$ and let $\mathcal{A}$ denote $\textsc{Advice-Complete-Tensor}$ with input $H$, output $T$, advice vertices $V$ and parameters $p$ and $q$. Then $\mathcal{A}$ runs in $\textnormal{poly}(N)$ time and satisfies
\begin{align*}
\mathcal{A}\left( \mathcal{M}_{\mathcal{E}_{[N], s}}\left( \mathcal{E}_{S \cup V, s}, \textnormal{Bern}(p), \textnormal{Bern}(q) \right) \right) &\sim \mathcal{M}_{\left( [N] \backslash V \right)^s}\left( S^{s}, \textnormal{Bern}(p), \textnormal{Bern}(Q) \right) \\
\mathcal{A}\left( \mathcal{M}_{\mathcal{E}_{[N], s}}\left( \textnormal{Bern}(q) \right) \right) &\sim \mathcal{M}_{\left( [N] \backslash V \right)^s}\left( \textnormal{Bern}(Q) \right)
\end{align*}
for all subsets $S \subseteq [N]$ disjoint from $V$.
\end{lemma}

\begin{proof}
First note that Step 2 of $\mathcal{A}$ is well defined since the fact that $|P(I)| = |\{i_1, i_2, \dots, i_s\}|$ implies that $\{v_1, v_2, \dots, v_{s - |P(I)|} \} \cup \{i_1, i_2, \dots, i_s\}$ is always a set of size $s$. We first consider the case in which $H \sim \mathcal{M}_{\mathcal{E}_{[N], s}}\left( \mathcal{E}_{S \cup V, s}, \textnormal{Bern}(p), \textnormal{Bern}(q) \right)$. By Lemma \ref{lem:bern-clone}, it follows that the hyperedge indicators of $H^{\sigma_1, \sigma_2}$ are all independent and distributed as
$$\mathbf{1}\left\{ e \in E\left( H^{\sigma_1, \sigma_2} \right) \right\} \sim \left\{ \begin{array}{ll} \textnormal{Bern}(p) &\textnormal{if } e \subseteq S \cup V \\ \textnormal{Bern}(Q) &\textnormal{otherwise} \end{array} \right.$$
for each $\sigma_1, \sigma_2 \in S_s$ and subset $e \subseteq [N]$ with $|e| = s$. We now observe that $T$ agrees in its entrywise marginal distributions with $\mathcal{M}_{\left( [N] \backslash V \right)^s}\left( S^{s}, \textnormal{Bern}(p), \textnormal{Bern}(Q) \right)$. In particular, we have that:
\begin{itemize}
\item if $(i_1, i_2, \dots, i_s)$ is such that $i_j \in S$ for all $1 \le j \le s$ then we have that $\{v_1, v_2, \dots, v_{s - |P(I)|} \} \cup \{i_1, i_2, \dots, i_s\} \subseteq S \cup V$ and hence
$$T_{i_1, i_2, \dots, i_s} = \mathbf{1}\left\{ \{v_1, v_2, \dots, v_{s - |P(I)|} \} \cup \{i_1, i_2, \dots, i_s\} \in E\left( H^{\tau_{\textnormal{P}}(I), \tau_{\textnormal{V}}(I)} \right) \right\} \sim \textnormal{Bern}(p)$$
\item if $(i_1, i_2, \dots, i_s)$ is such that there is some $j$ such that $i_jj \not \in S$, then $\{v_1, v_2, \dots, v_{s - |P(I)|} \} \cup \{i_1, i_2, \dots, i_s\} \not \subseteq S \cup V$ and $T_{i_1, i_2, \dots, i_s} \sim \text{Bern}(Q)$.
\end{itemize}
It suffices to verify that the entries of $T$ are independent. Since all of the hyperedge indicators of the $H^{\sigma_1, \sigma_2}$ are independent, it suffices to verify that the entries of $T$ are equal to distinct hyperedge indicators.

To show this, we will show that $\{i_1, i_2, \dots, i_s\}$, $\tau_{\textnormal{P}}(I)$ and $\tau_{\textnormal{V}}(I)$ determine the tuple $I = (i_1, i_2, \dots, i_s)$, from which the desired result follows. Consider the longest increasing subsequence of $\tau_{\textnormal{P}}(I)$ starting with $\left(\tau_{\textnormal{P}}(I) \right)(1)$. The elements of this subsequence partition $\tau_{\textnormal{P}}(I)$ into contiguous subsequences corresponding to the parts of $P(I)$. Thus $\tau_{\textnormal{P}}(I)$ determines $P(I)$. Now the first $|P(I)|$ elements of $\tau_{\textnormal{V}}(I)$ along with $\{i_1, i_2, \dots, i_s\}$ determine the values $v_j$ in Definition \ref{defn:tuple-stats} corresponding to $I$ on each part of $P(I)$. This uniquely determines the tuple $I$. Therefore the entries $T_{i_1, i_2, \dots, i_s}$ all correspond to distinct hyperedge indicators and are therefore independent. Applying this argument with $S = \emptyset$ yields the second identity in the statement of the lemma. This completes the proof of the lemma.
\end{proof}

We now analyze the additional subroutine \textsc{Iterate-and-Reduce}. This will show it suffices to design a reduction with low total variation error in order to show computational lower bounds for Tensor PCA. Let $k\pr{-pst}^s_E(N, k, p, q)$ denote the following \textit{planted subtensor} hypothesis testing problem with hypotheses
$$H_0 : T \sim \mathcal{M}_{[N]^s}\left( \textnormal{Bern}(q) \right) \quad \text{and} \quad H_1 : T \sim \mathcal{M}_{[N]^s}\left( S^{s}, \textnormal{Bern}(p), \textnormal{Bern}(q) \right)$$
where $S$ is chosen uniformly at random from all $k$-subsets of $[N]$ intersecting each part of $E$ in one element. The next lemma captures our key guarantee of \textsc{Iterate-and-Reduce}.

%
%
%

\begin{lemma}[Hardness of One-Sided Blackboxes by Reduction] \label{cor:one-side-reduction}
Fix a pair $0 < q < p \le 1$ with $\min\{q, 1 - q\} = \Omega(1)$, a constant $s$ and let $Q$ be as in Figure \ref{fig:hyp-to-tensors}. Suppose that there is a reduction mapping both hypotheses of $k\pr{-pst}^s_E(N - (s - 1)N/k, k - s + 1, p, Q)$ with $k = o(\sqrt{N})$ to the corresponding hypotheses $H_0$ and $H_1$ of a testing problem $\mP$ within total variation $O(N^{-s})$. Then the $k\pr{-hpc}^s$ or $k\pr{-hpds}^s$ conjecture for constant $0 < q < p \le 1$ implies that there cannot be a $\textnormal{poly}(n)$ time algorithm $\mathcal{A}$ solving $\mP$ with a low false positive probability of $\bP_{H_0}[\mathcal{A}(X) = H_1] = O(N^{-s})$, where $X$ denotes the observed variable in $\mP$.
\end{lemma}

\begin{proof}
Assume for contradiction that there is a such a $\textnormal{poly}(n)$ time algorithm $\mathcal{A}$ for $\mP$ with $\bP_{H_0}[\mathcal{A}(X) = H_1] = O(N^{-s})$ and Type I$+$II error
$$\bP_{H_0}[\mathcal{A}(X) = H_1] + \bP_{H_1}[\mathcal{A}(X) = H_0] \le 1 - \epsilon$$
for some $\epsilon = \Omega(1)$. Furthermore, let $\mathcal{R}$ denote the reduction described in the lemma. If $H_0'$ and $H_1'$ denote the hypotheses of $k\pr{-pst}^s_E(N - (s - 1)N/k, k - s + 1, p, Q)$ and $T$ denotes an instance of this problem, then $\mathcal{R}$ satisfies that
$$\TV\left( \mathcal{R}\left( \mL_{H_0'}(T) \right), \mL_{H_0}(T) \right) + \TV\left( \mathcal{R}\left( \mL_{H_1'}(T) \right), \mL_{H_1}(T) \right) = O(N^{-s})$$
Now consider applying \textsc{Iterate-and-Reduce} to: (1) a hard instance $H$ of $k\pr{-hpds}(N, k, p, q)$ with $k = o(\sqrt{N})$; and (2) the blackbox $\mathcal{B} = \mathcal{A} \circ \mathcal{R}$. Let $\pr{ir}(H) \in \{H_0'', H_1''\}$ denote the output of \textsc{Iterate-and-Reduce} on input $H$, and let $H_0''$ and $H_1''$ be the hypotheses of $k\pr{-hpds}(N, k, p, q)$. Furthermore, let $T_1, T_2, \dots, T_K$ denote the tensors formed in the $K = \left( \frac{N}{k} \right)^{s - 1} \binom{k}{s - 1}$ iterations of Step 1 of \textsc{Iterate-and-Reduce}. Note that each $T_i$ has all of its $s$ dimensions equal to $N - (s - 1)N/k$ since exactly $s - 1$ parts of $E$ of size $N/k$ are removed from $[N]$ in each iteration of Step 1 of \textsc{Iterate-and-Reduce}. First consider the case in which $H_0''$ holds. Each tensor in the sequence $T_1, T_2, \dots, T_K$ is marginally distributed as $\mathcal{M}_{[N - (s - 1)N/k]^s}\left( \textnormal{Bern}(Q) \right)$ by Lemma \ref{lem:completing}. By definition $\pr{ir}(H) = H_1''$ if and only if some application of $\mathcal{B}(T_i)$ outputs $H_1$. Now note that by a union bound, the definition of $\TV$ and the data-processing inequality, we have that
\begin{align*}
\bP_{H_0''}\left[\pr{ir}(H) = H_1''\right] &\le \sum_{i = 1}^K \bP_{H_0''}[\mathcal{A} \circ \mathcal{R}(T_i) = H_1] \\
&\le \sum_{i = 1}^K \left[ \bP_{H_0}[\mathcal{A}(X) = H_1] + \TV\left( \mathcal{R}\left( \mL_{H_0'}(T) \right), \mL_{H_0}(T) \right) \right] \\
&= O\left( K \cdot N^{-s} \right) = O(N^{-1})
\end{align*}
since $K = O(N^{s - 1})$. Now suppose that $H_1''$ holds and let $i^*$ be the first iteration of \textsc{Iterate-and-Reduce} in which each of the vertices $\{v_1, v_2, \dots, v_{s - 1}\}$ are in the planted dense sub-hypergraph of $H$. Lemma \ref{lem:completing} shows that $T_{i^*}$ is distributed as $\mathcal{M}_{[N - (s - 1)N/k]^s}\left( S^s, \text{Bern}(p), \textnormal{Bern}(Q) \right)$ where $S$ is chosen uniformly at random over all $(k - s + 1)$-subsets of $[N - (s - 1)N/k]$ with one element per part of the input partition $E$ associated with $H$. We now have that
\begin{align*}
\bP_{H_1''}\left[\pr{ir}(H) = H_0''\right] &\le 1 - \bP_{H_1''}\left[\pr{ir}(H) = H_1''\right] \le 1 - \bP_{H_1''}[\mathcal{A} \circ \mathcal{R}(T_{i^*}) = H_1] \\
&\le 1 - \bP_{H_1}[\mathcal{A}(X) = H_1] + \TV\left( \mathcal{R}\left( \mL_{H_1'}(T) \right), \mL_{H_1}(T) \right) \\
&= \bP_{H_1}[\mathcal{A}(X) = H_0] + O(N^{-s})
\end{align*}
Therefore the Type I$+$II error of \textsc{Iterate-and-Reduce} is
$$\bP_{H_0''}\left[\pr{ir}(H) = H_1''\right] + \bP_{H_1''}\left[\pr{ir}(H) = H_0''\right] = \bP_{H_1}[\mathcal{A}(X) = H_0] + O(N^{-1}) \le 1 - \epsilon + O(N^{-1})$$
and \textsc{Iterate-and-Reduce} solves $k\pr{-hpds}$, contradicting the $k\pr{-hpds}$ conjecture.
\end{proof}

%
%
%
%
%
%

\pagebreak

\part{Computational Lower Bounds from $\text{PC}_\rho$}
\label{part:lower-bounds}

\section{Secret Leakage and Hardness Assumptions}
\label{sec:2-secret-leakage}

In this section, we further discuss the conditions in the $\pr{pc}_\rho$ conjecture and provide evidence for it and for the specific hardness assumptions we use in our reductions. In Section \ref{subsec:2-sl-verifying}, we show that $k\pr{-hpc}^s$ is our strongest hardness assumption, explicitly give the $\rho$ corresponding to each of these hardness assumptions and show that the barriers in Conjecture \ref{conj:hard-conj} are supported by the $\pr{pc}_\rho$ conjecture for these $\rho$. In Section \ref{subsec:2-low-degree}, we give more general evidence for the $\pr{pc}_{\rho}$ conjecture through the failure of low-degree polynomial tests. We also discuss technical conditions in variants of the low-degree conjecture and how these relate to the $\pr{pc}_\rho$ conjecture. Finally, in Section \ref{subsec:2-sq}, we give evidence supporting several of the barriers in Conjecture \ref{conj:hard-conj} from statistical query lower bounds.

We remark that, as mentioned at the end of Section \ref{sec:1-PC}, all of our results and conjectures for $\pr{pc}_\rho$ appear to also hold for $\pr{pds}_\rho$ at constant edge densities $0 < q < p \le 1$. Evidence for these extensions to $\pr{pds}_\rho$ from the failure of low-degree polynomials and SQ algorithms can be obtained through computations analogous to those in Sections \ref{subsec:2-low-degree} and \ref{subsec:2-sq}.

\subsection{Hardness Assumptions and the $\pr{pc}_\rho$ Conjecture}
\label{subsec:2-sl-verifying}

In this section, we continue the discussion of the $\pr{pc}_{\rho}$ conjecture from Section \ref{sec:1-PC}. We first show that $k\pr{-hpc}^s$ reduces to the other conjectured barriers in Conjecture \ref{conj:hard-conj}. We then formalize the discussion in Section \ref{sec:1-PC} and explicitly construct secret leakage distributions $\rho$ such that the graph problems in Conjecture \ref{conj:hard-conj} can be obtained from instances of $\pr{pc}_\rho$ with these $\rho$. We then verify that the $\pr{pc}_\rho$ conjecture implies Conjecture \ref{conj:hard-conj} up to arbitrarily small polynomial factors. More precisely, we verify that these $\rho$, when constrained to be in the conjecturally hard parameter regimes in Conjecture \ref{conj:hard-conj}, satisfy the tail bound conditions on $p_\rho(s)$ in the $\pr{pc}_{\rho}$ conjecture.

\paragraph{The $k\pr{-hpc}^s$ Conjecture is the Strongest Hardness Assumption.} First note that when $s = 2$, our conjectured hardness for $k\pr{-hpc}^s$ is exactly our conjectured hardness for $k\pr{-pc}$ in Conjecture \ref{conj:hard-conj}. Thus it suffices to show that Conjecture \ref{conj:hard-conj} for $k\pr{-hpc}^s$ implies the conjecture for $k\pr{-bpc}$ and $\pr{bpc}$. This is the content of the following lemma.

\begin{lemma} \label{lem:khpc-strong}
Let $\alpha$ be a fixed positive rational number and $w = w(n)$ be an arbitrarily slow-growing function with $w(n) \to \infty$. Then there is a positive integer $s$ and a $\textnormal{poly}(n)$ time reduction from $k\pr{-hpc}^s(n, k, 1/2)$ with $k = o(\sqrt{n})$ to either $k\pr{-bpc}(M, N, k_M, k_N, 1/2)$ or $\pr{bpc}(M, N, k_M, k_N, 1/2)$ for some parameters satisfying $M = \Theta(N^\alpha)$ and $Cw^{-1} \sqrt{N} \le k_N = o(\sqrt{N})$ and $Cw^{-1} \sqrt{M} \le k_M = o(\sqrt{M})$ for some positive constant $C > 0$.
\end{lemma}

\begin{proof}
We first describe the desired reduction to $k\pr{-bpc}$. Let $\alpha = a/b$ for two fixed integers $a$ and $b$, and let $H$ be an input instance of $k\pr{-hpc}^{a + b}_E(n, k, 1/2)$ where $E$ is a fixed known partition of $[n]$. Suppose that $H$ is a nearly tight instance with $w^{-1/\max(a, b)} \sqrt{n} \le k = o(\sqrt{n})$. Now consider the following reduction:
\begin{enumerate}
\item Let $R_1, R_2, \dots, R_{a + b}$ be a partition of $[k]$ into $a + b$ sets of sizes differing by at most $1$, and let $E(R_j) = \bigcup_{i \in R_j} E_i$ for each $j \in [a + b]$.
\item Form the bipartite graph $G$ with left vertex set indexed by $V_1 = E(R_1) \times E(R_2) \times \cdots \times E(R_a)$ and right vertex set $V_2 = E(R_{a + 1}) \times E(R_{a + 2}) \times \cdots \times E(R_{a + b})$ such that $(u_1, u_2, \dots, u_a) \in V_1$ and $(v_1, v_2, \dots, v_b) \in V_2$ are adjacent if and only if $\{u_1, \dots, u_a, v_1, \dots, v_b\}$ is a hyperedge of $H$.
\item Output $G$ with left parts $E_{i_1} \times E_{i_2} \times \cdots \times E_{i_a}$ for all $(i_1, i_2, \dots, i_a) \in R_1 \times R_2 \times \cdots \times R_a$ and right parts $E_{i_1} \times E_{i_2} \times \cdots \times E_{i_b}$ for all $(i_1, i_2, \dots, i_b) \in R_{a+1} \times R_{a+2} \times \cdots \times R_{a+b}$, after randomly permuting the vertex labels of $G$ within each of these parts.
\end{enumerate}
Note that since $a + b = \Theta(1)$, we have that $|E(R_i)| = \Theta(n)$ for each $i$ and thus $N = |V_2| = \Theta(n^b)$ and $M = |V_1| = \Theta(n^a) = \Theta(N^\alpha)$. Under $H_0$, each possible hyperedge of $H$ is included independently with probability $1/2$. Since the edge indicators of $G$ corresponds to a distinct hyperedge indicator of $H$ in Step 2 above, it follows that each edge of $G$ is also included with probability $1/2$ and thus $G \sim \mG_B(M, N, 1/2)$.

In the case of $H_1$, suppose that $H$ is distributed according to the hypergraph planted clique distribution with clique vertices $S \subseteq [n]$ where $S \sim \mU_n(E)$. Examining the definition of the edge indicators in Step 2 above yields that $G$ is a sample from $H_1$ of $k\pr{-bpc}(M, N, k_M, k_N, 1/2)$ conditioned on having left biclique set $\prod_{i = 1}^a (S \cap E(R_i))$ and right biclique set $\prod_{i = a + 1}^{a+b} (S \cap E(R_i))$. Observe that these sets have exactly one vertex in $G$ in common with each of the parts described in Step 3 above. Now note that since $S$ has one vertex per part of $E$, we have that $|S \cap E(R_i)| = |R_i| = \Theta(k)$ since $a + b = \Theta(1)$. Thus $k_M = \left| \prod_{i = 1}^a (S \cap E(R_i)) \right| = \Theta(k^a)$ and $k_N = \Theta(k^b)$. The bound on $k$ now implies that the two desired bounds on $k_N$ and $k_M$ hold for a sufficiently small constant $C > 0$. Thus the permutations in Step 3 produce a sample exactly from $k\pr{-bpc}(M, N, k_M, k_N, 1/2)$ in the desired parameter regime. If instead of only permuting vertex labels within each part, we randomly permute all left vertex labels and all right vertex labels in Step 3, the resulting reduction produces $\pr{bpc}$ instead of $k\pr{-bpc}$. The correctness of this reduction follows from the same argument as for $k\pr{-bpc}$.
\end{proof}

We remark that since $m$ and $n$ are polynomial in each other in the setup in Conjecture \ref{conj:hard-conj} for $k\pr{-bpc}$ and $\pr{bpc}$, the lemma above fills out a dense subset of this entire parameter regime -- where $m = \Theta(n^\alpha)$ for some rational $\alpha$. In the case where $\alpha$ is irrational, the reduction in Lemma \ref{lem:khpc-strong}, when composed with our other reductions beginning with $k\pr{-bpc}$ and $\pr{bpc}$, shows tight computational lower bounds up to arbitrarily small polynomial factors $n^{\epsilon}$ by approximating $\alpha$ arbitrarily closely with a rational number.

\paragraph{Hardness Conjectures as Instances of $\pr{pc}_\rho$.} We now will verify that each of the graph problems in Conjecture \ref{conj:hard-conj} can be obtained from $\pr{pc}_\rho$. To do this, we explicitly construct several $\rho$ and give simple reductions from the corresponding instances of $\pr{pc}_\rho$ to these graph problems. We begin with $k\pr{-pc}$, $\pr{bpc}$ and $k\pr{-bpc}$ as their discussion will be brief. \\

\noindent \textit{Secrets for $k\pr{-pc}$, $\pr{bpc}$ and $k\pr{-bpc}$.} Below are the $\rho$ corresponding to these three graph problems. Both $\pr{bpc}$ and $k\pr{-bpc}$ can be obtained by restricting to bipartite subgraphs of the $\pr{pc}_\rho$ instances with these $\rho$.
\begin{itemize}
\item \textbf{$k$-partite \pr{pc}:} Suppose that $k$ divides $n$ and $E$ is a partition of $[n]$ into $k$ parts of size $n/k$. By definition, $k\pr{-pc}_E(n, k, 1/2)$ is $\pr{pc}_\rho(n, k, 1/2)$ where $\rho = \rho_{k\pr{-pc}}(E, n, k)$ is the uniform distribution $\mU_n(E)$ over all $k$-sets of $[n]$ intersecting each part of $E$ in one element.
\item \textbf{bipartite \pr{pc}:} Let $\rho_{\pr{bpc}}(m, n, k_m, k_n)$ be the uniform distribution over all $(k_n + k_m)$-sets of $[n + m]$ with $k_n$ elements in $\{1, 2, \dots, n\}$ and $k_m$ elements in $\{n + 1, n + 2, \dots, n + m\}$. An instance of $\pr{bpc}(m, n, k_m, k_n, 1/2)$ can then be obtained by outputting the bipartite subgraph of $\pr{pc}_\rho(m + n, k_m + k_n, 1/2)$ with this $\rho$, consisting of the edges between left vertex set $\{n + 1, n + 2, \dots, n + m\}$ and right vertex set $\{1, 2, \dots, n\}$.
\item \textbf{$k$-part bipartite \pr{pc}:} Suppose that $k_n$ divides $n$, $k_m$ divides $m$, and $E$ and $F$ are partitions of $[n]$ and $[m]$ into $k_n$ and $k_m$ parts of equal size, respectively. Let $\rho_{k\pr{-bpc}}(E, F, m, n, k_m, k_n)$ be uniform over all $(k_n + k_m)$-subsets of $[n+m]$ with exactly one vertex in each part of both $E$ and $n + F$. Here, $n + F$ denotes the partition of $\{n + 1, n + 2, \dots, n + m\}$ induced by shifting indices in $F$ by $n$. As with $\pr{bpc}$, $k\pr{-bpc}(m, n, k_m, k_n, 1/2)$ can be realized as the bipartite subgraph of $\pr{pc}_\rho(m+n, k_m + k_n, 1/2)$, with this $\rho$, between the vertex sets $\{n + 1, n + 2, \dots, n + m\}$ and $\{1, 2, \dots, n\}$.
\end{itemize}

\vspace{3mm}

\noindent \textit{Secret for $k\pr{-hpc}^s$.} We first will give the secret $\rho$ corresponding to $k\pr{-hpc}^s$ for even $s$, which can be viewed as roughly the pushforward of $\mU_n(E)$ after unfolding the adjacency tensor of $k\pr{-hpc}^s$. The secret for odd $s$ will then be obtained through a slight modification of the even case.

Suppose that $s = 2t$. Given a set $S \subseteq [n]$, let $P_t^n(S)$ denote the subset of $[n^t]$ given by
$$P_t^n(S) = \left\{ 1 + \sum_{j = 0}^{t - 1} (a_j - 1) n^j : a_0, a_1, \dots, a_{t - 1} \in S \right\}$$
In other words, $P_t^n(S)$ is the set of all numbers $x$ in $[n^t]$ such that the base-$n$ representation of $x - 1$ only has digits in $S - 1$, where $S - 1$ is the set of all $s - 1$ where $s \in S$. Note that if $|S| = k$ then $|P_t^n(S)| = k^t$. Given a partition $E$ of $[n]$ into $k$ parts of size $n/k$, let $\rho_{k\pr{-hpc}^s}(E, n, k)$ be the distribution over $k^t$-subsets of $[n^t]$ sampled by choosing $S$ at random from $\mU_n(E)$ and outputting $P_t^n(S)$. Throughout the rest of this section, we will let $I(a_0, a_1, \dots, a_{t - 1})$ denote the sum $1 + \sum_{j = 0}^{t - 1} (a_j - 1) n^j$. We now will show that $k\pr{-hpc}^s_E(n, k, 1/2)$ can be obtained from $\pr{pc}_\rho(n^t, k^t, 1/2)$ where $\rho = \rho_{k\pr{-hpc}^s}(E, n, k)$. Intuitively, this instance of $\pr{pc}_\rho$ has a subset of edges corresponding to the unfolded adjacency tensor of $\pr{-hpc}^s_E$. More formally, consider the following steps.
\begin{enumerate}
\item Let $G$ be an input instance of $\pr{pc}_\rho(n^t, k^t, 1/2)$ and let $H$ be the output hypergraph with vertex set $[n]$.
\item Construct $H$ as follows: for each possible hyperedge $e = \{a_1, a_2, \dots, a_{2t}\}$, with $1 \le a_1 < a_2 < \cdots < a_{2t} \le n$, include $e$ in $H$ if and only if there is an edge between vertices $I(a_1, a_2, \dots, a_t)$ and $I(a_{t + 1}, a_{t + 2}, \dots, a_{2t})$ in $G$.
\end{enumerate}
Under $H_0$, it follows that $G \sim \mG(n^t, 1/2)$. Note that each hyperedge $e$ in Step 2 identifies a unique pair of distinct vertices $I(a_1, a_2, \dots, a_t)$ and $I(a_{t + 1}, a_{t + 2}, \dots, a_{2t})$ in $G$, and thus the hyperedges of $H$ are independently included with probability $1/2$. Under $H_1$, it follows that the instance of $\pr{pc}_\rho(n^t, k^t, 1/2)$ is sampled from the planted clique distribution with clique vertices $P_t^n(S)$ where $S \sim \mU_n(E)$. By the definition of $P_t^n(S)$, it follows that $I(a_1, a_2, \dots, a_t)$ is in this clique if and only if $a_1, a_2, \dots, a_t \in S$. Examining the edge indicators of $H$ then yields that $H$ is a sample from the hypergraph planted clique distribution with clique vertex set $S$. Since $S \sim \mU_n(E)$, under both $H_0$ and $H_1$, it follows that $H$ is a sample from $k\pr{-hpc}^s$.

Now suppose that $s$ is odd with $s = 2t + 1$. The idea in this case is to pair up adjacent digits in base-$n$ expansions and use these pairs to label the vertices of $k\pr{-hpc}^s$. More precisely suppose that $n = N^2$ and $k = K^2$ for some positive integers $K$ and $N$. Let $E$ be a fixed partition of $[n]$ into $k = K^2$ equally sized parts and let $\rho_{k\pr{-hpc}^s}(E, n, k)$ be $\rho_{k\pr{-hpc}^{2s}}(F, N, K)$ as defined above for the even number $2s$, where $F$ is a fixed partition of $[N]$ into $K$ equally sized parts. We now will show that $k\pr{-hpc}^s_E(n, k, 1/2)$ can be obtained from $\pr{pc}_\rho(N^s, K^s, 1/2)$ where $\rho = \rho_{k\pr{-hpc}^{2s}}(F, N, K)$. Let $I'$ be the analogue of $I$ for base-$N$ expansions i.e. let $I(b_0, b_1, \dots, b_{t - 1})$ denote the sum $1 + \sum_{j = 0}^{t - 1} (b_j - 1) N^j$. Consider the following steps.
\begin{enumerate}
\item Let $G$ be an instance of $\pr{pc}_\rho(N^{s}, K^{s}, 1/2)$ and let $H$ be the output hypergraph with vertex set $[n]$.
\item Let $\sigma : [n] \to [n]$ be a bijection such that, for each $i \in [k]$, we have that
$$\sigma(E_i) = \left\{ I'(b_0, b_1) : b_0 \in F_{c_0} \text{ and } b_1 \in F_{c_1} \right\}$$
where $c_0, c_1$ are the unique elements of $[K]$ with $i - 1 = (c_0 - 1) + (c_1 - 1)K$. 
\item Construct $H$ as follows. For each possible hyperedge $e = \{a_1, a_2, \dots, a_{s}\}$, with $1 \le a_1 < a_2 < \cdots < a_{s} \le n$, let $b_{2i - 1}, b_{2i}$ be the unique elements of $[N]$ with $I'(b_{2i - 1}, b_{2i}) = \sigma(a_i)$ for each $i$. Now include $e$ in $H$ if and only if there is an edge between the two vertices $I(b_1, b_2, \dots, b_s)$ and $I(b_{s + 1}, b_{s + 2}, \dots, b_{2s})$ in $G$.
\item Permute the vertex labels of $H$ within each part $F_i$ uniformly at random.
\end{enumerate}
Note that $\sigma$ always trivially exists because the $K^2$ sets $E_1, E_2, \dots, E_{K^2}$ and the $K^2$ sets $F'_{i, j} = \{I'(b_0, b_1) : b_0 \in F_{i} \text{ and } b_1 \in F_{j} \}$ for $1 \le i, j \le K$ are both partitions of $[n]$ into parts of size $N^2/K^2$. As in the case where $s$ is even, under $H_1$ we have that $G \sim \mG(N^{2s}, 1/2)$ and the hyperedges of $H$ are independently included with probability $1/2$, since Step 3 identifies distinct pairs of vertices for each hyperedge $e$. Under $H_1$, let $S \sim \mU_N(F)$ be such that the clique vertices in $G$ are $P_s^N(S)$. By the same reasoning as in the even case, after Step 3, the hypergraph $H$ is distributed as a sample from the hypergraph planted clique distribution with clique vertex set $\sigma^{-1}(I'(S, S))$ where $I'(S, S) = \{ I'(s_0, s_1) : s_0, s_1 \in S\}$. The definition of $\sigma$ now ensures that this clique has one vertex per part of $E$. Step 4 ensures that the resulting hypergraph is exactly a sample from $H_1$ of $k\pr{-hpc}^{s}$. We remark that the conditions $n = N^2$ and $k = K^2$ do not affect our lower bounds when composing the reduction above with our other reductions. This is due to the subsequence criterion for computational lower bounds in Condition \ref{cond:lb}.

\paragraph{Verifying the Conditions of the $\pr{pc}_\rho$ Conjecture.} We now verify that the $\pr{pc}_{\rho}$ conjecture corresponds to the hard regimes in Conjecture \ref{conj:hard-conj} up to arbitrarily small polynomial factors. To do this, it suffices to verify the tail bound on $p_{\rho}(s)$ in the $\pr{pc}_{\rho}$ conjecture for each $\rho$ described above, which is done in the theorem below. In the next section, we will show that a slightly stronger variant of the $\pr{pc}_\rho$ conjecture implies Conjecture \ref{conj:hard-conj} exactly, without the small polynomial factors.

\begin{theorem}[$\pr{pc}_{\rho}$ Conjecture and Conjecture \ref{conj:hard-conj}] \label{thm:verify}
Suppose that $m$ and $n$ are polynomial in one another and let $\epsilon > 0$ be an arbitrarily small constant. Let $\rho$ be any one of the following distributions:
\begin{enumerate}
\item $\rho_{k\pr{-pc}}(E, n, k)$ where $k = O(n^{1/2 - \epsilon})$;
\item $\rho_{\pr{bpc}}(m, n, k_m, k_n)$ where $k_n = O(n^{1/2 - \epsilon})$ and $k_m = O(m^{1/2 - \epsilon})$;
\item $\rho_{k\pr{-bpc}}(E, F, m, n, k_m, k_n)$ where $k_n = O(n^{1/2 - \epsilon})$ and $k_m = O(m^{1/2 - \epsilon})$; and
\item $\rho_{k\pr{-hpc}^t}(E, n, k, 1/2)$ for $t \ge 3$ where $k = O(n^{1/2 - \epsilon})$.
\end{enumerate}
Then there is a constant $\delta > 0$ such that: for any parameter $d = O_n((\log n)^{1 + \delta})$, there is some $p_0 = o_n(1)$ such that $p_{\rho}(s)$ satisfies the tail bounds
$$p_{\rho}(s) \le p_0 \cdot \left\{ \begin{array}{ll} 2^{-s^2} &\textnormal{if } 1 \le s^2 < d \\ s^{-2d-4} &\textnormal{if } s^2 \ge d \end{array} \right.$$
\end{theorem}

\begin{proof}
We first prove the desired tail bounds hold for (1). Let $C > 0$ be a constant such that $k \le C n^{1/2 - \epsilon}$. Note that the probability that $S$ and $S'$ independently sampled from $\rho = \rho_{k\pr{-pc}}(E, n, k)$ intersect in their elements in $E_i$ is $1/|E_i| = k/n$ for each $1 \le i \le k$. Furthermore, these events are independent. Thus it follows that if $\rho = \rho_{k\pr{-pc}}(E, n, k)$, then $p_{\rho}$ is the PMF of $\text{Bin}(k, k/n)$. In particular, we have that
$$p_{\rho}(s) = \binom{k}{s} \left( \frac{k}{n} \right)^s \left( 1- \frac{k}{n} \right)^{k - s} \le k^s \cdot \left( \frac{k}{n} \right)^s = \left( \frac{k^2}{n} \right)^s \le C^{2s} \cdot n^{-2\epsilon s}$$
Let $p_0 = p_0(n)$ be a function tending to zero arbitrarily slowly. The bound above implies that $p_{\rho}(s) \le p_0 \cdot 2^{-s^2}$ as long as $s \le C_1 \log n$ for some sufficiently small constant $C_1 > 0$. Furthermore a direct computation verifies that $p_{\rho}(s) \le p_0 \cdot s^{-2d-4}$ as long as
$$s \ge \frac{C_2 d \log d}{\log n}$$
for some sufficiently large constant $C_2 > 0$. Thus if $d = O_n((\log n)^{1 + \delta})$ for some $\delta \in (0, 1)$, then $\frac{C_2 d \log d}{\log n} < \sqrt{d}$ and $C_1 \log n > \sqrt{d}$ for sufficiently large $n$. This implies the desired tail bound for (1).

The other three cases are similar. In the case of (3), if $S$ and $S'$ are independently sampled from $\rho = \rho_{k\pr{-bpc}}(E, F, m, n, k_m, k_n)$, then the probability that $S$ and $S'$ intersect in their elements in $E_i$ is $k_n/n$ for each $1 \le i \le k_n$, and the probability that they intersect in their elements in $n + F_i$ is $k_m/m$ for each $1 \le i \le k_m$. Thus $p_\rho$ is distributed as independent sum of samples from $\text{Bin}(k_m, k_m/m)$ and $\text{Bin}(k_n, k_n/n)$. It follows that
\begin{align}
p_{\rho}(s) &= \sum_{\ell = 0}^s \binom{k_n}{\ell} \left( \frac{k_n}{n} \right)^{\ell} \left( 1- \frac{k_n}{n} \right)^{k_n - \ell} \cdot \binom{k_m}{s - \ell} \left( \frac{k_m}{m} \right)^{s - \ell} \left( 1- \frac{k_m}{m} \right)^{k_m  - s + \ell} \nonumber \\
&\le \sum_{\ell = 0}^s \left( \frac{k_n^2}{n} \right)^{\ell} \left( \frac{k_m^2}{m} \right)^{s - \ell} \le s \cdot \max\left\{ \left( \frac{k_n^2}{n} \right)^s, \left( \frac{k_m^2}{m} \right)^s \right\} \label{eqn:case-3}
\end{align}
Repeating the bounding argument as in (1) shows that the desired tail bound holds for (3) if $d = O_n((\log n)^{1 + \delta})$ for some $\delta \in (0, 1)$. Since $m$ and $n$ are polynomial in one another implies that $\log m = \Theta(\log n)$, the $(k_m^2/m)^s$ term and the additional factor of $s$ do not affect this bounding argument other than changing the constants $C_1$ and $C_2$. In the case of (2), similar reasoning as in (3) yields that the distribution $p_{\rho}$ where $\rho = \rho_{\pr{bpc}}(m, n, k_m, k_n)$ is the independent sum of samples from $\text{Hyp}\left( n, k_n, k_n \right)$ and $\text{Hyp}\left( m, k_m, k_m \right)$. Now note that
$$\bP\left[ \text{Hyp}\left( n, k_n, k_n \right) = \ell \right] = \frac{\binom{k_n}{\ell} \binom{n - k_n}{k_n - \ell}}{\binom{n}{k_n}} \le \frac{k_n^{\ell} \binom{n - \ell}{k_n - \ell}}{\binom{n}{k_n}} = k_n^{\ell} \prod_{i = 0}^{\ell - 1} \frac{k_n - i}{n - i} \le \left( \frac{k_n^2}{n} \right)^{\ell}$$
This implies the same upper bound on $p_{\rho}(s)$ as in Equation (\ref{eqn:case-3}) also holds for $\rho$ in the case of (2). The argument above for (3) now establishes the desired tail bounds for (2).

We first handle the case in (4) where $t$ is even with $t = 2r$. We have that $\rho = \rho_{k\pr{-hpc}^t}(E, n, k, 1/2)$ can be sampled as $P^n_r(S) \subseteq [n^r]$ where $S \sim \mU_n(E)$. Thus $p_{\rho}(s)$ is the PMF of $|P^n_r(S) \cap P^n_r(S')|$ where $S, S' \sim_{\text{i.i.d.}} \mU_n(E)$. Furthermore the definition of $P^n_r$ implies that $|P^n_r(S) \cap P^n_r(S')| = |S \cap S'|^r$ and, from case (1), we have that $|S \cap S'| \sim \text{Bin}(k, k/n)$. It now follows that
$$p_{\rho}(s) = \left\{ \begin{array}{ll} \binom{k}{s^{1/r}} \left( \frac{k}{n} \right)^{s^{1/r}} \left( 1- \frac{k}{n} \right)^{k - s^{1/r}} &\text{if } s \text{ is an } r\text{th power} \\ 0 &\text{otherwise} \end{array} \right.$$
The same bounds as in case (1) therefore imply that $p_{\rho}(s) \le (k^2/n)^{s^{1/r}}$ for all $s \ge 0$. A similar analysis as in (1) now shows that $p_{\rho}(s) \le p_0 \cdot 2^{-s^2}$ holds if $s \le C_1 (\log n)^{r/(2r - 1)}$ for some sufficiently small constant $C_1 > 0$, and that $p_{\rho}(s) \le p_0 \cdot s^{-2d-4}$ holds if
$$s \ge C_2 \left( \frac{d \log d}{\log n} \right)^r$$
for some sufficiently large constant $C_2 > 0$. As long as $d = O_n((\log n)^{1 + \delta})$ for some $0 < \delta < 1/(2r - 1)$, we have that $C_2 \left( \frac{d \log d}{\log n} \right)^r < \sqrt{d}$ and $C_1 (\log n)^{r/(2r - 1)} > \sqrt{d}$ for sufficiently large $n$. Since $t$ and $r$ are constants here, $\delta$ can be taken to be constant as well. In the case where $t$ is odd, it follows that $\rho_{k\pr{-hpc}^t}(E, n, k, 1/2)$ is the same as $\rho_{k\pr{-hpc}^{2t}}(F, \sqrt{n}, \sqrt{k}, 1/2)$ for some partition $F$ as long as $n$ and $k$ are squares. The same argument establishes the desired tail bound for this prior, completing the case of (4) and proof of the theorem.
\end{proof}

\subsection{Low-Degree Polynomials and the $\pr{pc}_\rho$ Conjecture}
\label{subsec:2-low-degree}

In this section, we show that the low-degree conjecture -- that low-degree polynomials are optimal for a class of average-case hypothesis testing problems -- implies the $\pr{pc}_\rho$ conjecture. In particular, we will obtain a simple expression capturing the power of the optimal low-degree polynomial for $\pr{pc}_\rho$ in Proposition \ref{prop:sl-ld}. We then will apply this proposition to prove Theorem \ref{thm:sl-ld}, showing that the power of this optimal low-degree polynomial tends to zero under the tail bounds on $p_\rho$ in the $\pr{pc}_{\rho}$ conjecture. We also will discuss a stronger version of the $\pr{pc}_\rho$ conjecture that exactly implies Conjecture \ref{conj:hard-conj}. First, we informally introduce the low-degree conjecture and the technical conditions arising in its various formalizations in the literature.

\paragraph{Polynomial Tests and the Low-Degree Conjecture.} In this section, will draw heavily from similar discussions in \cite{hopkins2017efficient} and Hopkins's thesis \cite{hopkinsThesis}. Throughout, we will consider discrete hypothesis testing problems with observations taken without loss of generality to lie in the discrete hypercube $\{-1, 1\}^N$. For example, an $n$-vertex instance of planted clique can be represented in the discrete hypercube by the above-diagonal entries of its signed adjacency matrix when $N = \binom{n}{2}$. Given a hypothesis $H_0$, the term $D$-simple statistic refers to polynomials $f : \{-1, 1\}^N \to \mathbb{R}$ of degree at most $D$ in the coordinates of $\{-1, 1\}^N$ that are calibrated and normalized so that $\E_{H_0}f(X)=0$ and $\E_{H_0}f(X)^2=1$.

For a broad range of hypothesis testing problems, it has been observed in the literature that $D$-simple statistics seem to capture the full power of the SOS hierarchy \cite{hopkins2017efficient,hopkinsThesis}. This trend prompted a further conjecture that $D$-simple statistics often capture the full power of efficient algorithms, leading more concretely to the \textit{low-degree conjecture} which is stated informally below. This conjecture has been used to gather evidence of hardness for a number of natural detection problems and has generally emerged as a convenient tool to predict statistical-computational gaps \cite{hopkins2017efficient, hopkinsThesis, kunisky2019notes, bandeira2019computational}. Variants of this low-degree conjecture have appeared as Hypothesis 2.1.5 and Conjecture 2.2.4 in \cite{hopkinsThesis} and Conjectures 1.16 and 4.6 in \cite{kunisky2019notes}.

\begin{conjecture}[Informal -- Hypothesis 2.1.5 in \cite{hopkinsThesis}] \label{c:lowDeg}
For a broad class of hypothesis testing problems $H_0$ versus $H_1$, there is a test running in time $N^{\Ot(D)}$ with Type I$+$II error tending to zero if and only if there is a successful $D$-simple statistic i.e. a polynomial $f$ of degree at most $D$ such that $\E_{H_0}f(X)=0$ and $\E_{H_0}f(X)^2=1$ yet $\E_{H_1}f(X)\to \infty$.
\end{conjecture}

Detailed discussions of the low-degree conjecture and the connections between $D$-simple statistics and other types of algorithms can be found in \cite{kunisky2019notes} and \cite{holmgren2020counterexamples}. The informality in the conjecture above is the undefined ``broad class'' of hypothesis testing problems. In \cite{hopkinsThesis}, several candidate technical conditions defining this class were proposed and subsequently have been further refined in \cite{kunisky2019notes} and \cite{holmgren2020counterexamples}. These conditions are discussed in more detail later in this section.

The utility of the low-degree conjecture in predicting statistical-computational gaps arises from the fact that the optimal $D$-simple statistic can be explicitly characterized. By the Neyman-Pearson lemma, the optimal test with respect to Type I$+$II error is the the likelihood ratio test, which declares $H_1$ if $\lr(X) = \P_{H_1}(X)/\P_{H_0}(X) > 1$ and $H_0$ otherwise, given a sample $X$. Computing the likelihood ratio is typically intractable in problems in high-dimensional statistical inference. The low-degree likelihood ratio $\lrd$ is the orthogonal projection of the likelihood ratio onto the subspace of polynomials of degree at most $D$. When $H_0$ is a product distribution on the discrete hypercube $\{-1,1\}^N$, the following theorem asserts that $\lrd$ is the optimal test of a given degree. Here, the projection is with respect to the inner product $\la f,g\ra = \E_{H_0} f(X) g(X)$, which also defines a norm $\|f\|_2^2 = \la f,f\ra$.

\begin{theorem}[Page 35 of \cite{hopkinsThesis}] The optimal $D$-simple statistic is the low-degree likelihood ratio, i.e. it holds that
$$
\max_{{f\in \bR[x]_{{\leq D}}}\atop \E_{H_0}f(X)=0} \frac{\E_{H_1} f(X)}{\sqrt{\E_{H_0} f(X)^2}} = \|\lrd - 1\|_2
$$	
\end{theorem}

Thus existence of low-degree tests for a given problem boils down to computing the norm of the low-degree likelihood ratio. When $H_0$ is the uniform distribution on $\{-1, 1\}^N$, the norm above can be re-expressed in terms of the standard Boolean Fourier basis. Let the collection of functions $\{\chi_\alpha(X) = \prod_{e\in \alpha} X_e: \alpha \subseteq [N]\}$ denote this basis, which is orthonormal over the space $\{-1,1\}^{N}$ with inner product defined above. By orthonormality, any $\chi_\alpha$ with $1\leq |\alpha|\leq D$ satisfies that
$$
\la \chi_\alpha, \lrd -1\ra = \la \chi_\alpha, \lr \ra = \E_{H_0} \chi_\alpha(X) \lr(X) = \E_{H_1}\chi_\alpha(X)
$$
and $\E_{H_0} \lrd=\E_{H_1} 1=1$ so that  $\la 1, \lrd -1\ra=0$. 
It then follows by Parseval's identity that 
\begin{equation}\label{e:energy}
	\|\lrd - 1\|_2 = \left( \sum_{1\leq|\alpha|\leq D} \big(\E_{H_1}\chi_\alpha(X)\big)^2\right)^{1/2}
\end{equation}
which is exactly the Fourier energy up to degree $D$.

\paragraph{Technical Conditions, $S_n$-Invariance and Counterexamples.} While Conjecture \ref{c:lowDeg} is believed to accurately predict the computational barriers in nearly any natural high-dimensional statistical problem including all of the problems we consider, a precise set of criteria exactly characterizing this ``broad class'' has yet to be pinned down in the literature. The following was the first formalization of the low-degree conjecture, which appeared as Conjecture 2.2.4 in \cite{hopkinsThesis}.

\begin{conjecture}[Conjecture 2.2.4 in \cite{hopkinsThesis}] \label{c:low-deg-formal}
Let $\Omega$ be a finite set or $\mathbb{R}$, and let $k$ be a fixed integer. Let $N = \binom{n}{k}$, let $\nu$ be a product distribution on $\Omega^N$ and let $\mu$ be another distribution on $\Omega^N$. Suppose that $\mu$ is $S_n$-invariant and $(\log n)^{1 + \Omega(1)}$-wise almost independent with respect to $\nu$. Then no polynomial time test distinguishes $T_{\delta} \mu$ and $\nu$ with probability $1 - o(1)$, for any $\delta > 0$. Formally, for all $\delta > 0$ and every polynomial-time test $t : \Omega^N \to \{0, 1\}$ there exists $\delta' > 0$ such that for every large enough $n$,
$$\frac{1}{2} \bP_{x \sim \nu}\left[ t(x) = 0 \right] + \frac{1}{2} \bP_{x \sim T_{\delta} \mu}\left[ t(x) = 1 \right] \le 1 - \delta'$$
\end{conjecture}

This conjecture has several key technical stipulations attempting to conservatively pin down the $\tilde{O}$ in Conjecture \ref{c:lowDeg} and a set of \textit{sufficient conditions} to be in this ``broad class''. We highlight and explain these key conditions below.
\begin{enumerate}
\item The distribution $\mu$ is required to be $S_n$-invariant. Here, a distribution $\mu$ on $\Omega^N$ is said to be $S_n$-invariant if $\bP_\mu(x) = \bP_\mu(\pi \cdot x)$ for all $\pi \in S_n$ and $x \in \Omega^N$, where $\pi$ acts on $x$ by identifying the coordinates of $x$ with the $k$-subsets of $[n]$ and permuting these coordinates according to the permutation on $k$-subsets induced by $\pi$. 
\item The $(\log n)^{1 + \Omega(1)}$-wise almost independence requirement on $\mu$ essentially enforces that polynomials of degree at most $(\log n)^{1 + \Omega(1)}$ are unable to distinguish between $\mu$ and $\nu$. More formally, a distribution $\mu$ is $D$-wise almost independent with respect to $\nu$ if every $D$-simple statistic $f$, calibrated and normalized with respect to $\nu$, satisfies that $\bE_{x \sim \mu} f(x) = O(1)$. 
\item Rather than $\mu$, the distribution the conjecture asserts is hard to distinguish from $\nu$ is the result $T_\delta \mu$ of applying the noise operator $T_{\delta}$. Here, the distribution $T_{\delta} \mu$ is defined by first sampling $x \sim \mu$, then sampling $y \sim \nu$ and replacing each $x_i$ with $y_i$ independently with probability $\delta$.
\end{enumerate}
These technical conditions are intended to conservatively rule out specific pathological examples. As mentioned in \cite{hopkinsThesis}, the purpose of $T_\delta$ is to destroy algebraic structure that may lead to efficient algorithms that cannot be implemented with low-degree polynomials. For example, if $\mu$ uniform over the solution set to a satisfiable system of equations mod $2$ and $\nu$ is the uniform distribution, it is possible to distinguish these two distributions through Gaussian elimination while the lowest $D$ for which a $D$-simple statistic does so can be as large as $D = \Omega(N)$. The noise operator $T_{\delta}$ rules out distributions with this kind of algebraic structure. The $(\log n)^{1 + \Omega(1)}$-wise requirement on the almost independence of $\mu$ and the $\tilde{O}(D)$ in Conjecture \ref{c:lowDeg} are both to account for the fact that some common polynomial time algorithms for natural hypothesis testing problems can only be implemented as degree $O(\log n)$ polynomials. For example, Section 4.2.3 of \cite{kunisky2019notes} shows that spectral methods can typically be implemented as degree $O(\log n)$ polynomials.

In \cite{hopkinsThesis}, it was mentioned that the $S_n$-invariance condition was included in Conjecture \ref{c:low-deg-formal} mainly because most canonical inference problems satisfy this property and, furthermore, that there were no existing counterexamples to the conjecture without it. Recently, \cite{holmgren2020counterexamples} gave two construction of hypothesis testing problems based on efficiently-correctable binary codes and Reed-Solomon codes. The first construction is for binary $\Omega$ and admits a polynomial-time test despite being $\Omega(n)$-wise almost independent. This shows that $T_{\delta}$ is insufficient to always rule out high-degree algebraic structure that can be used in efficient algorithms. However, this construction also is highly asymmetric and ruled out by $S_n$-invariance condition in Conjecture \ref{c:low-deg-formal}. The second construction is for $\Omega = \mathbb{R}$ and admits a polynomial-time test despite being both $\Omega(n)$-wise almost independent and $S_n$-invariant, thus falsifying Conjecture \ref{c:low-deg-formal} as stated. However, as discussed in \cite{holmgren2020counterexamples}, the conjecture can easily be remedied by replacing $T_{\delta}$ with another operator, such as the Ornstein-Uhlenbeck noise operator. In this work, only the case of binary $\Omega$ will be relevant to the $\pr{pc}_\rho$ conjecture.

\paragraph{The $\pr{pc}_\rho$ Conjecture, Technical Conditions and a Generalization.} The $\pr{pc}_\rho$ hypothesis testing problems and their planted dense subgraph generalizations $\pr{pds}_\rho$ that we consider in this work can be shown to satisfy a wide range of properties sufficient to rule out known counterexamples to the low-degree conjecture. In particular, these problems almost satisfy all three conservative conditions proposed in \cite{hopkinsThesis}, instead satisfying a milder requirement for sufficient symmetry than full $S_n$-invariance.
\begin{enumerate}
\item By definition, a general instance of $\pr{pc}_\rho$ with an arbitrary $\rho$ is only invariant to permutations $\pi \in S_n$ that $\rho$ is also invariant to. However, each of the specific hardness assumptions we use in our reductions corresponds to a $\rho$ with a large amount of symmetry and that is invariant to large subgroups of $S_n$. For example, $k\pr{-pc}$ and $k\pr{-pds}$ are invariant to permutations within each part $E_i$, each of which has size $n/k = \omega(\sqrt{n})$. This symmetry seems sufficient to break the error-correcting code approach used to construct counterexamples to the low-degree conjecture in \cite{holmgren2020counterexamples}.
\item As will be shown subsequently in this section, the conditions in the $\pr{pc}_\rho$ conjecture imply that a $\pr{pc}_\rho$ instance be $(\log n)^{1 + \Omega(1)}$-wise almost independent for it to be conjectured to be hard.
\item While $\pr{pc}_\rho$ is not of the form $T_\delta \mu$, its generalization $\pr{pds}_\rho$ at any pair of constant edge densities $0 < q < p < 1$ always is. All of our reductions also apply to input instances of $\pr{pds}_\rho$ and thus a $\pr{pds}_\rho$ variant of the $\pr{pc}_\rho$ conjecture is sufficient to deduce our computational lower bounds. That said, we do not expect that the computational complexity of $\pr{pc}_\rho$ and $\pr{pds}_\rho$ to be different as long as $p$ and $q$ are constant. 
\end{enumerate}
As mentioned in Section \ref{sec:1-PC}, while we restrict our formal statement of the $\pr{pc}_\rho$ conjecture to the specific hardness assumptions we need for our reductions, we believe it should hold generally for $\rho$ with sufficient symmetry. A candidate condition is that $\rho$ is invariant to a subgroup $H \subseteq S_n$ of permutations such that, for each index $i \in [n]$, there are at least $n^{\Omega(n)}$ permutations $\pi \in H$ with $\pi(i) \neq i$. This ensures that $\rho$ has a large number of nontrivial symmetries that are not just permuting coordinates known not to lie in the clique.

We also remark that there are many examples of hypothesis testing problems where the three conditions in \cite{hopkinsThesis} are violated but low-degree polynomials still seem to accurately predict the performance of the best known efficient algorithms. As mentioned in \cite{holmgren2020counterexamples}, the spiked Wishart model does not quite satisfy $S_n$-invariance but still low-degree predictions are conjecturally accurate. Ordinary $\pr{pc}$ is not of the form $T_\delta \mu$ and the low-degree conjecture accurately predicts the $\pr{pc}$ conjecture, which is widely believed to be true.

\paragraph{The Degree Requirement and a Stronger $\pr{pc}_\rho$ Conjecture.} Furthermore, the degree requirement for the almost independence condition of Conjecture \ref{c:low-deg-formal} is often not exactly necessary. It is discussed in Section 4.2.5 of \cite{kunisky2019notes} that, for sufficiently nice distributions $H_0$ and $H_1$, low-degree predictions are often still accurate when the almost independence condition is relaxed to only be $\omega(1)$-wise for any $\omega(1)$ function of $n$. This yields the following stronger variant of the $\pr{pc}_\rho$ conjecture.

\begin{conjecture}[Informal -- Stronger $\pr{pc}_\rho$ Conjecture] \label{conj:inf-strong-slpc}
For sufficiently symmetric $\rho$, there is no polynomial time algorithm solving $\pr{pc}_\rho(n, k, 1/2)$ if there is some function $w(n) = \omega_n(1)$ such that the tail bounds on $p_\rho(s)$ in Conjecture \ref{conj:sl-conj} are only guaranteed to hold for all $d \le w(n)$.
\end{conjecture}

We conjecture that the $\rho$ in Conjecture \ref{conj:hard-conj} are symmetric enough for this conjecture to hold. A nearly identical argument to that in Theorem \ref{thm:verify} can be used to show that this stronger $\pr{pc}_\rho$ conjecture implies the exact boundaries in Conjecture \ref{conj:hard-conj}, without the small polynomial error factors of $O(n^\epsilon)$ and $O(m^\epsilon)$.

We now make several notes on the degree requirement in the $\pr{pc}_\rho$ conjecture, as stated in Conjecture \ref{conj:sl-conj}. As will be shown later in this section, the tail bounds on $p_{\rho}(s)$ for a particular $d$ directly imply the $d$-wise almost independence of $\pr{pc}_\rho$. Now note that for any $\rho$ and $k \gg \log n$, there is always a $d$-simple statistic solving $\pr{pc}_\rho$ with $d = O((\log n)^2)$. Specifically, $\mG(n, 1/2)$ has its largest clique of size less than $(2 + \epsilon) \log_2 n$ with probability $1 - o_n(1)$ and any instance of $H_1$ of $\pr{pc}_\rho$ with $k \gg \log n$ has $n^{\omega(1)}$ cliques of size $\lceil 3 \log_2 n \rceil$. Furthermore, the number of cliques of this size can be expressed as a degree $O((\log n)^2)$ polynomial in the edge indicators of a graph. Similarly, the largest clique in an $s$-uniform Erd\H{o}s-R\'{e}nyi hypergraph is in general of size $O((\log n)^{1/(s - 1)})$ and a simple clique-counting test distinguishing this from the planted clique hypergraph distribution can be expressed as an $O((\log n)^{s/(s - 1)})$ degree polynomial. This shows that for all $\rho$, the problem $\pr{pc}_\rho$ is not $O((\log n)^2)$-wise almost independent. Furthermore, for any $\delta > 0$, there is some $\rho$ corresponding to a hypergraph variant of $\pr{pc}$ such that $\pr{pc}_\rho$ is not $O((\log n)^{1+ \delta})$-wise almost independent. Thus the tail bounds in Conjecture \ref{conj:sl-conj} never hold for $\delta \ge 1$ and, for any $\delta' > 0$, there is some $\rho$ requiring $\delta \le \delta'$ for these tail bounds to be true.

Finally, we remark that there are highly asymmetric examples of $\rho$ for which Conjecture \ref{conj:inf-strong-slpc} is not true. Suppose that $n$ is even, let $c > 0$ be an arbitrarily large integer and let $S_1, S_2, \dots, S_{n^c} \subseteq [n/2]$ be a known family of subsets of size $\lceil 3 \log_2 n \rceil$. Now let $\rho$ be sampled by taking the union of an $S_i$ chosen uniformly at random and a size $k - \lceil 3 \log_2 n \rceil$ subset of $\{n/2 + 1, n/2 + 2, \dots, n\}$ chosen uniformly at random. The resulting $\pr{pc}_\rho$ problem can be solved in polynomial time by exhaustively searching for the subset $S_i$. However, this $\rho$ only violates the tail bounds on $p_\rho$ in Conjecture \ref{conj:sl-conj} for $d = \Omega_n(\log n/\log \log n)$. If $S_1, S_2, \dots, S_{n^c}$ are sufficiently pseudorandom, then the structure of this $\rho$ only appears in the tails of $p_\rho(s)$ when $s \ge \lceil 3 \log_2 n \rceil$. In particular, the probability that $s \ge \lceil 3 \log_2 n \rceil$ under $p_\rho$ is at least the chance that two independent samples from $\rho$ choose the same $S_i$, which occurs with probability $n^{-c}$. It can be verified the the tail bound of $p_0 \cdot s^{-2d-4}$ in Conjecture \ref{conj:sl-conj} only excludes this possibility when $d = \Omega_n(\log n/\log \log n)$. We remark though that this $\rho$ is highly asymmetric and any mild symmetry assumption that would effectively cause the number of $S_i$ to be super-polynomial would break this example.

\paragraph{The Low-Degree Conjecture and $\pr{pc}_\rho$.} We now will characterize the power of the optimal $D$-simple statistics for $\pr{pc}_\rho$. The following proposition establishes an explicit formula for $\lrd$ in $\pr{pc}_\rho$, which will be shown in the subsequent theorem to naturally yield the PMF decay condition in the $\pr{pc}_\rho$ conjecture.

\begin{proposition} \label{prop:sl-ld}
Let $\lrd$ be the low-degree likelihood ratio for the hypothesis testing problem $\pr{pc}_\rho(n, k, 1/2)$ between $\mG(n, 1/2)$ and $\mG_\rho(n, k, 1/2)$. For any $D \ge 1$, it follows that
$$\|\lrd - 1\|_2^2 = \bE_{S, S' \sim \rho^{\otimes 2}} \left[ \# \textnormal{ of nonempty edge subsets of } S \cap S' \textnormal{ of size at most } D \right]$$
\end{proposition}

\begin{proof}
In the notation above, let $N= \binom{n}{2}$ and identify $X \in \{-1, 1\}^N$ with the space of signed adjacency matrices $X$ of $n$-vertex graphs. Let $P_S$ be the distribution on graphs in this space induced by $\pr{pc}(n, k, 1/2)$ conditioned on the clique being planted on the vertices in the subset $S$ i.e. such that $X_{ij}=1$ if $i\in S$ and $j\in S$ and otherwise $X_{ij}=\pm 1$ with probability half each. Now let $\alpha\subseteq \mathcal{E}_0$ be a subset of possible edges. The set of functions $\{\chi_\alpha(X) = \prod_{e\in \alpha} X_e: \alpha \subseteq \mathcal{E}_0\}$ comprises the standard Fourier basis on $\{-1, 1\}^{\mathcal{E}_0}$. For each fixed clique $S$, because $\E_{P_S} X_e=0$ if $e\notin {S\choose 2}$ and non-clique edges are independent, we see that
$$\E_{P_S} [\chi_\alpha (X) ]= \mathbf{1} \{V(\alpha)\subseteq S\}$$
We therefore have that
$$\bE_{H_1} [\chi_\alpha (X) ] = \bE_{S \sim \rho} \E_{P_S} [\chi_\alpha (X) ] = \bE_{S \sim \rho} \left[ \mathbf{1} \{V(\alpha)\subseteq S\} \right] = \bP_{\rho}\left[ V(\alpha) \subseteq S \right]$$
Now suppose that $S'$ is drawn from $\rho$ independently of $S$. It now follows that
\begin{align*}
\bE_{H_1} [\chi_\alpha (X) ]^2 &= \bE_{S \sim \rho} \left[ \mathbf{1} \{V(\alpha)\subseteq S\} \right]^2 \\
&= \bE_{S \sim \rho} \left[ \mathbf{1} \{V(\alpha)\subseteq S\} \right] \cdot \bE_{S' \sim \rho} \left[ \mathbf{1} \{V(\alpha)\subseteq S' \} \right] \\
&= \bE_{S, S' \sim \rho^{\otimes 2}} \left[ \mathbf{1} \{V(\alpha)\subseteq S\} \cdot \mathbf{1} \{V(\alpha)\subseteq S'\} \right] \\
&= \bE_{S, S' \sim \rho^{\otimes 2}} \left[ \mathbf{1} \left\{V(\alpha)\subseteq S \cap S' \right\} \right]
\end{align*}
From Equation (\ref{e:energy}), we therefore have that
$$\|\lrd - 1\|_2^2 = \sum_{1\leq|\alpha|\leq D} \E_{H_1}\left[\chi_\alpha(X)\right]^2 = \bE_{S, S' \sim \rho^{\otimes 2}} \left[ \sum_{1\leq|\alpha|\leq D} \mathbf{1} \left\{V(\alpha)\subseteq S \cap S' \right\} \right]$$
Now observe that the sum
$$\sum_{1\leq|\alpha|\leq D} \mathbf{1} \left\{V(\alpha)\subseteq S \cap S' \right\}$$
counting the number of nonempty edge subsets of $S \cap S'$ of size at most $D$.
\end{proof}

This proposition now allows us to show the main result of this section, which is that the condition in the $\pr{pc}_\rho$ conjecture is enough to show the failure of low-degree polynomials for $\pr{pc}_\rho$. Combining the next theorem with Conjecture~\ref{c:lowDeg} would suggest that whenever the PMF decay condition of the $\pr{pc}_\rho$ condition holds, there is no polynomial time algorithm solving $\pr{pc}_\rho(n, k, 1/2)$.

\begin{theorem}[$\pr{pc}_\rho$ Implies Failure of Low-Degree] \label{thm:sl-ld}
Suppose that $\rho$ satisfies that for any parameter $d = O_n(\log n)$, there is some $p_0 = o_n(1)$ such that $p_{\rho}(s)$ satisfies the tail bounds
$$p_{\rho}(s) \le p_0 \cdot \left\{ \begin{array}{ll} 2^{-s^2} &\textnormal{if } 1 \le s^2 < d \\ s^{-2d-4} &\textnormal{if } s^2 \ge d \end{array} \right.$$
Let $\lrd$ be the low-degree likelihood ratio for the hypothesis testing problem $\pr{pc}_\rho(n, k, 1/2)$. Then it also follows that for any parameter $D = O_n(\log n)$, we have
$$\|\lrd - 1\|_2 = o_n(1)$$
\end{theorem}

\begin{proof}
First observe that the number of nonempty edge subsets of $S \cap S'$ of size at most $D$ can be expressed explicitly as
$$f_D(s) = \sum_{\ell = 1}^D \binom{s(s - 1)/2}{\ell}$$
if $s = |S \cap S'|$. Furthermore, we can crudely upper bound $f_D$ in two separate ways. Note that the number of nonempty edge subsets of $S \cap S'$ is exactly $2^{\binom{s}{2}} - 1$ if $s = |S \cap S'|$. Therefore we have that $f_D(s) \le 2^{\binom{s}{2}}$. Furthermore using the upper bound that $\binom{x}{\ell} \le x^\ell$, we have that if $s \ge 3$ then
$$f_D(s) = \sum_{\ell = 1}^D \binom{s(s - 1)/2}{\ell} \le \sum_{\ell = 1}^D \left( \frac{s(s - 1)}{2} \right)^\ell \le \frac{\left( \frac{s(s - 1)}{2} \right)^{D + 1} - 1}{\left( \frac{s(s - 1)}{2} \right) - 1} \le s^{2(D + 1)}$$
Combining these two crude upper bounds, we have that $f_D(s) \le \min\left\{ 2^{\binom{s}{2}}, s^{2(D + 1)} \right\}$. Also note that $f_D(0) = f_D(1) = 0$. Combining this with the given bounds on $p_{\rho}(s)$, we have that
\allowdisplaybreaks
\begin{align*}
\|\lrd - 1\|_2^2 &= \bE_{S, S' \sim \rho^{\otimes 2}} \left[ f_D(|S \cap S'|) \right] \\
&= \sum_{s = 2}^k p_{\rho}(s) \cdot f_D(s) \\
&\le p_0 \cdot \sum_{1 \le s^2 < D} 2^{-s^2} \cdot f_D(s) + p_0 \cdot \sum_{D \le s^2 \le k^2} s^{-2d-4} \cdot f_D(s) \\
&\le p_0 \cdot \sum_{1 \le s^2 < D} 2^{-s^2} \cdot 2^{\binom{s}{2}} + p_0 \cdot \sum_{D \le s^2 \le k^2} s^{-2D-4} \cdot s^{2(D + 1)} \\
&= p_0 \cdot \sum_{s = 1}^\infty 2^{-\binom{s+1}{2}} + p_0 \cdot \sum_{s = 1}^\infty s^{-2} = O_n(p_0)
\end{align*}
which completes the proof of the theorem.
\end{proof}

\subsection{Statistical Query Algorithms and the $\pr{pc}_\rho$ Conjecture}
\label{subsec:2-sq}

In this section, we verify that the lower bounds shown by \cite{feldman2013statistical} for \textsc{pc} for a generalization of statistical query algorithms hold essentially unchanged for SQ variants of $k\pr{-pc}, k\pr{-bpc}$ and $\pr{bpc}$. We remark at the end of this section why the statistical query model seems ill-suited to characterizing the computational barriers in problems that are tensor or hypergraph problems such as $k\pr{-hpc}$. Since it was shown in Section \ref{subsec:2-sl-verifying} that there are specific $\rho$ in $\pr{pc}_\rho$ corresponding to $k\pr{-hpc}$, it similarly follows that the SQ model seems ill-suited to characterizing the barriers $\pr{pc}_\rho$ for general $\rho$. Throughout this section, we focus on $k\pr{-pc}$, as lower bounds in the statistical query model for $k\pr{-bpc}$ and $\pr{bpc}$ will follow from nearly identical arguments.

\paragraph{Distributional Problems and SQ Dimension.} The Statistical Algorithm framework of \cite{feldman2013statistical} applies to distributional problems, where the input is a sequence of i.i.d. observations from a distribution $D$. In order to obtain lower bounds in the statistical query model supporting Conjecture \ref{conj:hard-conj}, we need to define a distributional analogue of \kpc. As in \cite{feldman2013statistical}, a natural distributional version can be obtained by considering a bipartite version of $k\pr{-pc}$, which we define as follows.

\begin{definition}[Distributional Formulation of $k\pr{-pc}$]
Let $k$ divide $n$ and fix a known partition $E$ of $[n]$ into $k$ parts $E_1, E_2, \dots, E_k$ with $|E_i|=n/k$. Let $S\subseteq [n]$ be a subset of indices with $|S\cap E_i|=1$ for each $i\in [k]$. The distribution $D_S$ over $\{0,1\}^n$ produces with probability $1-k/n$ a uniform point $X\sim \mathrm{Unif}(\{0,1\}^n)$ and with probability $k/n$ a point $X$ with $X_i=1$ for all $i\in S$ and $X_{S^c}\sim \mathrm{Unif}(\{0,1\})^{n-k}$. The distributional bipartite $k$-\textsc{pc} problem is to find the subset $S$ given some number of independent samples $m$ from $D_S$. 
\end{definition} 

In other words, the distribution $k\pr{-pc}$ problem is $k\pr{-bpc}$ with $n$ left and $n$ right vertices, a randomly-sized right part of the planted biclique and no $k$-partite structure on the right vertex set. We remark that many of our reductions, such as our reductions to $\pr{rsme}$, $\pr{neg-spca}$, $\pr{mslr}$ and $\pr{rslsr}$, only need the $k$-partite structure along one vertex set of $k\pr{-pc}$ or $k\pr{-bpc}$. This distributional formulation of $k\pr{-pc}$ is thus a valid starting point for these reductions. 

We now formally introduce the Statistical Algorithm framework of \cite{feldman2013statistical} and SQ dimension. Let $\cX=\{0,1\}^n$ denote the space of configurations and let $\cD$ be a set of distributions over $\cX$. Let $\cF$ be a set of solutions and $\cZ:\cD\to 2^\cF$ be a map taking each distribution $D\in \cD$ to a subset of solutions $\cZ(D)\subseteq \cF$ that are defined to be valid solutions for $D$. In our setting, $\cF$ corresponds to clique positions $S$ respecting the partition $E$. Furthermore, since each clique position is in one-to-one correspondence with distributions, there is a single clique $\cZ(D)$ corresponding to each distribution $D$. For $m>0$, the \emph{distributional search problem} $\cZ$ over $\cD$ and $\cF$ using $m$ samples is to find a valid solution $f\in \cZ(D)$ given access to $m$ random samples from an unknown $D\in \cD$.

Classes of algorithms in the framework of \cite{feldman2013statistical} are defined in terms of access to oracles. The most basic oracle is an unbiased oracle, which evaluates a simple function on a single sample as follows.

\begin{definition}[Unbiased Oracle]
Let $D$ be the true unknown distribution. A query to the oracle consists of any function $h:\cX\to \{0,1\}$, and the oracle then takes an independent random sample $X\sim D$ and returns $h(X)$.
\end{definition}

Algorithms with access to an unbiased oracle are referred to as \textit{unbiased statistical algorithms}. Since these algorithms access the sampled data only through the oracle, it is possible to prove \emph{unconditional} lower bounds using information-theoretic methods. Another oracle is the $VSTAT$, defined below, which is similar but also allowed to make an adversarial perturbation of the function evaluation. It is shown in \cite{feldman2013statistical} via a simulation argument that the two oracles are approximately equivalent. 

\begin{definition}[$VSTAT$ Oracle]
Let $D$ be the true distribution and $t>0$ a sample size parameter. A query to the $VSTAT(t)$ oracle again consists of any function $h:\cX\to [0,1]$, and the oracle returns an arbitrary value $v\in[\E_{ D}h(X)-\tau, \E_{ D}h(X)+\tau]$, where $\tau = \max\{1/t,\sqrt{\E_{ D}h(X)(1-\E_{ D}h(X))/t}\}$.
\end{definition}

We borrow some definitions from \cite{feldman2013statistical}. Given a distribution $D$, we define the inner product $\la f,g\ra_D = \E_{X\sim D}f(X) g(X)$ and the corresponding norm $\|f\|_D = \sqrt{\la f,f\ra_D}$. Given two distributions $D_1$ and $D_2$ both absolutely continuous with respect to $D$, their pairwise correlation is defined to be 
$$
\chi_D(D_1,D_2) = \Big|\Big\la \frac{D_1}D-1,\frac{D_2}D-1\Big\ra_D \Big|=|\la \Dh_1, \Dh_2\ra_D|\,.
$$ 
where $ \Dh_1 = \frac{D_1}D-1$. The \emph{average correlation} $\rho(\cD, D)$ of a set of distributions $\cD$ relative to distribution $D$ is then given by
$$
\rho(\cD, D) = \frac1{|\cD|^2} \sum_{D_1,D_2\in \cD}\chi_D(D_1,D_2) = \frac1{|\cD|^2} \sum_{D_1,D_2\in \cD}\Big|\Big\la \frac{D_1}D-1,\frac{D_2}D-1\Big\ra_D \Big|\,.
$$
Given these definitions, we can now introduce the key quantity from \cite{feldman2013statistical}, statistical dimension, which is defined in terms of average correlation. 

\begin{definition}[Statistical dimension]\label{d:statDimProblem}
Fix $\gamma>0,\eta>0$, and search problem $\cZ$ over set of solutions $\cF$ and class of distributions $\cD$ over $\cX$. 
We consider pairs $(D,\cD_D)$ consisting of a ``reference distribution" $D$ over $\cX$ and a finite set of distributions $\cD_D\subseteq \cD$ with the following property: for any solution $f\in \cF$, the set $\cD_f = \cD_D\setminus \cZ\inv (f)$ has size at least $(1-\eta)\cdot |\cD_D|$.
Let $\ell(D,\cD_D)$ be the largest integer $\ell$ so that for any subset $\cD'\subseteq \cD_f$ with $|\cD'|\geq |\cD_f|/\ell$, the average correlation is $|\rho(\cD',D)|<\gamma$ (if there is no such $\ell$ one can take $\ell=0$). The \emph{statistical dimension} with average correlation $\gamma$ and solution set bound $\eta$ is defined to be the largest $\ell(D,\cD_D)$ for valid pairs $(D,\cD_D)$ as described, and is  denoted by $\mathrm{SDA}(\cZ,\gamma,\eta)$.
\end{definition}

In \cite{feldman2013statistical}, it is shown that statistical dimension immediately yields a lower bound on the number of queries to an unbiased oracle or a $VSTAT$ oracle needed to solve a given distributional search problem.

\begin{theorem}[Theorems 2.7 and 3.17 of \cite{feldman2013statistical}]\label{t:sampleBound}
  Let $\cX$ be a domain and $\cZ$ a search problem over a set of solutions $\cF$ and a class of distributions $\cD$ over $\cX$. For $\gamma>0$ and $\eta\in (0,1)$, let $\ell = \mathrm{SDA}(\cZ,\gamma,\eta)$. Any (possibly randomized) statistical query algorithm that solves $\cZ$ with probability $\delta>\eta$ requires at least $\ell$ calls to the $VSTAT(1/(3\gamma))$ oracle to solve $\cZ$. 
  
Moreover, any statistical query algorithm requires at least  $m$ calls to the Unbiased Oracle for $m = \min\left\{ \frac{\ell(\delta- \eta)}{2(1-\eta)},\frac{(\delta-\eta)^2}{12\gamma}\right\}$. In particular, if $\eta \leq 1/6$, then any algorithm with success probability at least $2/3$ requires at least $\min\{ \ell/4,1/48\gamma\}$ samples from the Unbiased Oracle.
\end{theorem}

We remark that the number of queries to an oracle is a lower bound on the runtime of the statistical algorithm in question. Furthermore, the number of ``samples'' $m$ corresponding to a $VSTAT(t)$ oracle is $t$, as this is the number needed to approximately obtain the confidence interval of width $2\tau$ in the definition of the $VSTAT$ oracle above.  

\paragraph{SQ Lower Bounds for Distributional $k\pr{-pc}$.} We now will use the theorem above to deduce SQ lower bounds for distributional $k\pr{-pc}$. Let $\cS$ be the set of all $k$-subsets of $[n]$ respecting the partition $E$ i.e. $\cS = \{S:|S|=k\text{ and } |S\cap E_i|=1\text{ for }i\in [k]\}$. Note that $|\cS| = (n/k)^k$. We henceforth use $D$ to denote the uniform distribution on $\{0,1\}^n$. The following lemma is as in \cite{feldman2013statistical}, except that we further restrict $S$ and $T$ to be in $\cS$ rather than arbitrary size $k$ subsets of $[n]$, which does not change the bound.

\begin{lemma}[Lemma 5.1 in \cite{feldman2013statistical}]\label{l:avgCorr}
For $S,T\in \cS$, $\chi_D(D_S,D_T) = |\la \Dh_S, \Dh_T\ra_D|\leq 2^{|S\cap T|} k^2 / n^2$.  	
\end{lemma}

The following lemma is crucial to deriving the SQ dimension of distributional $k\pr{-pc}$ and is similar to Lemma 5.2 in \cite{feldman2013statistical}. Its proof is deferred to Appendix \ref{sec:appendix-4}.

\begin{lemma}[Modification of Lemma 5.2 in \cite{feldman2013statistical}]\label{l:avgCorrLargeSets}
	Let $\delta \geq 1/\log n$ and $k\leq n^{1/2 - \delta}$. For any integer $\ell \leq k$, $S\in \cS$, and set $A\subseteq \cS$ with $|A|\geq 2|\cS|/ n^{2\ell \delta}$, 
	$$
	\frac1{|A|} \sum_{T\in A} \big| \la \Dh_S, \Dh_T\ra_D \big| \leq 2^{\ell + 3}\frac{k^2}{n^2}\,.
	$$
\end{lemma}

This lemma now implies the following SQ dimension lower bound for distributional $k\pr{-pc}$.

\begin{theorem}[Analogue of Theorem 5.3 of \cite{feldman2013statistical}]\label{t:SQdim}
	For $\delta\geq 1/\log n$ and $k\leq n^{1/2-\delta}$ let $\cZ$ denote the distributional bipartite {\kpc } problem. If $\ell\leq k$ then $SDA(\cZ, 2^{\ell+3} k^2/n^2,\big(\frac nk\big) ^{-k} ) \geq n^{2\ell \delta}/8$. 
\end{theorem}
\begin{proof}
For each clique position $S$ let $\cD_S = \cD\setminus\{D_S\}$. Then $|\cD_S| = \big(\frac nk\big) ^k -1=\big(1-\big(\frac nk\big) ^{-k}\big)|\cD|$. Now for any $\cD'$ with $|\cD'|\geq 2|\cS|/ n^{2\ell \delta}$ we can apply Lemma~\ref{l:avgCorrLargeSets} to conclude that $\rho(\cD',D)\leq 2^{\ell + 3}k^2/n^2$. By Definition~\ref{d:statDimProblem} of statistical dimension this implies the bound stated in the theorem. 
\end{proof}

Applying Theorem~\ref{t:sampleBound} to this statistical dimension lower bound yields the following hardness for statistical query algorithms.

\begin{corollary}[SQ Lower Bound for Recovery in Distributional \kpc]
For any constant $\delta>0$ and $k\leq n^{1/2-\delta}$, any SQ algorithm that solves the distributional bipartite {\kpc } problem requires $\Omega(n^2/k^2\log n)=\tilde \Omega(n^{1+2\delta})$ queries to the Unbiased Oracle.
\end{corollary}

This is to be interpreted as impossible, as there are only $n$ right vertices vertices available in the actual bipartite graph. Because all the quantities in Theorem~\ref{t:SQdim} are the same as in \cite{feldman2013statistical} up to constants, the same logic as used there allows to deduce a statement regarding the hypothesis testing version, stated there as Theorems 2.9 and 2.10.

\begin{corollary}[SQ Lower Bound for Decision Variant of Distributional \kpc]
For any constant $\delta>0$, suppose $k\leq n^{1/2-\delta}$. Let $D=\mathrm{Unif}(\{0,1\}^n)$ and let $\cD$ be the set of all planted bipartite {\kpc } distributions (one for each clique position). Any SQ algorithm that solves the hypothesis testing problem between $\cD$ and $D$ with probability better than $2/3$ requires $\Omega(n^2/k^2)$ queries to the Unbiased Oracle.

A similar statement holds for VSTAT. There is a $t = n^{\Omega(\log n)}$ such that any randomized SQ algorithm that solves the hypothesis testing problem between $\cD$ and $D$ with probability better than $2/3$ requires at least $t$ queries to $VSTAT(n^{2-\delta}/k^2)$. 
\end{corollary}

We conclude this section by outlining how to extend these lower bounds to distributional versions of $k\pr{-bpc}$ and $\pr{bpc}$ and why the statistical query model is not suitable to deduce hardness of problems that are implicitly tensor or hypergraph problems such as $k\pr{-hpc}$.

\paragraph{Extending these SQ Lower Bounds.} Extending to the bipartite case is straightforward and follows by replacing the probability of including each right vertex from $k/n$ to $k_m/m$ where $k_m = O(m^{1/2 - \delta})$. This causes the upper bound in Lemma \ref{l:avgCorr} to become $\chi_D(D_S,D_T) = |\la \Dh_S, \Dh_T\ra_D|\leq 2^{|S\cap T|} k_m^2 / m^2$. Similarly, the upper bound in Lemma \ref{l:avgCorrLargeSets} becomes $2^{\ell + 3} k_m^2/m^2$, the relevant statistical dimension becomes $SDA(\cZ, 2^{\ell+3} k_m^2/n_m^2,\big(\frac nk\big) ^{-k} ) \geq n^{2\ell \delta}/8$ and the query lower bound in the final corollary becomes $\Omega(m^2/k_m^2 \log n) = \tilde{\Omega}(m^{1 + 2\delta})$ which yields the desired lower bound for $k\pr{-bpds}$. The lower bound for $\pr{bpds}$ follows by the same extension to the ordinary $\pr{pc}$ lower bound in \cite{feldman2013statistical}.

\paragraph{Hypergraph \pr{pc} and SQ Lower Bounds.} A key component of formulating SQ lower bounds is devising a distributional version of the problem with analogous limits in the SQ model. While there was a natural bipartite extension for $\pr{pc}$, for hypergraph \pr{pc}, such an extension does not seem to exist. Treating slices as individual samples yields a problem with statistical query algorithms that can detect a planted clique outside of polynomial time. Consider the function that given a slice, searches for a clique of size $k$ in the induced $(s - 1)$-uniform hypergraph on the neighbors of the vertex corresponding to the slice, outputting $1$ if such a clique is found. Without a planted clique, the probability a slice contains such a clique is exponentially small, while it is $k/n$ if there is a planted clique. An alternative is to consider individual entries as samples, but this discards the hypergraph structure of the problem entirely.

\section{Robustness, Negative Sparse PCA and Supervised Problems}
\label{sec:3-robust-and-supervised}

In this section, we apply reductions in Part \ref{part:reductions} to deduce computational lower bounds for robust sparse mean estimation, negative sparse PCA, mixtures of SLRs and robust SLR that follow from specific instantiations of the $\pr{pc}_\rho$ conjecture. Specifically, we apply the reduction $k\pr{-bpds-to-isgm}$ to deduce a lower bound for $\pr{rsme}$, the reduction $\pr{bpds-to-neg-spca}$ to deduce a lower bound for $\pr{neg-spca}$ and the reduction $k\pr{-bpds-to-mslr}$ to deduce lower bounds for $\pr{mslr}$, $\pr{uslr}$ and $\pr{rslr}$. This section is primarily devoted to summarizing the implications of these reductions and making explicit how their input parameters need to be set to deduce our lower bounds. The implications of these lower bounds and the relation between them and algorithms was previously discussed in Section \ref{sec:1-problems}. In cases where the discussion in Section \ref{sec:1-problems} was not exhaustive, such as the details of starting with different hardness assumptions, the number theoretic condition $\pr{(t)}$ or the adversary implied by our reductions for $\pr{rslr}$, we include omitted details in this section.

All lower bounds that will be shown in this section are \emph{computational lower bounds} in the sense introduced in the beginning of Section \ref{sec:1-problems}. To deduce our computational lower bounds from reductions, it suffices to verify the three criteria in Condition \ref{cond:lb}. We remark that this section is technical due to the number-theoretic constraints imposed by the prime number $r$ in our reductions. However, these technical details are tangential to the primary focus of the paper, which is reduction techniques.



\subsection{Robust Sparse Mean Estimation}
\label{subsec:3-rsme}

We first observe that the instances of $\pr{isgm}$ output by the reduction $k\pr{-bpds-to-isgm}$ are instances of $\pr{rsme}$ in Huber's contamination model. Let $r$ be a prime number and $\epsilon \ge 1/r$. It then follows that a sample from $\pr{isgm}_D(n, k, d, \mu, 1/r)$ is of the form
$$\pr{mix}_{\epsilon}\left( \mN(\mu \cdot \mathbf{1}_S, I_d), \mD_O \right)^{\otimes n} \quad \text{where} \quad \mD_O = \pr{mix}_{\epsilon^{-1} r^{-1}} \left( \mN(\mu \cdot \mathbf{1}_S, I_d), \mN(\mu' \cdot \mathbf{1}_S, I_d) \right)$$
for some possibly random $S$ with $|S| = k$ and where $(1 - r^{-1}) \mu + r^{-1} \cdot \mu' = 0$. Note that this is a distribution in the composite hypothesis $H_1$ of $\pr{rsme}(n, k, d, \tau, \epsilon)$ in Huber's contamination model with outlier distribution $\mD_O$ and where $\tau = \| \mu \cdot \mathbf{1}_S \|_2 = \mu \sqrt{k}$. This observation and the discussion in Section \ref{subsec:2-tvreductions} yields that it suffices to exhibit a reduction to $\pr{isgm}$ to show the lower bound for $\pr{rsme}$ in Theorem \ref{thm:rsme-lb}.

We now discuss the condition $\pr{(t)}$ and the number-theoretic constraint arising from applying Theorem \ref{thm:isgmreduction} to prove Theorem \ref{thm:rsme-lb}. As mentioned in Section \ref{subsec:1-problems-rsme}, while this condition does not restrict our computational lower bound for $\pr{rsme}$ in the main regime of interest where $\epsilon^{-1} = n^{o(1)}$, it also can be removed using the design matrices $R_{n, \epsilon}$ in place of $K_{r, t}$. Despite this, we introduce the condition $\pr{(t)}$ in this section as it will be a necessary condition in subsequent lower bounds in Part \ref{part:lower-bounds}.

As discussed in Section \ref{sec:2-supervised}, the prime power $r^t$ in $k\pr{-bpds-to-isgm}$ is intended to be a fairly close approximation to each of $k_n, \sqrt{n}$ and $\sqrt{N}$. We will now see that in order to show tight computational lower bounds for $\pr{rsme}$, this approximation needs to be very close to asymptotically exact, leading to the technical condition $\pr{(t)}$. First note that the level of signal $\mu$ produced by the reduction $k\pr{-bpds-to-isgm}$ is 
$$\mu \le \frac{\delta}{2 \sqrt{6\log (k_nmr^t) + 2\log (p - q)^{-1}}} \cdot \frac{1}{\sqrt{r^t(r - 1)(1 + (r - 1)^{-1})}} = \tilde{\Theta}\left( r^{-(t + 1)/2} \right)$$
where $\delta = \Theta(1)$ and the estimate above holds whenever $p$ and $q$ are constants. Therefore the corresponding $\tau$ is given by $\tau = \mu \sqrt{k} = \tilde{O}(k^{1/2} r^{-(t + 1)/2})$. Furthermore, in Theorem \ref{thm:isgmreduction}, the output number of samples $N$ is constrained to satisfy that $N = o(k_n r^t)$ and $n = O(k_n r^t)$. Combining this with the fact that in order to be starting with a hard $k\pr{-bpds}$ instance, we need $k_n = o(\sqrt{n})$ to hold, it is straightforward to see that these constraints together require that $N = o(r^{2t})$. If this is close to tight with $N = \tilde{\Theta}(r^{2t})$, the computational lower bound condition on $\tau$ becomes
$$\tau = \tilde{O}\left(k^{1/2} r^{-(t + 1)/2}\right) = \tilde{\Theta}\left(k^{1/2} \epsilon^{1/2} N^{-1/4} \right)$$
where we also use the fact that $\epsilon = \Theta(1/r)$. Note that this corresponds exactly to the desired computational lower bound of $N = \tilde{o}(k^2 \epsilon^2/\tau^4)$. Furthermore, if instead $N = \tilde{\Theta}(a^{-1}r^{2t})$ for some $a = \omega(1)$, then the lower bound we show degrades to $N = \tilde{o}(k^2 \epsilon^2/a\tau^4)$, and is suboptimal by a factor $a = \omega(1)$. Thus ideally we would like the pair of parameters $(N, r)$ to be such that there infinitely many $N$ with something like $N = \tilde{\Theta}(r^{2t})$ true for some positive integer $t \in \mathbb{N}$. This leads exactly to the condition $\pr{(t)}$ below.

\begin{definition}[Condition $\pr{(t)}$]
Suppose that $(N, r)$ is a pair of parameters with $N \in \mathbb{N}$ and $r = r(N)$ is non-decreasing. The pair $(N, r)$ satisfies $\pr{(t)}$ if either $r = N^{o(1)}$ as $N \to \infty$ or if $r = \tilde{\Theta}(N^{1/t})$ where $t \in \mathbb{N}$ is a constant even integer.
\end{definition}

The key property arising from condition $\pr{(t)}$ is captured in the following lemma.

\begin{lemma}[Property of $\pr{(t)}$] \label{lem:propT}
Suppose that $(N, r)$ satisfies $\pr{(t)}$ and let $r' = r'(N)$ be any non-decreasing positive integer parameter satisfying that $r' = \tilde{\Theta}(r)$. Then there are infinitely many values of $N$ with the following property: there exists $s \in \mathbb{N}$ such that $\sqrt{N} = \tilde{\Theta}\left( (r')^s \right)$.
\end{lemma}

\begin{proof}
If $r = \tilde{\Theta}(N^{1/t})$ where $t \in \mathbb{N}$ is a constant even integer, then this property is satisfied trivially by taking $s = t/2$. Now suppose that $r = N^{o(1)}$ and note that this also implies that $r' = N^{o(1)}$. Now consider the function
$$f(N) = \frac{\log N}{2\log r'(N)}$$
Since $r' = N^{o(1)}$, it follows that $f(N) \to \infty$ as $N \to \infty$. Suppose that $N$ is sufficiently large so that $f(N) > 1$. Note that, for each $N$, either $r'(N + 1) \ge r'(N) + 1$ or $r'(N + 1) = r'(N)$. If $r'(N + 1) = r'(N)$, then $f(N + 1) > f(N)$. If $r'(N + 1) \ge r'(N) + 1$, then
$$\frac{f(N+1)}{f(N)} \le \frac{g(N)}{g(r'(N))} \quad \text{where} \quad g(x) = \frac{\log (x + 1)}{\log x}$$
Note that $g(x)$ is a decreasing function of $x$ for $x \ge 2$. Since $f(N) > 1$, it follows that $r'(N) < N$ and hence the above inequality implies that $f(N + 1) < f(N)$. Summarizing these observations, every time $f(N)$ increases it must follow that $r'(N + 1) = r'(N)$. Fix a sufficiently large positive integer $s$ and consider the first $N$ for which $f(N) \ge s$. It follows by our observation that $r'(N) = r'(N - 1)$ and furthermore that $f(N - 1) < s$. This implies that $N - 1 < r'(N)^{2s}$ and $N \ge r'(N)^{2s}$. Since $r'(N)$ is a positive integer, it then must follow that $N = r'(N)^{2s}$. Since such an $N$ exists for every sufficiently large $s$, this completes the proof of the lemma.
\end{proof}

This condition $\pr{(t)}$ will arise in a number of others problems that we map to, including robust SLR and dense stochastic block models, for a nearly identical reason. We now formally prove Theorem \ref{thm:rsme-lb}. All remaining proofs in this section will be of a similar flavor and where details are similar, we only sketch them to avoid redundancy.

\begin{reptheorem}{thm:rsme-lb} [Lower Bounds for $\pr{rsme}$]
If $k, d$ and $n$ are polynomial in each other, $k = o(\sqrt{d})$ and $\epsilon < 1/2$ is such that $(n, \epsilon^{-1})$ satisfies $\pr{(t)}$, then the $k\pr{-bpc}$ conjecture or $k\pr{-bpds}$ conjecture for constant $0 < q < p \le 1$ both imply that there is a computational lower bound for $\pr{rsme}(n, k, d, \tau, \epsilon)$ at all sample complexities $n = \tilde{o}(k^2 \epsilon^2/\tau^4)$.
\end{reptheorem}

\begin{proof}
To prove this theorem, we will to show that Theorem \ref{thm:isgmreduction} implies that $k\pr{-bpds-to-isgm}$ fills out all of the possible growth rates specified by the computational lower bound $n = \tilde{o}(k^2 \epsilon^2/\tau^4)$ and the other conditions in the theorem statement. As discussed earlier in this section, it suffices to reduce in total variation to $\pr{isgm}(n, k, d, \mu, 1/r)$ where $1/r \le \epsilon$ and $\mu = \tau/\sqrt{k}$.

Fix a constant pair of probabilities $0 < q < p \le 1$ and any sequence of parameters $(n, k, d, \tau, \epsilon)$ all of which are implicitly functions of $n$ such that $(n, \epsilon^{-1})$ satisfies $\pr{(t)}$ and $(n, k, d, \tau, \epsilon)$ satisfy the conditions
$$n \le c \cdot \frac{k^2 \epsilon^2}{\tau^4 \cdot (\log n)^{2+2c'}} \quad \text{and} \quad w k^2 \le d$$
for sufficiently large $n$, an arbitrarily slow-growing function $w = w(n) \to \infty$ at least satisfying that $w(n) = n^{o(1)}$, a sufficiently small constant $c > 0$ and a sufficiently large constant $c' > 0$. In order to fulfill the criteria in Condition \ref{cond:lb}, we now will specify:
\begin{enumerate}
\item a sequence of parameters $(M, N, k_M, k_N, p, q)$ such that the $k\pr{-bpds}$ instance with these parameters is hard according to Conjecture \ref{conj:hard-conj}; and 
\item a sequence of parameters $(n', k, d, \tau, \epsilon)$ with a subsequence that satisfies three conditions: (2.1) the parameters on the subsequence are in the regime of the desired computational lower bound for $\pr{rsme}$; (2.2) they have the same growth rate as $(n, k, d, \tau, \epsilon)$ on this subsequence; and (2.3) such that $\pr{rsme}$ with the parameters on this subsequence can be produced by the reduction $k\pr{-bpds-to-isgm}$ with input $k\pr{-bpds}(M, N, k_M, k_N, p, q)$.
\end{enumerate}
By the discussion in Section \ref{subsec:2-tvreductions}, this would be sufficient to show the desired computational lower bound. We choose these parameters as follows:
\begin{itemize}
\item let $r$ be a prime with $r \ge \epsilon^{-1}$ and $r \le 2\epsilon^{-1}$, which exists by Bertrand's postulate and can be found in $\text{poly}(\epsilon^{-1}) \le \text{poly}(n)$ time;
\item let $t$ be such that $r^t$ is the closest power of $r$ to $\sqrt{n}$, let $n' = \lfloor w^{-2} r^{2t} \rfloor$, let $k_N = \lfloor \sqrt{n'} \rfloor$ and let $N = wk_N^2 \le k_N r^t$; and
\item set $\mu = \tau/\sqrt{k}$, $k_M = k$ and $M = w k^2$.
\end{itemize}
The given inequality and parameter settings above rearrange to the following condition on $n'$:
$$n' \le w^{-2} r^{2t} = O\left( \frac{r^{2t}}{n} \cdot \frac{k^2 \epsilon^2}{\tau^4 \cdot (\log n)^{2+2c'}} \right)$$
Furthermore, the given inequality yields the constraint on $\mu$ that
$$\mu = \tau \cdot k^{-1/2} \le \frac{c^{1/4} \epsilon^{1/2}}{n^{1/4} (\log n)^{(1 + c')/2}} = \Theta \left( \frac{r^{t/2}}{n^{1/4}} \cdot \frac{1}{\sqrt{r^{t + 1} (\log n)^{1+c'}}} \right)$$
As long as $\sqrt{n} = \tilde{\Theta}(r^t)$ then: (2.1) the inequality above on $n'$ would imply that $(n', k, d, \tau, \epsilon)$ is in the desired hard regime; (2.2) $n$ and $n'$ have the same growth rate since $w = n^{o(1)}$; and (2.3) taking $c'$ large enough would imply that $\mu$ satisfies the conditions needed to apply Theorem \ref{thm:isgmreduction} to yield the desired reduction. By Lemma \ref{lem:propT}, there is an infinite subsequence of the input parameters such that $\sqrt{n} = \tilde{\Theta}(r^t)$. This verifies the three criteria in Condition \ref{cond:lb}. Following the argument in Section \ref{subsec:2-tvreductions}, Lemma \ref{lem:3a} now implies the theorem.
\end{proof}

As alluded to in Section \ref{subsec:1-problems-rsme}, replacing $K_{r, t}$ with $R_{n, \epsilon}$ in the applications of dense Bernoulli rotations in $k\pr{-bpds-to-isgm}$ removes condition $(\pr{t})$ from this lower bound. Specifically, applying $k\pr{-bpds-to-isgm}_R$ and Corollary \ref{thm:mod-isgmreduction} in place of $k\pr{-bpds-to-isgm}$ and replacing the dimension $r^t$ with $L$ in the argument above yields the lower bound shown below. Note that condition $(\pr{t})$ in Theorem \ref{thm:rsme-lb} is replaced by the looser requirement that $\epsilon = \tilde{\Omega}(n^{-1/2})$. As discussed at the end of Section \ref{subsec:3-rsme-reduction}, this requirement arises from the condition $\epsilon \gg L^{-1} \log L$ in Corollary \ref{thm:mod-isgmreduction}. We remark that the condition $\epsilon = \tilde{\Omega}(n^{-1/2})$ is implicit in $(\pr{t})$ and hence the following corollary is strictly stronger than Theorem \ref{thm:rsme-lb}.

\begin{corollary}[Lower Bounds for $\pr{rsme}$ without Condition ($\pr{t}$)] \label{cor:rsme-lb-mod}
If $k, d$ and $n$ are polynomial in each other, $k = o(\sqrt{d})$ and $\epsilon < 1/2$ is such that $\epsilon = \tilde{\Omega}(n^{-1/2})$, then the $k\pr{-bpc}$ conjecture or $k\pr{-bpds}$ conjecture for constant $0 < q < p \le 1$ both imply that there is a computational lower bound for $\pr{rsme}(n, k, d, \tau, \epsilon)$ at all sample complexities $n = \tilde{o}(k^2 \epsilon^2/\tau^4)$.
\end{corollary}

We remark that only assuming the $k\pr{-pc}$ conjecture also yields hardness for $\pr{rsme}$. In particular $k\pr{-pc}$ can be mapped to the asymmetric bipartite case by considering the bipartite subgraph with $k/2$ parts on one size and $k/2$ on the other. Showing hardness for $\pr{rsme}$ from $k\pr{-pc}$ then reduces to the hardness yielded by $k\pr{-bpc}$ with $M = N$. Examining this restricted setting in the theorem above and passing through an analogous argument yields a computational lower bound at the slightly suboptimal rate
$$n = \tilde{o}\left(k^2 \epsilon/\tau^2\right) \quad \text{as long as} \quad \tau^2 \log n = o(\epsilon)$$
When $(\log n)^{-O(1)} \lesssim \epsilon \lesssim 1/\log n$, then the optimal $k$-to-$k^2$ gap is recovered up to $\text{polylog}(n)$ factors by this result. 

\subsection{Negative Sparse PCA}
\label{subsec:3-neg-spca}

In this section, we deduce Theorem \ref{thm:neg-spca-lb} on the hardness of $\pr{neg-spca}$ using the reduction $\pr{bpds-to-neg-spca}$ and Theorem \ref{thm:neg-spca}. Because this reduction does not bear the number-theoretic considerations of the reduction to $\pr{rsme}$, this proof will be substantially more straightforward.

\begin{reptheorem}{thm:neg-spca-lb} [Lower Bounds for $\pr{neg-spca}$]
If $k, d$ and $n$ are polynomial in each other, $k = o(\sqrt{d})$ and $k = o(n^{1/6})$, then the $\pr{bpc}$ or $\pr{bpds}$ conjecture for constant $0 < q < p \le 1$ both imply conjecture implies a computational lower bound for $\pr{neg-spca}(n, k, d, \theta)$ at all levels of signal $\theta = \tilde{o}(\sqrt{k^2/n})$.
\end{reptheorem}

\begin{proof}
We show that Theorem \ref{thm:neg-spca} implies that $\pr{bpds-to-neg-spca}$ fills out all of the possible growth rates specified by the computational lower bound $\theta = \tilde{o}(\sqrt{k^2/n})$ and the other conditions in the theorem statement. Fix a constant pair of probabilities $0 < q < p \le 1$ and a sequence of parameters $(n, k, d, \theta)$ all of which are implicitly functions of $n$ such that
$$\theta \le cw^{-1} \cdot \sqrt{\frac{k^2}{n (\log n)^2}}, \quad wk \le n^{1/6} \quad \text{and} \quad w k^2 \le d$$
for sufficiently large $n$, an arbitrarily slow-growing function $w = w(n) \to \infty$ where $w(n) = n^{o(1)}$ and a sufficiently small constant $c > 0$. In order to fulfill the criteria in Condition \ref{cond:lb}, we now will specify: a sequence of parameters $(M, N, k_M, k_N, p, q)$ such that the $\pr{bpds}$ instance with these parameters is hard according to Conjecture \ref{conj:hard-conj}, and such that $\pr{neg-spca}$ with the parameters $(n, k, d, \theta)$ can be produced by the reduction $\pr{bpds-to-neg-spca}$ applied to $\pr{bpds}(M, N, k_M, k_N, p, q)$. These parameters along with the internal parameter $\tau$ of the reduction can be chosen as follows:
\begin{itemize}
\item let $N = n$, $k_N = w^{-1} \sqrt{n}$, $k_M = k$ and $M = w k^2$; and
\item let $\tau > 0$ be such that
$$\tau^2 = \frac{4n\theta}{k_N k(1 - \theta)}$$
\end{itemize}
It is straightforward to verify that the inequality above upper bounding $\theta$ implies that $\tau \le 4c/\sqrt{\log n}$ and thus satisfies the condition on $\tau$ needed to apply Lemma \ref{lem:randomrotations} and Theorem \ref{thm:neg-spca} for a sufficiently small $c > 0$. Furthermore, this setting of $\tau$ yields
$$\theta = \frac{\tau^2 k_N k}{4n + \tau^2 k_N k}$$
Furthermore, note that $d \ge M$ and $n \gg M^3$ by construction. Applying Theorem \ref{thm:neg-spca} now verifies the desired property above. This verifies the criteria in Condition \ref{cond:lb} and, following the argument in Section \ref{subsec:2-tvreductions}, Lemma \ref{lem:3a} now implies the theorem.
\end{proof}

We remark the the constraint $k = o(n^{1/6})$, as mentioned in Section \ref{subsec:1-problems-negspca}, is a technical condition that we believe should not be necessary for the theorem to hold. This is similar to the constraint arising in the strong reduction to sparse PCA given by $\pr{Clique-to-Wishart}$ in \cite{brennan2019optimal}. In $\pr{Clique-to-Wishart}$, the random matrix comparison between Wishart and $\pr{goe}$ produced the technical condition that $k = o(n^{1/6})$ in a similar manner to how our comparison result between Wishart and inverse Wishart produces the same constraint here. We also remark that the reduction $\pr{Clique-to-Wishart}$ can be used here to yield the same hardness for $\pr{neg-spca}$ as in Theorem \ref{thm:neg-spca-lb} based only on the $\pr{pc}$ conjecture. This is achieved by the reduction that maps from $\pr{pc}$ to sparse PCA with $d = wk^2$ as a first step using $\pr{Clique-to-Wishart}$ and then uses the second step of $\pr{bpds-to-neg-spca}$ to map to $\pr{neg-spca}$.

\subsection{Mixtures of Sparse Linear Regressions and Robustness}
\label{subsec:3-slr}

In this section, we deduce Theorems \ref{thm:uslr-lb}, \ref{thm:mslr-lb} and \ref{thm:rslr-lb} on the hardness of unsigned, mixtures of and robust sparse linear regression, all using the reduction $k\pr{-bpds-to-mslr}$ with different parameters $(r, \epsilon)$ and Theorem \ref{thm:slr-reduction}. We begin by showing bounds for $\pr{uslr}(n, k, d, \tau)$.

We first make the following simple but important observation. Note that a single sample from $\pr{uslr}$ is of the form $y = | \tau \cdot \langle v_S, X \rangle + \mN(0, 1)|$, which has the same distribution as $|y'|$ where $y' = \tau r \cdot \langle v_S, X \rangle + \mN(0, 1)$ and $r$ is an independent Rademacher random variable. Note that $y'$ is a sample from $\pr{mslr}_D(n, k, d, \gamma, 1/2)$ with $\gamma = \tau$. Thus to show a computational lower bound for $\pr{uslr}(n, k, d, \tau)$, it suffices to show a lower bound for $\pr{mslr}(n, k, d, \tau)$.

\begin{theorem}[Lower Bounds for $\pr{uslr}$] \label{thm:uslr-lb}
If $k, d$ and $n$ are polynomial in each other, $k = o(\sqrt{d})$ and $k = o(n^{1/6})$, then the $k\pr{-bpc}$ or $k\pr{-bpds}$ conjecture for constant $0 < q < p \le 1$ both imply that there is a computational lower bound for $\pr{uslr}(n, k, d, \tau)$ at all sample complexities $n = \tilde{o}(k^2/\tau^4)$.
\end{theorem}

\begin{proof}
To prove this theorem, we will show that Theorem \ref{thm:slr-reduction} implies that $k\pr{-bpds-to-mslr}$ applied with $r = 2$ fills out all of the possible growth rates specified by the computational lower bound $n = \tilde{o}(k^2/\tau^4)$ and the other conditions in the theorem statement. As mentioned above, it suffices to reduce in total variation to $\pr{mslr}(n, k, d, \tau)$. Fix a constant pair of probabilities $0 < q < p \le 1$ and any sequence of parameters $(n, k, d, \tau)$ all of which are implicitly functions of $n$ with
$$n \le c \cdot \frac{k^2}{w^2 \cdot \tau^4 \cdot (\log n)^{4}}, \quad wk \le n^{1/6} \quad \text{and} \quad w k^2 \le d$$
for sufficiently large $n$, an arbitrarily slow-growing function $w = w(n) \to \infty$ and a sufficiently small constant $c > 0$. In order to fulfill the criteria in Condition \ref{cond:lb}, we now will specify: a sequence of parameters $(M, N, k_M, k_N, p, q)$ such that the $k\pr{-bpds}$ instance with these parameters is hard according to Conjecture \ref{conj:hard-conj}, and such that $\pr{mslr}$ with the parameters $(n, k, d, \tau, 1/2)$ can be produced by the reduction $k\pr{-bpds-to-mslr}$ applied with $r = 2$ to $\pr{bpds}(M, N, k_M, k_N, p, q)$. By the discussion in Section \ref{subsec:2-tvreductions}, this would be sufficient to show the desired computational lower bound. We choose these parameters as follows:
\begin{itemize}
\item let $t$ be such that $2^t$ is the smallest power of two greater than $w\sqrt{n}$, let $k_N = \lfloor \sqrt{n} \rfloor$ and let $N = wk_N^2 \le k_N 2^t$; and
\item set $k_M = k$ and $M = w k^2$.
\end{itemize}
Now note that $\tau^2$ is upper bounded by
$$\tau^2 \le \frac{c^{1/2} \cdot k}{wn^{1/2} \cdot (\log n)^{2}} = O\left( \frac{k_N k_M}{N \log (MN)} \right)$$
Furthermore, we have that
$$\tau^2 \le \frac{c^{1/2} \cdot k}{wn^{1/2} \cdot (\log n)^{2}} = \Theta\left( \frac{k_M}{2^{t + 1} \log (k_N M \cdot 2^t) \log n} \right)$$
Therefore $\tau$ satisfies the conditions needed to apply Theorem \ref{thm:slr-reduction} for a sufficiently small $c > 0$. Also note that $n \gg M^3$ and $d \ge M$ by construction. Applying Theorem \ref{thm:slr-reduction} now verifies the desired property above. This verifies the criteria in Condition \ref{cond:lb} and, following the argument in Section \ref{subsec:2-tvreductions}, Lemma \ref{lem:3a} now implies the theorem.
\end{proof}

The proof of the theorem above also directly implies Theorem \ref{thm:mslr-lb}. This yields our main computational lower bounds for $\pr{mslr}$, which are stated below.

\begin{reptheorem}{thm:mslr-lb} [Lower Bounds for $\pr{mslr}$]
If $k, d$ and $n$ are polynomial in each other, $k = o(\sqrt{d})$ and $k = o(n^{1/6})$, then the $k\pr{-bpc}$ or $k\pr{-bpds}$ conjecture for constant $0 < q < p \le 1$ both imply that there is a computational lower bound for $\pr{mslr}(n, k, d, \tau)$ at all sample complexities $n = \tilde{o}(k^2/\tau^4)$.
\end{reptheorem}

Now observe that the instances of $\pr{mslr}$ output by the reduction $k\pr{-bpds-to-mslr}$ applied with $r > 2$ are instances of $\pr{rslr}$ in Huber's contamination model. Let $r$ be a prime number and $\epsilon \ge 1/r$. Also let $X \sim \mN(0, I_d)$ and $y = \tau \cdot \langle v_S, X \rangle + \eta$ where $\eta \sim \mN(0, 1)$ where $|S| = k$. By Definition \ref{defn:mslr-imbalanced}, $\pr{mslr}_D(n, k, d, \tau, 1/r)$ is of the form
$$\pr{mix}_{\epsilon}\left( \mL(X, y), \mD_O \right)^{\otimes n} \quad \text{where} \quad \mD_O = \pr{mix}_{\epsilon^{-1} r^{-1}} \left( \mL(X, y), \mL' \right)$$
for some possibly random $S$ with $|S| = k$ and where $\mL'$ denotes the distribution on pairs $(X, y)$ that are jointly Gaussian with mean zero and $(d + 1) \times (d + 1)$ covariance matrix
$$\left[\begin{matrix} \Sigma_{XX} & \Sigma_{Xy} \\ \Sigma_{yX} & \Sigma_{yy} \end{matrix} \right] = \left[\begin{matrix} I_d + \frac{(a^2 - 1)\gamma^2}{1 + \gamma^2} \cdot v_S v_S^\top & -a\gamma \cdot v_S \\ -a\gamma \cdot v_S^\top & 1 + \gamma^2 \end{matrix} \right]$$
This yields a very particular construction of an adversary in Huber's contamination model, which we show in the next theorem yields a computational lower bound for $\pr{rslr}$. With the observations above, the proof of this theorem is similar to that of Theorem \ref{thm:rsme-lb} and is deferred to Appendix \ref{subsec:appendix-3-part-3}.

\begin{theorem}[Lower Bounds for $\pr{rslr}$ with Condition (\pr{t})]
If $k, d$ and $n$ are polynomial in each other, $\epsilon < 1/2$ is such that $(n, \epsilon^{-1})$ satisfies $\pr{(t)}$, $k = o(\sqrt{d})$ and $k = o(n^{1/6})$, then the $k\pr{-bpc}$ conjecture or $k\pr{-bpds}$ conjecture for constant $0 < q < p \le 1$ both imply that there is a computational lower bound for $\pr{rslr}(n, k, d, \tau, \epsilon)$ at all sample complexities $n = \tilde{o}(k^2 \epsilon^2/\tau^4)$.
\end{theorem}

Our main computational lower bound for $\pr{rslr}$ follows from the same argument applied to the reduction $k\pr{-bpds-to-mslr}_R$ instead of $k\pr{-bpds-to-mslr}$ and using Corollary \ref{thm:mod-slr-reduction} instead of Theorem \ref{thm:slr-reduction}. As in Corollary \ref{cor:rsme-lb-mod}, this replaces condition $(\pr{t})$ with the weaker condition that $\epsilon = \tilde{\Omega}(n^{-1/2})$.

\begin{reptheorem}{thm:rslr-lb} [Lower Bounds for $\pr{rslr}$]
If $k, d$ and $n$ are polynomial in each other, $\epsilon < 1/2$ is such that $\epsilon = \tilde{\Omega}(n^{-1/2})$, $k = o(\sqrt{d})$ and $k = o(n^{1/6})$, then the $k\pr{-bpc}$ conjecture or $k\pr{-bpds}$ conjecture for constant $0 < q < p \le 1$ both imply that there is a computational lower bound for $\pr{rslr}(n, k, d, \tau, \epsilon)$ at all sample complexities $n = \tilde{o}(k^2 \epsilon^2/\tau^4)$.
\end{reptheorem}


\section{Community Recovery and Partition Models}
\label{sec:3-all-community}

In this section, we devise several reductions based on $\pr{Bern-Rotations}$ and $\pr{Tensor-Bern-Rotations}$ using the design matrices and tensors from Section \ref{sec:2-bernoulli-rotations} to reduce from $k\pr{-pc}, k\pr{-pds}, k\pr{-bpc}$ and $k\pr{-bpds}$ to dense stochastic block models, hidden partition models and semirandom planted dense subgraph. These reductions are briefly outlined in Section \ref{subsec:1-tech-design-matrices}.

Furthermore, the heuristic presented at the end of Section \ref{subsec:1-tech-design-matrices} predicts the computational barriers for the problems in this section. The $\ell_2$ norm of the matrix $\bE[X]$ corresponding to a $k\pr{-pc}$ instance is $\Theta(k)$, which is just below $\tilde{\Theta}(\sqrt{n})$ when this $k\pr{-pc}$ is near its computational barrier. Furthermore, it can be verified that the $\ell_2$ norm of the matrices $\bE[X]$ corresponding to the problems in this section are:
\begin{itemize}
\item If $\gamma = P_{11} - P_0$ in the $\pr{isbm}$ notation of Section \ref{subsec:1-problems-sbm}, then a direct calculation yields that the $\ell_2$ norm corresponding to $\pr{isbm}$ is $\Theta(n\gamma/k)$.
\item In $\pr{ghpm}$ and $\pr{bhpm}$, the corresponding $\ell_2$ norm can be verified to be $\Theta(K\gamma\sqrt{r})$.
\item In our adversarial construction for $\pr{semi-cr}$, the corresponding $\ell_2$ norm is $\Theta(k \gamma)$ where $\gamma = P_1 - P_0$.
\end{itemize}
Following the heuristic, setting these equal to $\tilde{\Theta}(\sqrt{n})$ yields the predicted computational barriers of $\gamma^2 = \tilde{\Theta}(k^2/n)$ in $\pr{isbm}$, $\gamma^2 = \tilde{\Theta}(n/rK^2)$ in $\pr{ghpm}$ and $\pr{bhpm}$ and $\gamma^2 = \tilde{\Theta}(n/k^2)$ in $\pr{semi-cr}$. We now present our reduction to $\pr{isbm}$.

\subsection{Dense Stochastic Block Models with Two Communities}
\label{sec:3-community}

We begin by recalling the definition of the imbalanced 2-block stochastic block model from Section \ref{subsec:1-problems-sbm}.

\begin{definition}[Imbalanced 2-Block Stochastic Block Model]
Let $k$ and $n$ be positive integers such that $k$ divides $n$. The distribution $\pr{isbm}_D(n, k, P_{11}, P_{12}, P_{22})$ over $n$-vertex graphs $G$ is sampled by first choosing an $(n/k)$-subset $C \subseteq [n]$ uniformly at random and sampling the edges of $G$ independently with the following probabilities
$$\bP\left[ \{i, j \} \in E(G) \right] = \left\{ \begin{array}{ll} P_{11} &\textnormal{if } i, j \in C \\ P_{12} &\textnormal{if exactly one of } i, j \textnormal{ is in } C \\ P_{22} &\textnormal{if } i, j \in [n] \backslash C \end{array} \right.$$
\end{definition}

Given a subset $C \subseteq [n]$ of size $n/k$, we let $\pr{isbm}_D(n, C, P_{11}, P_{12}, P_{22})$ denote $\pr{isbm}$ as defined above conditioned on the latent subset $C$.  As discussed in Section \ref{subsec:2-formulations}, this naturally leads to a composite hypothesis testing problem between
$$H_0 : G \sim \mG\left(n, P_0 \right) \quad \text{and} \quad H_1 : G \sim \pr{isbm}_D(n, k, P_{11}, P_{12}, P_{22})$$
where $P_0$ is any edge density in $(0, 1)$. This section is devoted to showing reductions from $k\pr{-pds}$ and $k\pr{-pc}$ to $\pr{isbm}$ formulated as this hypothesis testing problem. In particular, we will focus on $P_{11}, P_{12}, P_{22}$ and $P_0$ all of which are bounded away from $0$ and $1$ by a constant, and which satisfy that
\begin{equation} \label{eqn:deg-isbm}
P_0 = \frac{1}{k} \cdot P_{11} + \left( 1 - \frac{1}{k} \right) P_{12} = \frac{1}{k} \cdot P_{12} + \left( 1 - \frac{1}{k} \right) P_{22}
\end{equation}
These two constraints allow $P_{11}, P_{12}, P_{22}$ to be reparameterized in terms of a signal parameter $\gamma$ as
\begin{equation} \label{eqn:isbm-param}
P_{11} = P_0 + \gamma, \quad P_{12} = P_0 - \frac{\gamma}{k - 1} \quad \text{and} \quad P_{22} = P_0 + \frac{\gamma}{(k - 1)^2}
\end{equation}
There are two main reasons why we restrict to the parameter regime enforced by the density constraints in (\ref{eqn:deg-isbm}) -- it creates a model with nearly uniform expected degrees and which is a mean-field analogue of recovering the first community in the $k$-block stochastic block model.
\begin{itemize}
\item \textit{Nearly Uniform Expected Degrees}: Observe that, conditioned on $C$, the expected degree of a vertex $i \in [n]$ in $\pr{isbm}(n, k, P_{11}, P_{12}, P_{22})$ is given by
$$\bE\left[ \deg(i) | C \right] = \left\{ \begin{array}{ll} \left( \frac{n}{k} - 1 \right) \cdot P_{11} + \frac{n(k - 1)}{k} \cdot P_{12} &\textnormal{if } i \in C \\ \frac{n}{k} \cdot P_{12} + \left( \frac{n(k - 1)}{k} - 1 \right) \cdot P_{22} &\textnormal{if } i \in [n] \backslash C \end{array} \right.$$
Thus the density constraints in (\ref{eqn:deg-isbm}) ensure that these differ by at most $1$ from each other and from $(n - 1)P_0$. Thus all of the vertices in $\pr{isbm}(n, k, P_{11}, P_{12}, P_{22})$ and the $H_0$ model $\mG\left(n, P_0 \right)$ have approximately the same expected degree. This precludes simple degree and total edge thresholding tests that are optimal in models of single community detection that are not degree-corrected. As discussed in Section \ref{subsec:1-problems-semicr}, the planted dense subgraph model has a detection threshold that differs from the conjectured Kesten-Stigum threshold for recovery of the planted dense subgraph. Thus to obtain computational lower bounds for a hypothesis testing problem that give tight recovery lower bounds, calibrating degrees is crucial. The main result of this section can be viewed as showing approximate degree correction is sufficient to obtain the Kesten-Stigum threshold for $\pr{isbm}$ through a reduction from $k\pr{-pds}$ and $k\pr{-pc}$.
\item \textit{Mean-Field Analogue of First Community Recovery in $k\pr{-sbm}$}: As discussed in Section \ref{subsec:1-problems-sbm}, the imbalanced 2-block stochastic block model $\pr{isbm}_D(n, k, P_{11}, P_{12}, P_{22})$ is roughly a mean-field analogue of recovering the first community $C_1$ in a $k$-block stochastic block model. More precisely, consider a graph $G$ wherein the vertex set $[n]$ is partitioned into $k$ latent communities $C_1, C_2, \dots, C_k$ each of size $n/k$ and edges are then included in the graph $G$ independently such that intra-community edges appear with probability $p$ while inter-community edges appear with probability $q < p$. The distribution $\pr{isbm}_D(n, k, P_{11}, P_{12}, P_{22})$ can be viewed as a mean-field analogue of recovering a first community $C = C_1$ in the $k$-block model, when
$$P_{11} = p, \quad P_{12} = q \quad \text{and} \quad P_{22} = \frac{1}{k - 1} \cdot p + \left(1 - \frac{1}{k - 1} \right) q$$
Here, $P_{22}$ approximately corresponds to the average edge density on the subgraph of the $k$-block model restricted to $[n] \backslash C_1$. This analogy between $\pr{isbm}$ and $k\pr{-sbm}$ is also why we choose to parameterize $\pr{isbm}$ in terms of $k$ rather than the size $n/k$ of $C$.
\end{itemize}

As discussed in Section \ref{subsec:1-problems-sbm}, if $k = o(\sqrt{n})$, the conjectured recovery threshold for efficient recovery in $k\pr{-sbm}$ is the Kesten-Stigum threshold of
$$\frac{(p - q)^2}{q(1 - q)} \gtrsim \frac{k^2}{n}$$
while the statistically optimal rate of recovery is when this level of signal is instead $\tilde{\Omega}(k^4/n^2)$. Furthermore, the information-theoretic threshold and conjectured computational barrier are the same for $\pr{isbm}$ in the regime defined by (\ref{eqn:deg-isbm}). Parameterizing $\pr{isbm}$ in terms of $\gamma$ as in (\ref{eqn:isbm-param}), the Kesten-Stigum threshold can be expressed as $\gamma^2 = \tilde{\Omega}(k^2/n)$. The objective of this section is give a reduction from $k\pr{-pds}$ to $\pr{isbm}$ in the dense regime with $\min\{P_0, 1 - P_0\} = \Omega(1)$ up to the Kesten-Stigum threshold.

The first reduction of this section $k$\textsc{-pds-to-isbm} is shown in Figure \ref{fig:isbm-reduction} and maps to the case where $P_0 = 1/2$ and (\ref{eqn:isbm-param}) is only approximately true. In a subsequent corollary, a simple modification of this reduction will map to all $P_0$ with $\min\{P_0, 1 - P_0\} = \Omega(1)$ and show (\ref{eqn:isbm-param}) holds exactly. The following theorem establishes the approximate Markov transition properties of $k$\textsc{-pds-to-isbm}. The proof of this theorem follows a similar structure to the proof of Theorem \ref{thm:isgmreduction}. Recall that $\Phi(x) = \frac{1}{\sqrt{2\pi}} \int_{-\infty}^x e^{-x^2/2} dx$ denotes the standard normal CDF.

\begin{theorem}[Reduction to $\pr{isbm}$] \label{thm:isbm}
Let $N$ be a parameter and $r = r(N) \ge 2$ be a prime number. Fix initial and target parameters as follows:
\begin{itemize}
\item \textnormal{Initial} $k\pr{-bpds}$ \textnormal{Parameters:} vertex count $N$, subgraph size $k = o(N)$ dividing $N$, edge probabilities $0 < q < p \le 1$ with $\min\{q, 1 - q\} = \Omega(1)$ and $p - q \ge N^{-O(1)}$, and a partition $E$ of $[N]$.
\item \textnormal{Target} $\pr{isbm}$ \textnormal{Parameters:} $(n, r)$ where $\ell = \frac{r^t - 1}{r - 1}$ and $n = kr\ell$ for some parameter $t = t(N) \in \mathbb{N}$ satisfying that that
$$m \le kr^t \le kr\ell \le \textnormal{poly}(N)$$
where $m$ is the smallest multiple of $k$ larger than $\left( \frac{p}{Q} + 1 \right) N$ and where
$$Q = 1 - \sqrt{(1 - p)(1 - q)} + \mathbf{1}_{\{ p = 1\}} \left( \sqrt{q} - 1 \right)$$
\item \textnormal{Target} $\pr{isbm}$ \textnormal{Edge Strengths:} $(P_{11}, P_{12}, P_{22})$ given by
$$P_{11} = \Phi\left( \frac{\mu(r - 1)^2}{r^{t +1}}\right), \quad P_{12} = \Phi\left( - \frac{\mu(r - 1)}{r^{t+1}}\right) \quad \textnormal{and} \quad P_{22} = \Phi\left( \frac{\mu}{r^{t +1}}\right)$$
where $\mu \in (0, 1)$ satisfies that
$$\mu \le \frac{1}{2 \sqrt{6\log (kr\ell) + 2\log (p - Q)^{-1}}} \cdot \min \left\{ \log \left( \frac{p}{Q} \right), \log \left( \frac{1 - Q}{1 - p} \right) \right\}$$
\end{itemize}
Let $\mathcal{A}(G)$ denote $k$\textsc{-pds-to-isbm} applied to the graph $G$ with these parameters. Then $\mathcal{A}$ runs in $\textnormal{poly}(N)$ time and it follows that
\begin{align*}
\TV\left( \mathcal{A}\left( \mG_E(N, k, p, q) \right), \, \pr{isbm}_D(n, r, P_{11}, P_{12}, P_{22}) \right) &= O\left( \frac{k}{\sqrt{N}} + e^{-\Omega(N^2/km)} + (kr\ell)^{-1} \right) \\
\TV\left( \mathcal{A}\left( \mG(N, q) \right), \, \mG(n, 1/2) \right) &= O\left( e^{-\Omega(N^2/km)} + (kr\ell)^{-1} \right)
\end{align*}
\end{theorem}

\begin{figure}[t!]
\begin{algbox}
\textbf{Algorithm} $k$\textsc{-pds-to-isbm}

\vspace{1mm}

\textit{Inputs}: $k$\pr{-pds} instance $G \in \mG_N$ with dense subgraph size $k$ that divides $N$, and the following parameters
\begin{itemize}
\item partition $E$ of $[N]$ into $k$ parts of size $N/k$, edge probabilities $0 < q < p \le 1$
\item let $m$ be the smallest multiple of $k$ larger than $\left( \frac{p}{Q} + 1 \right) N$ where
$$Q = 1 - \sqrt{(1 - p)(1 - q)} + \mathbf{1}_{\{ p = 1\}} \left( \sqrt{q} - 1 \right)$$
\item output number of vertices $n = kr\ell$ where $r$ is a prime number $r$, $\ell = \frac{r^t - 1}{r - 1}$ for some $t \in \mathbb{N}$ and
$$m \le kr^t \le kr\ell \le \text{poly}(N)$$
\item mean parameter $\mu \in (0, 1)$ satisfying that
$$\mu \le \frac{1}{2 \sqrt{6\log n + 2\log (p - Q)^{-1}}} \cdot \min \left\{ \log \left( \frac{p}{Q} \right), \log \left( \frac{1 - Q}{1 - p} \right) \right\}$$
\end{itemize}

\begin{enumerate}
\item \textit{Symmetrize and Plant Diagonals}: Compute $M_{\text{PD1}} \in \{0, 1\}^{m \times m}$ with partition $F$ of $[m]$ as
$$M_{\text{PD1}} \gets \pr{To-}k\textsc{-Partite-Submatrix}(G)$$
applied with initial dimension $N$, partition $E$, edge probabilities $p$ and $q$ and target dimension $m$.
\item \textit{Pad}: Form $M_{\text{PD2}} \in \{0, 1\}^{kr^t \times kr^t}$ by embedding $M_{\text{PD1}}$ as the upper left principal submatrix of $M_{\text{PD2}}$ and then adding $kr^t - m$ new indices for columns and rows, with all missing entries sampled i.i.d. from $\text{Bern}(Q)$. Let $F'_i$ be $F_i$ with $r^t - m/k$ of the new indices. Sample $k$ random permutations $\sigma_i$ of $F_i'$ independently for each $1 \le i \le k$ and permute the indices of the rows and columns of $M_{\text{PD2}}$ within each part $F'_i$ according to $\sigma_i$.
\item \textit{Bernoulli Rotations}: Let $F''$ be a partition of $[kr\ell]$ into $k$ equally sized parts. Now compute the matrix $M_{\text{R}} \in \mathbb{R}^{kr\ell \times kr\ell}$ as follows:
\begin{enumerate}
\item[(1)] For each $i, j \in [k]$, apply $\pr{Tensor-Bern-Rotations}$ to the matrix $(M_{\text{PD2}})_{F_i', F_j'}$ with matrix parameter $A_1 = A_2 = K_{r, t}$, rejection kernel parameter $R_{\pr{rk}} = kr\ell$, Bernoulli probabilities $0 < Q < p \le 1$, output dimension $r\ell$, $\lambda_1 = \lambda_2 = \sqrt{1 + (r - 1)^{-1}}$ and mean parameter $\mu$.
\item[(2)] Set the entries of $(M_{\text{R}})_{F''_i, F''_j}$ to be the entries in order of the matrix output in (1).
\end{enumerate}
\item \textit{Threshold and Output}: Now construct the graph $G'$ with vertex set $[kr\ell]$ such that for each $i > j$ with $i, j \in [kr\ell]$, we have $\{i, j \} \in E(G')$ if and only if $(M_{\text{R}})_{ij} \ge 0$. Output $G'$ with randomly permuted vertex labels.
\end{enumerate}
\vspace{0.5mm}

\end{algbox}
\caption{Reduction from $k$-partite planted dense subgraph to the dense imbalanced 2-block stochastic block model.}
\label{fig:isbm-reduction}
\end{figure}

To prove this theorem, we begin by proving a lemma analyzing the dense Bernoulli rotations step of $k$\textsc{-pds-to-isbm}. Define $v_{S, F', F''}(M)$ as in Section \ref{subsec:3-rsme-reduction}. The proof of the next lemma follows similar steps to the proof of Lemma \ref{lem:isgm-rotations}.

\begin{lemma}[Bernoulli Rotations for $\pr{isbm}$] \label{lem:isbm-rotations}
Let $F'$ and $F''$ be a fixed partitions of $[kr^t]$ and $[kr\ell]$ into $k$ parts of size $r^t$ and $r\ell$, respectively, and let $S \subseteq [kr^t]$ where $|S \cap F_i'| = 1$ for each $1 \le i \le k$. Let $\mathcal{A}_{\textnormal{3}}$ denote Step 3 of $k\pr{-pds-to-isbm}$ with input $M_{\textnormal{PD2}}$ and output $M_{\textnormal{R}}$. Suppose that $p, Q$ and $\mu$ are as in Theorem \ref{thm:isbm}, then it follows that
\begin{align*}
&\TV\Big( \mathcal{A}_{\textnormal{3}} \left( \mathcal{M}_{[kr^t] \times [kr^t]} \left( S \times S, \textnormal{Bern}(p), \textnormal{Bern}(Q) \right) \right), \\
&\quad \quad \quad \quad \left. \mL\left( \frac{\mu(r -1)}{r} \cdot v_{S, F', F''}(K_{r, t}) v_{S, F', F''}(K_{r, t})^\top + \mN(0, 1)^{\otimes kr\ell \times kr\ell} \right) \right) = O\left((kr\ell)^{-1}\right) \\
&\TV\left( \mathcal{A}_{\textnormal{3}} \left(\textnormal{Bern}(Q)^{\otimes kr^t \times kr^t} \right), \, \mN(0, 1)^{\otimes kr\ell \times kr\ell} \right) = O\left((kr\ell)^{-1}\right)
\end{align*}
\end{lemma}

\begin{proof}
First consider the case where $M_{\textnormal{PD2}} \sim \mathcal{M}_{[kr^t] \times [kr^t]} \left( S \times S, \textnormal{Bern}(p), \textnormal{Bern}(Q) \right)$. Observe that the submatrices of $M_{\textnormal{PD2}}$ are distributed as follows
$$(M_{\textnormal{PD2}})_{F_i', F_j'} \sim \pr{pb}\left(F_i' \times F_j', (S \cap F_i', S \cap F_j'), p, Q\right)$$
and are independent. Combining upper bound on the singular values of $K_{r, t}$ in Lemma \ref{lem:Krtsv} with Corollary \ref{cor:tensor-bern-rotations} implies that
$$\TV\left( (M_{\textnormal{R}})_{F''_i, F''_j}, \, \mL\left( \frac{\mu(r -1)}{r} \cdot (K_{r, t})_{\cdot, S \cap F_i'} (K_{r, t})_{\cdot, S \cap F_j'}^\top + \mN(0, 1)^{\otimes r\ell \times r\ell} \right) \right) = O\left(r^{2t} \cdot (kr\ell)^{-3} \right)$$
Since the submatrices $(M_{\textnormal{R}})_{F''_i, F''_j}$ are independent, the tensorization property of total variation in Fact \ref{tvfacts} implies that $\TV\left( M_{\textnormal{R}}, \mL(Z) \right) = O\left(k^2r^{2t} \cdot (kr\ell)^{-3} \right) = O\left((kr\ell)^{-1}\right)$ where the submatrices $Z_{F''_i, F_j''}$ are independent and satisfy
$$Z_{F''_i, F_j''} \sim \mL\left( \frac{\mu(r -1)}{r} \cdot (K_{r, t})_{\cdot, S \cap F_i'} (K_{r, t})_{\cdot, S \cap F_j'}^\top + \mN(0, 1)^{\otimes r\ell \times r\ell} \right)$$
Note that the entries of $Z$ are independent Gaussians each with variance $1$ and $Z$ has mean given by $\mu(1 + r^{-1}) \cdot v_{S, F', F''}(K_{r, t}) v_{S, F', F''}(K_{r, t})^\top$, by the definition of $v_{S, F', F''}(K_{r, t})$. This proves the first total variation upper bound in the statement of the lemma. Now suppose that $M_{\textnormal{PD2}} \sim \textnormal{Bern}(Q)^{\otimes kr^t \times kr^t}$. Corollary \ref{cor:tensor-bern-rotations} implies that
$$\TV\left( (M_{\textnormal{R}})_{F''_i, F''_j}, \, \mN(0, 1)^{\otimes r\ell \times r\ell} \right) = O\left(r^{2t} \cdot (kr\ell)^{-3} \right)$$
for each $1 \le i, j \le k$. Since the submatrices $(M_{\textnormal{R}})_{F''_i, F''_j}$ of $M_{\textnormal{R}}$ are independent, it follows that
$$\TV\left( M_{\textnormal{R}}, \, \mN(0, 1)^{\otimes kr\ell \times kr\ell} \right) = O\left(k^2r^{2t} \cdot (kr\ell)^{-3} \right) = O\left((kr\ell)^{-1}\right)$$
by the tensorization property of total variation in Fact \ref{tvfacts}, completing the proof of the lemma.
\end{proof}

The next lemma is immediate but makes explicit the precise guarantees for Step 4 of $k\pr{-pds-to-isbm}$.

\begin{lemma}[Thresholding for $\pr{isbm}$] \label{lem:thresholding-isbm}
Let $F', F'', S$ and $T$ be as in Lemma \ref{lem:isbm-rotations}. Let $\mathcal{A}_{\textnormal{4}}$ denote Step 4 of $k\pr{-pds-to-isbm}$ with input $M_{\textnormal{R}}$ and output $G'$. Then
\begin{align*}
\mathcal{A}_{\textnormal{4}}\left( \frac{\mu(r -1)}{r} \cdot v_{S, F', F''}(K_{r, t}) v_{S, F', F''}(K_{r, t})^\top + \mN(0, 1)^{\otimes kr\ell \times kr\ell} \right) &\sim \pr{isbm}_D(kr\ell, r, P_{11}, P_{12}, P_{22}) \\
\mathcal{A}_{\textnormal{4}} \left( \mN(0, 1)^{\otimes kr\ell \times kr\ell} \right) &\sim \mG(kr\ell, 1/2)
\end{align*}
where $P_{11}, P_{12}$ and $P_{22}$ are as in Theorem \ref{thm:isbm}.
\end{lemma}

\begin{proof}
First observe that, since Lemma \ref{lem:suborthogonalmatrices} implies that each column of $K_{r, t}$ contains exactly $(r - 1)\ell$ entries equal to $1/\sqrt{r^t(r - 1)}$ and $\ell$ entries equal to $(1 - r)/\sqrt{r^t(r - 1)}$, it follows that $v_{S, F', F''}(K_{r, t})$ contains $k(r - 1)\ell$ entries equal to $1/\sqrt{r^t(r - 1)}$ and $k\ell$ entries equal to $(1 - r)/\sqrt{r^t(r - 1)}$. Therefore there is a subset $T \subseteq [kr\ell]$ with $|T| = k\ell$ such that the $kr\ell \times kr\ell$ mean matrix $Z = v_{S, F', F''}(K_{r, t}) v_{S, F', F''}(K_{r, t})^\top$ has entries
$$Z_{ij} = \frac{1}{r^t(r - 1)} \cdot \left\{ \begin{array}{ll}  (r - 1)^2 &\textnormal{if } i, j \in S \\ -(r - 1) &\textnormal{if } i \in S \text{ and } j \not \in S \text{ or } i \not \in S \text{ and } j \in S \\ 1 &\textnormal{if } i, j \not \in S \end{array} \right.$$
Since the vertices of $G'$ are randomly permuted, it follows by definition now that if
$$M_{\textnormal{R}} \sim \mL\left( \frac{\mu(r -1)}{r} \cdot v_{S, F', F''}(K_{r, t}) v_{S, F', F''}(K_{r, t})^\top + \mN(0, 1)^{\otimes kr\ell \times kr\ell} \right)$$
then $G' \sim \pr{isbm}_D(kr\ell, k\ell, P_{11}, P_{12}, P_{22})$, proving the first distributional equality in the lemma. The second distributional equality follows from the fact that $\Phi(0) = 1/2$.
\end{proof}

We now complete the proof of Theorem \ref{thm:isbm} using a similar application of Lemma \ref{lem:tvacc} as in the proof of Theorem \ref{thm:isgmreduction}.

\begin{proof}[Proof of Theorem \ref{thm:isbm}]
We apply Lemma \ref{lem:tvacc} to the steps $\mathcal{A}_i$ of $\mathcal{A}$ under each of $H_0$ and $H_1$. Define the steps of $\mathcal{A}$ to map inputs to outputs as follows
$$(G, E) \xrightarrow{\mathcal{A}_1} (M_{\text{PD1}}, F) \xrightarrow{\mathcal{A}_2} (M_{\text{PD2}}, F') \xrightarrow{\mathcal{A}_3} (M_{\text{R}}, F'') \xrightarrow{\mathcal{A}_{\text{4}}} G'$$
Under $H_1$, consider Lemma \ref{lem:tvacc} applied to the following sequence of distributions
\allowdisplaybreaks
\begin{align*}
\mathcal{P}_0 &= \mG_E(N, k, p, q) \\
\mathcal{P}_1 &= \mathcal{M}_{[m] \times [m]}(S \times S, \textnormal{Bern}(p), \textnormal{Bern}(Q)) \quad \text{where } S \sim \mU_m(F) \\
\mathcal{P}_2 &= \mathcal{M}_{[kr^t] \times [kr^t]}(S \times S, \textnormal{Bern}(p), \textnormal{Bern}(Q)) \quad \text{where } S \sim \mU_{kr^t}(F') \\
\mathcal{P}_3 &= \frac{\mu(r -1)}{r} \cdot v_{S, F', F''}(K_{r, t}) v_{S, F', F''}(K_{r, t})^\top + \mN(0, 1)^{\otimes kr\ell \times kr\ell} \quad \text{where } S \sim \mU_{kr^t}(F') \\
\mathcal{P}_{\text{4}} &= \pr{isbm}_D(kr\ell, r, P_{11}, P_{12}, P_{22})
\end{align*}
Applying Lemma \ref{lem:submatrix}, we can take
$$\epsilon_1 = 4k \cdot \exp\left( - \frac{Q^2N^2}{48pkm} \right) + \sqrt{\frac{C_Q k^2}{2m}}$$
where $C_Q = \max\left\{ \frac{Q}{1 - Q}, \frac{1 - Q}{Q} \right\}$. The step $\mathcal{A}_2$ is exact and we can take $\epsilon_2 = 0$. Applying Lemma \ref{lem:isbm-rotations} and averaging over $S \sim \mU_{kr^t}(F')$ using the conditioning property of total variation in Fact \ref{tvfacts} yields that we can take $\epsilon_3 = O\left((kr\ell)^{-1}\right)$. By Lemma \ref{lem:thresholding-isbm}, Step 4 is exact and we can take $\epsilon_4 = 0$. By Lemma \ref{lem:tvacc}, we therefore have that
$$\TV\left( \mathcal{A}\left( \mG_E(N, k, p, q) \right), \, \pr{isbm}(n, r, P_{11}, P_{12}, P_{22}) \right) = O\left( \frac{k}{\sqrt{N}} + e^{-\Omega(N^2/km)} + (kr\ell)^{-1} \right)$$
which proves the desired result in the case of $H_1$. Under $H_0$, consider the distributions
\allowdisplaybreaks
\begin{align*}
\mathcal{P}_0 &= \mG(N, q) \\
\mathcal{P}_1 &= \text{Bern}(Q)^{\otimes m \times m} \\
\mathcal{P}_2 &= \text{Bern}(Q)^{\otimes kr^t \times kr^t} \\
\mathcal{P}_3 &= \mN(0, 1)^{\otimes kr\ell \times kr\ell} \\
\mathcal{P}_{\text{4}} &= \mG(kr\ell, 1/2)
\end{align*}
As above, Lemmas \ref{lem:submatrix}, \ref{lem:isbm-rotations} and \ref{lem:thresholding-isbm} imply that we can take
$$\epsilon_1 = 4k \cdot \exp\left( - \frac{Q^2N^2}{48pkm} \right), \quad \epsilon_2 = 0, \quad \epsilon_3 = O\left((kr\ell)^{-1}\right) \quad \text{and} \quad \epsilon_{\text{4}} = 0$$
By Lemma \ref{lem:tvacc}, we therefore have that
$$\TV\left( \mathcal{A}\left( \mG(N, q) \right), \mG(n, 1/2) \right) = O\left( e^{-\Omega(N^2/kn)} + (kr\ell)^{-1} \right)$$
which completes the proof of the theorem.
\end{proof}

We now prove that a slight modification to this reduction will map to all $P_0$ with $\min\{P_0, 1 - P_0\} = \Omega(1)$ and to the setting where the density constraints in (\ref{eqn:isbm-param}) hold exactly.

\begin{corollary}[Reduction to Arbitrary $P_0$] \label{thm:isbm-mod}
Let $0 < q < p \le 1$ be constant and let $N, r, k, E, \ell$ and $n$ be as in Theorem \ref{thm:isbm} with the additional condition that $kr^{3/2} = o(r^{2t})$. Suppose that $P_0$ satisfies $\min\{P_0, 1 - P_0 \} = \Omega(1)$ and $\gamma \in (0, 1)$ satisfies that
$$\gamma \le \frac{c}{r^{t - 1} \sqrt{\log (k r \ell)}}$$
for a sufficiently small constant $c > 0$. Then there is a $\textnormal{poly}(N)$ time reduction $\mathcal{A}$ from graphs on $N$ vertices to graphs on $n$ vertices satisfying that
\begin{align*}
&\TV\left( \mathcal{A}\left( \mG_E(N, k, p, q) \right), \, \pr{isbm}_D\left(n, r, P_0 + \gamma, P_0 - \frac{\gamma}{k - 1}, P_0 + \frac{\gamma}{(k - 1)^2} \right) \right) \\
&\quad \quad = O\left( \frac{k \mu^3 r^{3/2}}{r^{2t}} + \frac{k}{\sqrt{N}} + e^{-\Omega(N^2/km)} + (kr\ell)^{-1} \right) \\
&\TV\left( \mathcal{A}\left( \mG(N, q) \right), \, \mG(n, P_0) \right) = O\left( e^{-\Omega(N^2/km)} + (kr\ell)^{-1} \right)
\end{align*}
\end{corollary}

\begin{proof}
Consider the reduction $\mathcal{A}$ that adds a simple post-processing step to $k$\textsc{-pds-to-isbm} as follows. On input graph $G$ with $N$ vertices:
\begin{enumerate}
\item Form the graph $G_1$ by applying $k$\textsc{-pds-to-isbm} to $G$ with parameters $N, r, k, E, \ell, n$ and $\mu$ where $\mu$ is given by
$$\mu = \frac{r^{t + 1}}{(r - 1)^2} \cdot \Phi^{-1}\left( \frac{1}{2} + \frac{1}{2} \cdot \min\{P_0, 1 - P_0\}^{-1} \cdot \gamma \right)$$
and $\Phi^{-1}$ is the inverse of the standard normal CDF.
\item If $P_0 \le 1/2$, output the graph $G_2$ formed by independently including each edge of $G_1$ in $G_2$ with probability $2P_0$. If $P_0 > 1/2$, form $G_2$ instead by including each edge of $G_1$ in $G_2$ and including each non-edge of $G_1$ in $G_2$ as an edge independently with probability $2P_0 - 1$.
\end{enumerate}
This clearly runs in $\text{poly}(N)$ time and it suffices to establish its approximate Markov transition properties. Let $\mathcal{A}_1$ and $\mathcal{A}_2$ denote the two steps above with input-output pairs $(G, G_1)$ and $(G_1, G_2)$, respectively. Let $C \subseteq [n]$ be a fixed subset of size $n/r$ and define
\begin{align*}
&P_{11} = \Phi\left( \frac{\mu(r - 1)^2}{r^{t +1}}\right), \quad P_{12} = \Phi\left( - \frac{\mu(r - 1)}{r^{t+1}}\right) \quad \textnormal{and} \quad P_{22} = \Phi\left( \frac{\mu}{r^{t +1}}\right) \\
&P_{11}' = P_0 + \gamma, \quad P_{12}' = P_0 - \frac{\gamma}{r - 1} \quad \text{and} \quad P_{22}' = P_0 + \frac{\gamma}{(r - 1)^2}
\end{align*}
We will show that
\begin{equation} \label{eqn:density-comparison}
\TV\left( \mathcal{A}_2\left( \pr{isbm}_D\left(n, C, P_{11}, P_{12}, P_{22} \right) \right), \, \pr{isbm}_D\left(n, C, P_{11}', P_{12}', P_{22}' \right) \right) = O\left( \frac{k \mu^3 r^{3/2}}{r^{2t}} \right) = o(1)
\end{equation}
where the upper bound is $o(1)$ since $kr^{3/2} = o(r^{2t})$. First consider the case where $P_0 \le 1/2$. Step 2 above yields by construction that
$$\mathcal{A}_2\left( \pr{isbm}_D\left(n, C, P_{11}, P_{12}, P_{22} \right) \right) \sim \pr{isbm}_D\left(n, C, 2P_0 P_{11}, 2P_0 P_{12}, 2P_0 P_{22} \right)$$
Suppose that $X(r) \in \{0, 1\}^m$ is sampled by first sampling $X' \sim \text{Bin}(m, r)$ and then letting $X$ be selected uniformly at random from all elements of $\{0, 1\}^m$ with support size $X'$. It follows that $X(r) \sim \text{Bern}(r)^{\otimes m}$ since both distributions are permutation-invariant and their support sizes have the same distribution. Now the data-processing inequality in Fact \ref{tvfacts} implies that
$$\TV\left( \text{Bern}(r)^{\otimes m}, \, \text{Bern}(r')^{\otimes m} \right) = \TV\left( X(r), X(r') \right) \le \TV\left( \text{Bin}(m, r), \text{Bin}(m, r') \right)$$
which can be upper bounded with Lemma \ref{lem:bintv}. Using the fact that the edge indicators of $\pr{isbm}$ conditioned on $C$ are independent, the tensorization property in Fact \ref{tvfacts} and Lemma \ref{lem:bintv}, we now have that
\allowdisplaybreaks
\begin{align*}
&\TV\left( \pr{isbm}_D\left(n, C, 2P_0 P_{11}, 2P_0 P_{12}, 2P_0 P_{22} \right), \, \pr{isbm}_D\left(n, C, P_{11}', P_{12}', P_{22}' \right) \right) \\
&\quad \quad \le \TV\left( \text{Bern}(2P_0 P_{11})^{\otimes \binom{n/r}{2}}, \, \text{Bern}(P_{11}')^{\otimes \binom{n/r}{2}} \right) + \TV\left( \text{Bern}(2P_0 P_{12})^{\otimes \frac{n^2(r - 1)}{r^2}}, \, \text{Bern}(P_{12}')^{\otimes \frac{n^2(r - 1)}{r^2}} \right) \\
&\quad \quad \quad \quad + \TV\left( \text{Bern}(2P_0 P_{22})^{\otimes \binom{n(1 - 1/r)}{2}}, \, \text{Bern}(P_{22}')^{\otimes \binom{n(1 - 1/r)}{2}} \right) \\
&\quad \quad \le \left| 2P_0 P_{11} - P_{11}' \right| \cdot \sqrt{\frac{\binom{n/r}{2}}{2P'_{11}(1 - P'_{11})}} + \left| 2P_0 P_{12} - P_{12}' \right| \cdot \sqrt{\frac{n^2(r - 1)}{2r^2 P'_{12}(1 - P'_{12})}} \\
&\quad \quad \quad \quad + \left| 2P_0 P_{22} - P_{22}' \right| \cdot \sqrt{\frac{\binom{n(1 - 1/r)}{2}}{2P'_{22}(1 - P'_{22})}} \\
&\quad \quad \le \left| 2P_0 P_{11} - P_{11}' \right| \cdot O\left( \frac{n}{r} \right) + \left| 2P_0 P_{12} - P_{12}' \right| \cdot O\left( \frac{n}{\sqrt{r}} \right) + \left| 2P_0 P_{22} - P_{22}' \right| \cdot O(n)
\end{align*}
where the third inequality uses the fact that $P'_{11}, P'_{12}$ and $P'_{22}$ are each bounded away from $0$ and $1$. Observe that the definition of $\mu$ ensures
$$\frac{1}{2} + \frac{1}{2P_0} \cdot \gamma = \Phi\left( \frac{\mu (r - 1)^2}{r^{t + 1}} \right)$$
which implies that $2P_0 P_{11} = P_{11}'$. We now use a standard Taylor approximation for the error function $\Phi(x) - 1/2$ around zero, given by $\Phi(x) = \frac{1}{2} + \frac{x}{\sqrt{2\pi}} + O(x^3)$ when $x \in (-1, 1)$. Observe that
\begin{align*}
\left| 2P_0 P_{12} - P_{12}' \right| &= 2P_0 \cdot \left| \Phi\left( - \frac{\mu(r - 1)}{r^{t+1}}\right) - \frac{1}{2} + \frac{\gamma}{2P_0 (r - 1)} \right| \\
&= 2P_0 \cdot \left| \Phi\left( - \frac{\mu(r - 1)}{r^{t+1}}\right) - \frac{1}{2} + \frac{1}{r - 1} \left( \Phi\left( \frac{\mu (r - 1)^2}{r^{t + 1}} \right) - \frac{1}{2} \right) \right| \\
&= O\left( \frac{\mu^3 r^2}{r^{3t}} \right)
\end{align*}
An analogous computation shows that $\left| 2P_0 P_{22} - P_{22}' \right| = O\left( \mu^3/r^{3t - 1} \right)$. Combining all of these bounds now yields Equation (\ref{eqn:density-comparison}) after noting that $n = kr\ell = O(kr^t)$ implies that $n\mu^3r^{3/2}/r^{3t} = O(kr^{3/2}/r^{2t})$. A nearly identical argument considering the complement of the graph $G_1$ and replacing with $P_0$ with $1 - P_0$ establishes Equation (\ref{eqn:density-comparison}) in the case when $P_0 > 1/2$. Now observe that
$$\mathcal{A}_2 \left( \mG(n, 1/2) \right) \sim \mG(n, P_0)$$
by definition. Now consider applying Lemma \ref{lem:tvacc} to the steps $\mathcal{A}_1$ and $\mathcal{A}_2$ using an analogous recipe as in the proof of Theorem \ref{thm:isbm}. We have that $\epsilon_1$ is bounded by Theorem \ref{thm:isbm} and $\epsilon_2$ is bounded by the argument above. Note that in order to apply Theorem \ref{thm:isbm} here, it must follow that the required bound on $\mu$ is met. Observe that
$$\gamma = 2P_0 \left( \Phi\left( \frac{\mu (r - 1)^2}{r^{t + 1}} \right) - \frac{1}{2} \right) = \Theta\left( \frac{\mu}{r^{t - 1}} \right)$$
and hence if $\gamma$ satisfies the upper bound in the statement of the corollary for a sufficiently small constant $c$, then $\mu$ satisfies the requirement in Theorem \ref{thm:isbm} since $p$ and $q$ are constant. This application of Lemma \ref{lem:tvacc} now yields the desired two approximate Markov transition properties and completes the proof of the corollary.
\end{proof}

We now show that setting parameters in the reduction of Corollary \ref{thm:isbm-mod} as in the recipe set out in Theorems \ref{thm:rsme-lb} and \ref{thm:uslr-lb} now shows that we can fill out the parameter space for $\pr{isbm}$ obeying the edge density constraints of (\ref{eqn:isbm-param}) below the Kesten-Stigum threshold. This proves the following computational lower bound for $\pr{isbm}$. We remark that typically the parameter regime of interest for the $k$-block stochastic block model is when $k = n^{o(1)}$, and thus the conditions \pr{(t)} and $k = o(n^{1/3})$ are only mild restrictions here. Note that the condition $\pr{(t)}$ here is the same condition that was introduced in Section \ref{subsec:3-rsme}.

\begin{reptheorem}{thm:isbm-lb} [Lower Bounds for $\pr{isbm}$]
Suppose that $(n, k)$ satisfy condition \pr{(t)}, that $k$ is prime or $k = \omega_n(1)$ and $k = o(n^{1/3})$, and suppose that $P_0 \in (0, 1)$ satisfies $\min\{P_0, 1 - P_0 \} = \Omega_n(1)$. Consider the testing problem $\pr{isbm}(n, k, P_{11}, P_{12}, P_{22})$ where
$$P_{11} = P_0 + \gamma, \quad P_{12} = P_0 - \frac{\gamma}{k - 1} \quad \text{and} \quad P_{22} = P_0 + \frac{\gamma}{(k - 1)^2}$$
Then the $k\pr{-pc}$ conjecture or $k\pr{-pds}$ conjecture for constant $0 < q < p \le 1$ both imply that there is a computational lower bound for $\pr{isbm}(n, k, P_{11}, P_{12}, P_{22})$ at all levels of signal below the Kesten-Stigum threshold of $\gamma^2 = \tilde{o}(k^2/n)$.
\end{reptheorem}

\begin{proof}
It suffices to show that the reduction $\mathcal{A}$ in Corollary \ref{thm:isbm-mod} applied with $r \ge 2$ fills out all of the possible growth rates specified by the computational lower bound $\gamma^2 = \tilde{o}(k^2/n)$ and the other conditions in the theorem statement. Fix a constant pair of probabilities $0 < q < p \le 1$ and any sequence of parameters $(n, k, \gamma, P_0)$ all of which are implicitly functions of $n$ such that $(n, k)$ satisfies $\pr{(t)}$ and
$$\gamma^2 \le \frac{k^2}{w' \cdot n \log n}, \quad 2(w')^2 k \le n^{1/3} \quad \text{and} \quad \min\{P_0, 1 - P_0 \} = \Omega_n(1)$$
for sufficiently large $n$ and $w' = w'(n) = (\log n)^{c}$ for a sufficiently large constant $c > 0$. Now let $w = w(n) \to \infty$ be an arbitrarily slow-growing increasing positive integer-valued function at least satisfying that $w(n) = n^{o(1)}$. As in the proof of Theorem \ref{thm:rsme-lb}, we now specify the following in order to fulfill the criteria in Condition \ref{cond:lb}:
\begin{enumerate}
\item a sequence $(N, k_N)$ such that the $k\pr{-pds}(N, k_N, p, q)$ is hard according to Conjecture \ref{conj:hard-conj}; and 
\item a sequence $(n', k', \gamma, P_0)$ with a subsequence that satisfies three conditions: (2.1) the parameters on the subsequence are in the regime of the desired computational lower bound for $\pr{isbm}$; (2.2) they have the same growth rate as $(n, k, \gamma, P_0)$ on this subsequence; and (2.3) such that $\pr{isbm}$ with the parameters on this subsequence can be produced by $\mathcal{A}$ with input $k\pr{-pds}(N, k_N, p, q)$.
\end{enumerate}
As discussed in Section \ref{subsec:2-tvreductions}, this is sufficient to prove the theorem. We choose these parameters as follows:
\begin{itemize}
\item let $k' = r$ be the smallest prime satisfying that $k \le r \le 2k$, which exists by Bertrand's postulate and can be found in $\text{poly}(n)$ time;
\item let $t$ be such that $r^t$ is the closest power of $r$ to $\sqrt{n}$ and let
$$k_N = \left\lfloor \frac{1}{2}\left( 1 + \frac{p}{Q} \right)^{-1} w^{-2} \cdot \min\left\{ r^t, \sqrt{n} \right\} \right\rfloor$$
where $Q = 1 - \sqrt{(1 - p)(1 - q)} + \mathbf{1}_{\{ p = 1\}} \left( \sqrt{q} - 1 \right)$; and
\item let $n' = k_N r\ell$ where $\ell = \frac{r^t - 1}{r - 1}$ and let $N = wk_N^2$.
\end{itemize}
Note that we have that $w^2 r \le n^{1/3}$ since $r \le 2k$. Now observe that we have the following bounds
\allowdisplaybreaks
\begin{align*}
n' &\asymp k_N r^t \asymp \left( w^{-2} \cdot \min\left\{ \frac{r^t}{\sqrt{n}}, 1 \right\} \cdot \frac{r^t}{\sqrt{n}} \right) n \\
k_N r^{3/2} &\lesssim w^{-2} \cdot \min\left\{ r^t, \sqrt{n} \right\} \cdot w^{-3} \sqrt{n} \lesssim \left( w^{-4} \cdot \frac{n}{r^{2t}} \right) r^{2t} \\
m &\le 2\left( \frac{p}{Q} + 1 \right) wk_N^2 \le \left( w^{-3} \cdot \frac{\sqrt{n}}{r^t} \right) k_N r^t \\
k_N r \ell &\le \text{poly}(N) \\
\gamma^2 &\le \frac{k^2}{w' \cdot n\log n} = \frac{1}{w' \cdot r^{2t - 2} \log(k_N r \ell)} \cdot \frac{r^{2t} \log(k_N r \ell)}{n\log n} \\
\gamma^2 &\lesssim \frac{r^2}{w' \cdot n' \log n'} \left( w^{-2} \cdot \min\left\{ \frac{r^t}{\sqrt{n}}, 1 \right\} \cdot \frac{r^t}{\sqrt{n}} \right) \cdot \frac{\log n'}{\log n} \lesssim \frac{r^2}{w' \cdot w^2 \cdot n' \log n'} \cdot \frac{r^t}{\sqrt{n}} 
\end{align*}
where $m$ is the smallest multiple of $k_N$ larger $\left( \frac{p}{Q} + 1 \right) N$. Now observe that as long as $\sqrt{n} = \tilde{\Theta}(r^t)$ then: (2.1) the last inequality above on $\gamma^2$ would imply that $(n', k', \gamma, P_0)$ is in the desired hard regime; (2.2) $n$ and $n'$ have the same growth rate since $w = n^{o(1)}$, and $k$ and $k' = r$ have the same growth rate since either $k' = k$ or $k' = \Theta(k) = \omega(1)$; and (2.3) the middle four bounds above imply that taking $c$ large enough yields the conditions needed to apply Corollary \ref{thm:isbm-mod} to yield the desired reduction. By Lemma \ref{lem:propT}, there is an infinite subsequence of the input parameters such that $\sqrt{n} = \tilde{\Theta}(r^t)$, which concludes the proof as in Theorem \ref{thm:rsme-lb}.
\end{proof}

%

\subsection{Testing Hidden Partition Models}
\label{sec:3-hidden-partition}

In this section, we establish statistical-computational gaps based on the $k\pr{-pc}$ and $k\pr{-pds}$ conjectures for detection in the Gaussian and bipartite hidden partition models introduced in Sections \ref{subsec:1-problems-hidden-partition} and \ref{subsec:2-formulations}. These two models are bipartite analogues of the subgraph variants of the $k$-block stochastic block model in the constant edge density regime. Specifically, they are multiple-community variants of the subgraph stochastic block model considered in \cite{brennan2018reducibility}.

The motivation for considering these two models is to illustrate the versatility of Bernoulli rotations as a reduction primitive. These two models are structurally very different from planted clique yet can be produced through Bernoulli rotations for appropriate choices of the output mean vectors $A_1, A_2, \dots, A_m$. The mean vectors specified in the reduction are vectorizations of the slices of the design tensor $T_{r, t}$ constructed based on the incidence geometry of $\mathbb{F}_r^t$. The definition of $T_{r, t}$ and several of its properties can be found in Section \ref{subsec:2-design-tensors}. The reduction in this section demonstrates that natural applications of Bernoulli rotations can require more involved constructions than $K_{r, t}$ in order to produce tight computational lower bounds.

We begin by reviewing the definitions of the two main models considered in this section -- Gaussian and bipartite hidden partition models -- which were introduced in Sections \ref{subsec:1-problems-hidden-partition} and \ref{subsec:2-formulations}.

\begin{definition}[Gaussian Hidden Partition Models] \label{defn:ghpm}
Let $n, r$ and $K$ be positive integers, let $\gamma \in \mathbb{R}$ and let $C = (C_1, C_2, \dots, C_r)$ be a sequence of disjoint $K$-subsets of $[n]$. Let $D = (D_1, D_2, \dots, D_r)$ be another such sequence. The distribution $\pr{ghpm}_D(n, r, C, D, \gamma)$ over matrices $M \in \mathbb{R}^{n \times n}$ is such that $M_{ij} \sim_{\textnormal{i.i.d.}} \mN(\mu_{ij}, 1)$ where
$$\mu_{ij} = \left\{ \begin{array}{ll} \gamma &\textnormal{if } i \in C_h \textnormal{ and } j \in D_h \textnormal{ for some } h \in [r] \\ -\frac{\gamma}{r - 1} &\textnormal{if } i \in C_{h_1} \textnormal{ and } j \in D_{h_2} \textnormal{ where } h_1 \neq h_2 \\ 0 &\textnormal{otherwise} \end{array} \right.$$
for each $i, j \in [n]$. Furthermore, let $\pr{ghpm}_D(n, r, K, \gamma)$ denote the mixture over $\pr{ghpm}_D(n, r, C, D, \gamma)$ induced by choosing $C$ and $D$ independently and uniformly at random.
\end{definition}

\begin{definition}[Bipartite Hidden Partition Models] \label{defn:bhpm}
Let $n, r, K, C$ and $D$ be as in Definition \ref{defn:ghpm} and let $P_0, \gamma \in (0, 1)$ be such that $\gamma/r \le P_0 \le 1 - \gamma$. The distribution $\pr{bhpm}_D(n, r, C, D, P_0, \gamma)$ over bipartite graphs $G$ with two parts of size $n$, each indexed by $[n]$, such that each edge $(i, j)$ is included in $G$ independently with the following probabilities
$$\bP\left[ (i, j) \in E(G) \right] = \left\{ \begin{array}{ll} P_0 + \gamma &\textnormal{if } i \in C_h \textnormal{ and } j \in D_h \textnormal{ for some } h \in [r] \\ P_0 - \frac{\gamma}{r - 1} &\textnormal{if } i \in C_{h_1} \textnormal{ and } j \in D_{h_2} \textnormal{ where } h_1 \neq h_2 \\ P_0 &\textnormal{otherwise} \end{array} \right.$$
for each $i, j \in [n]$. Let $\pr{bhpm}_D(n, r, K, P_0, \gamma)$ denote the mixture over $\pr{bhpm}_D(n, r, C, D, P_0, \gamma)$ induced by choosing $C$ and $D$ independently and uniformly at random.
\end{definition}

The problems we consider in this section are the two simple hypothesis testing problems $\pr{ghpm}$ and $\pr{bhpm}$ from Section \ref{subsec:2-formulations}, given by
$$\begin{array}{lll}
H_0: M \sim \mN(0, 1)^{\otimes n \times n} &\text{and} &H_1: M \sim \pr{ghpm}(n, r, K, \gamma) \\
H_0: G \sim \mG_B(n, n, P_0) &\text{and} &H_1: G \sim \pr{bhpm}(n, r, K, P_0, \gamma)
\end{array}$$
An important remark is that the hypothesis testing formulations above for these two problems seem to have different computational and statistical barriers from the tasks of recovering $C$ and $D$. We now state the following lemma, giving guarantees for a natural polynomial-time test and exponential time test for $\pr{ghpm}$. The proof of this lemma is tangential to the main focus of this section -- computational lower bounds for $\pr{ghpm}$ and $\pr{bhpm}$ -- and is deferred to Appendix \ref{subsec:appendix-3-part-3}.

\begin{lemma}[Tests for \pr{ghpm}] \label{lem:ghpm-test}
Given a matrix $M \in \mathbb{R}^{n \times n}$, let $s_C(M) = \sum_{i, j = 1}^n M_{ij}^2 - n^2$ and
$$s_I(M) = \max_{C, D} \left\{ \sum_{h = 1}^r \sum_{i \in C_h} \sum_{j \in D_h} M_{ij} \right\}$$
where the maximum is over all pairs $(C, D)$ of sequences of disjoint $K$-subsets of $[n]$. Let $w = w(n)$ be any increasing function with $w(n) \to \infty$ as $n \to \infty$. We prove the following:
\begin{enumerate}
\item If $M \sim \pr{ghpm}_D(n, r, K, \gamma)$, then with probability $1 - o_n(1)$ it holds that
$$s_C(M) \ge rK^2\gamma^2 + \frac{rK^2}{r - 1} \cdot \gamma^2 - w\left(n + \gamma K \sqrt{r} + \frac{K\gamma}{r} \right) \quad \textnormal{and} \quad s_I(M) \ge rK^2 \gamma - wr^{1/2} K$$
\item If $M \sim \mN(0, 1)^{\otimes n \times n}$, then with probability $1 - o_n(1)$ it holds that
$$s_C(M) \le wn \quad \textnormal{and} \quad s_I(M) \le 2r K^{3/2} w\sqrt{\left(\log n + \log r \right)}$$
\end{enumerate}
\end{lemma}

This lemma implies upper bounds on the computational and statistical barriers for $\pr{ghpm}$. Specifically, it implies that the variance test $s_C$ succeeds above $\gamma_{\text{comp}}^2 = \tilde{\Theta}(n/rK^2)$ and the search test $s_I$ succeeds above $\gamma_{\text{IT}}^2 = \tilde{\Theta}(1/K)$. Thus, showing that there is a computational barrier at this level of signal $\gamma_{\text{comp}}$ is sufficient to show that there is a nontrivial statistical-computational gap for $\pr{ghpm}$. For $P_0$ with $\min\{P_0, 1 - P_0\} = \Omega(1)$, analogous tests show the same upper bounds on $\gamma_{\text{comp}}$ and $\gamma_{\text{IT}}$ for $\pr{bhpm}$.

Consider the case when $n = rK$, which corresponds to a testing variant of the bipartite $k$-block stochastic block model. In this case, the upper bounds shown by the previous lemma coincide at $\gamma_{\text{comp}}^2, \gamma_{\text{IT}}^2 = O(r/n)$ and hence do not support the existence of a statistical-computational gap. The subgraph formulation in which $rK \ll n$ seems crucial to yielding a testing problem with a statistical-computational gap. We also remark that while this testing formulation when $n = rK$ may not have a gap, the task of recovering $C$ and $D$ likely shares the gap conjectured in the $k$-block stochastic block model. Specifically, the conjectured computational barrier at the Kesten-Stigum threshold lies at $\gamma^2 = \tilde{\Theta}(r^2/n)$, which lies well above the $r/n$ limit in the testing formulation.

\begin{figure}[t!]
\begin{algbox}
\textbf{Algorithm} $k$\textsc{-pds-to-ghpm}

\vspace{1mm}

\textit{Inputs}: $k$\pr{-pds} instance $G \in \mG_N$ with dense subgraph size $k$ that divides $N$, and the following parameters
\begin{itemize}
\item partition $E$, edge probabilities $0 < q < p \le 1$, $Q \in (0, 1)$ and $m$ as in Figure \ref{fig:isbm-reduction}
\item refinement parameter $s$ and number of vertices $n = ksr^t$ where $r$ is a prime number, $\ell = \frac{r^t - 1}{r - 1}$ for some $t \in \mathbb{N}$ satisfy that $m \le ks(r - 1)\ell \le \text{poly}(N)$
\item mean parameter $\mu \in (0, 1)$ as in Figure \ref{fig:isbm-reduction}
\end{itemize}

\begin{enumerate}
\item \textit{Symmetrize and Plant Diagonals}: Compute $M_{\text{PD1}} \in \{0, 1\}^{m \times m}$ and $F$ as in Step 1 of Figure \ref{fig:isbm-reduction}.
\item \textit{Pad and Further Partition}: Form $M_{\text{PD2}}$ and $F'$ as in Step 2 of Figure \ref{fig:isbm-reduction} modified so that $M_{\text{PD2}}$ is a $ks(r-1)\ell \times ks(r-1)\ell$ matrix and each $F'_i$ has size $s(r-1)\ell$. Let $F^s$ be the partition of $[ks(r - 1)\ell]$ into $ks$ parts of size $(r - 1)\ell$ by refining $F'$ by splitting each of its parts into $s$ parts of equal size arbitrarily.
\item \textit{Bernoulli Rotations}: Let $F^o$ be a partition of $[ksr^t]$ into $ks$ equally sized parts. Now compute the matrix $M_{\text{R}} \in \mathbb{R}^{ksr^t \times ksr^t}$ as follows:
\begin{enumerate}
\item[(1)] For each $i, j \in [ks]$, flatten the $(r-1)\ell \times (r-1)\ell$ submatrix $(M_{\text{P}})_{F_i^s, F_j^s}$ into a vector $V_{ij} \in \mathbb{R}^{(r-1)^2 \ell^2}$ and let $A = M_{r, t}^\top \in \mathbb{R}^{r^{2t} \times (r-1)^2 \ell^2}$ as in Definition \ref{defn:unfolded-Trt}.
\item[(2)] Apply $\pr{Bern-Rotations}$ to $V_{ij}$ with matrix $A$, rejection kernel parameter $R_{\pr{rk}} = ksr^t$, Bernoulli probabilities $0 < Q < p \le 1$, output dimension $r^{2t}$, $\lambda = \sqrt{1 + (r - 1)^{-1}}$ and mean parameter $\mu$.
\item[(3)] Set the entries of $(M_{\text{R}})_{F^o_i, F^o_j}$ to be the entries of the output in (2) unflattened into a matrix.
\end{enumerate}
\item \textit{Permute and Output}: Output the matrix $M_{\text{R}}$ with its rows and columns independently permuted uniformly at random.
\end{enumerate}
\vspace{0.5mm}

\end{algbox}
\caption{Reduction from $k$-partite planted dense subgraph to gaussian hidden partition models.}
\label{fig:sbmtesting}
\end{figure}

The rest of this section is devoted to giving our main reduction $k$\textsc{-pds-to-ghpm} showing a computational barrier at $\gamma^2 = \tilde{o}(n/rK^2)$. This reduction is shown in Figure \ref{fig:sbmtesting} and its approximate Markov transition guarantees are stated in the theorem below. The intuition behind why our reduction is tight to the algorithm $s_C$ is as follows. Bernoulli rotations are approximately $\ell_2$-norm preserving in the signal to noise ratio if the output dimension is comparable to the input dimension with $m \asymp n$. Much of the effort in constructing $T_{r, t}$ and $M_{r, t}$ in Section \ref{subsec:2-design-tensors} was devoted to the linear functions $L$ which are crucial in designing $M_{r, t}$ to be nearly square and hence achieve $m \asymp n$ in Bernoulli rotations. Any reduction that is approximately $\ell_2$-norm preserving in the signal to noise ratio will be tight to a variance test such as $s_C$.

The key to the reduction $k$\textsc{-pds-to-ghpm} lies in the construction of $T_{r, t}$ and $M_{r, t}$ in Section \ref{subsec:2-design-tensors}. The rest of the proof of the following theorem is similar to the proofs in the previous section. We omit details that are similar for brevity. We recall from Section \ref{subsec:2-notation} that, given a matrix $M \in \mathbb{R}^{n \times n}$, the matrix $M_{S, T} \in \mathbb{R}^{k \times k}$ where $S, T$ are $k$-subsets of $[n]$ refers to the minor of $M$ restricted to the row indices in $S$ and column indices in $T$. Furthermore, $(M_{S, T})_{i, j} = M_{\sigma_S(i), \sigma_T(j)}$ where $\sigma_S : [k] \to S$ is the unique order-preserving bijection and $\sigma_T$ is analogously defined.

\begin{theorem}[Reduction to $\pr{ghpm}$] \label{thm:ghpm}
Let $N$ be a parameter and $r = r(N) \ge 2$ be a prime number. Fix initial and target parameters as follows:
\begin{itemize}
\item \textnormal{Initial} $k\pr{-bpds}$ \textnormal{Parameters:} $k, N, p, q$ and $E$ as in Theorem \ref{thm:isbm}.
\item \textnormal{Target} $\pr{ghpm}$ \textnormal{Parameters:} $(n, r, K, \gamma)$ where $n = ksr^t$, $K = kr^{t - 1}$ and $\ell = \frac{r^t - 1}{r - 1}$ for some parameters $t = t(N), s = s(N) \in \mathbb{N}$ satisfying that that
$$m \le ks(r - 1)\ell \le \textnormal{poly}(N)$$
where $m$ and $Q$ are as in Theorem \ref{thm:ghpm}. The target level of signal $\gamma$ is given by $\gamma = \frac{\mu(r - 1)}{r^t\sqrt{r}}$ where
$$\mu \le \frac{1}{2 \sqrt{6\log (ksr^t) + 2\log (p - Q)^{-1}}} \cdot \min \left\{ \log \left( \frac{p}{Q} \right), \log \left( \frac{1 - Q}{1 - p} \right) \right\}$$
\end{itemize}
Let $\mathcal{A}(G)$ denote $k$\textsc{-pds-to-ghpm} applied to the graph $G$ with these parameters. Then $\mathcal{A}$ runs in $\textnormal{poly}(N)$ time and it follows that
\begin{align*}
\TV\left( \mathcal{A}\left( \mG_E(N, k, p, q) \right), \, \pr{ghpm}_D(n, r, K, \gamma) \right) &= O\left( \frac{k}{\sqrt{N}} + e^{-\Omega(N^2/km)} + (ksr^t)^{-1} \right) \\
\TV\left( \mathcal{A}\left( \mG(N, q) \right), \, \mN(0, 1)^{\otimes n \times n} \right) &= O\left( e^{-\Omega(N^2/km)} + (ksr^t)^{-1} \right)
\end{align*}
\end{theorem}

In order to state the approximate Markov transition guarantees of the Bernoulli rotations step of $k$\textsc{-pds-to-ghpm}, we need the formalism from Section \ref{subsec:2-design-tensors} to describe the matrix $M_{r, t}$, tensor $T_{r, t}$ and their community alignment properties. While this will require a plethora of cumbersome notation, the goal of the ensuing discussion is simple -- we will show that Lemma \ref{lem:comm-align-tensors} guarantees that stitching together the individual applications of $\pr{Bern-Rotations}$ in Step 3 of $k$\textsc{-pds-to-ghpm} yields a valid instance of $\pr{ghpm}$.

Recall $\mathcal{C}(M^{1, 1}, M^{1, 2}, \dots, M^{ks, ks})$ denotes the concatenation of $k^2s^2$ matrices $M^{i, j} \in \mathbb{R}^{r^t \times r^t}$ into a $ksr^t \times ksr^t$ matrix, as introduced in Section \ref{subsec:2-design-tensors}. Given a partition $F$ of $[ksr^t]$ into $ks$ equally sized parts, let $\mathcal{C}_{F}(M^{1, 1}, M^{1, 2}, \dots, M^{ks, ks})$ denote the concatenation of the $M^{i, j}$, where now the entries of $M^{i, j}$ appear in $\mathcal{C}_{F}$ on the index set $F_i \times F_j$. For consistency, we fix a canonical embedding of the row and column indices of $\mathbb{R}^{r^t \times r^t}$ to $F_i \times F_j$ by always preserving the order of indices.

Let $F^o$ and $F^s$ be fixed partitions of $[ksr^t]$ and $[ks(r - 1)\ell]$ into $k$ parts of size $r^t$ and $(r - 1)\ell$, respectively, and let $S \subseteq [ks(r - 1)\ell]$ be such that $|S| = k$ and $S$ intersects each part of $F^s$ in at most one element. Now let $\mathbf{M}_{S, F^s, F^o}(T_{r, t}) \in \mathbb{R}^{ksr^t \times ksr^t}$ be the matrix
$$\mathbf{M}_{S, F^s, F^o}(T_{r, t}) = \mathcal{C}_{F^o}\left(M^{1, 1}, M^{1, 2}, \dots, M^{ks, ks}\right) \quad \text{where} \quad M^{i, j} = \left\{ \begin{array}{ll} T_{r, t}^{(V_{t_i}, V_{t_j}, L_{ij})} &\text{if } S \cap F_i^s \neq \emptyset \\ 0 &\text{otherwise} \end{array} \right.$$
where $t_i, t_j$ and $L_{ij}$ are given by:
\begin{itemize}
\item let $\sigma : [ks(r - 1)\ell] \to [ks(r - 1)\ell]$ be the unique bijection transforming the partition $F^s$ to the canonical contiguous partition $\{1, \dots, (r - 1)\ell\} \cup \cdots \cup \{(ks - 1)(r - 1)\ell + 1, \dots, ks(r -1)\ell\}$ while preserving ordering on each part $F^s_i$ for $1 \le i \le ks$;
\item let $s'_i$ be the unique element in $\sigma(S \cap F^s_i)$ for each $i$ for which this intersection is nonempty, and let $s_i$ be the unique positive integer with $1 \le s_i \le (r - 1)\ell$ and $s_i \equiv s_i' \pmod{(r - 1)\ell}$; and
\item $t_i, t_j$ and $L_{ij}$ are as in Lemma \ref{lem:comm-align-tensors} given these $s_i$ i.e. $t_i$ and $t_j$ are the unique $1 \le t_i, t_j \le \ell$ such that $t_i \equiv s_i \pmod{\ell}$ and $t_j \equiv s_j \pmod{\ell}$ and $L_{ij} : \mathbb{F}_r \to \mathbb{F}_r$ is given by $L_{ij}(x) = a_i x + a_j$ where $a_i = \lceil s_i/\ell \rceil$ and $a_j = \lceil s_j/\ell \rceil$.
\end{itemize}
The next lemma makes explicit the implications of Lemma \ref{lem:bern-rotations} and Lemma \ref{lem:comm-align-tensors} for the approximate Markov transition guarantees of Step 3 in $k$\textsc{-pds-to-ghpm}. The proof follows a similar structure to the proof of Lemma \ref{lem:isbm-rotations} and we omit identical details.

\begin{lemma}[Bernoulli Rotations for $\pr{ghpm}$] \label{lem:ghpm-rotations}
Let $F^o$ and $F^s$ be a fixed partitions of $[ksr^t]$ and $[ks(r - 1)\ell]$ into $k$ parts of size $r^t$ and $(r - 1)\ell$, respectively, and let $S \subseteq [ksr^t]$ be such that $|S| = k$ and $|S \cap F_i^s| \le 1$ for each $1 \le i \le ks$. Let $\mathcal{A}_{\textnormal{3}}$ denote Step 3 of $k\pr{-pds-to-ghpm}$ with input $M_{\textnormal{PD2}}$ and output $M_{\textnormal{R}}$. Suppose that $p, Q$ and $\mu$ are as in Theorem \ref{thm:isbm}, then it follows that
\begin{align*}
&\TV\Big( \mathcal{A}_{\textnormal{3}} \left( \mathcal{M}_{[ks(r - 1)\ell] \times [ks(r - 1)\ell]} \left( S \times S, \textnormal{Bern}(p), \textnormal{Bern}(Q) \right) \right), \\
&\quad \quad \quad \quad \left. \mL\left( \mu \sqrt{\frac{r - 1}{r}} \cdot \mathbf{M}_{S, F^s, F^o}(T_{r, t}) + \mN(0, 1)^{\otimes ksr^t \times ksr^t} \right) \right) = O\left((ksr^t)^{-1}\right) \\
&\TV\left( \mathcal{A}_{\textnormal{3}} \left(\textnormal{Bern}(Q)^{\otimes ks(r - 1)\ell \times ks(r - 1)\ell} \right), \, \mN(0, 1)^{\otimes ksr^t \times ksr^t} \right) = O\left((ksr^t)^{-1}\right)
\end{align*}
and furthermore, for all such subsets $S$, it holds that the matrix $\mathbf{M}_{S, F^s, F^o}(T_{r, t})$ has zero entries other than in a $kr^t \times kr^t$ submatrix, which is also $r$-block as defined in Section \ref{subsec:2-design-tensors}.
\end{lemma}

\begin{proof}
Define $s_i', s_i, t_i$ and $L_{ij}$ as in the preceding discussion for all $i, j$ with $S \cap F_i^s$ and $S \cap F_j^s$ nonempty. Let (1) and (2) denote the following two cases:
\begin{enumerate}
\item $M_{\textnormal{PD2}} \sim \mathcal{M}_{[ks(r - 1)\ell] \times [ks(r - 1)\ell]} \left( S \times S, \textnormal{Bern}(p), \textnormal{Bern}(Q) \right)$; and
\item $M_{\textnormal{PD2}} \sim \textnormal{Bern}(Q)^{\otimes ks(r - 1)\ell \times ks(r - 1)\ell}$.
\end{enumerate}
Now define the matrix $M_{\text{R}}'$ with independent entries such that
$$\left( M_{\text{R}}' \right)_{F_i^s, F_j^s} \sim \left\{ \begin{array}{ll} \mu \sqrt{\frac{r - 1}{r}} \cdot T_{r, t}^{(V_{t_i}, V_{t_j}, L_{ij})} + \mN(0, 1)^{\otimes r^t \times r^t} &\text{if (1) holds, } S \cap F_i^s \neq \emptyset \text{ and } S \cap F_j^s \neq \emptyset \\ \mN(0, 1)^{\otimes r^t \times r^t} &\text{otherwise if (1) holds or if (2) holds} \end{array} \right.$$
for each $1 \le i, j \le ks$. The vectorization and ordering conventions we adopt imply that if $S \cap F_i^s \neq \emptyset$ and $S \cap F_j^s \neq \emptyset$, then the unflattening of the row with index $(s_i - 1) (r - 1)\ell + s_j$ in $M_{r, t}$ is the approximate output mean of $\mathcal{A}_3$ on the minor $( M_{\text{R}} )_{F_i^s, F_j^s}$ when applying Lemma \ref{lem:bern-rotations} under (1). By Definition \ref{defn:unfolded-Trt} and the definitions of $a_i, t_i$ and $L_{ij}$, this unflattened row is exactly the matrix
$$M^{i, j} = T_{r, t}^{(V_{t_i}, V_{t_j}, L_{ij})}$$
Combining this observation with Lemmas \ref{lem:bern-rotations} and \ref{lem:Mrtsv} yields that under both (1) and (2), we have that
$$\TV\left( \left( M_{\text{R}} \right)_{F_i^s, F_j^s}, \left( M_{\text{R}}' \right)_{F_i^s, F_j^s} \right) = O\left( r^{2t} \cdot (ksr^t)^{-3} \right)$$
for all $1 \le i, j \le ks$. Through the same argument as in Lemma \ref{lem:isbm-rotations}, the tensorization property of total variation in Fact \ref{tvfacts} now yields that $\TV\left( \mL(M_{\text{R}}), \mL(M_{\text{R}}') \right) = O\left( (ksr^t)^{-1} \right)$ under both (1) and (2). Now note that the definition of $\mathcal{C}_{F^o}$ implies that
$$M_{\text{R}}' \sim \left\{ \begin{array}{ll} \mu \sqrt{\frac{r - 1}{r}} \cdot \mathbf{M}_{S, F^s, F^o}(T_{r, t}) + \mN(0, 1)^{\otimes ksr^t \times ksr^t} &\text{if (1) holds} \\ \mN(0, 1)^{\otimes ksr^t \times ksr^t} &\text{if (2) holds} \end{array} \right.$$
which completes the proof of the approximate Markov transition guarantees in the lemma statement. Now note that $\mathbf{M}_{S, F^s, F^o}(T_{r, t})$ is zero everywhere other than on the union $U$ of the $F^o_i$ over the $i$ such that $S \cap F^s_i \neq \emptyset$. There are exactly $k$ such $i$ and thus $|U| = kr^t$. Note that $r$-block matrices remain $r$-block matrices under permutations of column and row indices, and therefore Lemma \ref{lem:comm-align-tensors} implies the same conclusion if $\mathcal{C}$ is replaced by $\mathcal{C}_{F^o}$. Applying Lemma \ref{lem:comm-align-tensors} to the submatrix of $\mathbf{M}_{S, F^s, F^o}(T_{r, t})$ restricted to the indices of $U$ now completes the proof of the lemma.
\end{proof}

We now complete the proof of Theorem \ref{thm:ghpm}, again applying Lemma \ref{lem:tvacc} as in the proofs of Theorems \ref{thm:isgmreduction} and \ref{thm:isbm}. In this theorem, we let $\mathcal{U}_n^k(F)$ denote the uniform distribution over subsets $S \subseteq [n]$ of size $k$ intersecting each part of the partition $F$ in at most one element. When $F$ has exactly $k$ parts, this definition recovers the previously defined distribution $\mathcal{U}_n(F)$.

\begin{proof}[Proof of Theorem \ref{thm:ghpm}]
Let the steps of $\mathcal{A}$ to map inputs to outputs as follows
$$(G, E) \xrightarrow{\mathcal{A}_1} (M_{\text{PD1}}, F) \xrightarrow{\mathcal{A}_2} (M_{\text{PD2}}, F^s) \xrightarrow{\mathcal{A}_3} (M_{\text{R}}, F^o) \xrightarrow{\mathcal{A}_{\text{4}}} M_{\text{R}}'$$
where here $M_{\text{R}}'$ denotes the permuted form of $M_{\text{R}}$ after Step 4. Under $H_1$, consider Lemma \ref{lem:tvacc} applied to the following sequence of distributions
\allowdisplaybreaks
\begin{align*}
\mathcal{P}_0 &= \mG_E(N, k, p, q) \\
\mathcal{P}_1 &= \mathcal{M}_{[m] \times [m]}(S \times S, \textnormal{Bern}(p), \textnormal{Bern}(Q)) \quad \text{where } S \sim \mU_m(F) \\
\mathcal{P}_2 &= \mathcal{M}_{[ks(r - 1)\ell] \times [ks(r - 1)\ell]}(S \times S, \textnormal{Bern}(p), \textnormal{Bern}(Q)) \quad \text{where } S \sim \mU_{ks(r - 1)\ell}^k(F^s) \\
\mathcal{P}_3 &=  \mu \sqrt{\frac{r - 1}{r}} \cdot \mathbf{M}_{S, F^s, F^o}(T_{r, t}) + \mN(0, 1)^{\otimes ksr^t \times ksr^t} \quad \text{where } S \sim \mU_{ks(r - 1)\ell}^k(F^s) \\
\mathcal{P}_{\text{4}} &= \pr{ghpm}_D\left(ksr^t, r, kr^{t - 1}, \frac{\mu(r - 1)}{r^t \sqrt{r}} \right)
\end{align*}
Let $C_Q = \max\left\{ \frac{Q}{1 - Q}, \frac{1 - Q}{Q} \right\}$ and consider setting
$$\epsilon_1 = 4k \cdot \exp\left( - \frac{Q^2N^2}{48pkm} \right) + \sqrt{\frac{C_Q k^2}{2m}}, \quad \epsilon_2 = 0, \quad \epsilon_3 = O\left( (ksr^t)^{-1} \right) \quad \text{and} \quad \epsilon_4 = 0$$
As in the proof of Theorem \ref{thm:isbm}, Lemma \ref{lem:submatrix} implies this is a valid choice of $\epsilon_1$ and $\mathcal{A}_2$ is exact so we can take $\epsilon_2 = 0$. The choice of $\epsilon_3$ is valid by applying Lemma \ref{lem:ghpm-rotations} and averaging over $S \sim \mU^k_{ks(r - 1)\ell}(F^s)$ using the conditioning property of total variation in Fact \ref{tvfacts}. Now note that the $kr^t \times kr^t$ $r$-block submatrix of $\mathbf{M}_{S, F^s, F^o}(T_{r, t})$ has entries $\frac{r - 1}{r^t\sqrt{r - 1}}$ and $-\frac{1}{r^t \sqrt{r - 1}}$. Thus the matrix $\mu \sqrt{\frac{r - 1}{r}} \cdot \mathbf{M}_{S, F^s, F^o}(T_{r, t})$ is of the form of the mean matrix $(\mu_{ij})_{1 \le i, j \le ksr^t}$ in Definition \ref{defn:ghpm} for some choice of $C$ and $D$ where $K = kr^{t - 1}$ and
$$\gamma = \mu \sqrt{\frac{r - 1}{r}} \cdot \frac{r - 1}{r^t\sqrt{r - 1}} = \frac{\mu(r - 1)}{r^t \sqrt{r}}$$
This implies that permuting the rows and columns of $\mP_3$ yields $\mP_4$ exactly with $\epsilon_4 = 0$. Applying Lemma \ref{lem:tvacc} now yields the first bound in the theorem statement. Under $H_0$, consider the distributions
$$\mathcal{P}_0 = \mG(N, q), \quad \mathcal{P}_1 = \text{Bern}(Q)^{\otimes m \times m}, \quad \mathcal{P}_2 = \text{Bern}(Q)^{\otimes ks(r - 1)\ell \times ks(r - 1)\ell}, \quad \mathcal{P}_3 = \mathcal{P}_{\text{4}} = \mN(0, 1)^{\otimes ksr^t \times ksr^t}$$
As above, Lemmas \ref{lem:submatrix} and \ref{lem:ghpm-rotations} imply that we can take $\epsilon_1 = 4k \cdot \exp\left( - \frac{Q^2N^2}{48pkm} \right)$ and $\epsilon_2, \epsilon_3$ and $\epsilon_4$ as above. Lemma \ref{lem:tvacc} now yields the second bound in the theorem statement.
\end{proof}

We now append a final post-processing step to the reduction $k$\textsc{-pds-to-ghpm} to map to $\pr{bhpm}$. The proof of the following corollary is similar to that of Corollary \ref{thm:isbm-mod} and is deferred to Appendix \ref{subsec:appendix-3-part-3}.

\begin{corollary}[Reduction from $\pr{ghpm}$ to $\pr{bhpm}$] \label{cor:bhpm}
Let $0 < q < p \le 1$ be constant and let the parameters $k, N, E, r, \ell, n, s$ and $K$ be as in Theorem \ref{thm:ghpm} with the additional condition that $k\sqrt{r} = o(r^{2t})$. Let $\gamma \in (0, 1)$ be such that
$$\gamma \le \frac{c(r - 1)}{r^t \sqrt{r\log(ksr^t)}}$$
for a sufficiently small constant $c > 0$. Suppose that $P_0$ satisfies $\min\{P_0, 1 - P_0 \} = \Omega(1)$. Then there is a $\textnormal{poly}(N)$ time reduction $\mathcal{A}$ from graphs on $N$ vertices to graphs on $n$ vertices satisfying that
\begin{align*}
\TV\left( \mathcal{A}\left( \mG_E(N, k, p, q) \right), \, \pr{bhpm}_D(n, r, K, P_0, \gamma) \right) &= O\left( \frac{k \mu^3\sqrt{r}}{r^{2t}} + \frac{k}{\sqrt{N}} + e^{-\Omega(N^2/km)} + (ksr^t)^{-1} \right) \\
\TV\left( \mathcal{A}\left( \mG(N, q) \right), \, \mG_B(N, N, P_0) \right) &= O\left( e^{-\Omega(N^2/km)} + (ksr^t)^{-1} \right)
\end{align*}
\end{corollary}

Collecting the results of this section, we arrive at the following computational lower bounds for $\pr{ghpm}$ and $\pr{bhpm}$ matching the efficient test $s_C$ in Lemma \ref{lem:ghpm-test}.

\begin{reptheorem}{thm:ghpm-lb} [Lower Bounds for $\pr{ghpm}$ and $\pr{bhpm}$]
Suppose that $r^2 K^2 = \tilde{\omega}(n)$ and $(\lceil r^2 K^2/n \rceil, r)$ satisfies condition \pr{(t)}, suppose $r$ is prime or $r = \omega_n(1)$ and suppose that $P_0 \in (0, 1)$ satisfies $\min\{P_0, 1 - P_0 \} = \Omega_n(1)$. Then the $k\pr{-pc}$ conjecture or $k\pr{-pds}$ conjecture for constant $0 < q < p \le 1$ both imply that there is a computational lower bound for each of $\pr{ghpm}(n, r, K, \gamma)$ for all levels of signal $\gamma^2 = \tilde{o}(n/rK^2)$. This same lower bound also holds for $\pr{bhpm}(n, r, K, P_0, \gamma)$ given the additional condition $n = o(rK^{4/3})$.
\end{reptheorem}

\begin{proof}
The proof of this theorem will follow that of Theorem \ref{thm:isbm-lb} with several modifications. We begin by showing a lower bound for $\pr{ghpm}$. It suffices to show that the reduction $k\pr{-pds-to-ghpm}$ fills out all of the possible growth rates specified by the computational lower bound $\gamma^2 = \tilde{o}(n/rK^2)$ and the other conditions in the theorem statement. Fix a constant pair of probabilities $0 < q < p \le 1$ and any sequence of parameters $(n, r, K, \gamma)$ all of which are implicitly functions of $n$ such that $(\lceil r^2 K^2/n \rceil, r)$ satisfies $\pr{(t)}$ and
$$\gamma^2 \le c \cdot \frac{n}{w' \cdot rK^2 \log n} \quad \text{and} \quad r^2 K^2 \ge w'n$$
for sufficiently large $n$ and $w' = w'(n) = (\log n)^{c}$ for a sufficiently large constant $c > 0$. Now let $w = w(n) \to \infty$ be an arbitrarily slow-growing increasing positive integer-valued function at least satisfying that $w(n) = n^{o(1)}$. As in Theorem \ref{thm:isbm-lb}, we now specify the following parameters which are sufficient to establish the lower bound for $\pr{ghpm}$:
\begin{enumerate}
\item a sequence $(N, k_N)$ such that $k\pr{-pds}(N, k_N, p, q)$ is hard according to Conjecture \ref{conj:hard-conj}; and 
\item a sequence $(n', r', K', \gamma, s, t, \mu)$ with a subsequence that satisfies three conditions: (2.1) the parameters on the subsequence are in the regime of the desired computational lower bound for $\pr{ghpm}$; (2.2) the parameters $(n', r', K', \gamma)$ have the same growth rate as $(n, r, K, \gamma)$ on this subsequence; and (2.3) such that $\pr{ghpm}(n', r', K', \gamma)$ with the parameters on this subsequence can be produced by $k\pr{-pds-to-ghpm}$ with input $k\pr{-pds}(N, k_N, p, q)$ applied with additional parameters $s, t$ and $\mu$.
\end{enumerate}
We choose these parameters as follows:
\begin{itemize}
\item let $r' = r$ be the smallest prime satisfying that $r \le r' \le 2r$, which exists by Bertrand's postulate and can be found in $\text{poly}(n)$ time;
\item let $t$ be such that $(r')^{t}$ is the closest power of $r'$ to $r'K/\sqrt{n}$, let $s = \lceil n/r'K \rceil$ and let $\mu = \frac{\gamma (r')^t \sqrt{r'}}{r' - 1}$;
\item now let $k_N$ be given by
$$k_N = \left\lfloor \frac{1}{2}\left( 1 + \frac{p}{Q} \right)^{-1} w^{-2} \cdot \min\left\{ \frac{K}{(r')^{t - 1}}, \sqrt{n} \right\} \right\rfloor$$
where $Q = 1 - \sqrt{(1 - p)(1 - q)} + \mathbf{1}_{\{ p = 1\}} \left( \sqrt{q} - 1 \right)$; and
\item let $K' = k_N(r')^{t-1}$, let $n' = k_N s (r')^t$ and let $N = wk_N^2$.
\end{itemize}
Now observe that we have the following bounds
\allowdisplaybreaks
\begin{align*}
n' &\asymp k_N s (r')^t \asymp \left( w^{-2} \cdot \min\left\{ 1, \frac{(r')^{t -1}\sqrt{n}}{K} \right\} \right) n \\
K' &\asymp k_N (r')^{t - 1} = \frac{n'}{r's} \asymp \left( w^{-2} \cdot \min\left\{ 1, \frac{(r')^{t -1}\sqrt{n}}{K} \right\} \right) K \\
m &\le 2\left( \frac{p}{Q} + 1 \right) wk_N^2 \le \left( w^{-1} \cdot \min\left\{ \frac{K}{(r')^{t - 1} \sqrt{n}}, 1 \right\} \cdot \frac{K}{(r')^{t - 1} \sqrt{n}} \right) k_N s (r' - 1)\ell \\
k_N s (r' - 1)\ell&\le \text{poly}(N) \\
(r')^2 (K')^2 &\ge \left( \frac{r'K'}{rK} \right)^2 \cdot \frac{n}{n'} \cdot w' n' \\
\mu &= \frac{\gamma (r')^t \sqrt{r'}}{r' - 1} \le \sqrt{\frac{r'}{r}} \cdot \frac{(r')^{t - 1}\sqrt{n}}{K} \cdot \frac{2}{(w')^{1/2} \sqrt{\log n}} \\
\gamma^2 &\lesssim \frac{n'}{w' \cdot r'(K')^2 \log n'} \cdot \frac{r'}{r} \cdot \frac{n'}{n} \cdot \frac{(K')^2}{K^2} \cdot \frac{\log n'}{\log n}
\end{align*}
where $m$ is the smallest multiple of $k_N$ larger $\left( \frac{p}{Q} + 1 \right) N$ and $\ell = \frac{(r')^t - 1}{r' - 1}$. Now observe that as long as $r'K/\sqrt{n} = \tilde{\Theta}((r')^t)$ then: (2.1) the last inequality above on $\gamma^2$ would imply that $(n', r', K', \gamma)$ is in the desired hard regime; (2.2) the pairs of parameters $(n, n')$, $(K, K')$ and $(r, r')$ have the same growth rates since $w = n^{o(1)}$ and either $r' = r$ or $r' = \Theta(r) = \omega(1)$; and (2.3) the third through sixth bounds above imply that taking $c$ large enough yields the conditions needed to apply Corollary \ref{thm:isbm-mod} to yield the desired reduction. By Lemma \ref{lem:propT}, there is an infinite subsequence of the input parameters such that $r'K/\sqrt{n} = \tilde{\Theta}((r')^t)$, which concludes the proof of the lower bound for $\pr{ghpm}$ as in Theorems \ref{thm:rsme-lb} and \ref{thm:isbm-lb}.

The computational lower bound for $\pr{bhpm}$ follows from the same argument applied to $\mathcal{A}$ from Corollary \ref{cor:bhpm} with the following modification. The conditions in the theorem statement for $\pr{bhpm}$ add the initial condition that $rK^{4/3} \ge w'n$. The parameter settings above then imply that $k_N \sqrt{r'} = \tilde{o}((r')^{2t})$ holds on the parameter subsequence with $r'K/\sqrt{n} = \tilde{\Theta}((r')^t)$. The same reasoning above then yields the desired computational lower bound for $\pr{bhpm}$ and completes the proof of the theorem.
\end{proof}

\subsection{Semirandom Single Community Recovery}
\label{sec:semirandom}

In this section, we show that the $k$\pr{-pc} and $k$\pr{-pds} conjectures with constant edge density imply the $\pr{pds}$ Recovery Conjecture under a semirandom adversary in the regime of constant ambient edge density. The $\pr{pds}$ Recovery Conjecture and formulations of semirandom single community recovery here are as they were introduced in Sections \ref{subsec:1-problems-semicr} and \ref{subsec:2-formulations}. Our reduction from $k$\pr{-pds} to \pr{semi-cr} is shown in Figure \ref{fig:semirandreduction}. On a high level, our main observation is that an adversary in $\pr{semi-cr}$ with subgraph size $k$ can simulate the problem of detecting for the presence of a hidden $\pr{isbm}$ instance on a subgraph with $O(k)$ in an $n$-vertex Erd\H{o}s-R\'{e}nyi graph. Furthermore, combining the Bernoulli rotations step with $K_{3, t}$ as in $k\pr{-pds-to-isbm}$ with the partition refinement of $k\pr{-pds-to-ghpm}$ can be shown to map to this detection problem. Furthermore, it faithfully recovers the Kesten-Stigum bound from the $\pr{pds}$ Recovery Conjecture as opposed to the slower detection rate. The key proofs in this section resemble similar proofs in the previous two sections. We omit details that are similar for brevity.

Before proceeding with the main proofs of this section, we discuss the relationship between our results and the reduction of \cite{cai2015computational}. In \cite{cai2015computational}, the authors prove a detection-recovery gap in the context of sub-Gaussian submatrix localization based on the hardness of finding a planted $k$-clique in a random $n/2$-regular graph. This degree-regular formulation of $\pr{pc}$ was previously considered in \cite{deshpande2015finding} and differs in a number of ways from \pr{pc}. For example, it is unclear how to generate a sample from the degree-regular variant in polynomial time. We remark that the reduction of \cite{cai2015computational}, when instead applied the usual formulation of \pr{pc} produces a matrix with highly dependent entries. Specifically, the sum of the entries of the output matrix has variance $n^2/\mu$ where $\mu \ll 1$ is the mean parameter for the submatrix localization instance whereas an output matrix with independent entries of unit variance would have a sum of entries of variance $n^2$. Note that, in general, any reduction beginning with $\pr{pc}$ that also preserves the natural $H_0$ hypothesis cannot show the existence of a detection-recovery gap, as any lower bounds for localization would also apply to detection.

Formally, the goal of this section is to show that the reduction $k\pr{pds-to-semi-cr}$ in Figure \ref{fig:semirandreduction} maps from $k$\pr{-pc} and $k$\pr{-pds} to the following distribution under $H_1$, for a particular choice of $\mu_1, \mu_2$ and $\mu_3$ just below the $\pr{pds}$ Recovery Conjecture. We remark that $k\pr{-pds-to-semi-cr}$ maps to the specific case where $P_0 = 1/2$. This reduction is extended in Corollary \ref{cor:semi-cr-gen} to handle $P_0 \neq 1/2$ with $\min\{P_0, 1 - P_0\} = \Omega(1)$.

\begin{definition}[Target $\pr{semi-cr}$ Instance]
Given positive integers $k, k' \le n$ and $P_0, \mu_1, \mu_2, \mu_3 \in (0, 1)$ satisfying that $\mu_1, \mu_2 \le P_0 \le 1 - \mu_3$, let $\pr{tsi}(n, k, k', P_0, \mu_1, \mu_2, \mu_3)$ be the distribution over $G \in \mG_n$ sampled as follows:
\begin{enumerate}
\item choose two disjoint subsets $S \subseteq [n]$ and $S' \subseteq [n]$ of sizes $|S| = k$ and $|S'| = k'$, respectively, uniformly at random; and
\item include the edge $\{i, j\}$ in $E(G)$ independently with probability $p_{ij}$ where
$$p_{ij} = \left\{ \begin{array}{ll} P_0 &\textnormal{if } (i, j) \in S'^2 \\ P_0 - \mu_1 &\textnormal{if } (i, j) \in [n]^2 \backslash (S \cup S')^2 \\ P_0 - \mu_2 &\textnormal{if } (i, j) \in S \times S' \textnormal{ or } (i, j) \in S' \times S \\ P_0 + \mu_3 &\textnormal{if } (i, j) \in S^2 \end{array} \right.$$
\end{enumerate}
\end{definition}

Note that this distribution can be produced by a semirandom adversary in $\pr{semi-cr}(n, k, P_0 + \mu_3, P_0)$ under $H_1$ as follows:
\begin{enumerate}
\item samples $S'$ of size $k'$ uniformly at random from all $k'$-subsets of $[n] \backslash S$ where $S$ is the vertex set of the planted dense subgraph; and
\item if the edge $\{i, j \}$ is in $E(G)$, remove it from $G$ independently with probability $q_{ij}$ where
$$q_{ij} = \left\{ \begin{array}{ll} 0 &\text{if } (i, j) \in S^2 \cup S'^2 \\ \mu_1/P_0 &\text{if } (i, j) \not \in (S \cup S')^2 \\ \mu_2/P_0 &\text{if } (i, j) \in S \times S' \text{ or } (i, j) \in S' \times S \end{array} \right.$$
\end{enumerate}
Note that $\mG(n, P_0')$ can be produced by the adversary under $H_0$ of $\pr{semi-cr}(n, k, P_0 + \mu_1, P_0)$ as long as $P_0' \le P_0$ by removing all edges independently with probability $1 - P_0'/P_0$. Thus it suffices to map to a testing problem between some $\pr{tsi}(n, k, k', P_0, \mu_1, \mu_2, \mu_3)$ and $\mG(n, P_0')$.

The next theorem establishes our main Markov transition guarantees for the reduction $k\pr{pds-to-semi-cr}$, which map to such a testing problem when $P_0 = 1/2$.

\begin{theorem}[Reduction to $\pr{semi-cr}$] \label{thm:semi-cr-reduction}
Let $N$ be a parameter and fix other parameters as follows:
\begin{itemize}
\item \textnormal{Initial} $k\pr{-bpds}$ \textnormal{Parameters:} $k, N, p, q$ and $E$ as in Theorem \ref{thm:isbm}.
\item \textnormal{Target} $\pr{semi-cr}$ \textnormal{Parameters:} $(n, K, 1/2 + \gamma, 1/2)$ where $n = 3ks \cdot \frac{3^t - 1}{2}$ and $K = (3^t - 1)k$ for some parameters $t = t(N), s = s(N) \in \mathbb{N}$ satisfying that
$$m \le 3^t ks \le n \le \textnormal{poly}(N)$$
where $m$ and $Q$ are as in Theorem \ref{thm:ghpm}. The target level of signal $\gamma$ is given by $\gamma = \Phi\left( \frac{\mu}{3^{t}} \right) - 1/2$ and the target $\pr{tsi}$ densities are 
$$\mu_1 = \Phi\left( \frac{\mu}{3^{t+1}} \right) - \frac{1}{2} \quad \textnormal{and} \quad \mu_2 = \mu_3 = \Phi\left( \frac{\mu}{3^{t}} \right) - \frac{1}{2}$$
where $\mu \in (0, 1)$ satisfies that
$$\mu \le \frac{1}{2 \sqrt{6\log n + 2\log (p - Q)^{-1}}} \cdot \min \left\{ \log \left( \frac{p}{Q} \right), \log \left( \frac{1 - Q}{1 - p} \right) \right\}$$
\end{itemize}
Let $\mathcal{A}(G)$ denote $k$\textsc{-pds-to-semi-cr} applied to the graph $G$ with these parameters. Then $\mathcal{A}$ runs in $\textnormal{poly}(N)$ time and it follows that
\begin{align*}
\TV\left( \mathcal{A}\left( \mG_E(N, k, p, q) \right), \, \pr{tsi}(n, K, K/2, 1/2, \mu_1, \mu_2, \mu_3) \right) &= O\left( \frac{k}{\sqrt{N}} + e^{-\Omega(N^2/km)} + (3^t ks)^{-1} \right) \\
\TV\left( \mathcal{A}\left( \mG(N, q) \right), \, \mG\left(n, 1/2 - \mu_1 \right) \right) &= O\left( e^{-\Omega(N^2/km)} + (3^t ks)^{-1} \right)
\end{align*}
\end{theorem}

\begin{figure}[t!]
\begin{algbox}
\textbf{Algorithm} $k$\textsc{-pds-to-semi-cr}

\vspace{1mm}

\textit{Inputs}: $k$\pr{-pds} instance $G \in \mG_N$ with dense subgraph size $k$ that divides $N$, and the following parameters
\begin{itemize}
\item partition $E$, edge probabilities $0 < q < p \le 1$, $Q \in (0, 1)$ and $m$ as in Figure \ref{fig:isbm-reduction}
\item refinement parameter $s$ and number of vertices $n = 3ks \cdot \frac{3^t - 1}{2}$ for some $t \in \mathbb{N}$ satisfy that $m \le 3^t ks \le n \le \text{poly}(N)$
\item mean parameter $\mu \in (0, 1)$ as in Figure \ref{fig:isbm-reduction}
\end{itemize}

\begin{enumerate}
\item \textit{Symmetrize and Plant Diagonals}: Compute $M_{\text{PD1}} \in \{0, 1\}^{m \times m}$ and $F$ as in Step 1 of Figure \ref{fig:isbm-reduction}.
\item \textit{Pad and Further Partition}: Form $M_{\text{PD2}}$ and $F'$ as in Step 2 of Figure \ref{fig:isbm-reduction} modified so that $M_{\text{PD2}}$ is a $3^t ks \times 3^t ks$ matrix and each $F'_i$ has size $3^t s$. Let $F^s$ be the partition of $[3^t ks]$ into $ks$ parts of size $3^t$ by refining $F'$ by splitting each of its parts into $s$ parts of equal size arbitrarily.
\item \textit{Bernoulli Rotations}: Let $F^o$ be a partition of $[n]$ into $ks$ equally sized parts. Now compute the matrix $M_{\text{R}} \in \mathbb{R}^{n \times n}$ as follows:
\begin{enumerate}
\item[(1)] For each $i, j \in [k]$, apply $\pr{Tensor-Bern-Rotations}$ to the matrix $(M_{\text{P}})_{F_i^s, F_j^s}$ with matrix parameter $A_1 = A_2 = K_{3, t}$, Bernoulli probabilities $0 < Q < p \le 1$, output dimension $\frac{1}{2} (3^{t} - 1)$, $\lambda_1 = \lambda_2 = \sqrt{3/2}$ and mean parameter $\mu$.
\item[(2)] Set the entries of $(M_{\text{R}})_{F^o_i, F^o_j}$ to be the entries in order of the matrix output in (1).
\end{enumerate}
\item \textit{Threshold and Output}: Output the graph generated by Step 4 of Figure \ref{fig:isbm-reduction} modified so that $G'$ has vertex set $[n]$ and $M_{\text{R}}$ is thresholded at $\frac{\mu}{3^{t + 1}}$.
\end{enumerate}
\vspace{0.5mm}

\end{algbox}
\caption{Reduction from $k$-partite planted dense subgraph to semirandom community recovery.}
\label{fig:semirandreduction}
\end{figure}

To prove this theorem, we prove a lemma analyzing the Bernoulli rotations step in Figure \ref{fig:semirandreduction}. The proof of this lemma is similar to those of Lemmas \ref{lem:isbm-rotations} and \ref{lem:ghpm-rotations}. We omit details that are identical. Recall from Section \ref{subsec:3-rsme-reduction} the definition of the vector $v_{S, F^s, F^o}(M) \in \mathbb{R}^{ab}$ where $F^s$ and $F^o$ are partitions of $[ab]$ into $a$ equally sized parts and $S$ is a set intersecting each $F^s_i$ in exactly one element. Here we extend this definition to sets $S$ intersecting each $F^s_i$ in at most one element, by setting
$$\left( v_{S, F^s, F^o}(M) \right)_{F_i^o} = \left\{ \begin{array}{ll} M_{\cdot, S \cap F_i^s} &\text{if } S \cap F_i \neq \emptyset \\ 0 &\text{if } S \cap F_i = \emptyset \end{array} \right.$$
for each $1 \le i \le a$. We now can state the approximate Markov transition guarantees for the Bernoulli rotations step of $k\pr{-pds-to-semi-cr}$ in this notation.

\begin{lemma}[Bernoulli Rotations for $\pr{semi-cr}$] \label{lem:rotthres}
Let $F^s$ and $F^o$ be a fixed partitions of $[3^t ks]$ and $[n]$ into $ks$ parts of size $3^t$ and $\frac{1}{2}(3^t - 1)$, respectively, and let $S \subseteq [3^t ks]$ where $|S| = k$ and $|S \cap F_i^s| \le 1$ for each $1 \le i \le ks$. Let $\mathcal{A}_{\textnormal{3}}$ denote Step 3 of $k\pr{-pds-to-semi-cr}$ with input $M_{\textnormal{PD2}}$ and output $M_{\textnormal{R}}$. Suppose that $p, Q$ and $\mu$ are as in Theorem \ref{thm:semi-cr-reduction}, then it follows that
\begin{align*}
&\TV\Big( \mathcal{A}_{\textnormal{3}} \left( \mathcal{M}_{[3^t ks] \times [3^t ks]} \left( S \times S, \textnormal{Bern}(p), \textnormal{Bern}(Q) \right) \right), \\
&\quad \quad \quad \quad \left. \mL\left( \frac{2\mu}{3} \cdot v_{S, F^s, F^o}(K_{3, t}) v_{S, F^s, F^o}(K_{3, t})^\top + \mN(0, 1)^{\otimes n \times n} \right) \right) = O\left((3^t k s)^{-1}\right) \\
&\TV\left( \mathcal{A}_{\textnormal{3}} \left(\textnormal{Bern}(Q)^{\otimes 3^t ks \times 3^t ks} \right), \, \mN(0, 1)^{\otimes n \times n} \right) = O\left((3^t k s)^{-1}\right)
\end{align*}
\end{lemma}

\begin{proof}
Let (1) and (2) denote the following two cases:
\begin{enumerate}
\item $M_{\textnormal{PD2}} \sim \mathcal{M}_{[3^t ks] \times [3^t ks]} \left( S \times S, \textnormal{Bern}(p), \textnormal{Bern}(Q) \right)$; and
\item $M_{\textnormal{PD2}} \sim \textnormal{Bern}(Q)^{\otimes 3^t ks \times 3^t ks}$.
\end{enumerate}
Now define the matrix $M_{\text{R}}'$ with independent entries such that
$$M_{\text{R}}' \sim \left\{ \begin{array}{ll} \frac{2\mu}{3} \cdot v_{S, F^s, F^o}(K_{3, t}) v_{S, F^s, F^o}(K_{3, t})^\top + \mN(0, 1)^{\otimes n \times n} &\text{if (1) holds} \\ \mN(0, 1)^{\otimes n \times n} &\text{if (2) holds} \end{array} \right.$$
Similarly to Lemma \ref{lem:ghpm-rotations}, Lemmas \ref{lem:bern-rotations} and \ref{lem:Krtsv} yields that under both (1) and (2), we have that
$$\TV\left( \left( M_{\text{R}} \right)_{F_i^s, F_j^s}, \left( M_{\text{R}}' \right)_{F_i^s, F_j^s} \right) = O\left( 3^{2t} \cdot (3^t ks)^{-3} \right)$$
for all $1 \le i, j \le ks$. The tensorization property of total variation in Fact \ref{tvfacts} now yields that
$$\TV\left( \mL(M_{\text{R}}), \mL(M_{\text{R}}') \right) = O\left( (3^t ks)^{-1} \right)$$
under both (1) and (2), proving the lemma.
\end{proof}

We now complete the proof of Theorem \ref{thm:semi-cr-reduction}, which follows a similar structure as in Theorem \ref{thm:isbm}.

\begin{proof}[Proof of Theorem \ref{thm:semi-cr-reduction}]
Let the steps of $\mathcal{A}$ to map inputs to outputs as follows
$$(G, E) \xrightarrow{\mathcal{A}_1} (M_{\text{PD1}}, F) \xrightarrow{\mathcal{A}_2} (M_{\text{PD2}}, F^s) \xrightarrow{\mathcal{A}_3} (M_{\text{R}}, F^o) \xrightarrow{\mathcal{A}_{\text{4}}} G'$$
Under $H_1$, consider Lemma \ref{lem:tvacc} applied to the following sequence of distributions
\allowdisplaybreaks
\begin{align*}
\mathcal{P}_0 &= \mG_E(N, k, p, q) \\
\mathcal{P}_1 &= \mathcal{M}_{[m] \times [m]}(S \times S, \textnormal{Bern}(p), \textnormal{Bern}(Q)) \quad \text{where } S \sim \mU_m(F) \\
\mathcal{P}_2 &= \mathcal{M}_{[3^t ks] \times [3^t ks]}(S \times S, \textnormal{Bern}(p), \textnormal{Bern}(Q)) \quad \text{where } S \sim \mU_{3^t ks}^k(F^s) \\
\mathcal{P}_3 &= \frac{2\mu}{3} \cdot v_{S, F^s, F^o}(K_{3, t}) v_{S, F^s, F^o}(K_{3, t})^\top + \mN(0, 1)^{\otimes n \times n} \quad \text{where } S \sim \mU_{3^t ks}^k(F^s) \\
\mathcal{P}_{\text{4}} &= \pr{tsi}(n, K, K/2, 1/2, \mu_1, \mu_2, \mu_3)
\end{align*}
Let $C_Q = \max\left\{ \frac{Q}{1 - Q}, \frac{1 - Q}{Q} \right\}$ and consider setting
$$\epsilon_1 = 4k \cdot \exp\left( - \frac{Q^2N^2}{48pkm} \right) + \sqrt{\frac{C_Q k^2}{2m}}, \quad \epsilon_2 = 0, \quad \epsilon_3 = O\left( (3^t ks)^{-1} \right) \quad \text{and} \quad \epsilon_4 = 0$$
Lemma \ref{lem:submatrix} implies this is a valid choice of $\epsilon_1$, $\mathcal{A}_2$ is exact so we can take $\epsilon_2 = 0$ and $\epsilon_3$ is valid by applying Lemma \ref{lem:rotthres} and averaging over $S \sim  \mU_{3^t ks}^k(F^s)$ using the conditioning property of total variation in Fact \ref{tvfacts}. Now note that for each $S$ the definition of $v_{S, F^s, F^o}(K_{3, t})$ implies that there are sets $S_1$ and $S_2$ with $|S_1| = (3^t - 1)k$ and $|S_2| = \frac{3^t - 1}{2} \cdot k$ such that
$$\left( \frac{2\mu}{3} \cdot v_{S, F^s, F^o}(K_{3, t}) v_{S, F^s, F^o}(K_{3, t})^\top \right)_{ij} = \frac{\mu}{3^{t + 1}} + \left\{ \begin{array}{ll} \mu/3^t &\text{if } i, j \in S_1 \\ -\mu/3^t &\text{if } (i, j) \in S_1 \times S_2 \text{ or } (i, j) \in S_2 \times S_1 \\ 0 &\text{if } i, j \in S_2 \\ -\mu/3^{t + 1} &\text{if } i, j \not \in (S_1 \cup S_2) \end{array} \right.$$
for each $1 \le i, j \le n$. Permuting the rows and columns of $\mP_3$ therefore yields $\mP_4$ exactly with $\epsilon_4 = 0$. Lemma \ref{lem:tvacc} thus establishes the first bound. Under $H_0$, consider the distributions
\begin{align*}
&\mathcal{P}_0 = \mG(N, q), \quad \mathcal{P}_1 = \text{Bern}(Q)^{\otimes m \times m}, \quad \mathcal{P}_2 = \text{Bern}(Q)^{\otimes 3^t ks \times 3^t ks}, \\
&\mathcal{P}_3 = \mN(0, 1)^{\otimes n \times n} \quad \text{and} \quad \mP_4 = \mG\left(n, 1/2 - \mu_1 \right)
\end{align*}
As in Theorems \ref{thm:isbm} and \ref{thm:ghpm}, Lemmas \ref{lem:submatrix} and \ref{lem:rotthres} imply $\epsilon_1 = 4k \cdot \exp\left( - \frac{Q^2N^2}{48pkm} \right)$ and the choices of $\epsilon_2, \epsilon_3$ and $\epsilon_4$ above are valid. Lemma \ref{lem:tvacc} now yields the second bound and completes the proof of the theorem.
\end{proof}

We now add a simple final step to $k\pr{pds-to-semi-cr}$, reducing to arbitrary $P_0 \neq 1/2$. The guarantees for this modified reduction are captured in the following corollary.

\begin{corollary}[Arbitrary Bounded $P_0$] \label{cor:semi-cr-gen}
Define all parameters as in Theorem \ref{thm:semi-cr-reduction} and let $P_0 \in (0, 1)$ be such that $\eta = \min \{P_0, 1 - P_0\} = \Omega(1)$. Then there is a $\textnormal{poly}(N)$ time reduction $\mathcal{A}$ from graphs on $N$ vertices to graphs on $n$ vertices satisfying that
\begin{align*}
\TV\left( \mathcal{A}\left( \mG_E(N, k, p, q) \right), \, \pr{tsi}(n, K, K/2, P_0, 2\eta\mu_1, 2\eta\mu_2, 2\eta\mu_3) \right) &= O\left( \frac{k}{\sqrt{N}} + e^{-\Omega(N^2/km)} + (3^t ks)^{-1} \right) \\
\TV\left( \mathcal{A}\left( \mG(N, q) \right), \, \mG\left(n, P_0 - 2\eta \mu_1 \right) \right) &= O\left( e^{-\Omega(N^2/km)} + (3^t ks)^{-1} \right)
\end{align*}
\end{corollary}

\begin{proof}
This corollary follows from the same reduction in the first part of the proof of Corollary \ref{thm:isbm-mod}. Consider the reduction $\mathcal{A}$ that adds a simple post-processing step to $k$\textsc{-pds-to-semi-cr} as follows. On input graph $G$ with $N$ vertices:
\begin{enumerate}
\item Form the graph $G_1$ by applying $k$\textsc{-pds-to-cr} to $G$ with parameters $N, k, E, \ell, n, s, t$ and $\mu$.
\item Form $G_2$ as in $\mathcal{A}_2$ of Corollary \ref{thm:isbm-mod}.
\end{enumerate}
This clearly runs in $\text{poly}(N)$ time and the second step can be verified to map $\pr{tsi}(n, K, K/2, 1/2, \mu_1, \mu_2, \mu_3)$ to $\pr{tsi}(n, K, K/2, P_0, 2\eta\mu_1, 2\eta\mu_2, 2\eta\mu_3)$ and $\mG\left(n, 1/2 - \mu_1 \right)$ to $\mG\left(n, P_0 - 2\eta \mu_1 \right)$ exactly. Applying Theorem \ref{thm:semi-cr-reduction} and Lemma \ref{lem:tvacc} to each of these two steps proves the bounds in the corollary statement.
\end{proof}

Summarizing the results of this section, we arrive at the desired computational lower bound for $\pr{semi-cr}$. The proof of the next theorem follows the usual recipe for deducing computational lower bounds and is deferred to Appendix \ref{subsec:appendix-3-part-3}.

\begin{reptheorem}{thm:semi-cr-lb} [Lower Bounds for $\pr{semi-cr}$]
If $k$ and $n$ are polynomial in each other with $k = \Omega(\sqrt{n})$ and $0 < P_0 < P_1 \le 1$ where $\min\{P_0, 1 - P_0 \} = \Omega(1)$, then the $k\pr{-pc}$ conjecture or $k\pr{-pds}$ conjecture for constant $0 < q < p \le 1$ both imply that there is a computational lower bound for $\pr{semi-cr}(n, k, P_1, P_0)$ at $\frac{(P_1 - P_0)^2}{P_0(1 - P_0)} = \tilde{o}(n/k^2)$.
\end{reptheorem}

\section{Tensor Principal Component Analysis}
\label{sec:3-tensor}


\begin{figure}[t!]
\begin{algbox}
\textbf{Algorithm} $k$\textsc{-pst-to-tpca}

\vspace{1mm}

\textit{Inputs}: $k$\pr{-pst} instance $T \in \{0, 1\}^{N^{\otimes s}}$ of order $s$ with planted sub-tensor size $k$ that divides $N$, and the following parameters
\begin{itemize}
\item partition $F$ of $[N]$ into $k$ parts of size $N/k$ and edge probabilities $0 < q < p \le 1$
\item output dimension $n$ and a parameter $t \in \mathbb{N}$ satisfying that
$$n \le D = 2k(2^t - 1), \quad N \le 2^t k \quad \text{and} \quad t = O(\log N)$$
\item target level of signal $\theta \in (0, 1)$ where
$$\theta \le \frac{c \cdot \delta}{2^{st/2} \cdot \sqrt{t + \log (p - q)^{-1}}}$$
for a sufficiently small constant $c > 0$, where $\delta = \min \left\{ \log \left( \frac{p}{q} \right), \log \left( \frac{1 - q}{1 - p} \right) \right\}$.
\end{itemize}

\begin{enumerate}
\item \textit{Pad}: Form $T_{\text{PD}} \in \{0, 1\}^{2^t k \times 2^t k}$ by embedding $T$ as the upper left principal sub-tensor of $T_{\text{PD}}$ and then adding $2^t k - N$ new indices along each axis of $T$ and filling all missing entries with i.i.d. samples from $\text{Bern}(q)$. Let $F'_i$ be $F_i$ with $2^t - N/k$ of the new indices. Sample $k$ random permutations $\sigma_i$ of $F_i'$ independently for each $1 \le i \le k$ and permute the indices along each axis of $T_{\text{PD}}$ within each part $F'_i$ according to $\sigma_i$.
\item \textit{Bernoulli Rotations}: Let $F''$ be a partition of $[D]$ into $k$ equally sized parts. Now compute the matrix $T_{\text{R}} \in \mathbb{R}^{K^{\otimes s}}$ as follows:
\begin{enumerate}
\item[(1)] For each block index $(i_1, i_2, \dots, i_s) \in [k]$, apply $\pr{Tensor-Bern-Rotations}$ to the tensor $(T_{\text{PD}})_{F_{i_1}', F_{i_2}', \dots, F_{i_s}'}$ with matrix parameters $A_1 = A_2 = \cdots = A_s = K_{2, t}$, rejection kernel parameter $R_{\pr{rk}} = (2^t k)^s$, Bernoulli probabilities $0 < Q < p \le 1$, output dimension $D/k = 2(2^t - 1)$, singular value upper bounds $\lambda_1 = \lambda_2 = \cdots = \lambda_s = \sqrt{2}$ and mean parameter $\mu = \theta \cdot 2^{s(t+1)/2}$.
\item[(2)] Set the entries of $(T_{\text{R}})_{F_{i_1}'', F_{i_2}'', \dots, F_{i_s}''}$ to be the entries in order of the tensor output in (1).
\end{enumerate}
\item \textit{Subsample, Sign and Output}: Randomly choose a subset $U \subseteq [D]$ of size $|U| = n$ and randomly sample a vector $b \sim \text{Unif}\left[ \{-1, 1\}\right]^{\otimes D}$ output the tensor $b^{\otimes s} \odot T_{\text{R}}$ restricted to the indices in $U$, or in other words $\left(b^{\otimes s} \odot T_{\text{R}}\right)_{U, U, \dots, U}$, where $\odot$ denotes the entrywise product of two tensors.
\end{enumerate}
\vspace{0.5mm}

\end{algbox}
\caption{Reduction from $k$-partite Bernoulli planted sub-tensor to tensor PCA.}
\label{fig:tpca-reduction}
\end{figure}

In this section, we: (1) give our reduction $k$\textsc{-pst-to-tpca} from $k$-partite planted sub-tensor to tensor PCA; (2) combine this with the completing hypergraphs technique of Section \ref{sec:2-hypergraph-planting} to prove our main computational lower bound for the hypothesis testing formulation of tensor PCA, Theorem \ref{thm:tpca-lb}; and (3) we show that Theorem \ref{thm:tpca-lb} implies computational lower bounds for the recovery formulation of tensor PCA. We remark that the heuristic at the end of Section \ref{subsec:1-tech-design-matrices} yields the predicted computational barrier for $\pr{tpca}$. Specifically, the $\ell_2$ norm for the data tensor $\bE[X]$ corresponding to $k\pr{-hpc}^s$ is $\Theta(k^{s/2})$ which is $\tilde{\Theta}(n^{s/4})$ just below the conjectured computational barrier for $k\pr{-hpc}^s$. Furthermore, the corresponding $\ell_2$ norm for $H_1$ of $\pr{tpca}^s$ is $\tilde{\Theta}(\theta n^{s/2})$. Equating these norms correctly predicts the computational barrier of $\theta = \tilde{\Theta}(n^{-s/4})$.

Our reduction $k$\textsc{-pst-to-tpca} is shown in Figure \ref{fig:tpca-reduction}, which applies dense Bernoulli rotations with Kronecker products of the matrices $K_{2, t}$ to the planted sub-tensor problem. The following theorem establishes the approximate Markov transition properties of this reduction. Its proof is similar to the proofs of Theorems \ref{thm:isgmreduction} and \ref{thm:isbm}. We omit details that are similar for brevity.

\begin{theorem}[Reduction to Tensor PCA] \label{thm:tpca}
Fix initial and target parameters as follows:
\begin{itemize}
\item \textnormal{Initial} $k\pr{-pst}$ \textnormal{Parameters:} dimension $N$, sub-tensor size $k$ that divides $N$, order $s$, a partition $F$ of $[N]$ into $k$ parts of size $N/k$ and edge probabilities $0 < q < p \le 1$ where $\min\{q, 1 - q\} = \Omega_N(1)$.
\item \textnormal{Target} $\pr{tpca}$ \textnormal{ Parameters:} dimension $n$ and a parameter $t = t(N) \in \mathbb{N}$ satisfying that
$$n \le D = 2k(2^t - 1), \quad N \le 2^t k \quad \text{and} \quad t = O(\log N)$$
and target level of signal $\theta \in (0, 1)$ where
$$\theta \le \frac{c \cdot \delta}{2^{st/2} \cdot \sqrt{t + \log (p - q)^{-1}}}$$
for a sufficiently small constant $c > 0$, where $\delta = \min \left\{ \log \left( \frac{p}{q} \right), \log \left( \frac{1 - q}{1 - p} \right) \right\}$.
\end{itemize}
Let $\mathcal{A}(T)$ denote $k$\textsc{-pst-to-tpca} applied to the tensor $T$ with these parameters. Then $\mathcal{A}$ runs in $\textnormal{poly}(N)$ time and it follows that
\begin{align*}
\TV\left( \mathcal{A}\left( \mathcal{M}_{[N]^s}\left( S^{s}, \textnormal{Bern}(p), \textnormal{Bern}(q) \right) \right), \, \pr{tpca}^s_D(n, \theta) \right) &= O\left( k^{-2s} 2^{-2st} \right) \\
\TV\left( \mathcal{A}\left( \mathcal{M}_{[N]^s}\left( \textnormal{Bern}(q) \right) \right), \, \mN(0, 1)^{\otimes n^{\otimes s}} \right) &= O\left( k^{-2s} 2^{-2st} \right)
\end{align*}
for any set $S \subseteq [N]$ with $|S \cap E_i| = 1$ for each $1 \le i \le k$.
\end{theorem}

We now prove two lemmas stating the guarantees for the dense Bernoulli rotations step and final step of $k$\textsc{-pst-to-tpca}. Define $v_{S, F', F''}(M)$ as in Section \ref{subsec:3-rsme-reduction}. Note that the matrix $K_{2,t}$ has dimensions $2(2^t - 1) \times 2^t$. The proof of the next lemma follows from the same argument as in the proof of Lemma \ref{lem:isgm-rotations}.

\begin{lemma}[Bernoulli Rotations for $\pr{tpca}$] \label{lem:tpca-rotations}
Let $F'$ and $F''$ be fixed partitions of $[2^t k]$ and $[D]$ into $k$ parts of size $2^t$ and $2(2^t - 1)$, respectively, and let $S \subseteq [2^t k]$ where $|S \cap F_i'| = 1$ for each $1 \le i \le k$. Let $\mathcal{A}_{\textnormal{2}}$ denote Step 2 of $k$\textsc{-pst-to-tpca} with input $T_{\textnormal{PD}}$ and output $T_{\textnormal{R}}$. Suppose that $p, q$ and $\theta$ are as in Theorem \ref{thm:tpca}, then it follows that
\begin{align*}
&\TV\Big( \mathcal{A}_{\textnormal{2}} \left( \mathcal{M}_{[2^t k]^s} \left( S^s, \textnormal{Bern}(p), \textnormal{Bern}(q) \right) \right), \\
&\quad \quad \quad \quad \left. \mL\left( 2^{st/2} \theta \cdot v_{S, F', F''}(K_{2, t})^{\otimes s} + \mN(0, 1)^{\otimes D^{\otimes s}} \right) \right) = O\left( k^{-2s} 2^{-2st} \right) \\
&\TV\left( \mathcal{A}_{\textnormal{2}} \left( \mathcal{M}_{[2^t k]^s} \left( \textnormal{Bern}(q) \right) \right), \, \mN(0, 1)^{\otimes D^{\otimes s}} \right) = O\left( k^{-2s} 2^{-2st} \right) 
\end{align*}
\end{lemma}

\begin{proof}
This lemma follows from the same argument as in the proof of Lemma \ref{lem:isgm-rotations}. We outline the details that differ. Specifically, consider the case in which $T_{\textnormal{PD}} \sim \mathcal{M}_{[2^t k]^s} \left( S^s, \textnormal{Bern}(p), \textnormal{Bern}(q) \right)$. Observe that
$$(T_{\textnormal{PD2}})_{F'_{i_1}, F'_{i_2}, \dots, F'_{i_s}} \sim \pr{pb}\left(F_{i_1}' \times F'_{i_2} \times \cdots \times F_{i_s}', (S \cap F_{i_1}', S \cap F_{i_2}', \dots, S \cap F_{i_s}'), p, q\right)$$
for all $(i_1, i_2, \dots, i_s) \in [k]^s$.
The singular value upper bound on $K_{2, t}$ in Lemma \ref{lem:Krtsv} and the same application of Corollary \ref{cor:tensor-bern-rotations} as in Lemma \ref{lem:isgm-rotations} yields that
$$\TV\left( (T_{\textnormal{R}})_{F''_{i_1}, \dots, F''_{i_s}}, \, \mL\left( 2^{-s/2} \mu \cdot (K_{2, t})_{\cdot, S \cap F_{i_1}'} \otimes \cdots \otimes (K_{2, t})_{\cdot, S \cap F_{i_s}'} + \mN(0, 1)^{\otimes (D/k)^{\otimes s}} \right) \right) = O\left( k^{-3s} 2^{-2st} \right)$$
for all $(i_1, i_2, \dots, i_s) \in [k]^s$ since $\prod_{j = 1}^s \lambda_j = 2^{s/2}$. Note that the exponent of $8$ is guaranteed by changing the parameter in Gaussian rejection kernels from $n$ to $n^{10}$ to decrease their total variation error. Note that this step still runs in $\text{poly}(n^{10})$ time. Combining this bound for all such $(i_1, i_2, \dots, i_s)$ and the tensorization property of total variation in Fact \ref{tvfacts} yields that
$$\TV\left( T_{\textnormal{R}}, \, \mL\left( 2^{-s/2} \mu \cdot v_{S, F', F''}(K_{2, t})^{\otimes s} + \mN(0, 1)^{\otimes D^{\otimes s}} \right) \right) = O\left( k^{-2s} 2^{-2st} \right)$$
Combining this with the fact that $\mu = \theta \cdot 2^{s(t + 1)/2}$ now yields the first bound in the lemma. The second bound follows by the same argument but now applying Corollary \ref{cor:tensor-bern-rotations} to the distribution $(T_{\textnormal{PD2}})_{F'_{i_1}, \dots, F'_{i_s}} \sim \text{Bern}(q)^{(D/k)^{\otimes s}}$. This completes the proof of the lemma.
\end{proof}

\begin{lemma}[Signing for $\pr{tpca}$] \label{lem:signing}
Let $F', F''$ and $S$ be as in Lemma \ref{lem:tpca-rotations} and let $p, q$ and $\theta$ be as in Theorem \ref{thm:tpca}. Let $\mathcal{A}_{\textnormal{3}}$ denote Step 3 of $k$\textsc{-pst-to-tpca} with input $T_{\textnormal{R}}$ and output given by the output $T'$ of $\mathcal{A}$. Then
\begin{align*}
\mathcal{A}_{\textnormal{3}} \left( 2^{st/2} \theta \cdot v_{S, F', F''}(K_{2, t})^{\otimes s} + \mN(0, 1)^{\otimes D^{\otimes s}} \right) &\sim \pr{tpca}^s_D(n, \theta) \\
\mathcal{A}_{\textnormal{3}} \left( \mN(0, 1)^{\otimes D^{\otimes s}} \right) &\sim \mN(0, 1)^{\otimes n^{\otimes s}}
\end{align*}
\end{lemma}

\begin{proof}
Suppose that $T_{\textnormal{R}} \sim \mL\left( 2^{st/2} \theta \cdot v_{S, F', F''}(K_{2, t})^{\otimes s} + \mN(0, 1)^{\otimes D^{\otimes s}} \right)$ and let $b \sim \text{Unif}\left[ \{-1, 1\}\right]^{\otimes D}$ be as in Step 3 of $\mathcal{A}$. The symmetry of zero-mean Gaussians and independence among the entries of $\mN(0, 1)^{\otimes D^{\otimes s}}$ imply that
$$b^{\otimes s} \odot T_{\textnormal{R}} \sim \mL\left( 2^{st/2} \theta \cdot u^{\otimes s} + b^{\otimes s} \odot \mN(0, 1)^{\otimes D^{\otimes s}} \right) = \mL\left( 2^{st/2} \theta \cdot u^{\otimes s} + \mN(0, 1)^{\otimes D^{\otimes s}} \right)$$
where $u = b \odot v_{S, F', F''}(K_{2, t})$ and the two terms $u^{\otimes s}$ and $\mN(0, 1)^{\otimes D^{\otimes s}}$ above are independent. Now note that each entry of $v_{S, F', F''}(K_{2, t})$ is either $\pm 2^{-t/2}$ by the definition of $K_{2, t}$. This implies that $2^{t/2} u$ is distributed as $\text{Unif}\left[ \{-1, 1\}\right]^{\otimes D}$ and hence that
$$\mL\left(b^{\otimes s} \odot T_{\textnormal{R}} \right) = \mL\left( \theta \cdot b^{\otimes s} + \mN(0, 1)^{\otimes D^{\otimes s}} \right) = \pr{tpca}^s_D(D, \theta)$$
Subsampling the same set $U$ of $n$ coordinates of this tensor along each axis by definition yields $ \pr{tpca}(n, \theta)$, proving the first claim in the lemma. The second claim is immediate by the fact that if $T_{\textnormal{R}} \sim \mN(0, 1)^{\otimes D^{\otimes s}}$ then it also holds that $b^{\otimes s} \odot T_{\textnormal{R}} \sim \mN(0, 1)^{\otimes D^{\otimes s}}$. This completes the proof of the lemma.
\end{proof}

We now complete the proof of Theorem \ref{thm:tpca} by applying Lemma \ref{lem:tvacc} as in Theorems \ref{thm:isgmreduction} and \ref{thm:isbm}.

\begin{proof}[Proof of Theorem \ref{thm:tpca}]
Define the steps of $\mathcal{A}$ to map inputs to outputs as follows
$$(T, F) \xrightarrow{\mathcal{A}_1} (T_{\text{PD}}, F) \xrightarrow{\mathcal{A}_2} (T_{\text{R}}, F'') \xrightarrow{\mathcal{A}_{\text{3}}} T'$$
Consider Lemma \ref{lem:tvacc} applied to the following sequence of distributions
\allowdisplaybreaks
\begin{align*}
\mathcal{P}_0 &= \mathcal{M}_{[N]^s}\left( S^{s}, \textnormal{Bern}(p), \textnormal{Bern}(q) \right) \\
\mathcal{P}_1 &= \mathcal{M}_{[2^t k]^s} \left( S^s, \textnormal{Bern}(p), \textnormal{Bern}(q) \right) \quad \text{where } S \sim \mU_{2^t k}(F') \\
\mathcal{P}_2 &= 2^{st/2} \theta \cdot v_{S, F', F''}(K_{2, t})^{\otimes s} + \mN(0, 1)^{\otimes D^{\otimes s}} \quad \text{where } S \sim \mU_{2^t k}(F') \\
\mathcal{P}_3 &= \pr{tpca}^s_D(n, \theta)
\end{align*}
Consider applying Lemmas \ref{lem:tpca-rotations} and \ref{lem:signing} while averaging over $S \sim \mU_{2^t k}(F')$ and applying the conditioning property of total variation in Fact \ref{tvfacts}. This yields that we may take $\epsilon_1 = 0$, $\epsilon_2 = O\left( k^{-2s} 2^{-2st} \right)$ and $\epsilon_3 = 0$. Applying Lemma \ref{lem:tvacc} proves the first bound in the theorem. Now consider the following sequence of distributions
$$\mathcal{P}_0 = \mathcal{M}_{[N]^s}\left( \textnormal{Bern}(q) \right), \quad \mathcal{P}_1 = \mathcal{M}_{[2^t k]^s}\left( \textnormal{Bern}(q) \right), \quad \mathcal{P}_2 = \mN(0, 1)^{\otimes D^{\otimes s}} \quad \text{and} \quad \mathcal{P}_3 = \mN(0, 1)^{\otimes n^{\otimes s}}$$
Lemmas \ref{lem:tpca-rotations} and \ref{lem:signing} imply we can again take $\epsilon_1 = 0$, $\epsilon_2 = O\left( k^{-2s} 2^{-2st} \right)$ and $\epsilon_3 = 0$. The second bound in the theorem now follows from Lemma \ref{lem:tvacc}.
\end{proof}

We now apply this theorem to deduce our main computational lower bounds for tensor PCA by verifying its guarantees are sufficient to apply Lemma \ref{cor:one-side-reduction}.

\begin{reptheorem}{thm:tpca-lb} [Lower Bounds for $\pr{tpca}$]
Let $n$ be a parameter and $s \ge 3$ be a constant, then the $k\pr{-hpc}^s$ or $k\pr{-hpds}^s$ conjecture for constant $0 < q < p \le 1$ both imply a computational lower bound for $\pr{tpca}^s(n, \theta)$ at all levels of signal $\theta = \tilde{o}(n^{-s/4})$ against $\textnormal{poly}(n)$ time algorithms $\mathcal{A}$ solving $\pr{tpca}^s(n, \theta)$ with a low false positive probability of $\bP_{H_0}[\mathcal{A}(T) = H_1] = O(n^{-s})$.
\end{reptheorem}

\begin{proof}
We will verify that the approximate Markov transition guarantees for $k$\textsc{-pst-to-tpca} in Theorem \ref{thm:tpca} are sufficient to apply Lemma \ref{cor:one-side-reduction} for the set of $\mP = \pr{tpca}^s(n, \theta)$ with parameters $(n, \theta)$ that fill out the region $\theta = \tilde{o}(n^{-s/4})$. Fix a constant pair of probabilities $0 < Q < p \le 1$, a constant positive integer $s$ and any sequence of parameters $(n, \theta)$ where $\theta \in (0, 1)$ is implicitly a function of $n$ with
$$\theta \le \frac{c}{w^{s/2}n^{s/4} \sqrt{\log n}}$$
for sufficiently large $n$, an arbitrarily slow-growing function $w = w(n) \to \infty$ and a sufficiently small constant $c > 0$. Now consider the parameters $(N,  k)$ and input $t$ to $k\pr{-pst-to-tpca}$ defined as follows:
\begin{itemize}
\item let $t$ be such that $2^t$ is the smallest power of two greater than $w\sqrt{n}$; and
\item let $k = \lceil w^{-1} \sqrt{n} \rceil$ and let $N$ be the largest multiple of $k$ less than $n$.
\end{itemize}
Now observe that these choices of parameters ensure that $k$ divides $N$, it holds that $k = o(\sqrt{N})$ and
$$N \le n \le 2^t k \le D = 2k(2^t - 1)$$
Furthermore, we have that $N = \Theta(n)$ and $2^t = \Theta(w \sqrt{n})$. For a sufficiently small choice of $c > 0$, we also have that
$$\theta \le \frac{c}{w^{s/2}n^{s/4} \sqrt{\log n}} \le \frac{c' \cdot \delta}{2^{st/2} \cdot \sqrt{t + \log(p - Q)^{-1}}}$$
where $c' > 0$ is the constant and $\delta$ is as in Theorem \ref{thm:tpca}. This verifies all of the conditions needed to apply Theorem \ref{thm:tpca}, which implies that $k$\textsc{-pst-to-tpca} maps $k\pr{-pst}_E^s(N, k, p, Q)$ to $\pr{tpca}^s(n, \theta)$ under both $H_0$ and $H_1$ to within total variation error $O\left( k^{-2s} 2^{-2st} \right) = O(n^{-2s})$. By Lemma \ref{cor:one-side-reduction}, the $k\pr{-hpds}^s$ conjecture for $k\pr{-hpds}^s_E(N', k', p, q)$ where $N = N' - (s - 1)N'/k'$ and $k = k' - s + 1$ now implies that there are is no $\textnormal{poly}(n)$ time algorithm $\mathcal{A}$ solving $\pr{tpca}^s(n, \theta)$ with a low false positive probability of $\bP_{H_0}[\mathcal{A}(T) = H_1] = O(n^{-s})$. This completes the proof of the theorem.
\end{proof}

We conclude this section with the following lemma observing that this theorem implies a computational lower bound for estimating $v$ in $\pr{tpca}^s(n, \theta)$ where $\theta = \tilde{\omega}(n^{-s/2})$ and $\theta = \tilde{o}(n^{-s/4})$. Note that the requirement $\theta = \tilde{\omega}(n^{-s/2})$ is weaker than the condition $\theta = \tilde{\omega}(n^{(1-s)/2})$, which is necessary for recovering $v$ to be information-theoretically possible, as discussed in Section \ref{subsec:1-problems-tpca}. The next lemma shows that any estimator yields a test in the hypothesis testing formulation of tensor PCA that must have a low false positive probability of error, since thresholding $\langle \hat{v}, T\rangle$ where $\hat{v}$ is an estimator of $v$, yields a means to distinguish $H_0$ and $H_1$ with high probability. We remark that the requirement $\langle v, \hat{v} \rangle = \Omega(\| v \|_2)$ is weaker than the condition $\| v - \hat{v} \cdot \sqrt{n} \|_2 = o(\sqrt{n})$ when $\hat{v}$ is a unit vector and $v \in \{-1, 1\}^n$. Thus any estimation algorithm with $\ell_2$ error $o(\sqrt{n})$, directly yields an algorithm $\mathcal{A}_E$ satisfying the conditions of the lemma.

\begin{lemma}[One-Sided Blackboxes from Estimation in Tensor PCA] \label{lem:one-side-estimation}
Let $s \ge 2$ be a fixed constant and suppose that there is a $\textnormal{poly}(n)$ time algorithm $\mathcal{A}_E$ that, on input sampled from $\theta v^{\otimes s} + \mN(0, 1)^{\otimes n^{\otimes s}}$ where $v \in \{-1, 1\}^n$ is fixed but unknown to $\mathcal{A}_E$ and $\theta = \omega(n^{-s/2} \sqrt{s \log n})$, outputs a unit vector $\hat{v} \in \mathbb{R}^n$ with $\langle v, \hat{v} \rangle = \Omega(\| v \|_2)$. Then there is a $\textnormal{poly}(n)$ time algorithm $\mathcal{A}_D$ solving $\pr{tpca}^s(n, \theta)$ with a low false positive probability of $\bP_{H_0}[\mathcal{A}_D(T) = H_1] = O(n^{-s})$.
\end{lemma}

\begin{proof}
Let $T$ be an instance of $\pr{tpca}^s(n, \theta)$ with $T = \theta v^{\otimes s} + G$ under $H_1$ and $T = G$ under $H_0$ where $G \sim \mN(0, 1)^{\otimes n^{\otimes s}}$. Consider the following algorithm $\mathcal{A}_D$ for $\pr{tpca}^s(n, \theta)$:
\begin{enumerate}
\item Independently sample $G' \sim \mN(0, 1)^{\otimes n^{\otimes s}}$ and form $T_1 = \frac{1}{\sqrt{2}} (T + G')$ and $T_2 = \frac{1}{\sqrt{2}} (T - G')$.
\item Compute $\hat{v}(T_1)$ as the output of $\mathcal{A}_E$ applied to $T_1$.
\item Output $H_0$ if $\langle \hat{v}(T_1)^{\otimes s}, T_2 \rangle < 2\sqrt{s \log n}$ and output $H_1$ otherwise.
\end{enumerate}
First note that the entries of $\frac{1}{\sqrt{2}} (G + G')$ and $\frac{1}{\sqrt{2}} (G - G')$ are jointly Gaussian but uncorrelated, which implies that these two tensors are independent. This implies that $T_1$ and $T_2$ are independent. Since $\hat{v}(T_1)$ is a unit vector and independent of $T_2$, it follows that $\langle \hat{v}(T_1)^{\otimes s}, T_2 \rangle$ is distributed as $\mN(0, 1)$ conditioned on $\hat{v}(T_1)$ if $T$ is distributed according to $H_0$ of $\pr{tpca}^s(n, \theta)$. Now we have that
$$\bP_{H_0}[\mathcal{A}_D(T) = H_1] = \mP\left[ \langle \hat{v}(T_1)^{\otimes s}, T_2 \rangle \ge 2 \sqrt{s \log n} \right] = O(n^{-2s})$$
where the second equality follows from standard Gaussian tail bounds. If $T$ is distributed according to $H_1$, then $\langle \hat{v}(T_1)^{\otimes s}, T_2 \rangle \sim \mN( \theta \langle \hat{v}(T_1), v \rangle^s, 1)$. In this case, $\mathcal{A}_E$ ensures that $\langle \hat{v}(T_1), v \rangle^s = \Omega(n^{s/2})$ since $\| v \|_2 = \sqrt{n}$, and therefore $\theta \langle \hat{v}(T_1), v \rangle^s = \omega(\sqrt{s \log n})$. It therefore follows that
$$\bP_{H_1}[\mathcal{A}_D(T) = H_0] \le \mP\left[ \langle \hat{v}(T_1)^{\otimes s}, T_2 \rangle - \theta \langle \hat{v}(T_1), v \rangle^s < - 2 \sqrt{s \log n} \right] = O(n^{-2s})$$
Thus $\mathcal{A}_D$ has Type I$+$II error that is $o(1)$ and the desired low false positive probability, which completes the proof of the lemma.
\end{proof}

\section{Universality of Lower Bounds for Learning Sparse Mixtures}
\label{sec:universality}

In this section, we combine our reduction to $\pr{isgm}$ from Section \ref{subsec:3-rsme-reduction} with symmetric 3-ary rejection kernels, which were introduced and analyzed in Section \ref{subsec:srk}. We remark that the $k$-partite promise in $k\pr{-pds}$ is crucially used in our reduction to obtain this universality. In particular, this promise ensures that the entries of the intermediate $\pr{isgm}$ instance are from one of three distinct distributions, when conditioned on the part of the mixture the sample is from. This is necessary for our application of symmetric 3-ary rejection kernels. An overview of the ideas in this section can be found in Section \ref{subsec:1-tech-universality}.

Our general lower bound holds given tail bounds on the likelihood ratios between the planted and noise distributions, and applies to a wide range of natural distributional formulations of learning sparse mixtures. For example, our general lower bound recovers the tight computational lower bounds for sparse PCA in the spiked covariance model from \cite{gao2017sparse, brennan2018reducibility, brennan2019optimal}. The results in this section can also be interpreted as a universality principle for computational lower bounds in sparse PCA. We prove the approximate Markov transition guarantees for our reduction to $\pr{glsm}$ in Section \ref{subsec:universalitybounds} and discuss the universality conditions needed for our lower bounds in Section \ref{subsec:universalitydiscussion}.

\subsection{Reduction to Generalized Learning Sparse Mixtures}
\label{subsec:universalitybounds}

\begin{figure}[t!]
\begin{algbox}
\textbf{Algorithm} $k$\textsc{-bpds-to-glsm}

\vspace{1mm}

\textit{Inputs}: Matrix $M \in \{0, 1\}^{m \times n}$, dense subgraph dimensions $k_m$ and $k_n$ where $k_n$ divides $n$ and the following parameters
\begin{itemize}
\item partition $F$, edge probabilities $0 < q < p \le 1$ and $w(n)$ as in Figure \ref{fig:isgmreduction}
\item target $\pr{glsm}$ parameters $(N, k_m, d)$ satisfying $wN \le n$ and $m \le d$, a mixture distribution $\mD$ and target distributions $\{ \mP_{\nu} \}_{\nu \in \mathbb{R}}$ and $\mQ$
\end{itemize}

\begin{enumerate}
\item \textit{Map to Gaussian Sparse Mixtures}: Form the sample $Z_1, Z_2, \dots, Z_N \in \mathbb{R}^d$ by setting
$$(Z_1, Z_2, \dots, Z_N) \gets k\pr{-bpds-to-isgm}(M, F)$$
where $k\pr{-bpds-to-isgm}$ is applied with $r = 2$, slow-growing function $w(n)$, $t = \lceil \log_2(n/k_n) \rceil$, target parameters $(N, k_m, d)$, $\epsilon = 1/2$ and $\mu = c_1\sqrt{\frac{k_n}{n \log n}}$ for a sufficiently small constant $c_1 > 0$.
\item \textit{Truncate and 3-ary Rejection Kernels}: Sample $\nu_1, \nu_2, \dots, \nu_N \sim_{\text{i.i.d.}} \mD$, truncate the $\nu_i$ to lie within $[-1, 1]$ and form the vectors $X_1, X_2, \dots, X_N \in \mathbb{R}^d$ by setting
$$X_{ij} \gets 3\pr{-srk}(\pr{tr}_{\tau}(Z_{ij}), \mP_{\nu_i}, \mP_{-\nu_i}, \mQ)$$
for each $i \in [N]$ and $j \in [d]$. Here $3\pr{-srk}$ is applied with $N_{\text{it}} = \lceil 4 \log (dN) \rceil$ iterations and with the parameters
\begin{align*}
a &= \Phi(\tau) - \Phi(-\tau), \quad \mu_1 = \frac{1}{2} \left( \Phi(\tau + \mu) - \Phi(\tau - \mu) \right), \\
\mu_2 &= \frac{1}{2} \left( 2 \cdot \Phi(\tau) - \Phi(\tau + \mu) - \Phi(\tau - \mu) \right)
\end{align*}
\item \textit{Output}: The vectors $(X_1, X_2, \dots, X_N)$.
\end{enumerate}
\vspace{0.5mm}

\end{algbox}
\caption{Reduction from $k$-part bipartite planted dense subgraph to general learning sparse mixtures.}
\label{fig:universalityreduction}
\end{figure}

In this section, we combine symmetric 3-ary rejection kernels with the reduction $k\pr{-bpds-to-isgm}$ to map from $k\pr{-bpds}$ to generalized sparse mixtures. The details of this reduction $k$\textsc{-bpds-to-glsm} are shown in Figure \ref{fig:universalityreduction}. As mentioned in Sections \ref{subsec:1-tech-universality} and \ref{subsec:srk}, to reduce to sparse mixtures near their computational barrier, it is crucial to produce multiple planted distributions. Previous rejection kernels do not have enough degrees of freedom to map to three output distributions given their binary inputs. The symmetric 3-ary rejection kernels introduced in Section \ref{subsec:srk} overcome this issue by mapping three input to three output distributions. In particular, we will see in this section that their approximate Markov transition guarantees established in Lemma \ref{lem:srk} exactly lead to tight computational lower bounds for $\pr{glsm}$. Throughout this section, we will adopt the definitions of $\pr{glsm}$ and $\pr{glsm}_D$ introduced in Sections \ref{subsec:1-problems-universality} and \ref{subsec:2-formulations}.

In order to establish computational lower bounds for $\pr{glsm}$, it is crucial to define a meaningful notion of the level of signal in a set of target distributions $\mD, \mQ$ and $\{ \mP_{\nu} \}_{\nu \in \mathbb{R}}$. This level of signal was defined in Section \ref{subsec:1-problems-universality} and is reviewed below for convenience. We remark that this definition will turn out to coincide with the conditions needed to apply symmetric 3-ary rejection kernels. This notion of signal also implicitly defines the universality class over which our computational lower bounds hold.

\begin{repdefinition}{defn:univ-signal} [Universal Class and Level of Signal]
Given a parameter $N$, define the collection of distributions $\mathcal{U} = (\mD, \mQ, \{ \mP_{\nu} \}_{\nu \in \mathbb{R}})$ implicitly parameterized by $N$ to be in the universality class $\pr{uc}(N)$ if
\begin{itemize}
\item the pairs $(\mP_{\nu}, \mQ)$ are all computable pairs, as in Definition \ref{def:computable}, for all $\nu \in \mathbb{R}$;
\item $\mD$ is a symmetric distribution about zero and $\bP_{\nu \sim \mD}[\nu \in [-1, 1]] = 1 - o(N^{-1})$; and
\item there is a level of signal $\tau_{\mathcal{U}} \in \mathbb{R}$ such that for all $\nu \in [-1, 1]$ such that for any fixed constant $K > 0$, it holds that
$$\left| \frac{d\mP_{\nu}}{d\mQ} (x) - \frac{d\mP_{-\nu}}{d\mQ} (x) \right| = O_N\left(\tau_{\mathcal{U}} \right) \quad \textnormal{and} \quad \left|\frac{d\mP_{\nu}}{d\mQ} (x) + \frac{d\mP_{-\nu}}{d\mQ} (x) - 2 \right| = O_N\left( \tau_{\mathcal{U}}^2 \right)$$
with probability at least $1 - O\left(N^{-K}\right)$ over each of $\mP_{\nu}, \mP_{-\nu}$ and $\mQ$.
\end{itemize}
\end{repdefinition}

In our reduction $k\pr{-bpds-to-isgm}$, we truncate Gaussians to generate the input distributions $\text{Tern}$. In Figure \ref{fig:universalityreduction}, $\pr{tr}_{\tau} : \mathbb{R} \to \{-1, 0, 1\}$ denotes the truncation map given by
$$\pr{tr}_{\tau}(x) = \left\{ \begin{array}{ll} 1 &\text{if } x > |\tau| \\ 0 &\text{if } -|\tau| \le x \le |\tau| \\ -1 &\text{if } x < -|\tau| \end{array} \right.$$
The following simple lemma on truncating symmetric triples of Gaussian distributions will be important in the proofs in this section. Its proof is a direct computation and is deferred to Appendix \ref{subsec:appendix-3-part-3}.

\begin{lemma}[Truncating Gaussians] \label{lem:truncgauss}
Let $\tau > 0$ be constant, $\mu > 0$ be tending to zero and let $a, \mu_1, \mu_2$ be such that
\begin{align*}
&\pr{tr}_\tau(\mN(\mu, 1)) \sim \textnormal{Tern}(a, \mu_1, \mu_2) \\
&\pr{tr}_\tau(\mN(-\mu, 1)) \sim \textnormal{Tern}(a, -\mu_1, \mu_2) \\
&\pr{tr}_\tau(\mN(0, 1)) \sim \textnormal{Tern}(a, 0, 0)
\end{align*}
Then it follows that $a > 0$ is constant, $0 < \mu_1 = \Theta(\mu)$ and $0 < \mu_2 = \Theta(\mu^2)$.
\end{lemma}

We now will prove our main approximate Markov transition guarantees for $k\textsc{-bpds-to-glsm}$. The proof follows from combining Theorem \ref{thm:isgmreduction}, Lemma \ref{lem:srk} and an application of tensorization of $\TV$.

\begin{theorem}[Reduction from $k$\pr{-bpds} to $\pr{glsm}$] \label{lem:univlem}
Let $n$ be a parameter and $w(n) = \omega(1)$ be a slow-growing function. Fix initial and target parameters as follows:
\begin{itemize}
\item \textnormal{Initial} $k\pr{-bpds}$ \textnormal{Parameters:} vertex counts on each side $m$ and $n$ that are polynomial in one another, dense subgraph dimensions $k_m$ and $k_n$ where $k_n$ divides $n$, constant edge probabilities $0 < q < p \le 1$ and a partition $F$ of $[n]$.
\item \textnormal{Target} $\pr{glsm}$ \textnormal{Parameters:} $(N, d)$ satisfying $wN \le n$, $N \ge n^{c'}$ for some constant $c' > 0$ and $m \le d \le \textnormal{poly}(n)$, target distribution collection $\mathcal{U} = (\mD, \mQ, \{ \mP_{\nu} \}_{\nu \in \mathbb{R}}) \in \pr{uc}(N)$ satisyfing that
$$0 < \tau_{\mU} \le c \cdot \sqrt{\frac{k_n}{n\log n}}$$
for a sufficiently small constant $c > 0$.
\end{itemize}
Let $\mathcal{A}(M)$ denote $k\textsc{-bpds-to-glsm}$ applied to the adjacency matrix $M$ with these parameters. Then $\mathcal{A}$ runs in $\textnormal{poly}(m, n)$ time and it follows that
\begin{align*}
\TV\left( \mathcal{A}\left( \mathcal{M}_{[m] \times [n]}(S \times T, p, q) \right), \, \pr{glsm}_D(N, S, d, \mU) \right) &= o(1) + O\left( w^{-1} + k_n^{-2}m^{-2}r^{-2t} + n^{-2} + N^{-3} d^{-3} \right) \\
\TV\left( \mathcal{A}\left( \textnormal{Bern}(q)^{\otimes m \times n} \right), \, \mQ^{\otimes d \times N} \right) &= O\left( k_n^{-2}m^{-2}r^{-2t} + n^{-2} + N^{-3} d^{-3} \right)
\end{align*}
for all subsets $S \subseteq [m]$ with $|S| = k_m$ and subsets $T \subseteq [n]$ with $|T| = k_n$ and $|T \cap F_i| = 1$ for each $1 \le i \le k_n$.
\end{theorem}

\begin{proof}
Let $\mathcal{A}_1$ denote Step 1 of $\mathcal{A}$ with input $M$ and output $(Z_1, Z_2, \dots, Z_N)$. First note that $2^t = \Theta(n/k_n)$ by the definition of $t$ and $\log m = \Theta(\log n)$ since $m$ and $n$ are polynomial in one another. Thus for a small enough choice of $c_1 > 0$, we have
$$\mu = c_1 \cdot \sqrt{\frac{k_n}{n \log n}} \le \frac{2^{-(t + 1)/2}}{2 \sqrt{6\log (k_n m \cdot 2^t) + 2\log (p - q)^{-1}}} \cdot \min \left\{ \log \left( \frac{p}{q} \right), \log \left( \frac{1 - q}{1 - p} \right) \right\}$$
since $p$ and $q$ are constants. Therefore $\mu$ satisfies the conditions needed to apply Theorem \ref{thm:isgmreduction} to $\mathcal{A}_1$. Now let $\mathcal{A}_2$ denote Step 2 of $\mathcal{A}$ with input $(Z_1, Z_2, \dots, Z_N)$ and output $(X_1, X_2, \dots, X_N)$. First suppose that $(Z_1, Z_2, \dots, Z_N) \sim \pr{isgm}_D(N, S, d, \mu, 1/2)$ or in other words where
$$Z_i \sim_{\text{i.i.d.}} \pr{mix}_{1/2}\left( \mN( \mu \cdot \mathbf{1}_S, I_d), \mN( -\mu \cdot \mathbf{1}_S, I_d) \right)$$
For the next part of this argument, we condition on: (1) the entire vector $\nu = (\nu_1, \nu_2, \dots, \nu_N)$; and (2) the subset $P \subseteq [N]$ of sample indices corresponding to the positive part $\mN(\mu \cdot \mathbf{1}_S, I_d)$ of the mixture. Let $\mathcal{C}(\nu, P)$ denote the event corresponding to this conditioning. After truncating according to $\pr{tr}_{\tau}$, by Lemma \ref{lem:truncgauss} the resulting entries are distributed as
$$\pr{tr}_{\tau}(Z_{ij}) \sim \left\{ \begin{array}{ll} \text{Tern}(a, \mu_1, \mu_2) &\text{if } (i, j) \in S \times P \\ \text{Tern}(a, -\mu_1, \mu_2) &\text{if } (i, j) \in S \times P^C \\ \text{Tern}(a, 0, 0) &\text{if } i \not \in S \end{array} \right.$$
Furthermore, these entries are all independent conditioned on $(\nu, P)$. Since $\tau$ is constant, Lemma \ref{lem:truncgauss} also implies that $a \in (0, 1)$ is constant, $\mu_1 = \Theta(\mu)$ and $\mu_2 = \Theta(\mu^2)$. Let $S_\nu$ be
$$S_{\nu} = \left\{ x \in X : 2|\mu_1| \ge \left| \frac{d\mP_{\nu}}{d\mQ} (x) - \frac{d\mP_{-\nu}}{d\mQ} (x) \right| \quad \textnormal{and} \quad \frac{2|\mu_2|}{\max\{a, 1 - a\}} \ge \left|\frac{d\mP_{\nu}}{d\mQ} (x) + \frac{d\mP_{-\nu}}{d\mQ} (x) - 2 \right| \right\}$$
as in Lemma \ref{lem:srk}. Since $\mU = (\mD, \mQ, \{ \mP_{\nu} \}_{\nu \in \mathbb{R}}) \in \pr{uc}(N)$ has level of signal $\tau_{\mU} \le c' \cdot \mu$ for a sufficiently small constant $c' > 0$, we have by definition that $\{x \in S_{\nu_i}\}$ occurs with probability at least $1 - \delta_1$ where $\delta_1 = O(n^{-4 - K_1})$ over each of $\mP_{\nu_i}, \mP_{-\nu_i}$ and $\mQ$, where $K_1 > 0$ is a constant for which $d = O(n^{K_1})$. Here, we are implicitly using the fact that $N \ge n^{c'}$ for some constant $c' > 0$.

Now consider applying Lemma \ref{lem:srk} to each application of $3\pr{-srk}$ in Step 2 of $\mathcal{A}$. Note that $|\mu_1|^{-1} = O(\sqrt{n \log n})$ and $|\mu_2|^{-1} = O(n \log n)$ since $\mu = \Omega(\sqrt{k_n/n\log n})$ and $k_n \ge 1$. Now consider the $d$-dimensional vectors $X_1', X_2', \dots, X_N'$ with independent entries distributed as
$$X'_{ij} \sim \left\{ \begin{array}{ll} \mP_{\nu_i} &\text{if } (i, j) \in S \times P \\ \mP_{-\nu_i} &\text{if } (i, j) \in S \times P^C \\ \mQ &\text{if } i \not \in S \end{array} \right.$$
The tensorization property of $\TV$ from Fact \ref{tvfacts} implies that
\begin{align*}
&\TV\left( \mL(X_1, X_2, \dots, X_N | \nu, P), \mL(X_1', X_2', \dots, X_N'| \nu, P) \right) \\
&\quad \quad \le \sum_{i = 1}^N \sum_{j = 1}^d \TV\left( \mL(X_{ij} | \nu, P), \mL(X_{ij}' | \nu, P) \right) \\
&\quad \quad \le \sum_{i = 1}^N \sum_{j = 1}^d \TV\left( 3\pr{-srk}(\pr{tr}_{\tau}(Z_{ij}), \mP_{\nu_i}, \mP_{-\nu_i}, \mQ), \mL(X_{ij}' | \nu, P) \right) \\
&\quad \quad \le Nd \left[ 2\delta_1 \left(1 + |\mu_1|^{-1} + |\mu_2|^{-1} \right) + \left( \frac{1}{2} + \delta_1 \left( 1 + |\mu_1|^{-1} + |\mu_2|^{-1} \right) \right)^{N_{\text{it}}} \right] \\
&\quad \quad = O\left( n^{-2} + N^{-3} d^{-3} \right)
\end{align*}
since $N \le n$, $\delta_1 = O(n^{-4} d^{-1})$, $N_{\text{it}} = \lceil 4 \log(dN) \rceil$ and by the total variation upper bounds in Lemma \ref{lem:srk}. 

We now will drop the conditioning on $(\nu, P)$ and average over $\nu \sim \mD'$ and $P \sim \text{Unif}\left[2^{[N]}\right]$. Observe that, when not conditioned on $(\nu, P)$, it holds that
$$(X_1', X_2', \dots, X_N') \sim \pr{glsm}_D\left(N, S, d, \left( \mD', \mQ, \{ \mP_{\nu} \}_{\nu \in \mathbb{R}} \right) \right)$$
where $\mD'$ is $\mD$ conditioned to lie in $[-1, 1]$. Note that here we used the fact that $\mD$ and therefore $\mD'$ is symmetric about zero. Coupling the latent $\nu_1, \nu_2, \dots, \nu_N$ sampled from $\mD$ and $\mD'$ and then applying the tensorization property of Fact \ref{tvfacts} yields that
\begin{align*}
&\TV\left( \pr{glsm}_D\left(N, S, d, \left( \mD', \mQ, \{ \mP_{\nu} \}_{\nu \in \mathbb{R}} \right) \right), \pr{glsm}_D\left(N, S, d, \left( \mD, \mQ, \{ \mP_{\nu} \}_{\nu \in \mathbb{R}} \right) \right) \right) \\
&\quad \quad \le \TV( \mD^{\otimes n}, \mD'^{\otimes n}) \le N \cdot \TV( \mD, \mD') \le N \cdot o(N^{-1}) = o(1)
\end{align*}
where $\TV( \mD, \mD') = o(N^{-1})$ follow from the conditioning property of $\TV$ from Fact \ref{tvfacts} and the fact that $\bP_{\nu \sim \mD}[\nu \in [-1, 1]] = 1 - o(N^{-1})$. The triangle inequality and conditioning property of $\TV$ in Fact \ref{tvfacts} now imply that
\begin{align*}
&\TV\left( \mathcal{A}_2\left( \pr{isgm}_D(N, S, d, \mu, 1/2) \right), \pr{glsm}_D\left(N, S, d, \mU \right) \right) \\
&\quad \quad \le \TV\left( \mL(X_1, X_2, \dots, X_N), \mL(X_1', X_2', \dots, X_N') \right) + \TV\left( \mL(X_1', X_2', \dots, X_N'), \pr{glsm}_D\left(N, S, d, \mU \right) \right) \\
&\quad \quad \le \bE_{\nu \sim \mD'} \, \bE_{P \sim \text{Unif}\left[ 2^{[N]} \right]} \, \TV\left( \mL(X_1, X_2, \dots, X_N | \nu, P), \mL(X_1', X_2', \dots, X_N'| \nu, P) \right) \\
&\quad \quad \quad \quad + \TV\left( \pr{glsm}_D\left(N, S, d, \left( \mD', \mQ, \{ \mP_{\nu} \}_{\nu \in \mathbb{R}} \right) \right), \pr{glsm}_D\left(N, S, d, \mU \right) \right) \\
&\quad \quad = o(1) + O\left( n^{-2} + N^{-3} d^{-3} \right)
\end{align*}
Now consider the case when $Z_1, Z_2, \dots, Z_N \sim_{\text{i.i.d.}} \mN(0, I_d)$. Repeating the argument above with $S = \emptyset$ and observing that $(X'_1, X_2', \dots, X_N') \sim \mQ^{\otimes N}$ yields that
$$\TV\left( \mathcal{A}_2\left( \mN(0, I_d)^{\otimes N} \right), \mQ^{\otimes d \times N} \right) = O\left( n^{-2} + N^{-3} d^{-3} \right)$$
We now apply Lemma \ref{lem:tvacc} to the steps $\mathcal{A}_1$ and $\mathcal{A}_2$ under each of $H_0$ and $H_1$, as in the proof of Theorem \ref{thm:isgmreduction}. Under $H_1$, consider Lemma \ref{lem:tvacc} applied to the following sequence of distributions
$$\mathcal{P}_0 = \mathcal{M}_{[m] \times [n]}(S \times T, p, q), \quad \mathcal{P}_1 = \pr{isgm}_D(N, S, d, \mu, 1/2) \quad \text{and} \quad \mathcal{P}_2 = \pr{glsm}_D\left(N, S, d, \mU \right)$$
By Theorem \ref{thm:isgmreduction} and the argument above, we can take
$$\epsilon_1 = O\left( w^{-1} + k_n^{-2}m^{-2}r^{-2t} + n^{-2} + N^{-3} d^{-3} \right) \quad \text{and} \quad \epsilon_2 = o(1) + O\left( n^{-2} + N^{-3} d^{-3} \right)$$
By Lemma \ref{lem:tvacc}, we therefore have that
$$\TV\left( \mathcal{A}\left( \mathcal{M}_{[m] \times [n]}(S \times T, p, q) \right), \, \pr{glsm}_D(N, S, d, \mU) \right) = o(1) + O\left( w^{-1} + k_n^{-2}m^{-2}r^{-2t} + n^{-2} + N^{-3} d^{-3} \right)$$
which proves the desired result in the case of $H_1$. Under $H_0$, similarly applying Theorem \ref{thm:isgmreduction}, the argument above and Lemma \ref{lem:tvacc} to the distributions
$$\mathcal{P}_0 = \textnormal{Bern}(q)^{\otimes m \times n}, \quad \mathcal{P}_1 = \mN(0, I_d)^{\otimes N} \quad \text{and} \quad \mathcal{P}_2 = \mQ^{\otimes d \times N}$$
yields the total variation bound
$$\TV\left( \mathcal{A}\left( \textnormal{Bern}(q)^{\otimes m \times n} \right), \, \mQ^{\otimes d \times N} \right) = O\left( k_n^{-2}m^{-2}r^{-2t} + n^{-2} + N^{-3} d^{-3} \right)$$
which completes the proof of the theorem.
\end{proof}

We now use this theorem to deduce our universality principle for lower bounds in $\pr{glsm}$. The proof of this next theorem is similar to that of Theorems \ref{thm:rsme-lb} and \ref{thm:uslr-lb} and is deferred to Appendix \ref{subsec:appendix-3-part-3}.

\begin{reptheorem}{thm:glsm-lb} [Computational Lower Bounds for \pr{glsm}]
Let $n, k$ and $d$ be polynomial in each other and such that $k = o(\sqrt{d})$. Suppose that the collections of distributions $\mU = (\mD, \mQ, \{ \mP_{\nu} \}_{\nu \in \mathbb{R}})$ is in $\pr{uc}(n)$. Then the $k\pr{-bpc}$ conjecture or $k\pr{-bpds}$ conjecture for constant $0 < q < p \le 1$ both imply a computational lower bound for $\pr{glsm}\left(n, k, d, \mU \right)$ at all sample complexities $n = \tilde{o}\left(\tau_{\mU}^{-4}\right)$.
\end{reptheorem}

\subsection{The Universality Class UC$(n)$ and Level of Signal $\tau_{\mU}$}
\label{subsec:universalitydiscussion}

The result in Theorem \ref{thm:glsm-lb} shows universality of the computational sample complexity of $n = \tilde{\Omega}(\tau_{\mU}^{-4})$ for learning sparse mixtures under the mild conditions of $\pr{uc}(n)$. In this section, we discuss this lower bound, its implications, the universality class $\pr{uc}(n)$ and the level of signal $\tau_{\mU}$.

\paragraph{Remarks on UC$(n)$ and $\tau_{\mU}$.} The conditions for $\mU = (\mD, \mQ, \{ \mP_{\nu} \}_{\nu \in \mathbb{R}}) \in \pr{uc}(n)$ and the definition of $\tau_{\mU}$ have the following two notable properties.
\begin{itemize}
\item \textit{They are defined in terms of marginals}: The class $\pr{uc}(n)$ and $\tau_{\mU}$ are defined entirely in terms of the likelihood ratios $d\mP_\nu/d\mQ$ between the planted and non-planted marginals. In particular, they are independent of the sparsity level $k$ and other high-dimensional properties of the distribution $\pr{glsm}$ constructed from the $\mP_{\nu}$ and $\mQ$. Theorem \ref{thm:glsm-lb} thus establishes a computational lower bound for $\pr{glsm}$ at a sample complexity entirely based on properties of the marginals of $\mP_{\nu}$ and $\mQ$.
\item \textit{Their dependence on $n$ is negligible}: The parameter $n$ only enters the definitions of $\pr{uc}(n)$ and $\tau_{\mU}$ through requirements on tail probabilities. When the likelihood ratios $d\mP_\nu/d\mQ$ are relatively concentrated, the dependence of the conditions in $\pr{uc}(n)$ and $\tau_{\mU}$ on $n$ is nearly negligible. If the ratios $d\mP_\nu/d\mQ$ are concentrated under $\mP_{\nu}$ and $\mQ$ with exponentially decaying tails, then the tail probability bound requirement of $O(n^{-K})$ only appears as a $\text{polylog}(n)$ factor in $\tau_{\mU}$. This will be the case in the examples that appear later in this section.
\end{itemize}

\paragraph{$\mD$ and Parameterization over $[-1,1]$.} $\mD$ and the indices of $\mP_{\nu}$ can be reparameterized without changing the underlying problem. The assumption that $\mD$ is symmetric and mostly supported on $[-1, 1]$ is for notational convenience. As in the case of $\tau_{\mU}$ and the examples later in this section, the tail probability requirement of $o(n^{-1})$ for $\mD$ only appears as a $\text{polylog}(n)$ factor in the computational lower bound of $n = \tilde{\Omega}(\tau_{\mU}^{-4})$ if $\mD$ is concentrated with exponential tails.

While the output vectors $(X_1, X_2, \dots, X_N)$ of our reduction $k$\textsc{-bpds-to-glsm} are independent, their coordinates have dependence induced by the mixture $\mD$. The fact that our reduction samples the $\nu_i$ implies that if these values were revealed to the algorithm, the problem would still remain hard: an algorithm for the latter could be used together with the reduction to solve \kpc. However, even given the $\nu_i$ for the $i$th sample, our reduction is such that whether the planted marginals in the $i$th sample are distributed according to $\mP_{\nu_i}$ or $\mP_{-\nu_i}$ remains unknown to the algorithm. Intuitively, our setup chooses to parameterize the distribution $\mD$ over $[-1, 1]$ such that the sign ambiguity between $\mP_{\nu_i}$ or $\mP_{-\nu_i}$ is what is producing hardness below the sample complexity of $n = \tilde{\Omega}(\tau_{\mU}^{-4})$.

\paragraph{Implications for Concentrated LLR.} We now give several remarks on $\tau_{\mU}$ in the case that the log-likelihood ratios (LLR) $\log d\mP_{\nu}/d\mQ (x)$ are sufficiently well-concentrated if $x \sim \mQ$ or $x \sim \mP_{\nu}$. Suppose that $\mU = (\mD, \mQ, \{ \mP_{\nu} \}_{\nu \in \mathbb{R}}) \in \pr{uc}(n)$, fix some arbitrarily large constant $c > 0$ and fix some $\nu \in [-1,1]$. If $S_{\mQ}$ is the common support of the $\mP_{\nu}$ and $\mQ$, define $S$ to be
$$S = \left\{ x \in S_{\mQ} : c \cdot \tau_{\mU} \ge \left| \frac{d\mP_{\nu}}{d\mQ} (x) - \frac{d\mP_{-\nu}}{d\mQ} (x) \right| \quad \textnormal{and} \quad c \cdot \tau_{\mU}^2 \ge \left|\frac{d\mP_{\nu}}{d\mQ} (x) + \frac{d\mP_{-\nu}}{d\mQ} (x) - 2 \right| \right\}$$
Suppose that $\tau_{\mU} = \Omega(n^{-K})$ for some constant $K > 0$ and let $c$ be large enough that $S$ occurs with probability at least $1 - O(n^{-K})$ under each of $\mP_{\nu}, \mP_{-\nu}$ and $\mQ$. Note that such a constant $c$ is guaranteed by Definition \ref{defn:univ-signal}. Now observe that
\begin{align*}
\TV\left( \mP_{\nu}, \mP_{-\nu} \right) &= \frac{1}{2} \cdot \bE_{x \in \mQ} \left[ \left| \frac{d\mP_{\nu}}{d\mQ} (x) - \frac{d\mP_{-\nu}}{d\mQ} (x) \right| \right] \\
&\le \frac{1}{2} \cdot \bE_{x \in \mQ} \left[ \left| \frac{d\mP_{\nu}}{d\mQ} (x) - \frac{d\mP_{-\nu}}{d\mQ} (x) \right| \cdot \mathbf{1}_S(x) \right] + \frac{1}{2} \cdot \mP_{\nu}\left[S^C\right] + \frac{1}{2} \cdot \mP_{-\nu}\left[S^C\right] \\
&\le c \cdot \tau_{\mU} + O\left(n^{-K}\right) = O\left(\tau_{\mU}\right)
\end{align*}
A similar calculation with the second condition defining $S$ shows that
$$\TV\left( \pr{mix}_{1/2}\left(\mP_{\nu}, \mP_{-\nu} \right), \mQ \right) = O\left( \tau_{\mU}^2 \right)$$
If the LLRs $\log d\mP_{\nu}/d\mQ$ are sufficiently well-concentrated, then the random variables
$$\left| \frac{d\mP_{\nu}}{d\mQ} (x) - \frac{d\mP_{-\nu}}{d\mQ} (x) \right| \quad \text{and} \quad \left|\frac{d\mP_{\nu}}{d\mQ} (x) + \frac{d\mP_{-\nu}}{d\mQ} (x) - 2 \right|$$
will also concentrate around their means if $x \sim \mQ$. LLR concentration also implies that this is true if $x \sim \mP_\nu$ or $x \sim \mP_{-\nu}$. Thus, under sufficient concentration, the definition of the level of signal $\tau_{\mU}$ reduces to the much more interpretable pair of upper bounds
$$\TV\left( \mP_{\nu}, \mP_{-\nu} \right) = O\left(\tau_{\mU}\right) \quad \text{and} \quad \TV\left( \pr{mix}_{1/2}\left(\mP_{\nu}, \mP_{-\nu} \right), \mQ \right) = O\left(\tau_{\mU}^2 \right)$$
These conditions directly measure the amount of statistical signal present in the planted marginals $\mP_{\nu}$. The relevant calculations for an example application of Theorem \ref{thm:glsm-lb} when the LLR concentrates is shown below for sparse PCA. In \cite{brennan2019universality}, various assumptions of concentration of the LLR and analogous implications for computational lower bounds in submatrix detection are analyzed in detail. We refer the reader to Sections 3 and 9 of \cite{brennan2019universality} for the calculations needed to make the discussion here precise.

We remark that, assuming sufficient concentration on the LLR, the analysis of the $k$-sparse eigenvalue statistic from \cite{berthet2013complexity} yields an information-theoretic upper bound for $\pr{glsm}$. Given $\pr{glsm}$ samples $(X_1, X_2, \dots, X_n)$, consider forming the LLR-processed samples $Z_i$ with
$$Z_{ij} = \bE_{\nu \sim \mD} \left[ \log \frac{d\mP_{\nu}}{d\mQ} (X_{ij}) \right]$$
for each $i \in [n]$ and $j \in [d]$. Now consider taking the $k$-sparse eigenvalue of the samples $Z_1, Z_2, \dots, Z_n$. Under sub-Gaussianity assumptions on the $Z_{ij}$, the analysis in Theorem 2 of \cite{berthet2013complexity} applies. Similarly, the analysis in Theorem 5 of \cite{berthet2013complexity} continues to hold, showing that the semidefinite programming algorithm for sparse PCA yields an algorithmic upper bound for $\pr{glsm}$. As information-theoretic limits and algorithms are not the focus of this paper, we omit the technical details needed to make this rigorous.

In many setups captured by $\pr{glsm}$ such as sparse PCA, learning sparse mixtures of Gaussians and learning sparse mixtures of Rademachers, these analyses and our lower bound in Theorem \ref{thm:glsm-lb} together yield a $k$-to-$k^2$ statistical-computational gap. How our lower bound yields a $k^2$ dependence in the computational barriers for these problems is discussed below.

\paragraph{Sparse PCA and Specific Distributions.} One specific example captured by our universality principle and that falls under the concentrated LLR setup discussed above is sparse PCA in the spiked covariance model. The statistical-computational gaps of sparse PCA have been characterized based on the planted clique conjecture in a line of work \cite{berthet2013optimal, berthet2013complexity, wang2016statistical, gao2017sparse, brennan2018reducibility, brennan2019optimal}. We show that our universality principle faithfully recovers the $k$-to-$k^2$ gap for sparse PCA shown in \cite{berthet2013optimal, berthet2013complexity, wang2016statistical, gao2017sparse, brennan2018reducibility} assuming the $k\pr{-bpds}$ conjecture. As discussed in Section \ref{sec:2-secret-leakage}, also the $k\pr{-bpc}$, $k\pr{-pds}$ or $k\pr{-pc}$ conjectures therefore yields nontrivial lower bounds. We remark that \cite{brennan2019optimal} shows stronger hardness based on weaker forms of the $\pr{pc}$ conjecture.

We show in the next lemma that sparse PCA corresponds to $\pr{glsm}\left(n, k, d, \mU \right)$ for a proper choice of $\mU = (\mD, \mQ, \{ \mP_{\nu} \}_{\nu \in \mathbb{R}}) \in \pr{uc}(n)$ and $\tau_{\mU}$ so that the lower bound $n = \tilde{\Omega}(\tau_{\mU}^{-4})$ exactly corresponds to the conjectured computational barrier in Sparse PCA. Recall that the hypothesis testing problem $\pr{spca}(n, k, d, \theta)$ has hypotheses
\begin{align*}
&H_0 : (X_1, X_2, \dots, X_n) \sim_{\textnormal{i.i.d.}} \mN(0, I_d) \\
&H_1 : (X_1, X_2, \dots, X_n) \sim_{\textnormal{i.i.d.}} \mN\left(0, I_d + \theta vv^\top\right)
\end{align*}
where $v$ is a $k$-sparse unit vector in $\mathbb{R}^d$ chosen uniformly at random among all such vectors with nonzero entries equal to $1/\sqrt{k}$. 

\begin{lemma}[Lower Bounds for Sparse PCA]
If, then $\pr{spca}(n, k, d, \theta)$ can be expressed as $\pr{glsm}(n, k, d, \mU)$ where $\mU = (\mD, \mQ, \{ \mP_{\nu} \}_{\nu \in \mathbb{R}}) \in \pr{uc}(n)$ is given by
$$\mP_{\nu} = \mN\left( 2\nu \sqrt{\frac{\theta \log n}{k}}, 1 \right) \textnormal{ for all } \nu \in \mathbb{R}, \quad \mQ = \mN(0, 1) \quad \textnormal{and} \quad \mD = \mN\left(0, \frac{1}{4\log n} \right)$$
and has valid level of signal $\tau_{\mU} = \Theta\left( \sqrt{\frac{\theta (\log n)^2}{k}} \right)$ if it holds that $\theta (\log n)^2 = o(k)$.
\end{lemma}

\begin{proof}
Note that if $X \sim \mN\left(0, I_d + \theta vv^\top \right)$ then $X$ can be written as
$$X = 2\sqrt{\theta \log n} \cdot gv + G \quad \text{where } g \sim \mN\left(0, \frac{1}{4\log n} \right) \text{ and } G \sim \mN(0, I_d)$$
and where $g$ and $G$ are independent. This follows from the fact that the random variable on the right-hand side above is a jointly Gaussian vector with covariance matrix given by the sum of the covariance matrices of the individual terms. This observation implies that $\pr{spca}(n, k, d, \theta)$ is exactly the problem $\pr{glsm}(n, k, d, \mU)$. Now observe that the probability that $x \sim \mD$ satisfies $x \in [-1, 1]$ is $1 - o(n^{-1})$ by standard Gaussian tail bounds. Fix some $\nu\in [-1, 1]$ and let $t = 2\nu \sqrt{\frac{\theta \log n}{k}}$. Note that
$$\left|\frac{d\mP_\nu}{d\mQ}(x) - \frac{d\mP_{-\nu}}{d\mQ}(x) \right| = \left| e^{tx - t^2/2} - e^{-tx - t^2/2}\right| = \Theta \left( |tx| \right)$$
if $|tx| = o(1)$. As long as $x = O(\sqrt{\log n})$, it follows that $|tx| = O(\tau_{\mU}) = o(1)$ from the definition of $\tau_{\mU}$ and fact that $\theta (\log n)^2 = o(k)$. Note that $x = O(\sqrt{\log n})$ occurs with probability at least $1 - O(n^{-K})$ for any constant $K > 0$ under each of $\mP_{\nu}$ where $\nu \in [-1, 1]$ and $\mQ$ by standard Gaussian tail bounds. Now observe that
$$\left|\frac{d\mP_\nu}{d\mQ}(x) + \frac{d\mP_{-\nu}}{d\mQ}(x) - 2\right| = \left| e^{tx - t^2/2} + e^{-tx - t^2/2} - 2\right| = \Theta(t^2)$$
holds if $|tx| = o(1)$, which is true as long as $x = O(\sqrt{\log n})$ and thus holds with probability $1 - O(n^{-K})$ for any fixed $K > 0$. Since $t^2 = O(\tau_{\mU}^2)$ for any $\nu \in [-1, 1]$, this completes the proof that $\mU \in \pr{uc}(n)$ with level of signal $\tau_{\mU}$.
\end{proof}

Combining this lemma with Theorem \ref{thm:glsm-lb} yields the $k\pr{-bpds}$ conjecture implies a computational lower bound for Sparse PCA at the barrier $n = \tilde{o}(k^2/\theta^2)$ as long as $\theta(\log n)^2 = o(k)$ and $k = o(\sqrt{d})$, which matches the planted clique lower bounds in \cite{berthet2013optimal, berthet2013complexity, wang2016statistical, gao2017sparse, brennan2018reducibility}. Similar calculations to those in the above corollary can be used to identify the computational lower bound implied by Theorem \ref{thm:glsm-lb} for many other choices of $\mU = (\mD, \mQ, \{ \mP_{\nu} \}_{\nu \in \mathbb{R}}) \in \pr{uc}(n)$. Some examples are:
\begin{itemize}
\item Balanced sparse Gaussian mixtures where $\mQ = \mN(0, 1)$, $\mP_{\nu} = \mN(\theta \nu, 1)$ where $\mD$ is any symmetric distribution over $[-1, 1]$ can be shown to satisfy that $\tau_{\mU} = \Theta\left(\theta \sqrt{\log n}\right)$ if $\theta \sqrt{\log n} = o(1)$.
\item The Bernoulli case where $\mQ = \text{Bern}(1/2)$, $\mP_{\nu} = \text{Bern}(1/2 + \theta \nu)$ and $\mD$ is any symmetric distribution over $[-1, 1]$ can be shown to satisfy that $\tau_{\mU} = \Theta\left(\theta \right)$ if $\theta \le 1/2$.
\item Sparse mixtures of exponential distributions where $\mQ = \text{Exp}(\lambda)$, $\mP_{\nu} = \text{Exp}(\lambda + \theta \nu)$ and $\mD$ is any symmetric distribution over $[-1, 1]$ can be shown to satisfy that $\tau_{\mU} = \tilde{\Theta}\left( \theta \lambda^{-1} \log n \right)$ if it holds that $\theta \log n = o(\lambda)$.
\item Sparse mixtures of centered Gaussians with difference variances where $\mQ = \mN(0, 1)$, $\mP_{\nu} = \mN(0, 1 + \theta \nu)$ and $\mD$ is any symmetric distribution over $[-1, 1]$ can be shown to satisfy that $\tau_{\mU} = \Theta\left(\theta \log n \right)$ if $\theta \log n = o(1)$.
\end{itemize}
We remark that $\tau_{\mU}$ can be calculated for many more choices of $\mD, \mQ$ and $\mP_{\nu}$ using the computations outlined in the discussion above on the implications of our result for concentrated LLR.

\section{Computational Lower Bounds for Recovery and Estimation}
\label{subsec:2-estimation}

In this section, we outline several ways to deduce that our reductions to the hypothesis testing formulations in the previous section imply computational lower bounds for natural recovery and estimation formulations of the problems introduced in Section \ref{sec:1-problems}. We first introduce a notion of average-case reductions in total variation between recovery problems and note that most of our reductions satisfy these stronger conditions in addition to those in Section \ref{subsec:2-tvreductions}. We then discuss alternative methods of obtaining hardness of recovery and estimation in the problems that we consider directly from computational lower bounds for detection.  

In the previous section, we showed that lower bounds for our detection formulations of $\pr{rsme}$ and $\pr{glsm}$ directly imply lower bounds for natural estimation and recovery variants, respectively. In Section \ref{sec:3-tensor}, we showed that our lower bounds against blackboxes solving the detection formulation of tensor PCA with a low false positive probability of error directly implies hardness of estimating $v$ in $\ell_2$ norm. As discussed in Section \ref{subsec:1-problems-hidden-partition}, the problems of recovering the hidden partitions in $\pr{ghpm}$ and $\pr{bhpm}$ have very different barriers than the testing problem we consider in this work. In this section, we will discuss recovery and estimation hardness for the remaining problems from Section \ref{sec:1-problems}.

\subsection{Our Reductions and Computational Lower Bounds for Recovery}

Similar to the framework in Section \ref{subsec:2-tvreductions} for reductions showing hardness of detection, there is a natural notion of a reduction in total variation transferring computational lower bounds between recovery problems. Let $\mP(n, \tau)$ denote the recovery problem of estimating $\theta \in \Theta_\mP$ within some small loss $\ell_{\mP}(\theta, \hat{\theta}) \le \tau$ given an observation from the distribution $\mP_D(\theta)$. Here, $n$ is any parameterization such that this observation has size $\text{poly}(n)$ and, as per usual, $\ell_\mP$, $\Theta_\mP$ and $\tau$ are implicitly functions of $n$. Define the problem $\mP'(N, \tau')$ analogously. The following is the definition of a reduction in total variation between $\mP$ and $\mP'$.

\begin{definition}[Reductions in Total Variation between Recovery Problems] \label{defn:tvreductions-recovery}
A $\textnormal{poly}(n)$ time algorithm $\mathcal{A}$ sending valid inputs for $\mP(n, \tau)$ to valid inputs for $\mP'(N, \tau')$ is a reduction in total variation from $\mP$ to $\mP'$ if the following criteria are met for all $\theta \in \Theta_{\mP}$:
\begin{enumerate}
\item There is a distribution $\mD(\theta)$ over $\Theta_{\mP'}$ such that
$$\TV\left( \mathcal{A}(\mP_D(\theta)), \, \bE_{\theta' \sim \mD(\theta)} \, \mP'_D(\theta')\right) = o_n(1)$$
\item There is a $\textnormal{poly}(n)$ time randomized algorithm $\mathcal{B}(X, \hat{\theta'})$ mapping instances $X$ of $\mP(n, \tau)$ and $\hat{\theta'} \in \Theta_{\mP'}$ to $\hat{\theta} \in \Theta_{\mP}$ with the following property: if $X \sim \mP_D(\theta)$, $\theta'$ is an arbitrary element of $\textnormal{supp} \, \mD(\theta)$ and $\hat{\theta'}$ is guaranteed to satisfy that $\ell_{\mP'}(\theta', \hat{\theta'}) \le \tau'$, then $\mathcal{B}(X, \hat{\theta'})$ outputs some $\hat{\theta}$ with $\ell_{\mP}(\theta, \hat{\theta}) \le \tau$ with probability $1 - o_n(1)$.
\end{enumerate}
\end{definition}

While this definition has a number of technical conditions, it is conceptually simple. A randomized algorithm $\mathcal{A}$ is a reduction in total variation from $\mP$ to $\mP'$ if it maps a sample from the conditional distribution $\mP_D(\theta)$ approximately to a sample from a mixture of $\mP_D(\theta')$, where the mixture is over a distribution $\mD(\theta)$ determined by $\theta$. Furthermore, there must be an efficient way $\mathcal{B}$ to recover a good estimate $\hat{\theta}$ of $\theta$ given a good estimate $\hat{\theta'}$ of $\theta'$ and the original instance $X$ of $\mP$. The reason that (2) must be true for any $\theta' \in \textnormal{supp} \, \mD(\theta)$ is that, to transfer recovery hardness from $\mP$ to $\mP'$, the algorithm $\mathcal{B}$ will be applied to the output $\theta'$ of a blackbox solving $\mP'$ applied to $\mathcal{A}(X)$. In this setting, $\theta'$ and $X$ are dependent and allowing $\theta' \in \textnormal{supp} \, \mD(\theta)$ in the definition above accounts for this. Note that, as per usual, $\mathcal{A}$ must satisfy the properties in the definition above oblivious to $\theta$. The following lemma shows that Definition \ref{defn:tvreductions-recovery} fulfills its objective and transfers hardness of recovery from $\mP$ to $\mP'$. Its proof is simple and deferred to Appendix \ref{subsec:appendix-2-tv}.

\begin{lemma} \label{lem:tvreductions-recovery}
Suppose that there is reduction $\mathcal{A}$ from $\mP(n, \tau)$ to $\mP'(N, \tau')$ satisfying the conditions in Definition \ref{defn:tvreductions-recovery}. If there is a polynomial time algorithm $\mathcal{E}'$ solving $\mP'(N, \tau')$ with probability at least $p$, then there is a polynomial time algorithm $\mathcal{E}$ solving $\mP(n, \tau)$ with probability at least $p - o_n(1)$.
\end{lemma}

The recovery variants of the problems we consider all take the form of $\mP(n, \tau)$. For example, $\Theta_{\mP}$ is the set of $k$-sparse vectors of bounded norm and $\ell_{\mP}$ is $\ell_2$ in $\pr{mslr}$, and $\Theta_{\mP}$ is the set of $(n/k)$-subsets of $[n]$ and $\ell_{\mP}$ is the size of the symmetric difference between two $(n/k)$-subsets in $\pr{isbm}$. In $\pr{rslr}$, $\Theta_{\mP}$ can be taken to be the set of al $(u, \mathcal{A})$ where $u$ is a $k$-sparse vector of bounded norm and $\mathcal{A}$ is a valid adversary. The loss $\ell_{\mP}$ is then independent of $\mathcal{A}$ and given by the $\ell_2$ norm on $u$. Throughout Parts \ref{part:reductions} and \ref{part:lower-bounds}, the guarantees we proved for our reductions among the hypothesis testing formulations from Section \ref{subsec:2-formulations} generally took the form of condition (1) in Definition \ref{defn:tvreductions-recovery}. Some reductions had a post-processing step where coordinates in the output instance are randomly permuted or subsampled, but these can simply be removed to yield a guarantee matching the form of (1). In light of this and Lemma \ref{lem:tvreductions-recovery}, it suffices to show that our reductions also satisfy condition (2) in Definition \ref{defn:tvreductions-recovery}. We outline how to construct these algorithms $\mathcal{B}$ for each of our remaining problems below.

\paragraph{Reductions from $\pr{bpc}$ and $k\pr{-bpc}$.} All of our reductions from $\pr{bpc}$ and $k\pr{-bpc}$ to $\pr{rsme}$, $\pr{neg-spca}$, $\pr{mslr}$ and $\pr{rslr}$ map from an instance with left biclique vertex set $S$ with $|S| = k_m$ to an instance with hidden vector $u = \gamma \cdot k_m^{-1/2} \cdot \mathbf{1}_S$ for some $\gamma \in (0, 1)$. In the notation of Definition \ref{defn:tvreductions-recovery}, $\mD(S)$ is a point mass on $u$. We now outline how such reductions imply hardness of estimation up to any $\ell_2$ error $\tau' = o(\gamma)$.

To verify condition (2) of Definition \ref{defn:tvreductions-recovery}, it suffices to give an efficient algorithm $\mathcal{B}$ recovering $S$ and the right biclique vertices $S'$ from the original $\pr{bpc}$ or $k\pr{-bpc}$ instance $G$ and an estimate $\hat{u}$ satisfying that $\| \hat{u} - \gamma \cdot k_m^{-1/2} \cdot \mathbf{1}_S \|_2 \le \tau'$. Suppose that $|S| = k_m$ and $|S'|$ are both $\omega(\log n)$. Let $\hat{S}$ be the set of the largest $k_m$ entries of $\hat{u}$ and note that $\| \gamma^{-1} \cdot \hat{u} - k_m^{-1/2} \cdot \mathbf{1}_S \|_2 = o(1)$, which can be verified to imply that at least $(1 - o(1))k_m$ of $\hat{S}$ must be in $S$. A union bound and Chernoff bound can be used to show that, in a $\pr{bpc}$ instance with left and right biclique sets $S$ and $S'$, there is no right vertex in $[n] \backslash S'$ with at least $3k_m/4$ neighbors in $S$ with probability $1 - o_n(1)$ if $k_m \gg \log n$. Therefore $S'$ is exactly the set of right vertices with at least $5k_m/6$ neighbors in $\hat{S}$ with probability $1 - o_n(1)$. Taking the common neighbors of $S'$ now recovers $S$ with high probability. Thus this procedure of taking the $k_m$ largest entries $\hat{S}$ of $\hat{u}$, taking right vertices with many neighbors in $\hat{S}$ and then taking their common neighborhoods, exactly solves the $\pr{bpc}$ and $k\pr{-bpc}$ recovery problems. We remark that hardness for these exact recovery problems follows from detection hardness, as bipartite Erd\H{o}s-R\'{e}nyi random graphs do not contain bicliques of left and right sizes $\omega(\log n)$ with probability $1 - o_n(1)$.

We remark that for the values of $\gamma$ in our reductions, the condition $\tau = o(\gamma)$ implies tight computational lower bounds for estimation in $\pr{rsme}$, $\pr{neg-spca}$, $\pr{mslr}$ and $\pr{rslr}$. In particular, for $\pr{rsme}$, $\pr{mslr}$ and $\pr{rslr}$, we may take $\tau'$ to be arbitrarily close to $\tau$ in our detection lower bound as long as $\tau' = o(\tau)$. For $\pr{neg-spca}$, a natural estimation analogue is to estimate some $k$-sparse $v$ within $\ell_2$ norm $\tau'$ given $n$ i.i.d. samples from $\mN(0, I_d + vv^\top)$. For this estimation formulation, we may take $\tau' = o(\sqrt{\theta})$ where $\theta$ is as in our detection lower bound.

\paragraph{Reductions from $k\pr{-pc}$.} We now outline how to construct such an algorithm $\mathcal{B}$ for $\pr{isbm}$. We only sketch the details of this construction as a more direct and simpler way to deduce hardness of recovery for $\pr{isbm}$ will be discussed in the next section. We remark that a similar construction of $\mathcal{B}$ also verifies condition (2) for our reduction to $\pr{semi-cr}$. 

For simplicity, first consider $k\pr{-pds-to-isbm}$ without the initial $\pr{To-}k\pr{-Partite-Submatrix}$ step and the random permutations of vertex labels in Steps 2 and 4. Let $S \subseteq [kr^t]$ be the vertex set of the planted dense subgraph in $M_{\text{PD2}}$ and let $F'$ and $F''$ be the given partitions of the indices $[kr^t]$ of $M_{\text{PD2}}$ and the vertices $[kr\ell]$ of the output graph, respectively. Lemma \ref{lem:isbm-rotations} shows that the output instance of $\pr{isbm}$ has its smaller hidden community $C_1$ of size $k\ell$ on the vertices corresponding to the negative entries of the vector $v_{S, F', F''}(K_{r, t})$. Note that, as a function of this set $S$, the mixture distribution $\mD(S)$ is again a point mass. We now will outline how to approximately recover $S$ given a close estimate $\hat{C}_1$ of $C_1$. Suppose that $\hat{C}_1$ is a $k\ell$-subset of $[kr\ell]$ such that $|C_1 \cap \hat{C}_1| \ge (1 - o(1))k\ell$. Construct the vector $\hat{v}$ given by
$$\hat{v}_i = \frac{1}{\sqrt{r^t(r - 1)}} \cdot \left\{ \begin{matrix} 1 & \textnormal{if } i \not\in \hat{C}_1 \\ 1 - r & \textnormal{if } i \in \hat{C}_1 \end{matrix} \right.$$
Since $\ell = \Theta(r^{t - 1})$, a direct calculation shows that $\| \hat{v} - v_{S, F', F''}(K_{r, t}) \|_2 = o(\sqrt{k})$. For each part $F_i''$, consider the vector in $\mathbb{R}^{r\ell}$ formed by restricting $\hat{v}$ to the indices in $F_i''$ and identifying these indices with $[r\ell]$ in increasing order. For each such vector, find the closest column of $K_{r, t}$ to this vector in $\ell_2$ norm. If the index of this column is $j$, add the $j$th smallest element of $F_i'$ to $\hat{S}$. We claim that the resulting set $\hat{S}$ contains at least $(1 - o(1))k$ elements of $S$. The singular values of $K_{r, t}$ computed in Lemma \ref{lem:Krtsv} can be used to show that any two columns of $K_{r, t}$ are separated by an $\ell_2$ distance of $\Omega(1)$. Any part $F_i'$ for which the correct $j \in S \cap F_i'$ was not added to $\hat{S}$ must have satisfied that $\hat{v}$ restricted to the part $F_i''$ was an $\ell_2$ distance of $\Omega(1)$ from the corresponding restriction of $v_{S, F', F''}(K_{r, t})$. Since $\| \hat{v} - v_{S, F', F''}(K_{r, t}) \|_2 = o(\sqrt{k})$, the number of such $j$ incorrectly added to $\hat{S}$ is $o(k)$, verifying the claim.

Now consider $k\pr{-pds-to-isbm}$ with its first step and the random permutations. Since the random index permutation in $\pr{To-}k\pr{-Partite-Submatrix}$ and the subsequent random permutations in Steps 2 and 4 are all generated by the reduction, they can also be remembered and used in the algorithm $\mathcal{B}$ recovering the clique of the input $k\pr{-pc}$ instance. When combined with the subroutine recovering $\hat{S}$ from $\hat{C}_1$, these permutations are sufficient to identify a set of $k$ vertices overlapping with the clique in at least $(1 - o(1))k$ vertices. Now using a similar procedure to the one mentioned above for $\pr{bpc}$, together with the input $k\pr{-pc}$ instance $G$, this is sufficient to exactly recover the hidden clique vertices.

\subsection{Relationship Between Detection and Recovery}

As shown in the previous section, computational lower bounds from recovery can generally be deduced from our reductions because they are also reductions in total variation between recovery problems. We now will outline how our computational lower bounds for detection all either directly or almost directly imply hardness of recovery. As in Section 10 of \cite{brennan2018reducibility}, our approach is to produce two independent instances $X$ and $X'$ from $\mP_D(\theta)$ without knowing $\theta$, to use $X$ to recover an estimate $\hat{\theta}$ of $\theta$ and then to verify that $\hat{\theta}$ is a good estimate of $\theta$ using $X'$. If $\hat{\theta}$ is confirmed to closely approximate $\theta$ using $X'$, then output $H_1$, and otherwise output $H_0$. This recipe shows detection is easier than recovery as long as there are efficient ways to produce the pair $(X, X')$ and to verify $\hat{\theta}$ is a good estimate given a fresh sample $X'$. In general, the purpose of cloning into the pair $(X, X')$ is to sidestep the fact that $X$ and $\hat{\theta}$ are dependent random variables, which complicates analyzing the verification step. In contrast, $\hat{\theta}$ and $X'$ are conditionally independent given $\theta$. We now show that this recipe applies to each of our problems.

\paragraph{Sample Splitting.} In problems with samples, a natural way to produce $X$ and $X'$ is to simply split the set of samples into two groups. This yields a means to directly transfer computational lower bounds from detection to recovery for $\pr{rsme}$, $\pr{neg-spca}$, $\pr{mslr}$ and $\pr{rslr}$. As we already discussed one way our reductions imply computational lower bounds for the recovery variants of these problems in the previous section, we only sketch the main ideas here. 

We first show an efficient algorithm for recovery in $\pr{mslr}$ yields an efficient algorithm for detection. Consider the detection problem $\pr{mslr}(2n, k, d, \tau)$, and assume there is a blackbox $\mathcal{E}$ solving the recovery problem $\pr{mslr}(n, k, d, \tau')$ with probability $1 - o_n(1)$ for some $\tau' = o(\tau)$. If the samples from $\pr{mslr}(2n, k, d, \tau)$ are $(X_1, y_1), (X_2, y_2), \dots, (X_{2n}, y_{2n})$, apply $\mathcal{E}$ to $(X_1, y_1), \dots, (X_{n}, y_{n})$ to produce an estimate $\hat{u}$. Under $H_1$, there is some true $u = \tau \cdot k^{-1/2} \cdot \mathbf{1}_S$ for some $k$-set $S$ and it holds that $\| \hat{u} - u \|_2 = o(\tau)$. As in the previous section, taking the largest $k$ coordinates of $\hat{u}$ yields a set $\hat{S}$ containing at least $(1 - o(1))k$ elements of $S$. The idea is now that we almost know the true set $S$, detection using the second group of $n$ samples essentially reduces to $\pr{mslr}$ without sparsity and is easy down to the information-theoretic limit. More precisely, consider using the second half of the samples to form the statistic
$$Z = \frac{1}{\tau^2 (1 + \tau^2)} \sum_{i = n + 1}^{2n} \left( y_i^2 - 1 - \tau^2 \right) \cdot \left\langle (X_i)_{\hat{S}}, \hat{u}_{\hat{S}} \right\rangle^2$$
where $v_{\hat{S}}$ denotes the vector equal to $v$ on the indices in $\hat{S}$ and zero elsewhere. Note that conditioned on $S$, the second group of $n$ samples is independent of $\hat{S}$. Under $H_0$, it can be verified that $\bE[Z] = 0$ and $\text{Var}[Z] = O(n)$. Under $H_1$, it can be verified that $\| \hat{u} \|_2$ and $\| \hat{u}_{\hat{S}} \|_2$ are both $(1 + o(1))\tau$ and furthermore that $\langle u, \hat{u}_{\hat{S}} \rangle \ge (1 - o(1)) \tau^2$. Now note that since $y_i = R_i \cdot \langle X_i, u \rangle + g_i$ where $g_i \sim \mN(0, 1)$ and $R_i \sim \text{Rad}$, we have that
\begin{align*}
\left( y_i^2 - 1 - \tau^2 \right) \cdot \left\langle (X_i)_{\hat{S}}, \hat{u}_{\hat{S}} \right\rangle^2 &= \langle X_i, u \rangle^2 \cdot \left\langle X_i, \hat{u}_{\hat{S}} \right\rangle^2 - \tau^2 \cdot \left\langle X_i, \hat{u}_{\hat{S}} \right\rangle^2 + 2R_i g_i \cdot \langle X_i, u \rangle \cdot \left\langle X_i, \hat{u}_{\hat{S}} \right\rangle^2 \\
&\quad \quad + (g_i^2 - 1) \cdot \left\langle X_i, \hat{u}_{\hat{S}} \right\rangle^2 
\end{align*}
The last two terms are mean zero and the second term has expectation $-(1 + o(1))\tau^4$ since $\| \hat{u}_{\hat{S}} \|_2 = (1 + o(1))\tau$. Directly expanding the first term in terms of the components of $X_i$ yields that its expectation is given by $2 \langle u, \hat{u}_{\hat{S}} \rangle^2 + \| u \|_2^2 \cdot \| \hat{u}_{\hat{S}} \|_2^2 \ge 3(1 - o(1))\tau^4$. Combining these computations yields that $\bE[Z] \ge 2n(1 - o(1)) \tau^2$, and it can again be verified that $\text{Var}[Z] = O(n)$. Chebyshev's inequality now yields that thresholding $Z$ at $n\tau^2$ distinguishes $H_0$ and $H_1$ as long as $\tau^2 \sqrt{n} \gg 1$. Since the information-theoretic limit of the detection formulation of $\pr{mslr}$ is when $n = \Theta(k \log d/\tau^4)$ \cite{fan2018curse}, whenever this problem is possible it holds that $\tau^2 \sqrt{n} \gg 1$. Therefore, whenever detection is possible, the reduction outlined above shows how to produce a test solving detection in $\pr{mslr}$ using an estimator with $\ell_2$ error $\tau' = o(\tau)$.

Similar reductions transfer hardness of recovery to detection for $\pr{neg-spca}$, $\pr{rsme}$ and $\pr{rslr}$. For $\pr{neg-spca}$ and $\pr{rsme}$, the same argument as above can be shown to work with the test statistic given by $Z = \sum_{i = 1}^{2n} \langle X_i, \hat{u}_{\hat{S}} \rangle^2$, and the same $Z$ used above for $\pr{mslr}$ suffices in the case of $\pr{rslr}$. We remark that to show these statistics $Z$ solve the detection variants of $\pr{rsme}$ and $\pr{rslr}$, it is important to use detection formulations incorporating the exact form of our adversarial constructions, which are $\pr{isgm}$ in the case of $\pr{rsme}$ and the adversary described in Section \ref{sec:2-supervised} in the case of $\pr{rslr}$. An arbitrary adversary could corrupt instances of $\pr{rsme}$ and $\pr{rslr}$ to cause these statistics $Z$ to not distinguish between $H_0$ and $H_1$. Because our detection lower bounds apply to these fixed adversaries rather than requiring an arbitrary adversary, this argument yields the desired hardness of estimation for $\pr{rsme}$ and $\pr{rslr}$.

\paragraph{Post-Reduction Cloning.} In problems without samples, producing the pair $(X, X')$ requires an additional reduction step. We now outline how to produce such a pair and verification step for $\pr{isbm}$. The high-level idea is to stop our reduction to $\pr{isbm}$ before the final thresholding step, apply Gaussian cloning as in Section 10 of \cite{brennan2018reducibility}, then to continue the reduction with both copies, eventually using one to verify the output of a recovery blackbox applied to the other. A similar argument can be used to show computational lower bounds for recovery in $\pr{semi-cr}$.

Consider the reduction $k\pr{-pds-to-isbm}$ without the final thresholding step, outputting the matrix $M_{\text{R}} \in \mathbb{R}^{kr\ell \times kr\ell}$ at the end of Step 3. Now consider adding the following three steps to this reduction, given access to a recovery blackbox $\mathcal{E}$. More precisely, given an instance of $\pr{isbm}(n, k, P_{11}, P_{12}, P_{22})$ with 
$$P_{11} = P_0 + \gamma, \quad P_{12} = P_0 - \frac{\gamma}{k - 1} \quad \text{and} \quad P_{22} = P_0 + \frac{\gamma}{(k - 1)^2}$$
as in Section \ref{sec:3-community}, suppose $\mathcal{E}$ is guaranteed to output an $(n/k)$-subset of vertices $\hat{C}_1 \subseteq [n]$ with $|C_1 \cap \hat{C}_1| \ge (1 + \epsilon)n/k^2$ with probability $1 - o_n(1)$ for some $\epsilon = \Omega(1)$. Here, $C_1$ is the true hidden smaller community of the input $\pr{isbm}$ instance. Observe that when $\epsilon = \Theta(1)$, the blackbox $\mathcal{E}$ has the weak guarantee of recovering marginally more than a trivial $1/k$ fraction of $C_1$. This exactly matches the notion of weak recovery discussed in Section \ref{subsec:1-problems-sbm}.
\begin{enumerate}
\item Sample $W \sim \mN(0, 1)^{\otimes n \times n}$ and form
$$M_{\text{R}}^1 = \frac{1}{\sqrt{2}} \left( M_{\text{R}} + W \right) \quad \text{and} \quad M_{\text{R}}^2 = \frac{1}{\sqrt{2}} \left( M_{\text{R}} - W \right)$$
\item Using each of $M_{\text{R}}^1$ and $M_{\text{R}}^2$, complete the reduction $k\pr{-pds-to-isbm}$ omitting the random permutation in Step 4, and complete the additional steps from Corollary \ref{thm:isbm-mod} replacing $\mu$ with $\mu/\sqrt{2}$. Let the two output graphs be $G^1$ and $G^2$.
\item Let $\hat{C}_1$ be the output of $\mathcal{E}$ applied to $G^1$. Output $H_0$ if the subgraph of $G^2$ restricted to $\hat{C}_1$ has at least $M$ edges, and output $H_1$ otherwise.
\end{enumerate}
We now outline how this solves the detection variant of $\pr{isbm}$. Let $C_1$ be the true hidden smaller community of the instance that $k\pr{-pds-to-isbm}$ would produce if completed using $M_{\text{R}}$. We claim that $G^1$ and $G^2$ are $o(1)$ total variation from independent copies of $\pr{isbm}(n, C_1, P_{11}, P_{12}, P_{22})$ where $P_{11}, P_{12}$ and $P_{22}$ are as above and $\gamma$ is as in Corollary \ref{thm:isbm-mod}, but defined using $\mu/\sqrt{2}$ instead of $\mu$. To see this, note that $M_{\text{R}}$ is $o(1)$ total variation from the distribution
$$M'_{\text{R}} = \frac{\mu(r - 1)}{r} \cdot v(C_1) v(C_1)^\top + Y \quad \text{where} \quad v(C_1)_i = \frac{1}{\sqrt{r^t(r - 1)}} \cdot \left\{ \begin{matrix} 1 & \textnormal{if } i \not\in C_1 \\ 1 - r & \textnormal{if } i \in C_1 \end{matrix} \right.$$
by Lemma \ref{lem:isbm-rotations}, where $Y \sim \mN(0, 1)^{\otimes n \times n}$ and $t$ is the internal parameter used in $k\pr{-pds-to-isbm}$. Now it follows that $M_{\text{R}}^1$ and $M_{\text{R}}^2$ are respectively $o(1)$ total variation from
\begin{align*}
\left( M_{\text{R}}^1 \right)' &= \frac{\mu(r - 1)}{r\sqrt{2}} \cdot v(C_1) v(C_1)^\top + \frac{1}{\sqrt{2}} \left( Y + W \right) \quad \text{and} \\
\left( M_{\text{R}}^2 \right)' &= \frac{\mu(r - 1)}{r\sqrt{2}} \cdot v(C_1) v(C_1)^\top + \frac{1}{\sqrt{2}} \left( Y - W \right)
\end{align*}
The entries of $\frac{1}{\sqrt{2}} \left( Y + W \right)$ and $\frac{1}{\sqrt{2}} \left( Y - W \right)$ are all jointly Gaussian and have variance $1$. Furthermore, they can all be verified to be uncorrelated, implying that these two matrices are independent copies of $\mN(0, 1)^{\otimes n \times n}$ and thus $\left( M_{\text{R}}^1 \right)'$ and $\left( M_{\text{R}}^2 \right)'$ are independent conditioned on $C_1$. Note that $\mu$ has essentially been scaled down by a factor of $\sqrt{2}$ in both of these instances as well. Thus Step 2 above ensures that $G^1$ and $G^2$ are $o(1)$ total variation from independent copies of $\pr{isbm}(n, C_1, P_{11}, P_{12}, P_{22})$. 

Now consider Step 3 above applied to two exact independent copies of $\pr{isbm}(n, C_1, P_{11}, P_{12}, P_{22})$. The guarantee for $\mathcal{E}$ ensures that $|C_1 \cap \hat{C}_1| \ge (1 + \epsilon)n/k^2$ with probability $1 - o_n(1)$. The variance of the number of edges in the subgraph of $G^2$ restricted to $\hat{C}_1$ is $O(n^2/k^2)$ under both $H_0$ and $H_1$, and the expected number of edges in this subgraph is $P_0 \binom{n/k}{2}$ under $H_0$. Under $H_1$, the expected number of edges is
\begin{align*}
\bE\left[ |E(G[\hat{C}_1])| \right] &= \binom{|C_1 \cap \hat{C}_1|}{2} P_{11} + |C_1 \cap \hat{C}_1| \cdot \left( \frac{n}{k} - |C_1 \cap \hat{C}_1| \right) P_{12} + \binom{\frac{n}{k} - |C_1 \cap \hat{C}_1|}{2} P_{22} \\
&= P_0 \binom{n/k}{2} + \frac{\gamma}{2(k - 1)^2} \cdot \left( k|C_1 \cap \hat{C}_1| - \frac{n}{k} \right)^2 - \frac{\gamma}{2(k - 1)} \cdot \left( (k - 2) \cdot |C_1 \cap \hat{C}_1| + \frac{n}{k} \right) \\
&= P_0 \binom{n/k}{2} + \Omega\left( \frac{\gamma \epsilon^2 n^2}{k^4} \right)
\end{align*}
where the last bound holds since $\epsilon = \Omega(1)$ and $k^2 \ll n$.

By Chebyshev's inequality, Step 3 solves the hypothesis testing problem exactly when this difference $\Omega( \gamma \epsilon^2 n^2/k^4)$ grows faster than the $O(n/k)$ standard deviations in the number of edges in the subgraph under $H_0$ and $H_1$. This implies that Step 3 succeeds if it holds that $\gamma \epsilon^2 \gg k^3/n$. The Kesten-Stigum threshold corresponds to $\gamma^2 = \tilde{\Theta}(k^2/n)$ and therefore as long as $\epsilon^4 n = \tilde{\omega}(k^4)$, this argument solves the detection problem just below the Kesten-Stigum threshold. When $\epsilon = \Theta(1)$, this argument shows a computational lower bound up to the Kesten-Stigum threshold for weak recovery in $\pr{isbm}$. Since $k^2 = o(n)$ is always true in our formulation of $\pr{isbm}$, setting $\epsilon = \Theta(\sqrt{k})$ yields that for all $k$ it is hard to recover a $\Theta(1/\sqrt{k})$ fraction of the hidden community $C_1$. This guarantee is much stronger than the analysis in the previous section, which only showed hardness for a blackbox recovering a $1 - o(1)$ fraction of the hidden community. We remark that the same trick used in Step 1 above to produce two independent copies of a matrix with Gaussian noise was used to show estimation lower bounds for tensor PCA in Section \ref{sec:3-tensor}.

\paragraph{Pre-Reduction Cloning.} We remark that there is a general alternative method to obtain the pairs $(X, X')$ in our reductions that we sketch here. Consider applying Bernoulli cloning either directly to the input $\pr{pc}$ or $\pr{pds}$ instance or to the output of $\pr{To-}k\pr{-Partite-Submatrix}$, in the case of reductions from $k\pr{-pc}$, and then running the remaining parts of our reductions on each of the two resulting copies. Ignoring post-processing steps where we permute vertex labels or subsample the output instance, this general approach can be used to yield two copies of the outputs of our reductions that have the same hidden structure and are conditionally independent given this hidden structure. The same verification steps outlined above can then be applied to obtain our computational lower bounds for recovery.

\pagebreak

\section*{Acknowledgements}

We are greatly indebted to Jerry Li for introducing the conjectured statistical-computational gap for robust sparse mean estimation and for discussions that helped lead to this work. We thank Ilias Diakonikolas for pointing out the statistical query model construction in \cite{diakonikolas2017statistical}. We thank the anonymous reviewers for helpful feedback that greatly improved the exposition. We also thank Frederic Koehler, Sam Hopkins, Philippe Rigollet, Enric Boix-Adser\`{a}, Dheeraj Nagaraj, Rares-Darius Buhai, Alex Wein, Ilias Zadik, Dylan Foster and Austin Stromme for inspiring discussions on related topics. This work was supported in part by MIT-IBM Watson AI Lab and NSF CAREER award CCF-1940205.

\bibliography{GB_BIB.bib}
\bibliographystyle{alpha}

\pagebreak

\part{Appendix}

\appendix

\section{Deferred Proofs from Part \ref{part:reductions}}
\label{sec:appendix-2}

\subsection{Proofs of Total Variation Properties}
\label{subsec:appendix-2-tv}

In this section, we present several deferred proofs from Sections \ref{subsec:2-tvreductions} and \ref{subsec:2-estimation}. We first prove Lemma \ref{lem:tvacc}. 

\begin{proof}[Proof of Lemma \ref{lem:tvacc}]
This follows from a simple induction on $m$. Note that the case when $m = 1$ follows by definition. Now observe that by the data-processing and triangle inequalities of total variation, we have that if $\mathcal{B} = \mathcal{A}_{m-1} \circ \mathcal{A}_{m-2} \circ \cdots \circ \mathcal{A}_1$ then
\begin{align*}
\TV\left( \mathcal{A}(\mP_0), \mP_m \right) &\le \TV\left( \mathcal{A}_m \circ \mathcal{B}(\mP_0), \mathcal{A}_m(\mP_{m - 1}) \right) + \TV\left(\mathcal{A}_m(\mP_{m - 1}), \mP_m \right) \\
&\le \TV\left( \mathcal{B}(\mP_0), \mP_{m - 1} \right) + \epsilon_m \\
&\le \sum_{i = 1}^m \epsilon_i
\end{align*}
where the last inequality follows from the induction hypothesis applied with $m - 1$ to $\mathcal{B}$. This completes the induction and proves the lemma.
\end{proof}

We now prove Lemma \ref{lem:bernproduct} upper bounding the total variation distance between vectors of unplanted and planted samples from binomial distributions.

\begin{proof}[Proof of Lemma \ref{lem:bernproduct}]
Given some $P \in [0, 1]$, we begin by computing $\chi^2\left( \textnormal{Bern}(P) + \textnormal{Bin}(m - 1, Q), \textnormal{Bin}(m, Q) \right)$. For notational convenience, let $\binom{a}{b} = 0$ if $b > a$ or $b < 0$. It follows that
\allowdisplaybreaks
\begin{align*} 
&1 + \chi^2\left( \textnormal{Bern}(P) + \textnormal{Bin}(m - 1, Q), \textnormal{Bin}(m, Q) \right) \\
&\quad \quad = \sum_{t = 0}^{m} \frac{\left((1 - P) \cdot \binom{m - 1}{t} Q^t (1 - Q)^{m - 1 - t} + P \cdot \binom{m - 1}{t - 1} Q^{t - 1} (1 - Q)^{m - t} \right)^2}{\binom{m}{t} Q^t (1 - Q)^{m - t}} \\
&\quad \quad = \sum_{t = 0}^{m} \binom{m}{t} Q^t (1 - Q)^{m - t} \left( \frac{m - t}{m} \cdot \frac{1 - P}{1 - Q} + \frac{t}{m} \cdot \frac{P}{Q} \right)^2 \\
&\quad \quad = \bE\left[ \left( \frac{m - X}{m} \cdot \frac{1 - P}{1 - Q} + \frac{X}{m} \cdot \frac{P}{Q} \right)^2 \right] \\
&\quad \quad = \bE\left[ \left( 1 + \frac{X - mQ}{m} \cdot \frac{P - Q}{Q(1 - Q)} \right)^2 \right] \\
&\quad \quad = 1 + \frac{2(P - Q)}{mQ(1 - Q)} \cdot \bE[X - mQ] + \frac{(P - Q)^2}{m^2Q^2(1 - Q)^2} \cdot \bE\left[(X - Qm)^2\right] \\
&\quad \quad = 1 + \frac{(P - Q)^2}{mQ(1 - Q)}
\end{align*}
where $X \sim \textnormal{Bin}(m, Q)$ and the second last equality follows from $\bE[X] = Qm$ and $\bE[(X - Qm)^2] = \text{Var}[X] = Q(1 - Q)m$. The concavity of $\log$ implies that $\KL(\mP, \mQ) \le \log\left( 1 + \chi^2(\mP, \mQ) \right) \le \chi^2(\mP, \mQ)$ for any two distributions with $\mP$ absolutely continuous with respect to $\mQ$. Pinsker's inequality and tensorization of $\KL$ now imply that
\begin{align*}
&2 \cdot \TV\left( \otimes_{i = 1}^k \left( \textnormal{Bern}(P_i) + \textnormal{Bin}(m - 1, Q) \right), \textnormal{Bin}(m, Q)^{\otimes k} \right)^2 \\
&\quad \quad \le \KL\left( \otimes_{i = 1}^k \left( \textnormal{Bern}(P_i) + \textnormal{Bin}(m - 1, Q) \right), \textnormal{Bin}(m, Q)^{\otimes k} \right) \\
&\quad \quad = \sum_{i = 1}^k \KL\left( \textnormal{Bern}(P_i) + \textnormal{Bin}(m - 1, Q), \textnormal{Bin}(m, Q) \right) \\
&\quad \quad \le \sum_{i = 1}^k \chi^2\left( \textnormal{Bern}(P_i) + \textnormal{Bin}(m - 1, Q), \textnormal{Bin}(m, Q) \right) = \sum_{i = 1}^k \frac{(P_i - Q)^2}{mQ(1 - Q)}
\end{align*}
which completes the proof of the lemma.
\end{proof}

We now prove Lemma \ref{lem:bintv} on the total variation distance between two binomial distributions.

\begin{proof}[Proof of \ref{lem:bintv}]
By applying the data processing inequality for $\TV$ to the function taking the sum of the coordinates of a vector, we have that
\begin{align*}
2 \cdot \TV\left( \textnormal{Bin}(n, P), \textnormal{Bin}(n, Q) \right)^2 &\le 2 \cdot \TV\left( \textnormal{Bern}(P)^{\otimes n}, \textnormal{Bern}(Q)^{\otimes n} \right)^2 \\
&\le \KL\left( \textnormal{Bern}(P)^{\otimes n}, \textnormal{Bern}(Q)^{\otimes n} \right) \\
&= n \cdot \KL\left( \textnormal{Bern}(P), \textnormal{Bern}(Q) \right) \\
&\le n \cdot \chi^2\left( \textnormal{Bern}(P), \textnormal{Bern}(Q) \right) \\
&= n \cdot \frac{(P - Q)^2}{Q(1 - Q)}
\end{align*}
The second inequality is an application of Pinsker's, the first equality is tensorization of $\KL$ and the third inequality is the fact that $\chi^2$ upper bounds $\KL$ by the concavity of $\log$. This completes the proof of the lemma.
\end{proof}

We conclude this section with a proof of Lemma \ref{lem:tvreductions-recovery}, establishing the key property of reductions in total variation among recovery problems.

\begin{proof}[Proof of Lemma \ref{lem:tvreductions-recovery}]
As in the proof of Lemma \ref{lem:3a} from \cite{brennan2018reducibility}, this lemma follows from a simple application of the definition of $\TV$. Suppose that there is such an $\mathcal{E}'$. Now consider the algorithm $\mathcal{E}$ that proceeds as follows on an input $X$ of $\mP(n, \tau)$:
\begin{enumerate}
\item compute $\mathcal{A}(X)$ and the output $\hat{\theta'}$ of $\mathcal{E}'$ on input $\mathcal{A}(X)$; and
\item output the result $\hat{\theta} \gets \mathcal{B}(X, \hat{\theta'})$.
\end{enumerate}
Suppose that $X \sim \mP_D(\theta)$ for some $\theta \in \Theta_{\mP}$. Consider a coupling of $X$, the randomness of $\mathcal{A}$ and $Y \sim \bE_{\theta' \sim \mD(\theta)} \, \mP'_D(\theta')$ such that $\P[\mathcal{A}(X) \neq Y] = o_n(1)$. Since $Y$ is distributed as a mixture of $\mP'_D(\theta')$, conditioned on $\theta'$, it holds that $\mathcal{E}'$ succeeds with probability
$$\bP\left[ \ell_{\mP'}(\mathcal{E}'(Y), \theta') \le \tau' \, \Big| \, \theta' \right] \ge p$$
Marginalizing this over $\theta'$ yields that $\bP\left[ \ell_{\mP'}(\mathcal{E}'(Y), \theta') \le \tau' \text{ for some } \theta' \in \textnormal{supp} \, \mD(\theta) \right] \ge p$. Now since $\mathcal{A}(X) = Y$ is a probability $1 - o_n(1)$ event, we have that the intersection of this and the event above occurs with probability $p - o_n(1)$. Therefore
$$\bP\left[ \ell_{\mP'}(\theta', \hat{\theta'}) \le \tau' \text{ for some } \theta' \in \textnormal{supp} \, \mD(\theta) \right] \ge \bP\left[ \mathcal{A}(X) = Y \textnormal{ and } \mathcal{E}' \textnormal{ succeeds} \right] \ge p - o_n(1)$$
Now note that the definition of $\mathcal{B}$ implies that
\begin{align*}
\bP\left[ \ell_{\mP}(\theta, \hat{\theta}) \le \tau \right] &\ge \bP\left[ \ell_{\mP'}(\theta', \hat{\theta'}) \le \tau' \text{ for some } \theta' \in \textnormal{supp} \, \mD(\theta) \text{ and } \mathcal{B} \text{ succeeds} \right] \\
&\ge \bP\left[ \ell_{\mP'}(\theta', \hat{\theta'}) \le \tau' \text{ for some } \theta' \in \textnormal{supp} \, \mD(\theta) \right] -  \bP\left[ \mathcal{B} \text{ fails} \right] \\
&\ge p - o_n(1)
\end{align*}
which completes the proof of the lemma.
\end{proof}

\subsection{Proofs for To-$k$-Partite-Submatrix}
\label{subsec:appendix-2-k-partite}

In this section, we prove Lemma \ref{lem:submatrix}, which establishes the approximate Markov transition properties of the reduction $\pr{To-}k\textsc{-Partite-Submatrix}$. We first establish analogue of Lemma 6.4 from \cite{brennan2019universality} in the $k$-partite case to analyze the planted diagonal entries in Step 2 of $\pr{To-}k\textsc{-Partite-Submatrix}$.

\begin{lemma}[Planting $k$-Partite Diagonals] \label{lem:plantingdiagonals}
Suppose that $0 < Q < P \le 1$ and $n \ge \left( \frac{P}{Q} + 1 \right) N$ is such that both $N$ and $n$ are divisible by $k$ and $k \le QN/4$. Suppose that for each $t \in [k]$,
$$z_1^t \sim \textnormal{Bern}(P), \quad z_2^t \sim \textnormal{Bin}(N/k - 1, P) \quad \textnormal{and} \quad z_3^t \sim \textnormal{Bin}(n/k, Q)$$
are independent. If $z_4^t = \max \{ z_3^t - z_1^t - z_2^t, 0 \}$, then it follows that
\begin{align*}
\TV\left( \otimes_{t = 1}^k \mL(z_1^t, z_2^t + z_4^t), \left( \textnormal{Bern}(P) \otimes \textnormal{Bin}(n/k - 1, Q) \right)^{\otimes k} \right) &\le 4k \cdot \exp \left( - \frac{Q^2N^2}{48Pkn} \right) + \sqrt{\frac{C_Q k^2}{2n}} \\
\TV\left( \otimes_{t = 1}^k \mL(z_1^t + z_2^t + z_4^t), \textnormal{Bin}(n/k, Q)^{\otimes k} \right) &\le 4k \cdot \exp \left( - \frac{Q^2N^2}{48Pkn} \right)
\end{align*}
where $C_Q = \max \left\{ \frac{Q}{1 - Q}, \frac{1 - Q}{Q} \right\}$.
\end{lemma}

\begin{proof}
Throughout this argument, let $v$ denote a vector in $\{0, 1\}^k$. Now define the event
$$\mathcal{E} = \bigcap_{t = 1}^k \left\{ z_3^t = z_1^t + z_2^t + z_4^t \right\}$$
Now observe that if $z_3^t \ge Qn/k - QN/2k + 1$ and $z_2^t \le P(N/k - 1) + QN/2k$ then it follows that $z_3^t \ge 1 + z_2^t \ge v_t + z_2^t$ for any $v_t \in \{0, 1\}$ since $Qn \ge (P+Q)N$. Now union bounding the probability that $\mathcal{E}$ does not hold conditioned on $z_1$ yields that
\allowdisplaybreaks
\begin{align*}
\bP\left[ \mathcal{E}^C \Big| z_1 = v \right] &\le \sum_{t = 1}^k \bP\left[ z_3^t < v_t + z_2^t \right] \\
&\le \sum_{t = 1}^k \bP\left[ z_3^t < \frac{Qn}{k} - \frac{QN}{2k} + 1 \right] + \sum_{t = 1}^k \bP\left[ z_2^t > P\left(\frac{N}{k} - 1\right) + \frac{QN}{2k} \right] \\
&\le k \cdot \exp\left( - \frac{\left(QN/2k - 1 \right)^2}{3Qn/k} \right) + k \cdot \exp\left( - \frac{\left(QN/2k \right)^2}{2P(N/k - 1)} \right) \\
&\le 2k \cdot \exp \left( - \frac{Q^2N^2}{48Pkn} \right)
\end{align*}
where the third inequality follows from standard Chernoff bounds on the tails of the binomial distribution. Marginalizing this bound over $v \sim \mL(z_1) = \text{Bern}(P)^{\otimes k}$, we have that
$$\bP\left[ \mathcal{E}^C \right] = \bE_{v \sim \mL(z_1)} \bP\left[ \mathcal{E}^C \Big| z_1 = v \right] \le 2k \cdot \exp \left( - \frac{Q^2N^2}{48Pkn} \right)$$
Now consider the total variation error induced by conditioning each of the product measures $\otimes_{t = 1}^k \mL(z_1^t + z_2^t + z_4^t)$ and $\otimes_{t = 1}^k \mL(z_3^t)$ on the event $\mathcal{E}$. Note that under $\mathcal{E}$, by definition, we have that $z_3^t = z_1^t + z_2^t + z_4^t$ for each $t \in [k]$. By the conditioning property of $\TV$ in Fact \ref{tvfacts}, we have
\begin{align*}
\TV\left( \otimes_{t = 1}^k \mL(z_1^t + z_2^t + z_4^t), \mL\left( \left(z_3^t : t \in [k]\right) \Big| \mathcal{E} \right) \right) &\le \bP\left[ \mathcal{E}^C \right] \\
\TV\left( \otimes_{t = 1}^k \mL(z_3^t), \mL\left( \left(z_3^t : t \in [k]\right) \Big| \mathcal{E} \right) \right) &\le \bP\left[ \mathcal{E}^C \right]
\end{align*}
The fact that $\otimes_{t = 1}^k \mL(z_3^t) = \text{Bin}(n/k, Q)^{\otimes k}$ and the triangle inequality now imply that
$$\TV\left( \otimes_{t = 1}^k \mL(z_1^t + z_2^t + z_4^t), \textnormal{Bin}(n/k, Q)^{\otimes k} \right) \le 2 \cdot \bP\left[ \mathcal{E}^C \right] \le 4k \cdot \exp \left( - \frac{Q^2N^2}{48Pkn} \right)$$
which proves the second inequality in the statement of the lemma. It suffices to establish the first inequality. A similar conditioning step as above shows that for all $v \in \{0, 1\}^k$, we have that
\begin{align*}
\TV\left( \otimes_{t = 1}^k \mL\left(v_t + z_2^t + z_4^t \Big| z_1^t = v_t\right), \mL\left( \left(v_t + z_2^t + z_4^t : t \in [k]\right) \Big| z_1 = v \text{ and } \mathcal{E} \right) \right) &\le \bP\left[ \mathcal{E}^C \Big| z_1 = v \right] \\
\TV\left( \otimes_{t = 1}^k \mL\left(z_3^t \Big| z_1^t = v_t \right), \mL\left( \left(z_3^t : t \in [k]\right) \Big| z_1 = v \text{ and } \mathcal{E} \right) \right) &\le \bP\left[ \mathcal{E}^C \Big| z_1 = v \right]
\end{align*}
The triangle inequality and the fact that $z_3 \sim \text{Bin}(n/k, Q)^{\otimes k}$ is independent of $z_1$ implies that
$$\TV\left( \otimes_{t = 1}^k \mL\left(v_t + z_2^t + z_4^t \Big| z_1^t = v_t\right), \text{Bin}(n/k, Q)^{\otimes k} \right) \le 4k \cdot \exp \left( - \frac{Q^2N^2}{48Pkn} \right)$$
By Lemma \ref{lem:bernproduct} applied with $P_t = v_t \in \{0, 1\}$, we also have that
$$\TV\left( \otimes_{t = 1}^k \left( v_t + \text{Bin}(n/k - 1, Q) \right), \text{Bin}(n/k, Q)^{\otimes k} \right) \le \sqrt{\sum_{t = 1}^k \frac{k(v_t - Q)^2}{2nQ(1 - Q)}} \le \sqrt{\frac{C_Q k^2}{2n}}$$
The triangle now implies that for each $v \in \{0, 1\}^k$,
\allowdisplaybreaks
\begin{align*}
&\TV\left( \otimes_{t = 1}^k \mL\left(z_2^t + z_4^t \Big| z_1^t = v_t\right), \text{Bin}(n/k - 1, Q)^{\otimes k} \right) \\
&\quad \quad = \TV\left( \otimes_{t = 1}^k \mL\left(v_t + z_2^t + z_4^t \Big| z_1^t = v_t\right), \otimes_{t = 1}^k \left( v_t + \text{Bin}(n/k - 1, Q) \right) \right) \\
&\quad \quad \le 4k \cdot \exp \left( - \frac{Q^2N^2}{48Pkn} \right) + \sqrt{\frac{C_Q k^2}{2n}}
\end{align*}
We now marginalize over $v \sim \mL(z_1) = \text{Bern}(P)^{\otimes k}$. The conditioning on a random variable property of $\TV$ in Fact \ref{tvfacts} implies that
\begin{align*}
&\TV\left( \otimes_{t = 1}^k \mL(z_1^t, z_2^t + z_4^t), \left( \textnormal{Bern}(P) \otimes \textnormal{Bin}(n/k - 1, Q) \right)^{\otimes k} \right) \\
&\quad \quad \le \bE_{v \sim \text{Bern}(P)^{\otimes k}} \, \TV\left( \otimes_{t = 1}^k \mL\left(z_2^t + z_4^t \Big| z_1^t = v_t\right), \text{Bin}(n/k - 1, Q)^{\otimes k} \right)
\end{align*}
which, when combined with the inequalities above, completes the proof of the lemma.
\end{proof}

We now apply this lemma to prove Lemma \ref{lem:submatrix}. The proof of this lemma is a $k$-partite variant of the argument used to prove Theorem 6.1 in \cite{brennan2019universality}. However, it involves several technical subtleties that do not arise in the non $k$-partite case.

\begin{proof}[Proof of Lemma \ref{lem:submatrix}]
Fix some subset $R \subseteq [N]$ such that $|R \cap E_i| = 1$ for each $i \in [k]$. We will first show that $\mathcal{A}$ maps an input $G \sim \mG(N, R, p, q)$ approximately in total variation to a sample from the planted submatrix distribution $\mathcal{M}_{[n] \times [n]} \left(\mU_n(F), \textnormal{Bern}(p), \textnormal{Bern}(Q) \right)$. By AM-GM, we have that
$$\sqrt{pq} \le \frac{p + q}{2} = 1 - \frac{(1 - p) + (1 - q)}{2} \le 1 - \sqrt{(1 - p)(1 - q)}$$
If $p \neq 1$, it follows that $P = p > Q = 1 - \sqrt{(1 - p)(1 - q)}$. This implies that $\frac{1 - p}{1 - q} = \left( \frac{1 - P}{1 - Q} \right)^2$ and the inequality above rearranges to $\left( \frac{P}{Q} \right)^2 \le \frac{p}{q}$. If $p = 1$, then $Q = \sqrt{q}$ and $\left( \frac{P}{Q} \right)^2 = \frac{p}{q}$. Furthermore, the inequality $\frac{1 - p}{1 - q} \le \left( \frac{1 - P}{1 - Q} \right)^2$ holds trivially. Therefore we may apply Lemma \ref{lem:graphcloning}, which implies that $(G_1, G_2) \sim \mG(N, R, p, Q)^{\otimes 2}$.

Let the random set $U = \{ \pi_1^{-1}(R \cap E_1), \pi_2^{-1}(R \cap E_2), \dots, \pi_k^{-1}(R \cap E_k) \}$ denote the support of the $k$-subset of $[n]$ that $R$ is mapped to in the embedding step of $\pr{To-}k\textsc{-Partite-Submatrix}$. Now fix some $k$-subset $R' \subseteq [n]$ with $|R' \cap F_i| = 1$ for each $i \in [k]$ and consider the distribution of $M_{\text{PD}}$ conditioned on the event $U = R'$. Since $(G_1, G_2) \sim \mG(n, R, p, Q)^{\otimes 2}$, Step 2 of $\pr{To-}k\textsc{-Partite-Submatrix}$ ensures that the off-diagonal entries of $M_{\text{PD}}$, given this conditioning, are independent and distributed as follows:
\begin{itemize}
\item $M_{ij} \sim \text{Bern}(p)$ if $i \neq j$ and $i, j \in R'$; and
\item $M_{ij} \sim \text{Bern}(Q)$ if $i \neq j$ and $i \not \in R'$ or $j \not \in R'$.
\end{itemize}
which match the corresponding entries of $\mathcal{M}_{[n] \times [n]} \left(R' \times R', \textnormal{Bern}(p), \textnormal{Bern}(Q) \right)$. Furthermore, these entries are independent of the vector $\text{diag}(M_{\text{PD}}) = \left( (M_{\text{PD}})_{ii} : i \in [k] \right)$ of the diagonal entries of $M_{\text{PD}}$. It therefore follows that
\begin{align*}
&\TV\left( \mL \left( M_{\text{PD}} \Big| U = R' \right), \mathcal{M}_{[n] \times [n]} \left(R' \times R', \textnormal{Bern}(p), \textnormal{Bern}(Q) \right) \right) \\
&\quad \quad = \TV\left( \mL \left( \text{diag}(M_{\text{PD}}) \Big| U = R' \right), \mathcal{M}_{[n]} \left(R', \textnormal{Bern}(p), \textnormal{Bern}(Q) \right) \right)
\end{align*}
Let $(S_1', S_2', \dots, S_k')$ be any tuple of fixed subsets such that $|S_t'| = N/k$, $S_i' \subseteq F_t$ and $R' \cap F_t \in S_t'$ for each $t \in [k]$. Now consider the distribution of $\text{diag}(M_{\text{PD}})$ conditioned on both $U = R'$ and $(S_1, S_2, \dots, S_k) = (S_1', S_2', \dots, S_k')$. It holds by construction that the $k$ vectors $\text{diag}(M_{\text{PD}})_{F_t}$ are independent for $t \in [k]$ and each distributed as follows:
\begin{itemize}
\item $\text{diag}(M_{\text{PD}})_{S_t'}$ is an exchangeable distribution on $\{0, 1\}^{N/k}$ with support of size $s_1^t \sim \text{Bin}(N/k, p)$, by construction. This implies that $\text{diag}(M_{\text{PD}})_{S_t'} \sim \text{Bern}(p)^{\otimes N/k}$. This can trivially be restated as $\left(M_{R' \cap F_t, R' \cap F_t}, \text{diag}(M_{\text{PD}})_{S_t' \backslash R'}\right) \sim \text{Bern}(p) \otimes \text{Bern}(p)^{\otimes N/k - 1}$.
\item $\text{diag}(M_{\text{PD}})_{F_t \backslash S_t'}$ is an exchangeable distribution on $\{0, 1\}^{N/k}$ with support of size $z_4^t = \max\{s_2^t - s_1^t, 0\}$. Furthermore, $\text{diag}(M_{\text{PD}})_{F_t \backslash S_t'}$ is independent of $\text{diag}(M_{\text{PD}})_{S_t'}$.
\end{itemize}
For each $t \in [k]$, let $z_1^t = M_{R' \cap F_t, R' \cap F_t} \sim \text{Bern}(p)$ and $z_2^t \sim \text{Bin}(N/k - 1, p)$ be the size of the support of $\text{diag}(M_{\text{PD}})_{S_t' \backslash R'}$. As shown discussed in the first point above, we have that $z_1^t$ and $z_2^t$ are independent and $z_1^t + z_2^t = s_1^t$.

Now consider the distribution of $\text{diag}(M_{\text{PD}})$ relaxed to only be conditioned on $U = R'$, and no longer on $(S_1, S_2, \dots, S_k) = (S_1', S_2', \dots, S_k')$. Conditioned on $U = R'$, the $S_t$ are independent and each uniformly distributed among all $N/k$ size subsets of $F_t$ that contain the element $R' \cap F_t$. In particular, this implies that the distribution of $\text{diag}(M_{\text{PD}})_{F_t \backslash R'}$ is an exchangeable distribution on $\{0, 1\}^{n/k - 1}$ with support size $z_2^t + z_4^t$ for each $t$. Note that any $v \sim \mathcal{M}_{[n]} \left(R', \textnormal{Bern}(p), \textnormal{Bern}(Q) \right)$ also satisfies that $v_{F_t \backslash R'}$ is exchangeable. This implies that $\mathcal{M}_{[n]} \left(R', \textnormal{Bern}(p), \textnormal{Bern}(Q) \right)$ and $\text{diag}(M_{\text{PD}})$ are identically distributed when conditioned on their entries with indices in $R'$ and on their support sizes within the $k$ sets of indices $F_t \backslash R'$. The conditioning property of Fact \ref{tvfacts} therefore implies that
\begin{align*}
&\TV\left( \mL \left( \text{diag}(M_{\text{PD}}) \Big| U = R' \right), \mathcal{M}_{[n]} \left(R', \textnormal{Bern}(p), \textnormal{Bern}(Q) \right) \right) \\
&\quad \quad \le \TV\left( \otimes_{t = 1}^k \mL(z_1^t, z_2^t + z_4^t), \left( \textnormal{Bern}(p) \otimes \textnormal{Bin}(n/k - 1, Q) \right)^{\otimes k} \right) \\
&\quad \quad \le 4k \cdot \exp \left( - \frac{Q^2N^2}{48Pkn} \right) + \sqrt{\frac{C_Q k^2}{2n}}
\end{align*}
by the first inequality in Lemma \ref{lem:plantingdiagonals}. Now observe that $U \sim \mU_n(F)$ and thus marginalizing over $R' \sim \mL(U) = \mU_n(F)$ and applying the conditioning property of Fact \ref{tvfacts} yields that
\begin{align*}
&\TV\left( \mathcal{A}(G(N, R, p, q)), \mathcal{M}_{[n] \times [n]} \left(\mU_n(F), \textnormal{Bern}(p), \textnormal{Bern}(Q) \right) \right) \\
&\quad \quad \le \bE_{R' \sim \mU_n(F)} \, \TV\left( \mL \left( M_{\text{PD}} \Big| U = R' \right), \mathcal{M}_{[n] \times [n]} \left(R' \times R', \textnormal{Bern}(p), \textnormal{Bern}(Q) \right) \right)
\end{align*}
since $M_{\text{PD}} \sim \mathcal{A}(\mG(N, R, p, q))$. Applying an identical marginalization over $R \sim \mU_N(E)$ completes the proof of the first inequality in the lemma statement.

It suffices to consider the case where $G \sim \mG(N, q)$, which follows from an analogous but simpler argument. By Lemma \ref{lem:graphcloning}, we have that $(G_1, G_2) \sim \mG(N, Q)^{\otimes 2}$. It follows that the entries of $M_{\text{PD}}$ are distributed as $(M_{\text{PD}})_{ij} \sim_{\text{i.i.d.}} \text{Bern}(Q)$ for all $i \neq j$ independently of $\text{diag}(M_{\text{PD}})$. Now note that the $k$ vectors $\text{diag}(M_{\text{PD}})_{F_t}$ for $t \in [k]$ are each exchangeable and have support size $s_1^t + \max\{ s_2^t - s_1^t, 0 \} = z_1^t + z_2^t + z_4^t$ where $z_1^t \sim \text{Bern}(p)$, $z_2^t \sim \text{Bin}(N/k - 1, p)$ and $s_2^t \sim \text{Bin}(n/k, Q)$ are independent. By the same argument as above, we have that
\begin{align*}
\TV\left( \mL(M_{\text{PD}}), \text{Bern}(Q)^{\otimes n \times n} \right) &= \TV\left( \mL(\text{diag}(M_{\text{PD}})), \text{Bern}(Q)^{\otimes n} \right) \\
&= \TV\left( \otimes_{t = 1}^k \mL\left( z_1^t + z_2^t + z_4^t \right), \text{Bin}(n/k, Q) \right) \\
&\le 4k \cdot \exp \left( - \frac{Q^2N^2}{48Pkn} \right)
\end{align*}
by Lemma \ref{lem:plantingdiagonals}. Since $M_{\text{PD}} \sim \mathcal{A}(\mG(N, q))$, this completes the proof of the lemma.
\end{proof}

\subsection{Proofs for Symmetric 3-ary Rejection Kernels}
\label{subsec:appendix-3-ary}

In this section, we establish the approximate Markov transition properties for symmetric 3-ary rejection kernels introduced in Section \ref{subsec:srk}.

\begin{proof}[Proof of Lemma \ref{lem:srk}]
Define $\mL_1, \mL_2 : X \to \mathbb{R}$ to be
$$\mL_1(x) = \frac{d\mP_+}{d\mQ} (x) - \frac{d\mP_-}{d\mQ} (x) \quad \text{and} \quad \mL_2(x) = \frac{d\mP_+}{d\mQ} (x) + \frac{d\mP_-}{d\mQ} (x) - 2$$
Note that if $x \in S$, then the triangle inequality implies that
\begin{align*}
P_A(x, 1) &\le \frac{1}{2} \left( 1 + \frac{a}{4|\mu_2|} \cdot |\mL_2(x)| + \frac{1}{4|\mu_1|} \cdot |\mL_1(x)| \right) \le 1 \\
P_A(x, 1) &\ge \frac{1}{2} \left( 1 - \frac{a}{4|\mu_2|} \cdot |\mL_2(x)| - \frac{1}{4|\mu_1|} \cdot |\mL_1(x)| \right) \ge 0
\end{align*}
Similar computations show that $0 \le P_A(x, 0) \le 1$ and $0 \le P_A(x, -1) \le 1$, implying that each of these probabilities is well-defined. Now let $R_1 = \bP_{X \sim \mP_+}[X \in S]$, $R_0 = \bP_{X \sim \mQ}[X \in S]$ and $R_{-1} = \bP_{X \sim \mP_-}[X \in S]$ where $R_1, R_0, R_{-1} \ge 1 - \delta$ by assumption.

We now define several useful events. For the sake of analysis, consider continuing to iterate Step 2 even after $z$ is set for the first time for a total of $N$ iterations. Let $A_i^1$, $A_i^0$ and $A_i^{-1}$ be the events that $z$ is set in the $i$th iteration of Step 2 when $B = 1$, $B = 0$ and $B = -1$, respectively. Let $B_i^1 = (A_1^1)^C \cap (A_2^1)^C \cap \cdots \cap (A^1_{i - 1})^C \cap A_i^1$ be the event that $z$ is set for the first time in the $i$th iteration of Step 2. Let $C^1 = A_1^1 \cup A_2^1 \cup \cdots \cup A_N^1$ be the event that $z$ is set in some iteration of Step 2. Define $B_i^0$, $C^0$, $B_i^{-1}$ and $C^{-1}$ analogously. Let $z_0$ be the initialization of $z$ in Step 1.

Now let $Z_1 \sim \mD_1 = \mL(3\textsc{-srk}(1))$, $Z_0 \sim \mD_0 = \mL(3\textsc{-srk}(0))$ and $Z_{-1} \sim \mD_{-1} = \mL(3\textsc{-srk}(-1))$. Note that $\mL(Z_t|B_i^t) = \mL(Z_t|A_i^t)$ for each $t \in \{-1, 0, 1\}$ since $A_i^t$ is independent of $A_1^t, A_2^t, \dots, A_{i-1}^t$ and the sample $z'$ chosen in the $i$th iteration of Step 2. The independence between Steps 2.1 and 2.3 implies that
\allowdisplaybreaks
\begin{align*}
\bP\left[A_i^1\right] &= \bE_{x \sim \mQ}\left[ \frac{1}{2} \left( 1+ \frac{a}{4\mu_2} \cdot \mL_2(x) + \frac{1}{4\mu_1} \cdot \mL_1(x) \right) \cdot \mathbf{1}_{S}(x) \right] \\
&= \frac{1}{2} R_0 + \frac{a}{8\mu_2} \left( R_1 + R_{-1} - 2R_0 \right) + \frac{1}{8\mu_1} \left( R_1 - R_{-1} \right) \ge \frac{1}{2} - \frac{\delta}{2} \left( 1 + \frac{a}{2}|\mu_2|^{-1} + \frac{1}{4}|\mu_1|^{-1} \right) \\
\bP\left[A_i^0 \right] &= \bE_{x \sim \mQ}\left[ \frac{1}{2} \left( 1 - \frac{1 - a}{4\mu_2} \cdot \mL_2(x) \right) \cdot \mathbf{1}_{S}(x) \right] \\
&= \frac{1}{2} R_0 - \frac{1 - a}{8\mu_2} \left( R_1 + R_{-1} - 2R_0 \right) \ge \frac{1}{2} - \frac{\delta}{2} \left( 1 + \frac{1 - a}{4} \cdot |\mu_2|^{-1} \right) \\
\bP\left[A_i^{-1}\right] &= \bE_{x \sim \mQ}\left[ \frac{1}{2} \left( 1+ \frac{a}{4\mu_2} \cdot \mL_2(x) - \frac{1}{4\mu_1} \cdot \mL_1(x) \right) \cdot \mathbf{1}_{S}(x) \right] \\
&= \frac{1}{2} R_0 + \frac{a}{8\mu_2} \left( R_1 + R_{-1} - 2R_0 \right) - \frac{1}{4\mu_1} \left( R_1 - R_{-1} \right) \ge \frac{1}{2} - \frac{\delta}{2} \left( 1 + \frac{a}{2}|\mu_2|^{-1} + \frac{1}{4}|\mu_1|^{-1} \right)
\end{align*}
The independence of the $A_i^t$ for each $t \in \{-1, 0, 1\}$ implies that
$$1 - \bP\left[ C^t \right] = \prod_{i = 1}^N \left( 1 - \bP\left[A_i^t\right] \right) \le \left( \frac{1}{2} + \frac{\delta}{2} \left( 1 + \frac{1}{2}|\mu_2|^{-1} + |\mu_1|^{-1} \right) \right)^N$$
Note that $\mL(Z_t|A_i^t)$ are each absolutely continuous with respect to $\mQ$ or each $t \in \{-1, 0, 1\}$, with Radon-Nikodym derivatives given by
\allowdisplaybreaks
\begin{align*}
\frac{d\mL(Z_1|B_i^1)}{d\mQ} (x) = \frac{d\mL(Z_1|A_i^1)}{d\mQ} (x) &= \frac{1}{2\cdot \bP\left[A_i^1\right]} \left( 1+ \frac{a}{4\mu_2} \cdot \mL_2(x) + \frac{1}{4\mu_1} \cdot \mL_1(x) \right) \cdot \mathbf{1}_S(x) \\
\frac{d\mL(Z_0|B_i^0)}{d\mQ} (x) = \frac{d\mL(Z_0|A_i^0)}{d\mQ} (x) &= \frac{1}{2\cdot \bP\left[A_i^1\right]} \left( 1 - \frac{1 - a}{4\mu_2} \cdot \mL_2(x) \right) \cdot \mathbf{1}_S(x) \\
\frac{d\mL(Z_{-1}|B_i^{-1})}{d\mQ} (x) = \frac{d\mL(Z_{-1}|A_i^{-1})}{d\mQ} (x) &= \frac{1}{2\cdot \bP\left[A_i^1\right]} \left( 1+ \frac{a}{4\mu_2} \cdot \mL_2(x) - \frac{1}{4\mu_1} \cdot \mL_1(x) \right) \cdot \mathbf{1}_S(x)
\end{align*}
Fix one of $t \in \{-1, 0, 1\}$ and note that since the conditional laws $\mL(Z_t|B_i^t)$ are all identical, we have that
$$\frac{d\mD_t}{d\mQ} (x) = \bP\left[C^t \right] \cdot \frac{d\mL(Z_t|B_1^t)}{d\mQ} (x) + \left( 1 - \bP\left[C^t \right] \right) \cdot \mathbf{1}_{z_0}(x)$$
Therefore it follows that
\begin{align*}
\TV\left( \mD_t, \mL(Z_t|B_1^t) \right) &= \frac{1}{2} \cdot \bE_{x \sim \mQ} \left[\left| \frac{d\mD_t}{d\mQ} (x) - \frac{d\mL(Z_t|B_1^t)}{d\mQ} (x) \right| \right] \\
&\le \frac{1}{2} \left( 1 - \bP\left[ C^t \right] \right) \cdot \bE_{x \sim \mQ} \left[ \mathbf{1}_{z_0}(x) + \frac{d\mL(Z_t|B_1^t)}{d\mQ} (x) \right] = 1 - \bP\left[ C^t \right]
\end{align*}
by the triangle inequality. Since $1+ \frac{a}{4\mu_2} \cdot \mL_2(x) + \frac{1}{4\mu_1} \cdot \mL_1(x) \ge 0$ for $x \in S$, we have that
\allowdisplaybreaks
\begin{align*}
&\bE_{x \sim \mQ} \left[\left| \frac{d\mL(Z_1|B_1^1)}{d\mQ} (x) - \left( 1+ \frac{a}{4\mu_2} \cdot \mL_2(x) + \frac{1}{4\mu_1} \cdot \mL_1(x) \right) \right| \right] \\ 
&\quad \quad = \left|\frac{1}{2\cdot \bP\left[A_i^1\right]} - 1 \right| \cdot \bE_{x \sim \mQ^*_n} \left[\left( 1+ \frac{a}{4\mu_2} \cdot \mL_2(x) + \frac{1}{4\mu_1} \cdot \mL_1(x) \right) \cdot \mathbf{1}_S(x) \right] \\
&\quad \quad \quad \quad + \bE_{x \sim \mQ} \left[ \left| 1+ \frac{a}{4\mu_2} \cdot \mL_2(x) + \frac{1}{4|\mu_1|} \cdot \mL_1(x) \right| \cdot \mathbf{1}_{S^C}(x) \right] \\
&\quad \quad \le \left| \frac{1}{2} - \bP[A_i^1] \right| + \bE_{x \sim \mQ} \left[ \left( 1+ \frac{a}{4|\mu_2|} \cdot \left( \frac{d\mP_+}{d\mQ} (x) + \frac{d\mP_-}{d\mQ} (x) +2 \right) \right) \cdot \mathbf{1}_{S^C}(x) \right] \\
&\quad \quad \quad \quad + \bE_{x \sim \mQ} \left[ \frac{1}{4|\mu_1|} \cdot \left( \frac{d\mP_+}{d\mQ} (x) + \frac{d\mP_-}{d\mQ} (x) \right) \cdot \mathbf{1}_{S^C}(x) \right] \\
&\quad \quad \le \frac{\delta}{2} \left( 1 + \frac{a}{2}|\mu_2|^{-1} + \frac{1}{4}|\mu_1|^{-1} \right) + \delta \left( 1 + a|\mu_2|^{-1} + \frac{1}{2}|\mu_1|^{-1} \right) = \delta \left( \frac{3}{2} + \frac{5}{4} |\mu_2|^{-1} + \frac{5}{8} |\mu_1|^{-1} \right)
\end{align*}
By analogous computations, we have that
\allowdisplaybreaks
\begin{align*}
\bE_{x \sim \mQ} \left[\left| \frac{d\mL(Z_0|B_1^0)}{d\mQ} (x) - \left( 1 - \frac{1 - a}{4\mu_2} \cdot \mL_2(x) \right) \right| \right] &\le 2\delta \left(1 + |\mu_1|^{-1} + |\mu_2|^{-1} \right) \\ 
\bE_{x \sim \mQ} \left[\left| \frac{d\mL(Z_{-1}|B_1^{-1})}{d\mQ} (x) - \left( 1+ \frac{a}{4\mu_2} \cdot \mL_2(x) - \frac{1}{4\mu_1} \cdot \mL_1(x) \right) \right| \right] &\le 2\delta \left(1 + |\mu_1|^{-1} + |\mu_2|^{-1} \right) 
\end{align*}
Now observe that
\allowdisplaybreaks
\begin{align*}
\frac{d\mP_+}{d\mQ}(x) &= \left( \frac{1 - a}{2} + \mu_1 + \mu_2 \right) \cdot \left( 1+ \frac{a}{4\mu_2} \cdot \mL_2(x) + \frac{1}{4\mu_1} \cdot \mL_1(x) \right) + (a - 2\mu_2) \cdot \left( 1 - \frac{1 - a}{4\mu_2} \cdot \mL_2(x) \right) \\
&\quad \quad + \left( \frac{1 - a}{2} - \mu_1 + \mu_2 \right) \cdot \left( 1+ \frac{a}{4\mu_2} \cdot \mL_2(x) - \frac{1}{4\mu_1} \cdot \mL_1(x) \right) \\
1 &= \frac{1 - a}{2} \cdot \left( 1+ \frac{a}{4\mu_2} \cdot \mL_2(x) + \frac{1}{4\mu_1} \cdot \mL_1(x) \right) + a \cdot \left( 1 - \frac{1 - a}{4\mu_2} \cdot \mL_2(x) \right) \\
&\quad \quad +\frac{1 - a}{2} \cdot \left( 1+ \frac{a}{4\mu_2} \cdot \mL_2(x) - \frac{1}{4\mu_1} \cdot \mL_1(x) \right) \\
\frac{d\mP_-}{d\mQ}(x) &= \left( \frac{1 - a}{2} - \mu_1 + \mu_2 \right) \cdot \left( 1+ \frac{a}{4\mu_2} \cdot \mL_2(x) + \frac{1}{4\mu_1} \cdot \mL_1(x) \right) + (a - 2\mu_2) \cdot \left( 1 - \frac{1 - a}{4\mu_2} \cdot \mL_2(x) \right) \\
&\quad \quad + \left( \frac{1 - a}{2} + \mu_1 + \mu_2 \right) \cdot \left( 1+ \frac{a}{4\mu_2} \cdot \mL_2(x) - \frac{1}{4\mu_1} \cdot \mL_1(x) \right)
\end{align*}
Let $\mD^*$ be the mixture of $\mL(Z_1 | B_1^1), \mL(Z_0 | B_1^0)$ and $\mL(Z_{-1} | B_1^{-1})$ with weights $\frac{1 - a}{2} + \mu_1 + \mu_2, a - 2\mu_2$ and $\frac{1 - a}{2} - \mu_1 + \mu_2$, respectively. It then follows by the triangle inequality that
\allowdisplaybreaks
\begin{align*}
&\TV\left( 3\textsc{-srk}(\textnormal{Tern}(a, \mu_1, \mu_2)), \mP_+ \right) \\
&\quad \quad \le \TV\left( \mD^*, \mP_+ \right) + \TV\left( \mD^*, 3\textsc{-srk}(\textnormal{Tern}(a, \mu_1, \mu_2)) \right) \\
&\quad \quad \le \left( \frac{1 - a}{2} + \mu_1 + \mu_2 \right) \cdot \bE_{x \sim \mQ} \left[\left| \frac{d\mL(Z_1|B_1^1)}{d\mQ} (x) - \left( 1+ \frac{a}{4\mu_2} \cdot \mL_2(x) + \frac{1}{4\mu_1} \cdot \mL_1(x) \right) \right| \right] \\
&\quad \quad \quad \quad + \left( a - 2\mu_2 \right) \cdot \bE_{x \sim \mQ} \left[\left| \frac{d\mL(Z_0|B_1^0)}{d\mQ} (x) - \left( 1 - \frac{1 - a}{4\mu_2} \cdot \mL_2(x) \right) \right| \right] \\
&\quad \quad \quad \quad + \left( \frac{1 - a}{2} - \mu_1 + \mu_2 \right) \cdot \bE_{x \sim \mQ} \left[\left| \frac{d\mL(Z_{-1}|B_1^{-1})}{d\mQ} (x) - \left( 1+ \frac{a}{4\mu_2} \cdot \mL_2(x) - \frac{1}{4\mu_1} \cdot \mL_1(x) \right) \right| \right] \\
&\quad \quad \quad \quad + \left( \frac{1 - a}{2} + \mu_1 + \mu_2 \right) \cdot \TV\left( \mD_1, \mL(Z_1|B_1^1) \right)  + \left( a - 2\mu_2 \right) \cdot \TV\left( \mD_1, \mL(Z_0|B_1^0) \right) \\
&\quad \quad \quad \quad + \left( \frac{1 - a}{2} - \mu_1 + \mu_2 \right) \cdot \TV\left( \mD_{-1}, \mL(Z_{-1}|B_1^{-1}) \right) \\
&\quad \quad \le 2\delta \left(1 + |\mu_1|^{-1} + |\mu_2|^{-1} \right) + \left( \frac{1}{2} + \delta \left( 1 + |\mu_1|^{-1} + |\mu_2|^{-1} \right) \right)^N
\end{align*}
A symmetric argument shows analogous upper bounds on both $\TV\left( 3\textsc{-srk}(\textnormal{Tern}(a, -\mu_1, \mu_2)), \mP_- \right)$ and $\TV\left( 3\textsc{-srk}(\textnormal{Tern}(a, 0, 0)), \mQ \right)$, completing the proof of the lemma.
\end{proof}

\subsection{Proofs for Label Generation}
\label{sec:app-label-generation}

In this section, we give the two deferred proofs from Section \ref{subsec:2-mixtures-slr}.

\begin{proof}[Proof of Lemma \ref{lem:imbalanced-planted-label}]
This lemma follows from a similar argument to Lemma \ref{lem:planted-label}. As in Lemma \ref{lem:planted-label}, the given conditions on $C, \gamma, \mu'$ and $N$ imply that
$$2\left( \frac{\gamma \cdot y'}{\mu'(1 + \gamma^2)} \right)^2 \le 1$$
and thus $X'$ is well-defined almost surely. First observe that if $Z = \mu'' \cdot u + G'$ where $G' \sim \mN(0, I_d)$ then
$$X' = \frac{a\gamma \cdot y'}{1 + \gamma^2} \cdot u + \frac{\gamma \cdot y'}{\mu'(1 + \gamma^2)} \cdot G' + \frac{1}{\sqrt{2}} \cdot \sqrt{1 - 2\left( \frac{\gamma \cdot y'}{\mu'(1 + \gamma^2)} \right)^2} \cdot G + \frac{1}{\sqrt{2}} \cdot W$$
where $a = \mu''/\mu'$. Thus by the same argument as in Lemma \ref{lem:planted-label}, we have that
$$\mL(X' | y') = \mN\left( \frac{a\gamma \cdot y}{1 + \gamma^2} \cdot u, \, I_d - \frac{\gamma^2}{1 + \gamma^2} \cdot uu^\top \right)$$
Now note that by the conditioning property of multivariate Gaussians, we have that
$$\mL(X|y) = \mN\left(\Sigma_{Xy}\Sigma_{yy}^{-1} \cdot y, \, \Sigma_{XX} - \Sigma_{Xy} \Sigma_{yy}^{-1} \Sigma_{yX} \right)$$
It is easily verified that
$$\Sigma_{Xy}\Sigma_{yy}^{-1} = \frac{a\gamma}{1 + \gamma^2} \cdot u \quad \text{and} \quad \Sigma_{XX} - \Sigma_{Xy} \Sigma_{yy}^{-1} \Sigma_{yX} = I_d - \frac{\gamma^2}{1 + \gamma^2} \cdot uu^\top$$
and thus $\mL(X|y)$ and $\mL(X'|y')$ are equidistributed. Since $y \sim \mN(0, 1 + \gamma^2)$, it follows by the same application of the conditioning property in Fact \ref{tvfacts} as in Lemma \ref{lem:planted-label} implies that
$$\TV\left( \mL(X, y), \mL(X', y') \right) \le \TV\left( \mL(y), \mL(y') \right) = O\left( N^{-C^2/2} \right)$$
which completes the proof of the lemma.
\end{proof}

\begin{proof}[Proof of Lemma \ref{lem:unplanted-label}]
This lemma follows from a similar argument to Lemma \ref{lem:planted-label}. As in Lemmas \ref{lem:planted-label} and \ref{lem:imbalanced-planted-label}, the given conditions imply that $X'$ is well-defined almost surely. Conditioned on $y'$, it holds that $Z, G$ and $W$ are independent. Therefore the three terms in the definition of $X'$ are independent and distributed as 
\begin{align*}
&\frac{\gamma \cdot y'}{\mu'(1 + \gamma^2)} \cdot Z \sim \mN\left(0, \, \left( \frac{\gamma \cdot y'}{\mu'(1 + \gamma^2)} \right)^2 \cdot I_d  \right), \\
&\frac{1}{\sqrt{2}} \cdot \sqrt{1 - 2\left( \frac{\gamma \cdot y'}{\mu'(1 + \gamma^2)} \right)^2} \cdot G \sim \mN\left(0, \, \frac{1}{2} \cdot I_d - \left( \frac{\gamma \cdot y'}{\mu'(1 + \gamma^2)} \right)^2 \cdot I_d \right) \quad \text{and} \\
&\frac{1}{\sqrt{2}} \cdot W \sim \mN\left(0, \, \frac{1}{2} \cdot I_d \right)
\end{align*}
conditioned on $y'$. It follows that $X' | y' \sim \mN(0, I_d)$ and thus $X'$ is independent of $y'$. Now let $X \in \mathbb{R}^d$ and $y \in \mathbb{R}$ be such that $X \sim \mN(0, I_d)$ and $y \sim \mN(0, 1 + \gamma^2)$ are independent. The same application of the conditioning property in Fact \ref{tvfacts} as in Lemmas \ref{lem:planted-label} and \ref{lem:imbalanced-planted-label} now completes the proof of the lemma.
\end{proof}

\section{Deferred Proofs from Part \ref{part:lower-bounds}}
\label{sec:appendix-3}

\subsection{Proofs from Secret Leakage and the $\pr{pc}_\rho$ Conjecture}
\label{sec:appendix-4}

In this section, we present the deferred proof of Lemma \ref{l:avgCorrLargeSets} from Section \ref{sec:2-secret-leakage}. The proof of this lemma is similar to the proof of Lemma 5.2 in \cite{feldman2013statistical}.

\begin{proof}[Proof of Lemma \ref{l:avgCorrLargeSets}]
The proof is almost identical to Lemma 5.2 in \cite{feldman2013statistical} and we give a sketch here. Lemma~\ref{l:avgCorr} implies that $\sum_{T\in A}\big| \la \Dh_S, \Dh_T\ra_D \big| \leq \sum_{T\in A} 2^{|S\cap T|} k^2 / n^2$. If the only constraint on $A$ is its cardinality, then the maximum value for the RHS is obtained by adding $S$ to $A$, next $\{T:|T\cap S|=k-1\}$, and so forth with decreasing size of $|T\cap S|$, and we assume that $A$ is defined in this manner. Letting $T_\lambda = \{T: |T\cap S|=\lambda\}|$, set $\lambda_0 = \min\{\lambda: T_\lambda\neq \varnothing\}$ so that $T_\lambda\subseteq A$ for $\lambda>\lambda_0$. We bound the ratio
$$
\frac{|T_j|}{|T_{j+1}|} = \frac{{k\choose j}\big(\frac nk\big)^{k-j}}{{k\choose j+1}\big(\frac nk\big)^{k-j-1}}\geq \frac{jn}{k^2}=j n^{2\delta}\quad \text{hence}\quad |T_j|\leq \frac{|T_0|}{(j-1)! n^{2\delta j}}\leq \frac{|\cS|}{(j-1)! n^{2\delta j}}\,.
$$ 
Now
$$
|A|\leq \sum_{j\geq \lambda_0} |T_j| \leq |\cS|n^{-2\delta \lambda_0}\sum_{j\geq \lambda_0} \frac1{(j-1)!n^{2\delta(j-\lambda_0)}}\leq 2 |\cS|n^{-2\delta \lambda_0}
$$ for $n$ greater than some constant. Thus if $|A|\geq 2|\cS|/ n^{2\ell \delta}$, we must conclude that $\ell \geq \lambda_0$. We bound the quantity $\sum_{T\in A} 2^{|S\cap T|} \leq \sum_{j=\lambda_0}^k 2^j|T_j\cap A|\leq 2^{\lambda_0}|T_{\lambda_0}\cap A|+\sum_{j=\lambda_0+1}^k 2^j|T_j|\leq 2^{\lambda_0}|A| + 2^{\lambda_0+2}|T_{\lambda_0+1}|\leq 2^{\lambda_0+3}|A|\leq 2^{\ell +3}|A|$. Here we used that $|T_{j+1}|\leq |T_j| n^{-2\delta}$ to bound by a geometric series and also that $T_{\lambda_0+1}\subseteq A$. Rearranging and combining with the inequality at the start of the proof concludes the argument.
\end{proof}

\subsection{Proofs for Reductions and Computational Lower Bounds}
\label{subsec:appendix-3-part-3}

In this section, we present a number of deferred proofs from Part \ref{part:lower-bounds}. The majority of these proofs are similar to other proofs presented in the main body of the paper.

\begin{proof}[Proof of Theorem \ref{thm:rslr-lb}]
To prove this theorem, we will to show that Theorem \ref{thm:slr-reduction} implies that $k\pr{-bpds-to-mslr}$ applied with $r > 2$ fills out all of the possible growth rates specified by the computational lower bound $n = \tilde{o}(k^2 \epsilon^2/\tau^4)$ and the other conditions in the theorem statement. As discussed above, it suffices to reduce in total variation to $\pr{mslr}(n, k, d, \tau, 1/r)$ where $1/r \le \epsilon$.

Fix a constant pair of probabilities $0 < q < p \le 1$ and any sequence of parameters $(n, k, d, \tau, \epsilon)$ all of which are implicitly functions of $n$ such that $(n, \epsilon^{-1})$ satisfies $\pr{(t)}$ and $(n, k, d, \tau, \epsilon)$ satisfy the conditions
$$n \le c \cdot \frac{k^2 \epsilon^2}{w^2 \cdot \tau^4 \cdot (\log n)^{4+2c'}}, \quad wk \le n^{1/6} \quad \text{and} \quad w k^2 \le d$$
for sufficiently large $n$, an arbitrarily slow-growing function $w = w(n) \to \infty$ at least satisfying that $w(n) = n^{o(1)}$, a sufficiently small constant $c > 0$ and a sufficiently large constant $c' > 0$. The rest of this proof will follow that of Theorem \ref{thm:rsme-lb} very closely. In order to fulfill the criteria in Condition \ref{cond:lb}, we specify $M, N, k_M, k_N$ and $n'$ exactly as in Theorem \ref{thm:rsme-lb}. As in Theorem \ref{thm:rsme-lb}, we have the inequalities
$$n' \le w^{-2} r^{2t} = O\left( \frac{r^{2t}}{n} \cdot \frac{k^2 \epsilon^2}{\tau^4 \cdot (\log n)^{2+2c'}} \right)$$
$$\tau \le \frac{c^{1/4} \epsilon^{1/2} k^{1/2}}{n^{1/4} (\log n)^{(2 + c')/2}} = \Theta \left( \frac{r^{t/2}}{n^{1/4}} \cdot \frac{k_M^{1/2}}{\sqrt{r^{t + 1} (\log n)^{2+c'}}} \right)$$
Furthermore, we also have that
$$\tau^2 \le \frac{c^{1/2} \cdot k}{wn^{1/2} \cdot (\log n)^{2+c'}} = O\left( \frac{r^t}{n} \cdot \frac{k_N k_M}{N \log (MN)} \right)$$
As long as $\sqrt{n} = \tilde{\Theta}(r^t)$ then: (2.1) the inequality above on $n'$ would imply that $(n', k, d, \tau, \epsilon)$ is in the desired hard regime; (2.2) $n$ and $n'$ have the same growth rate since $w = n^{o(1)}$; and (2.3) $n \gg M^3$, $d \ge M$ and taking $c'$ large enough would imply that $\tau$ satisfies the bounds needed to apply Theorem \ref{thm:slr-reduction} to yield the desired reduction. By Lemma \ref{lem:propT}, there is an infinite subsequence of the input parameters such that $\sqrt{n} = \tilde{\Theta}(r^t)$, which concludes the proof as in Theorem \ref{thm:rsme-lb}.
\end{proof}

\begin{proof}[Proof of Lemma \ref{lem:ghpm-test}]
First suppose that $M \sim \pr{ghpm}_D(n, r, C, D, \gamma)$ where $C$ and $D$ are each sequences of $r$ disjoint sets of size $K$. Since the $M_{ij}$ are independent for $1 \le i, j \le n$, we now have that
\begin{align*}
\bE[s_C(M)] &= \sum_{i, j = 1}^n \bE\left[M_{ij}^2 - 1\right] = rK^2 \cdot \gamma^2 + \frac{rK^2}{r - 1} \cdot \gamma^2 \\
\text{Var}\left[ s_C(M) \right] &= \sum_{i, j = 1}^n \text{Var}\left[M_{ij}^2 - 1\right] = rK^2 \cdot 4\gamma^2 + \frac{rK^2}{(r - 1)^3} \cdot \gamma^2 + 2n^2
\end{align*}
Here, we have used the following facts. If $X \sim \mN(0, 1)$, then
\begin{align*}
&\bE[(\gamma + X)^2 - 1] = \gamma^2, \quad \bE\left[\left(\frac{\gamma}{r - 1} + X\right)^2 - 1\right] = \frac{\gamma^2}{(r - 1)^2} \\
&\text{Var}[X^2 - 1] = 2, \quad \text{Var}[(\gamma + X)^2 - 1] = 4\gamma^2 + 2, \quad \text{Var}\left[\left(\frac{\gamma}{r - 1} + X\right)^2 - 1\right] = \frac{\gamma^2}{(r - 1)^4} + 2
\end{align*}
Note that $s_C(M)$ is invariant to permuting the rows and columns of $M$ and thus $s_C(M)$ is equidistirbuted under $M \sim \pr{ghpm}_D(n, r, C, D, \gamma)$ and $M \sim \pr{ghpm}_D(n, r, K, \gamma)$. Now Chebyshev's inequality implies the desired lower bound on $s_C(M)$ in (1) holds with probability $1 - o_n(1)$. Now observe that
$$s_I(M) \ge \sum_{h = 1}^r \sum_{i \in C_h} \sum_{j \in D_h} M_{ij} = Y$$
holds almost surely by definition when $M \sim \pr{ghpm}_D(n, r, C, D, \gamma)$. Note that $Y \sim \mN(rK^2 \gamma, rK^2)$ conditioned on $C$ and $D$ and therefore it holds that $Y \ge rK^2 \gamma - wr^{1/2} K$ with probability $1 - o_n(1)$. The second lower bound in (1) now follows since $s_I(M)$ is equidistirbuted under $M \sim \pr{ghpm}_D(n, r, C, D, \gamma)$ and $M \sim \pr{ghpm}_D(n, r, K, \gamma)$.

Now suppose that $M \sim \mN(0, 1)^{\otimes n \times n}$. In this case, $s_C(M) + n^2$ is distributed as $\chi^2(n^2)$ and the first upper bound in (2) holds by Chebyshev's inequality and the fact that $\chi^2(n^2)$ has variance $2n^2$. Now note
$$Y(C, D) = \sum_{h = 1}^r \sum_{i \in C_h} \sum_{j \in D_h} M_{ij} \sim \mN(0, rK^2)$$
Standard gaussian tail bounds imply that
\begin{align*}
\bP\left[ Y(C, D) > 2r K^{3/2}w \sqrt{\left(\log n + \log r \right)} \right] &\le \frac{1}{\sqrt{2\pi}} \exp\left( -\frac{1}{2rK^2} \left( 2r K^{3/2} w \sqrt{\left(\log n + \log r \right)} \right)^2 \right) \\
&\le (nr)^{-2rKw^2}
\end{align*}
A crude upper bound on the number of pairs $(C, D)$ is
$$\left( \binom{n}{rK} r^{rK} \right)^2 = o\left( (nr)^{2rK} \right)$$
and therefore a union bound implies that $s_I(M) = \max_{C, D} Y(C, D) \le 2r K^{3/2}w \sqrt{\left(\log n + \log r \right)}$ with probability $1 - o_n(1)$. This completes the proof of the lemma.
\end{proof}

\begin{proof}[Proof of Corollary \ref{cor:bhpm}]
Consider the following reduction $\mathcal{A}$ that adds a simple post-processing step to $k$\textsc{-pds-to-ghpm} as in Corollary \ref{thm:isbm-mod}. On input graph $G$ with $N$ vertices:
\begin{enumerate}
\item Form the graph $M_{\text{R}}$ by applying $k$\textsc{-pds-to-ghpm} to $G$ with parameters $N, r, k, E, \ell, n, s$ and $\mu$ where $\mu$ is given by
$$\mu = \frac{r^{t} \sqrt{r}}{(r - 1)} \cdot \Phi^{-1}\left( \frac{1}{2} + \frac{1}{2} \cdot \min\{P_0, 1 - P_0\}^{-1} \cdot \gamma \right)$$
and $\Phi^{-1}$ is the inverse of the standard normal CDF.
\item Let $G_1$ be the graph where each edge $(i, j)$ with is in $G_1$ if and only if $(M_{\text{R}})_{ij} \ge 0$. Now form $G_2$ as in Step 2 of Corollary \ref{thm:isbm-mod}, while restricting to edges between the two parts.
\end{enumerate}
This clearly runs in $\text{poly}(N)$ time and it suffices to establish its approximate Markov transition properties. Let $\mathcal{A}_1$ denote the first step with input $G$ and output $M_{\text{R}}$, and let $\mathcal{A}_2$ denote the second step with input $M_{\text{R}}$ and output $G_2$. Let $C$ and $D$ be two fixed sequences, each consisting of $r$ disjoint subsets of $[ksr^t]$ of size $kr^{t - 1}$. Let $P_1, P_2 \in (0, 1)$ be
$$P_{1} = \Phi\left( \frac{\mu(r - 1)}{r^t \sqrt{r}}\right) \quad \text{and} \quad P_{2} = \Phi\left( - \frac{\mu}{r^t \sqrt{r}} \right)$$
Note that by the definition of $\mu$, we have that $P_1 = \frac{1}{2} + \frac{1}{2} \cdot \min\{P_0, 1 - P_0\}^{-1} \cdot \gamma$. Now note that $\mathcal{A}_2$ applied to $M_{\text{R}} \sim \pr{ghpm}_D(ksr^t, r, C, D, \gamma)$ yields an instance of $\pr{bhpm}_D(ksr^t, r, C, D, \gamma)$ with the following modified edge probabilities:
\begin{enumerate}
\item The edge probabilities between vertices $C_h$ and $D_h$ for each $1 \le h \le r$ are still $P_0 + \gamma$.
\item The edge probabilities between $C_{h_1}$ and $D_{h_2}$ for each $h_1 \neq h_2$ are now
$$P_0 + 2\min\{P_0, 1 - P_0\} \cdot \left( \Phi\left( - \frac{\mu}{r^t \sqrt{r}} \right) - \frac{1}{2} \right) = P_0 + 2 \min\{P_0, 1 - P_0\} \cdot \left( P_2 - \frac{1}{2} \right)$$
\item All other edge probabilities are still $P_0$.
\end{enumerate}
We now apply a similar sequence of inequalities as in Corollary \ref{thm:isbm-mod}. For now assume that $P_0 \le 1/2$. Using the fact that all of the edge indicators of this model and the usual definition of $\pr{bhpm}$ are independent, the tensorization property in Fact \ref{tvfacts} and Lemma \ref{lem:bintv}, we now have that
\allowdisplaybreaks
\begin{align*}
&\TV\left( \mathcal{A}_2\left( \pr{ghpm}_D(ksr^t, r, C, D, \gamma) \right), \, \pr{bhpm}_D(ksr^t, r, C, D, \gamma) \right) \\
&\quad \quad \le \TV\left( \text{Bern}\left( P_0 - \frac{\gamma}{r - 1} \right)^{\otimes k^2r^{2t - 1}(r - 1)}, \, \text{Bern}\left( P_0 + 2 P_0 \cdot \left( P_2 - \frac{1}{2} \right) \right)^{\otimes k^2r^{2t - 1}(r - 1)} \right) \\
&\quad \quad \le \left| \frac{\gamma}{r - 1} + 2P_0 \cdot \left( P_2 - \frac{1}{2} \right) \right| \cdot \sqrt{\frac{k^2r^{2t - 1}(r - 1)}{2\left( P_0 - \frac{\gamma}{r - 1} \right) \left(1 - P_0 + \frac{\gamma}{r - 1} \right)}} \\
&\quad \quad \le \left| \frac{\gamma}{r - 1} + 2P_0 \cdot \left( P_2 - \frac{1}{2} \right) \right| \cdot O\left( kr^{t} \right)
\end{align*}
where the third inequality uses the fact that $P_0$ is bounded away from $0$ and $1$ and $\gamma = o(1)$. Now note that
$$\frac{\gamma}{r - 1} = \frac{2P_0}{r - 1} \cdot \left( \Phi\left( \frac{\mu(r - 1)}{r^t \sqrt{r}}\right) - \frac{1}{2} \right)$$
Using the standard Taylor approximation for $\Phi(x) - 1/2$ around zero when $x \in (-1, 1)$, we have
\begin{align*}
\left| \frac{\gamma}{r - 1} + 2P_0 \cdot \left( P_2 - \frac{1}{2} \right) \right| &= 2P_0 \cdot \left| \frac{1}{r - 1} \left( \Phi\left( \frac{\mu(r - 1)}{r^t \sqrt{r}}\right) - \frac{1}{2} \right) - \left( \Phi\left( - \frac{\mu}{r^t \sqrt{r}} \right) - \frac{1}{2} \right) \right| \\
&= O\left( \frac{\mu^3 \sqrt{r}}{r^{3t}} \right)
\end{align*}
Therefore we have that
$$\TV\left( \mathcal{A}_2\left( \pr{ghpm}_D(ksr^t, r, C, D, \gamma) \right), \, \pr{bhpm}_D(ksr^t, r, C, D, \gamma) \right) = O\left( \frac{k\mu^3 \sqrt{r}}{r^{2t}} \right)$$
A nearly identical argument considering the complement of the graph $G_1$ and replacing with $P_0$ with $1 - P_0$ establishes this bound in the case when $P_0 > 1/2$. Observe that $\mathcal{A}_2 \left( \mathcal{N}(0, 1)^{\otimes n \times n} \right) \sim \mG_B(n, n, P_0)$. Now consider applying Lemma \ref{lem:tvacc} to the steps $\mathcal{A}_1$ and $\mathcal{A}_2$ as in Corollary \ref{thm:isbm-mod}. It can be verified that the given bound on $\gamma$ yields the condition on $\mu$ needed to apply Theorem \ref{thm:ghpm} if $c > 0$ is sufficiently small. Thus $\epsilon_1$ is bounded by Theorem \ref{thm:ghpm} and $\epsilon_2$ is bounded by the argument above after averaging over $C$ and $D$ and applying the conditioning property of Fact \ref{tvfacts}. This application of Lemma \ref{lem:tvacc} therefore yields the desired two approximate Markov transition properties and completes the proof of the corollary.
\end{proof}

\begin{proof}[Proof of Theorem \ref{thm:semi-cr-lb}]
As discussed in the beginning of this section, it suffices to map to $\mG(n, P_0 - \mu_1)$ under $H_0$ and $\pr{tsi}(n, k, k_1, P_0, \mu_1, \mu_2, \mu_3)$ under $H_1$ where $\mu_3 = P_1 - P_0$ and $\mu_1, \mu_2 \ge 0$. Thus it suffices to show that the reduction $\mathcal{A}$ in Corollary \ref{cor:semi-cr-gen} fills out all of the possible growth rates specified by the computational lower bound $\frac{(P_1 - P_0)^2}{P_0(1 - P_0)} = \tilde{o}(n/k^2)$ and the other conditions in the theorem statement. Fix a constant pair of probabilities $0 < q < p \le 1$ and any sequence of parameters $(n, k, P_1, P_0)$ all of which are implicitly functions of $n$ such that
$$\frac{(P_1 - P_0)^2}{P_0(1 - P_0)} \le c \cdot \frac{n}{w^3 \cdot k^2 \log n} \quad \text{and} \quad \min\{P_0, 1 - P_0 \} = \Omega_n(1)$$
for sufficiently large $n$, sufficiently small constant $c > 0$ and an arbitrarily slow-growing increasing positive integer-valued function $w = w(n) \to \infty$ at least satisfying that $w(n) = n^{o(1)}$. As in the proof of Theorem \ref{thm:rsme-lb}, it suffices to specify:
\begin{enumerate}
\item a sequence $(N, k_N)$ such that the $k\pr{-pds}(N, k_N, p, q)$ is hard according to Conjecture \ref{conj:hard-conj}; and 
\item a sequence $(n', k', P_1, P_0, s, t, \mu)$ satisfying: (2.1) the parameters $(n', k', P_1, P_0)$ are in the regime of the desired computational lower bound for $\pr{semi-cr}$; (2.2) $(n', k')$ have the same growth rates as $(n, k)$; and (2.3) such that $\mG(n', P_0 - \mu_1)$ and $\pr{tsi}(n', k', k'/2, P_0, \mu_1, \mu_2, P_1 - P_0)$, where $k'$ is even and $\mu_1, \mu_2 \ge 0$, can be produced by $\mathcal{A}$ with input $k\pr{-pds}(N, k_N, p, q)$.
\end{enumerate}
We choose these parameters as follows:
\begin{itemize}
\item let $t$ be such that $3^t$ is the smallest power of 3 larger than $k/\sqrt{n}$ and let $s = \lceil 2n/3k \rceil$;
\item let $\mu \in (0, 1)$ be given by
$$\mu = 3^t \cdot \Phi^{-1} \left( \frac{1}{2} + \frac{1}{2} \cdot \min\{P_0, 1 - P_0 \}^{-1} (P_1 - P_0) \right)$$
\item now let
$$k_N = \left\lfloor \frac{1}{2}\left( 1 + \frac{p}{Q} \right)^{-1} w^{-2} \cdot \sqrt{n} \right\rfloor$$
where $Q = 1 - \sqrt{(1 - p)(1 - q)} + \mathbf{1}_{\{ p = 1\}} \left( \sqrt{q} - 1 \right)$; and
\item let $n' = 3k_N s \cdot \frac{3^t - 1}{2}$, let $k' = (3^t - 1)k_N$ and let $N = wk_N^2$.
\end{itemize}
Note that $3^t = \Theta(k/\sqrt{n})$, $s = \Theta(n/k)$ and $3^t k_N s \le \text{poly}(N)$. Note that this choice of $\mu$ implies that
$$P_1 = P_0 + 2\min\{P_0, 1 - P_0 \} \cdot \left( \Phi\left( \frac{\mu}{3^t} \right) - \frac{1}{2} \right)$$
which implies that the instance of $\pr{tsi}$ output by $\mathcal{A}$ has edge density $P_1$ on its $k'$-vertex the planted dense subgraph. It follows that
\begin{align*}
n' &\asymp 3^t k_N s \asymp \frac{k}{\sqrt{n}} \cdot \frac{n}{k} w^{-2} \cdot \sqrt{n} \asymp w^{-2} \cdot n \quad \text{and} \quad k' \asymp 3^t k_N \asymp w^{-2} k \\
\frac{(P_1 - P_0)^2}{P_0(1 - P_0)} &\le c \cdot \frac{n}{w^3 \cdot k^2 \log n} \lesssim c \cdot \frac{n'}{w \cdot (k')^2 \log n'} \\
m &\le 2\left( \frac{p}{Q} + 1 \right) wk_N^2 \le w^{-1} \sqrt{n} \cdot k_N \le 3^t k_N s \\
\mu &\lesssim 3^t \cdot (P_1 - P_0) \lesssim 3^t \cdot \frac{\sqrt{n}}{w^{3/2} \cdot k \sqrt{\log n'}} \le \frac{c}{w^{3/2} \sqrt{\log n'}}
\end{align*}
where the last bound above follows from the fact that $\Phi(x) - 1/2 \sim x$ if $|x| \to 0$. Here, $m$ is the smallest multiple of $k_N$ larger $\left( \frac{p}{Q} + 1 \right) N$. Now note that: (2.1) the third inequality above on $(P_1 - P_0)^2/P_0(1 - P_0)$ implies that $(n', k', P_1, P_0)$ is in the desired hard regime; (2.2) $(n, n')$ and $(k, k')$ have the same growth rates since $w = n^{o(1)}$; and (2.3) the last two bounds above imply that taking $c$ small enough yields the conditions needed to apply Corollary \ref{cor:semi-cr-gen} to yield the desired reduction. This completes the proof of the theorem.
\end{proof}

\begin{proof}[Proof of Lemma \ref{lem:truncgauss}]
The parameters $a, \mu_1, \mu_2$ for which these distributional statements are true are given by
\allowdisplaybreaks
\begin{align*}
a &= \Phi(\tau) - \Phi(-\tau) \\
\mu_1 &= \frac{1}{2} \left( (1 - \Phi(\tau - \mu)) - \Phi(-\tau - \mu) \right) = \frac{1}{2} \left( \Phi(\tau + \mu) - \Phi(\tau - \mu) \right) \\
\mu_2 &= \frac{1}{2} \left( \Phi(\tau) - \Phi(-\tau) \right) - \frac{1}{2} \left( \Phi(\tau + \mu) - \Phi(-\tau + \mu) \right) = \frac{1}{2} \left( 2 \cdot \Phi(\tau) - \Phi(\tau + \mu) - \Phi(\tau - \mu) \right)
\end{align*}
Now note that
$$\mu_1 = \frac{1}{2} \left( \Phi(\tau + \mu) - \Phi(\tau - \mu) \right) = \frac{1}{2\sqrt{2\pi}} \int_{\tau - \mu}^{\tau + \mu} e^{-t^2/2} dt = \Theta(\mu)$$
and is positive since $e^{-t^2/2}$ is bounded on $[\tau - \mu, \tau + \mu]$ as $\tau$ is constant and $\mu \to 0$. Furthermore, note that
\begin{align*}
\mu_2 &= \frac{1}{2} \left( 2 \cdot \Phi(\tau) - \Phi(\tau + \mu) - \Phi(\tau - \mu) \right) = \frac{1}{2\sqrt{2\pi}} \int_{\tau - \mu}^{\tau} e^{-t^2/2} dt - \frac{1}{2\sqrt{2\pi}} \int_{\tau}^{\tau + \mu} e^{-t^2/2} dt \\
&= \frac{1}{2\sqrt{2\pi}} \int_{\tau}^{\tau + \mu} \left( e^{-(t - \mu)^2/2} - e^{-t^2/2}\right) dt = \frac{1}{2\sqrt{2\pi}} \int_{\tau}^{\tau + \mu}  e^{-t^2/2} \left(e^{t\mu - \mu^2/2} - 1 \right) dt 
\end{align*}
Now note that as $\mu \to 0$ and for $t \in [\tau, \tau + \mu]$, it follows that $0 < e^{t\mu - \mu^2/2} - 1= \Theta(\mu)$. This implies that $0 < \mu_2 = \Theta(\mu^2)$, as claimed.
\end{proof}

\begin{proof}[Proof of Theorem \ref{thm:glsm-lb}]
To prove this theorem, we will to show Theorem \ref{lem:univlem} implies that $k\textsc{-bpds-to-glsm}$ fills out all of the possible growth rates specified by the computational lower bound $n = \tilde{o}\left(\tau_{\mU}^{-4}\right)$ and the other conditions in the theorem statement, as in the proof of Theorems \ref{thm:rsme-lb} and \ref{thm:uslr-lb}. Fix a constant pair of probabilities $0 < q < p \le 1$ and any sequence $(n, k, d, \mU)$ where $\mU = \left( \mD, \mQ, \{ \mP_{\nu} \}_{\nu \in \mathbb{R}} \right) \in \pr{uc}(n)$ all of which are implicitly functions of $n$ with
$$n \le \frac{c}{\tau_{\mU}^4 \cdot w^2 \cdot (\log n)^{2}} \quad \text{and} \quad w k^2 \le d$$
for sufficiently large $n$, an arbitrarily slow-growing function $w = w(n) \to \infty$ and a sufficiently small constant $c > 0$. Now consider specifying the parameters $M, N, k_M, k_N$ and $t$ exactly as in Theorem \ref{thm:uslr-lb}. Now note that under these parameter settings, we have that
$$\tau_{\mU} \le \frac{c^{1/4}}{n^{1/4} w^{1/2} \sqrt{\log n}} \le 2c^{1/4} \cdot \sqrt{\frac{k_N}{N \log N}}$$
Therefore $\tau_{\mU}$ satisfies the conditions needed to apply Theorem \ref{lem:univlem} for a sufficiently small $c > 0$. The other parameters $(n, k, d, \mU)$ and $(M, N, k_M, k_N, p, q)$ can also be verified to satisfy the conditions of this theorem. We now have that $k\pr{-bpds}(M, N, k_M, k_N, p, q)$ is hard according to Conjecture \ref{conj:hard-conj}, and that $\pr{glsm}(n, k, d, \mU)$ can be produced by the reduction $k\textsc{-bpds-to-glsm}$ applied to $\pr{bpds}(M, N, k_M, k_N, p, q)$. This verifies the criteria in Condition \ref{cond:lb} and, following the argument in Section \ref{subsec:2-tvreductions}, Lemma \ref{lem:3a} now implies the theorem.
\end{proof}

\end{document}